\documentclass[11pt,a4paper,twoside,openright]{book}
\usepackage[T1]{fontenc}
\usepackage[utf8]{inputenc}
\usepackage{psfrag}
\usepackage{graphicx}
\usepackage{bbold}
\usepackage{amssymb}
\usepackage{float}
\usepackage{geometry} 
\usepackage{subfigure}
\usepackage{amsmath}
\usepackage{physics}

\usepackage{geometry}

\usepackage{todonotes}
\usepackage{tikz}
\usetikzlibrary{arrows, decorations.markings}
\usetikzlibrary{decorations.pathreplacing}
\usepackage{algpseudocode,algorithm,algorithmicx}

\usepackage{amsthm}
\theoremstyle{definition}
\newtheorem{definition}{Definition}[section]
\theoremstyle{theorem}
\newtheorem{theorem}{Theorem}[section]
\theoremstyle{example}
\newtheorem{example}{Example}[section]
\theoremstyle{remark}
\newtheorem{remark}{Remark}[section]
\newcommand{\R}{\mathbb{R}}
\newcommand{\N}{\mathbb{N}}
\newcommand{\NN}{\mathcal{N}}
\newcommand{\x}{\mathbf{x}}
\newcommand{\ZZ}{\mathfrak{Z}}

\newcommand{\SSSS}{\mathbf{\Sigma}}
\newcommand{\XXX}{\mathfrak{X}}
\newcommand{\X}{\mathbf{X}}
\newcommand{\Y}{\mathbf{Y}}
\newcommand{\eg}{e.g.}
\newcommand{\ie}{\textit{i.e.}}
\newcommand{\uu}{\mathbf{u}}
\newcommand{\vv}{\mathbf{v}}
\newcommand{\rhog}{\mathbf{f}_{\NN(\mu, \SSSS)}}
\newcommand{\Fg}{\mathbf{F}_{\NN(\mu, \SSSS)}}
\newcommand{\UU}{\mathfrak{U}}
\newcommand{\U}{\mathbf{U}}

\newcommand{\C}{\mathbf{C}}
\newcommand{\CC}{\mathcal{C}}
\newcommand{\MM}{\mathcal{M}}
\newcommand{\1}{\mathbb{1}}

\newcommand{\F}{{F_{\text{T}}}}
\newcommand{\f}{{f_{\text{T}}}}
\newcommand{\E}{\mathbb{\mathbb{E}}}

\newcommand{\ii}{\mathbf{i}}
\newcommand{\SSd}{\SSSS_{(\text{d})}}

\usepackage{algpseudocode,algorithm,algorithmicx}
\usepackage{listings}
\title{Selected Methods for non-Gaussian Data Analysis}
\author{Krzysztof Domino}
\date{November 27, 2018}

\makeatletter
\let\runauthor\@author
\let\runtitle\@title
\makeatother

\usepackage{fancyhdr}
\begin{document}

\thispagestyle{empty}

\begin{flushleft}
    \vspace*{-36pt}
    \hrule
    \vspace{24pt}
    \huge{\bf \runtitle}
    \vspace{24pt}
    \hrule
    \vspace{24pt}
    \LARGE{\bf \runauthor}
    \vspace{24pt}

%    \hrule
%    \vspace{12pt}

\end{flushleft}

%\maketitle
\chapter*{Foreword}
\addcontentsline{toc}{chapter}{Foreword}

The primary goal of computer engineering is the analysis of data. Such data are 
often large data sets distributed according to various distribution models. In 
this manuscript, we focus on the analysis of non-Gaussian distributed data. In 
the case of univariate data analysis, we discuss stochastic processes with 
auto-correlated increments and univariate distributions derived from specific 
stochastic processes, \ie~L\'evy and Tsallis distributions. Inspired by the 
fact that there is an increasing interest in financial technology and 
computing, we use the Ising model applicable on the quantum computation on the 
D-Wave machine to discuss the stochastic process of real financial data 
generation and analysis.

A crucial observation is that stochastic processes with auto-correlated 
increments may lead to non-Gaussian distributed data that are relatively common 
among real-life data. One can consider here computer network traffic data, data 
generated for the machine learning purposes, audio signals, multiple sensors 
data, weather data, various medical data, and cosmological data, or finally 
financial data where the failure of a Gaussian predictive model may make a 
hazard in economy and lead in some cases to the bankruptcy. The motivation for 
this manuscript comes from the fact that it is expected from computer 
scientists to develop algorithms to handle real-life data such as non-Gaussian 
distributed ones. While analysing non-Gaussian distributed data, there may 
appear a temptation to assume a priori Gaussian distribution and disregard 
extreme data values that do not fit the assumption. Such a naive approach may 
worsen the outcome of the analysis of non-Gaussian distributed data, especially 
if simultaneous extreme values of many marginals are possible. Such extreme 
events are not predicted by a simple Gaussian model but rather by a 
non-Gaussian multivariate frequency distribution.

What is essential, in-depth investigation of multivariate non-Gaussian 
distributions requires the copula approach. A copula is a component of 
multivariate distribution (especially non-Gaussian one) that models the mutual 
interdependence between marginals. There are many copula families characterised 
by various measures of the dependence between marginals. Importantly, one of 
those is `tail' dependencies that model the simultaneous appearance of extreme 
values in many marginals. Those extreme events may reflect a crisis given 
financial data, outliers in machine learning, or traffic congestion.

In this manuscript, we discuss copula-based data generation algorithms 
implemented in the Julia programming language that is an efficient open-source 
programming language suitable for scientific computation. The implementation is 
available on the GitHub repository, and the code is available for scientists 
for analysis and further development. We use a variety of copula families, 
especially those that can be applied in real-life data analysis (Gaussian, 
$t$-Student, Fr\'echet, Archimedean). Using those generators we perform 
experiments, demonstrating how different methods of features extraction or 
selection can distinguish between multivariate data distributed according to 
Gaussian or non-Gaussian copulas.

Having discussed non-Gaussian multivariate probabilistic models, we discuss 
higher order multivariate cumulants that are non-zero if the multivariate 
distribution is non-Gaussian. Nevertheless, the relation between those 
cumulants and copulas is not straight forward, but 
the $d$\textsuperscript{th} order multivariate cumulant encloses the natural 
measure of the $d$-variate cross-correlation between marginals. We discuss the 
application of those cumulants to extract information about non-Gaussian 
multivariate distributions, such that information about non-Gaussian copulas.  
The use of higher order multivariate cumulants in computer science is inspired 
by financial data analysis, especially by the safe investment portfolio 
evaluation. Apart from this, there are many other applications of higher order 
multivariate cumulants in data engineering, especially in: signal processing, 
non-linear system identification, blind sources separation, and direction 
finding algorithms of multi-source signals. 

Another promising computer science discipline, where higher order multivariate 
cumulants are used, involves analysis of data obtained from hyper-spectral 
imaging. In this book, we discuss the small target detection scenario where the 
analysis of the non-Gaussian distribution of features is beneficial. We show on 
the real-life data example, from a forensic analysis, a need for non-Gaussian 
algorithms using copulas and higher order multivariate cumulants. Given those, 
we evolve algorithms based on higher order cumulants applicable for features 
selection and features extraction.  We show by experiments the application of 
those methods in detecting subsets of marginals with non-Gaussian copulas, 
including copulas with `tail' dependencies reflecting the appearance of 
simultaneous high values in many marginals being extreme events.  For further 
real-life examples, we discuss through the manuscript applications of mentioned 
methods to analyse real-life biomedical data as well.
\vspace{12pt}

%\section*{Keywords}
%
%non-Gaussian distributions, copulas, higher order  
%multivariate cumulants, tensor analysis, features selection, information 
%extraction, financial data analysis.

\tableofcontents

\addcontentsline{toc}{chapter}{List of symbols}
\begin{table}[h]
	\renewcommand{\arraystretch}{1.2}
	\centering
	\begin{tabular}{lp{0.44\textwidth}}
		\textbf{Symbol} & \textbf{Description/explanation}  \\ \hline 
		$\XXX$ & univariate random variable \\ \hline 
		$X = [x_1, \ldots, x_t]^{\intercal}$ & vector of its realisations \\ 
		\hline $\E(\XXX), \E(\XXX^2), \ldots$  &
		expectation value operators\\ \hline 
		$f(x), F(x)$ & univariate PDF and CDF functions \\ \hline
		$\mathcal{N}(\mu, \sigma^2)$ & normal univariate distribution with mean 
		$\mu$ and variance $\sigma^2$ \\ \hline
		$\text{Uniform}([0,1])$ & uniform univariate distribution on segment 
		$[0,1]$ \\ \hline
		$(1:n)$ &  a vector $[1,2,\ldots, n]$ \\ \hline
		$\XXX^{(n)}$, $\XXX_i$ & $n$-variate random vector and its 
		the $i$\textsuperscript{th} marginal \\ \hline 
		$\mathbf{X} \in \mathbb{R}^{t \times n}$ & matrix of $t$
		realisations of $n$-variate random vector, with elements $x_{j,i}$\\ 
		\hline $\mathbf{x}_j = 
		[x_{j,1}, \ldots, x_{j,n}]$ & the single $j$\textsuperscript{th} 
		realisation of $n$-variate random 
		vector\\ \hline
		$\mathfrak{Z}_i, \mathfrak{Z}^{(n)}_i$ & the $i$\textsuperscript{th} 
		increment of an univariate or a multivariate 
		random variable. \\ 
		\hline 
		$\mathbf{f}(\x), \mathbf{F}(\x)$ & multivariate PDF and CDF functions\\ 
		\hline 
		$\mathfrak{U}^{(n)}$ & $n$-variate random vector with all marginals 
		uniformly distributed on $[0,1]$ segment\\ \hline
		$\mathbf{c}(\uu), \C(\uu)$ & copula density and copula function \\ 
		\hline  
		$\mathbf{A} \in \R^{n_1 \times n_2}$ & matrix with elements $a_{i_1, 
			i_2}$  \\ \hline$\SSSS \in \R^{[n,2]}$ & covariance matrix with 
		elements $s_{i_1, i_2}$  \\ \hline
		$\mathcal{N}(\mu, \SSSS)$ & normal multivariate distribution 
		parametrised by the mean vector $\mu$ and the covariance matrix $\SSSS$ 
		\\ \hline $\mathcal{T} \in \R^{n_1 \times \cdots \times n_d}$ & $d$ 
		mode tensor of size
		$n_1 \times \ldots \times n_d$, with elements $t_{i_1, \ldots, i_d}$ \\ 
		\hline
		$\mathcal{T} \in \R^{[n,d]}$ &
		$d$ mode super-symmetric tensor of size $n \times \ldots \times n$, 
		with elements 
		$t_{i_1, \ldots, i_d}$\\ 
		\hline
		$\CC_d \in \mathbb{R}^{[n, d]}$ $\MM_d \in \mathbb{R}^{[n, d]}$ & 
		$d$\textsuperscript{th} cumulant, moment
		tensor with elements $c_{i_1, \ldots, i_d}$, $m_{i_1, \ldots, i_d}$\\ 
	\end{tabular}
	\caption{Symbols used in the book.}\label{tab::symbols}
\end{table}

\chapter{Introduction}

The basic goal of computer engineering is the analysis of data. Such 
data are often large data sets distributed according to various distribution 
models. In this manuscript we focus on the analysis of non-Gaussian distributed 
data. 
To show that such data are rather common among real-life data, we can mention 
data 
with auto-correlated increments that may lead to 
non-Gaussian distributions both in an univariate (see 
Chapter~\ref{ch::univ}) and in a multivariate 
data case. For comparison with multivariate Gaussian models see 
Chapter~\ref{ch::mvn}. Deep investigation of 
non-Gaussian multivariate distributions requires the copula 
approach~\cite{nelsen1999introduction}, see 
Chapter~\ref{chap::cops}. A copula is an component of 
 multivariate distribution (especially non-Gaussian one) that models the mutual 
 interdependence between marginals. To extract probabilistic information 
about non-Gaussian multivariate distributions we use multivariate higher order 
cumulants, see Chapter~\ref{ch::hostats} that are non-zero if data are 
non-Gaussian 
distributed~\cite{kendall1946advanced,lukacs1970characteristics}. Having 
introduced higher order multivariate cumulants we use them in 
Chapter~\ref{cha::cumml} to 
discuss and develop some machine learning algorithms that can detect 
non-Gaussian 
features. For the programming implementation we use the Julia 
programming language \cite{bezanson2012julia, bezanson2014julia}, that is 
modern, efficient, open source and high level programming language suitable for 
scientific computations.

The 
motivation for this book comes from the fact that there are many types 
of non-Gaussian distributed real life data, and it is expected from computer 
scientists to develop algorithms to analyse such data. While handling 
non-Gaussian distributed data, 
there may appear a temptation to assume a priory Gaussian distribution 
 and disregard extreme data, that do not fit such Gaussian distribution and are
 out of priory  assumed model. Gaussian models are discussed in 
 Chapter~\ref{ch::mvn}. Such naive approach may worsen the outcome of 
the analysis of non-Gaussian distributed data. For a meaningful example 
consider financial data. Here simultaneously extreme 
values of many marginals, that are not expected by a simple Gaussian model, may 
appear especially during a crisis~\cite{cherubini2004copula}. Such extremes 
results being the failure of a 
Gaussian predictive model, may lead 
in some cases to the bankruptcy~\cite{vasconcelos2004guided}. 

To show wider application of non-Gaussian models, one can observe that computer 
network 
 traffic data do not follow the Gaussian 
 distribution~\cite{domanska2014}. This is due to the fact that computer 
 network 
 data may possess long range 
auto-correlations~\cite{norros1995use, li2009fractional, 
chandrasekaran2009survey} of increases, or may be modelled by the L\'evy 
process~\cite{terdik2009levy}. The same is true for financial 
data~\cite{vasconcelos2004guided}. To see more analogies between those two 
types of data, observe that to model both multivariate 
computer network 
traffic data~\cite{dong2015copula, 
	dong2017copula} and multivariate financial data~\cite{cherubini2004copula} 
	one uses non-Gaussian copulas. 
Beside those distinct examples, there are many other example of non-Gaussian 
distributed real-life data that 
are expected to be handled by computer scientists. One of them 
are audio signals~\cite{kidmose2000alpha, swami1997bibliography} (see also
\cite{porat1991direction, moulines1991second, cardoso1995asymptotic,	
	liang2009joint}), or multiple sensors data~\cite{iyengar2011decision}. 
	The other are weather data
\cite{cong2012interdependence, domino2014meteo, ozga2016snow}, various medical
data \cite{pougaza2008using, iyengar2010quantifying}, and cosmological data 
\cite{scherrer2009finance}. It is also worth to mention data generated for 
machine learning purposes using non-Gaussian 
copulas \cite{eban2013dynamic}. 

The probabilistic model for multivariate non-Gaussian distributed data,  
includes the copula that model interdependency between marginals this is 
discussed in Chapter~\ref{chap::cops}. There are 
many copula families, all imposing different measures of the 
interdependence between marginals. One of such measures are 
`tail' dependencies that model the simultaneous 
appearance of
extreme values in many marginals. Such extreme events may reflect: a crisis 
while analysing financial data, outliers data in machine learning, or a traffic 
congestion. In this book we discuss widely copulas and copula-based data 
generation algorithms. We 
have implemented those algorithms in the Julia 
programming language that is an efficient open source programming language 
suitable for scientific computation. The 
implementation is available in the GitHub repository~\cite{cop} and the code is 
available for scientists for analysis and further development.
In the implementation we use variety of copula families (such as 
Gaussian, $t$-Student, Archimedean, Fr\'echet and Marshal-Olkin), that can be 
found in real life data analysis. For the application of copulas in
decision 
making algorithms~\cite{sundaresan2011copula} or the machine 
learning, see copulas application in the 
remotely sensed 
images~\cite{mercier2008conditional, voisin2014supervised} analysis.

Unfortunately, higher order 
multivariate cumulants are hard to calculate, especially for large 
multivariate data sets. This is due to high computational 
complexity of cumulants calculation while using straight forward algorithms. 
See~\cite{de2012multivariate} 
for discussion and the multivariate cumulans calculation algorithm implemented 
in 
the \texttt{R} programming language. To overcome this problem 
in~\cite{domino2018efficient} a fast parallel algorithm for multivariate 
high order cumulants calculation was introduced. Here multivariate cumulants 
were represented in the tensor form structured in the form of blocks to take 
advantage of their super-symmetry and reduce to minimum redundant operations. 
In this book we discuss algorithms implemented in Julia programming 
language~\cite{cum}. Our 
approach makes the multivariate 
cumulants analysis more tractable and practical.

To understand the meaning of multivariate 
cumulants, one should observe that an 
univariate frequency distribution can be characterised by 
the series of scalar cumulants: the first is a mean, the second a variance, the 
third 
is proportional to an asymmetry, the fourth to a kurtosis, etc. The 
multivariate 
frequency distribution can be 
characterised by a series of cumulants as well, but those cumulants are no more 
scalars. If we consider multivariate data with $n$ 
marginals (or features using the machine learning terminology) the first 
multivariate cumulant is a mean vector of size $\R^n$, the second one 
is a covariance matrix of size $\R^{n \times n}$. Multivariate cumulants of 
order $d$, where $d \geq 3$ are 
called higher order cumulants, and as discussed further in this book, they 
carry 
information about the non-Gaussian multivariate distribution of data. Those 
higher order 
cumulants can be represented in the form of $d$-dimensional array of size 
$n^d$~\cite{domino2018efficient}, called the $d$-mode tensor. Hence if 
computing or storing them, a 
computer memory requirement 
and a processor time requirement rises rapidly with $d$, especially for high 
$n$ as in many real life data cases. This requires more sophisticated 
algorithms, as those introduced in~\cite{domino2018efficient} to reduce 
computational power and memory requirements.

 Having overcame 
computational difficulties in 
calculating higher order multivariate cumulants, we can move to their practical 
application in non-Gaussian data analysis. In general, the use of higher 
order multivariate cumulants in computer science is inspired on financial data 
analysis. See \cite{mittelhammer1996mathematical, sornette2000nonlinear, 
muzy2001multifractal, domino2016use}, 
where the practical use of such cumulants for 
safe investment portfolio determination is discussed. For the computer science 
application of those cumulants, consider for example signal processing 
\cite{nikias1993signal, swami1997bibliography, mendel1991tutorial, 
nikiashigher}, non-linear system 
identification \cite{tugnait1987identification, nam1994application}, blind 
sources separation \cite{cardoso1991super, cardoso1990eigen}, and direction 
finding 
algorithms of multi-source signals \cite{porat1991direction}. 
Another promising 
computer science discipline, where higher order multivariate cumulants are 
used, 
involves data analysis from hyper-spectral imaging. In \cite{geng2015joint} new 
higher order cumulant based algorithm of features selection for the 
small 
target detection of hyper-spectral data is discussed. This type of algorithms 
is 
discussed and developed 
in this manuscript. Finally, consider the analysis of computer 
network 
traffic data. There 
are some recent results concerning the use of higher order univariate 
cumulants~\cite{bai2012modeling, salagean2010anomaly, lam2009transportation} to 
analyse such data. 
For the discussion on the potential application of higher order multivariate 
cumulants, for such data scenario, see the introduction 
of~\cite{gibberd2017characterising}. Hence, one can hope for new results the 
use of higher order multivariate cumulants to model computer network 
traffic data.

Let us use the example of financial data, to discuss why data may be 
non-Gaussian 
distributed, despite the fact, that in many cases data 
records are 
sums on many increases and the Central Limit Theorem should hold leading to the 
Gaussian distribution of such data. In this case we consider 
financial data being a share prices traded on the stock exchange. 
The current share price is the price of last transaction. Such transactions 
may occur frequently. For major companies shares 
traded on major stock exchanges we have many transactions per second. The 
transaction price is derived from buy and 
sell orders on the stock exchange. However, as discussed in 
\cite{osborne1977stock}, the relation is not straight forward. In general, as 
time pass 
new orders are recorded, and new 
transactions are executed. Let $x_0 \in \R^{+}$ be an initial price and $x_1 
\in 
\R^{+}$ a price after the transaction is recorded. In this case the increment 
is:
\begin{equation}
z_1 = x_1 - x_0.
\end{equation}
The price after $k$ transactions is given by the following sum of increments,
\begin{equation}\label{eq::sumofincrease}
x = x_0 +\sum_{i = 1}^{k} z_i
\end{equation}
If increments $z_1, \ldots, z_k$ were independent, sampled from identical 
distribution with finite variance, than according to the Central Limit Theorem 
\cite{vandewalle1997coherent} $x$ would be a sample from the Gaussian 
distribution.
In practice for financial data, conditions of the  
Central Limit Theorem may not be fulfilled due to:
\begin{enumerate}
	\item long range auto-correlations of increments, 
	see \cite{mandelbrot1968fractional},
	\item increments from distribution without finite variance, 
	see~\cite{wim2003levy},
	\item varying distributions of increments \cite{engle1982autoregressive}.
\end{enumerate}
Hence, financial data may not be Gaussian distributed. 

As discussed in Chapter~\ref{ch::univ}, to demonstrate that it is 
true not only for financial data analysis, we can refer to some computer 
science data models inspired on financial data models. For this end consider 
that autoregressive models \cite{engle1982autoregressive, engle2001garch, 
engle2002new} 
rewarded by the Nobel 
Price in 2003 for financial data modelling. Analogically such models are used 
for computer network traffic~\cite{chandrasekaran2009survey}. 
Further~\cite{bak1997price} there are 
sophisticated 
financial data models concerning
games theory, or analogies 
with complex physical systems \cite{sornette2009stock}. Such analogies explains 
financial crisis using 
stochastic auto-correlated processes, fractals and multi-fractals, and 
non-Gaussian distributions. What more, there are hierarchical models of 
investors behaviour~\cite{sornette1998hierarchical}, and we have 
hierarchical models of the computer network traffic as 
well~\cite{altschaffel2014hierarchical}. In Section~\ref{sec::ising} we discuss 
the Ising model of financial data~\cite{gligor2001econophysics} concerning a 
hierarchical relation between investors and we propose the quantum computing 
scheme applicable on the D-Wave machine to implement this model. Such 
approach is based on the 
discussion of features and limitations of the Chimera graph, being a scheme of 
a processor unit of the D-Wave machine. By analogy such model may be 
applicable to analyse other types of computer science data.

After introducing stochastic models leading do non-Gaussian distributions, 
we can discuss now real life applications of multivariate probabilistic models 
of non-Gaussian distributed data. One should note that, there is a vital real 
life issue of multivariate financial data analysis. It is an evaluation of a 
risk of investment portfolio composed of many assets. In the portfolio 
management practice one either minimise risk, given expected revenue, or 
maximise revenue given acceptable risk \cite{best2000implementing}. This can be 
generalised to the computer science multi-parameter optimisation problem.
The classical financial engineering method \cite{best2000implementing} 
identifies 
risk with the 
variance of the portfolio value and assumes that financial data are 
multivariate Gaussian distributed. Next, it uses eigenvalue/eigenvector 
decomposition of the covariance matrix. However, as discussed before, 
financial data are often non-Gaussian distributed, hence the covariance matrix 
and a mean vector does not carry all information about their frequency 
distribution. It is why the classical financial engineering method fails 
to~anticipate cross-correlated extreme events \cite{sornette2000nonlinear}, 
especially during the crash, imposing the bankruptcy risk. 
 
Such failure of the Gaussian distribution in financial engineering, gives  
the motivation to investigate machine learning methods applicable for 
non-Gaussian distributed financial data at first and other types of data as 
well. In machine learning, there are 
many feature extraction and features selection methods, based on the 
multivariate Gaussian 
distribution assumption~\cite{duda2012pattern}. There exists a 
temptation 
to use such methods without testing if data are Gaussian or non-Gaussian 
distributed. 
Examples of Gaussian distribution based machine learning methods are Singular 
Value Decomposition (SVD) of 
the covariance 
matrix or the Principle Component Analysis (PCA) \cite{pearson1901liii, 
jolliffe2002principal}.

To discuss features selection, let us consider hyper-spectral camera
data. We can consider spectral channels of such 
camera as
$n$ marginals of a random vector and light intensities recorded at each pixel 
as $t$ realisations of such random vector. In practice $n  \sim 100$, and $t 
\sim 10^5$.
In the spectral 
approach ~\cite{benediktsson2015spectral}, we analyse multivariate statistics 
of such data, hence we ignore spatial information due to pixel's relative 
positions. Due to the presence of phenomena such as spectral mixing and feature 
redundancy \cite{ghamisi2017advances}, we are often interested in selecting the 
subset of features (spectral canals) that carries 
information relevant for further processing. There are classical methods of 
features selection, that use
only information stored in a covariance matrix, see for example the 
Maximum Ellipsoid Volume (MEV) \cite{sheffield1985selecting}.

Alternatively to the features selection scenario we 
may be interested in linear combinations of features that carries meaningful 
information. The 
classical method of such features extraction is the Principle 
Component Analysis (PCA) \cite{jolliffe2002principal}. This method uses the 
Gaussian 
distribution assumption of data as well. 
Obviously there are methods of features extraction or features selection that 
does not require the Gaussian distribution assumption of data, see 
introduction in~\cite{geng2015joint}. Those methods in general are either 
supervised or analyse only marginal frequency distributions of data. 
Supervised methods are dependent on the training samples what may lead to 
instability of the solution, therefore unsupervised methods may be more robust 
\cite{jia2016novel}. On the other hand, methods that analyses only marginal 
frequency distributions are not sensitive on 
copulas~\cite{nelsen1999introduction}, what may decrease their predictive 
power for non-Gaussian copula of data. Besides those methods, consider the 
Independent 
Component Analysis~\cite{stone2004independent} (ICA) 
that is useful for features extraction~\cite{hyvarinen2001independent} and 
selection~\cite{du2003band} of non-Gaussian 
distributed data. Unfortunately, ICA results are probabilistic and do not 
analyse 
standard correlations of data, what makes the ICA less robust.

There are few machine learning methods using higher order multivariate 
cumulants, despite the fact that such cumulants carry meaningful information 
about 
non-Gaussian joint frequency distribution of data. Hence, analogically to the 
financial 
engineering case, higher order multivariate cumulants should give better 
machine 
learning 
algorithms dedicated to non-Gaussian distributed data. In this book we 
discuss and improve some features selection and features extraction methods 
that 
use mentioned higher 
order multivariate cumulants. Furthermore we show by experiments the 
application of 
those methods in detecting subsets of marginals modelled by non-Gaussian 
copulas, 
including copulas with `tail' dependencies reflecting the appearance of 
simultaneous high values of many marginals~\ie~extreme events.

The implementation described in this book utilises 
tensors~\cite{kolda2009tensor} to store high 
order 
multivariate cumulants. We use 
tensor operations presented in the form of 
tensor networks to increase the clarity of the formulas. Observe that such 
tensor network 
presentation is inspired by the quantum mechanics, see 
\cite{sornette2000nonlinear} where the method of even 
order cumulants calculation has been developed by analogy to 
Feynman diagrams used in the
particle quantum physics. Such approach is used for financial data analysis, 
however priory Weibull class distribution of data is 
assumed. In this book we provide a general 
algorithm of any order multivariate cumulants calculation, regardless of the 
data 
distribution. Concluding the analogy between tensor networks of higher order 
cumulants and Feynman diagrams is worth to notice.

\chapter{Univariate data  models}\label{ch::univ}

In this chapter, we discuss univariate models of data employing stochastic 
processes. For this purpose, we assume that data are generated by the random 
process and we consider both random processes with independent or correlated 
increments.

\section{Random variable and its increments}\label{sec::stoch}

We start with the standard approach of independent increments and the 
Central Limit Theorem.

\subsection{Central Limit Theorem}\label{sec::clt}

Let us start with a simple random process. 
Suppose that we start with initial value $x_0 = 0$. Next in a unit time we 
increment 
it by $z_i = \pm \delta$. The sign of an increment is chosen each time 
randomly, 
hence all 
increments are independent. The probability of being in position  
$x$ after $k$ steps is binomial
\begin{equation} \label{eq::binomdist}
f_k(x) = \left(\frac{1}{2}\right)^k\left(\begin{array}{c}
k \\
\frac{\left(\frac{x}{\delta}\right)+k}{2}
\end{array}\right).
\end{equation}
Here $f_k(x)$ is the discrete function since 
$\frac{\left(\frac{x}{\delta}\right)+k}{2}$ 
must be an integer. One can show that $f_k(x)$ is normalised: $\sum_x f_k(x) = 
1$ and positive valued. The discrete function fulfilling those conditions is 
called in literature the 
probability mass function. If increments are independent, $f_k(x)$ for 
relatively high $k$ can be approximated by means of the continuous Gaussian PDF
(Probability Distribution Function), see Figure~\ref{fig::ras}. 
\begin{figure}[ht!]
    \subfigure[Independent 
    increments\label{fig::ras}]{\includegraphics{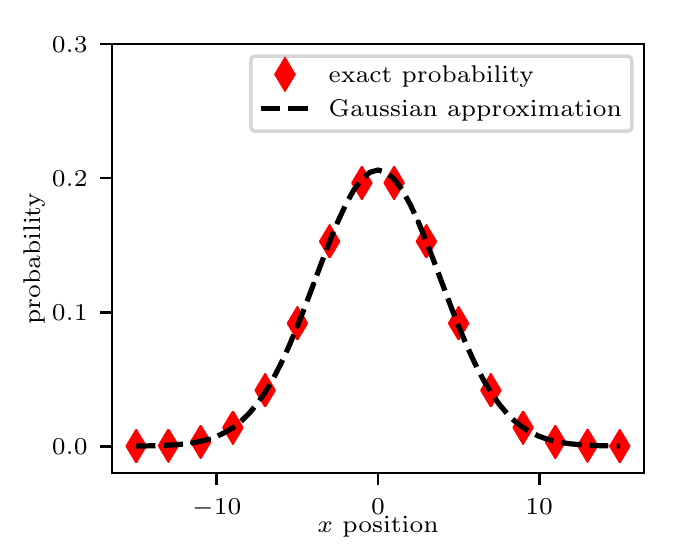}}
    \subfigure[Fully correlated increments\label{fig::cor}]	
    {\includegraphics{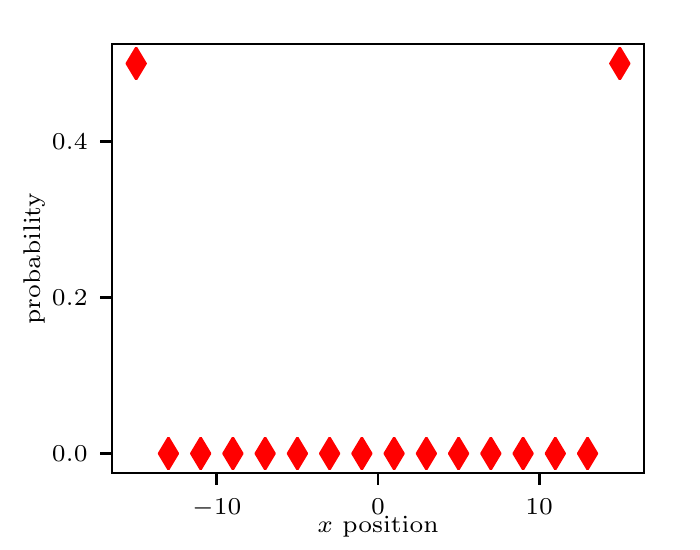}}
    \caption{Probabilities of simple random process after $k = 15$ steps, 
        $\delta = 1$}
    \label{fig::binfin}
\end{figure}

Switching to a random variable notation, $x$ is a single realisation of the 
following random variable
\begin{equation}\label{eq::Z}
\mathfrak{X}_k = \sum_{i=1}^k \ZZ_i,
\end{equation} where $\ZZ_1, \ldots, \ZZ_k$ are 
independent random variables (increments). The random variable $\ZZ_i$ has two 
possible 
realisations $z_i = \delta$  with probability $\frac{1}{2}$, or $z_i = -\delta$ 
with probability~$\frac{1}{2}$. 
It is easy to show that the expecting value of $\XXX_k$ is $\E(\XXX_k) = 0$. 
Following 
\cite{havlin2002diffusion}, the 
variance of $\XXX_k$ is
\begin{equation}\begin{split}\label{eq::independentrandvar}
\sigma^2(\XXX_k) &= \E\left(\left(\XXX_k-\E(\XXX_k)\right)^{2}\right) = \E 
\left(\sum_{i = 
    1}^{k}\ZZ_{i}\right)^{2} = \E \left(\sum_{i = 0}^{k}\ZZ_{i}^{2}\right) + 2 
    \E 
\left(\sum_{i > 
    j}^{k}\ZZ_{i} 
\ZZ_{j}\right) \\ &= \delta^2 k + 0 = \delta^2 k,
\end{split}
\end{equation}
and the standard deviation is $\sigma = \delta 
\sqrt{k}$, hence $\sigma \propto \sqrt{k}$.
\begin{remark}\label{rem::eoperator}
    For the random variable $\XXX$ with 
    realisations $x_1, \ldots, x_t$ occurring with probabilities $p_1, \ldots, 
    p_t$ 
    such that $p_i \geq 0 \wedge \sum p_i = 1$ the expecting value fulfils 
    \begin{equation}
    \E(\XXX) = \sum_{i=1}^t p_i x_i.
    \end{equation}
    While estimating the expected value we can assume that each realisation is 
    equally probable, hence the estimator would be
    \begin{equation}
    \E(\XXX) = \frac{1}{t}\sum_{i=1}^t x_i.
    \end{equation}
\end{remark}
According to Central Limit Theorem (Theorem~\ref{t::clt}) the 
distribution of $\XXX_k$ can be approximated by
the Gaussian one. The higher the $k$, the better the approximation. 

The Gaussian distribution has the following Probability Distribution 
Function (PDF)
\begin{equation}
f(x) = \frac{1}{\sqrt{2 \pi} \sigma 
}\exp\bigg(- \frac{x^{2}}{2\sigma^{2}}\bigg),  
\end{equation}
and the following Cumulative Distribution Function (CDF)
\begin{equation} F(x) = p(\XXX \leq x) = \int_{-\infty}^{x} 
\frac{1}{\sqrt{2 \pi} \sigma}\exp \left(- \frac{u^{2}}{2 \sigma^{2}} 
\right) d u. \end{equation}

Let us move to the formal definition of the Central Limit Theorem 
\cite{vandewalle1997coherent}.
\begin{theorem}\label{t::clt}
    Let $\ZZ_{1}, \ZZ_{2},..., \ZZ_{k}$ be univariate random variables which
    \begin{enumerate}
        \item are independent,
        \item are identically distributed,
        \item have defined finite variance $\delta^2 = \sigma^2(\ZZ_i)$ and 
        the mean value 
        $\mu_z$.
    \end{enumerate} 
The random variable
    \begin{equation}
    \XXX_k' = \frac{1}{\sqrt{k}}\sum_{i=1}^k \left(\ZZ_i - \mu_z \right)
    \end{equation}
    has the frequency distribution that converges to
    $\mathcal{N}(0, 
    \delta^2)$ as $k \rightarrow \infty$.
\end{theorem}

\begin{remark}
    One of the important features of a Gaussian distribution is a fact that it 
    is stable in the following sense. Suppose we have two independent random 
    variables $\ZZ_{1} \sim N(\mu_1, \sigma_1^2)$ and $\ZZ_{2} \sim 
    \mathcal{N}(\mu_2, \sigma_2^{2})$, their sum is Gaussian distributed
    as well,
    \begin{equation}
    \ZZ_{1} + \ZZ_{2} \sim \mathcal{N}(\mu_1 + \mu_2 , 
    \sigma_1^{2} 
    + \sigma_2^{2}).
    \end{equation}
    If $\ZZ_1$ and $\ZZ_2$ are Gaussian distributed increments their sum would 
    be Gaussian distributed as well.
\end{remark}
Let us now discuss a few examples where the Central Limit Theorem holds or does 
not hold.
\begin{example}
    Suppose that $\ZZ_1, \ZZ_2, \ldots \ZZ_k$ are identically distributed with 
    mean 
    $\mu$ and variance $\delta^2$, but not independent.
    The distribution of $\XXX$ will be Gaussian (but with changed parameters) 
    if 
    the auto-correlation between $\ZZ_i, 
    \ZZ_j$ is short range~\ie~the auto-correlation function $\text{acf}_{i,j}$ 
    falls rapidly with $|i 
    - j|$, see \cite{billingsley1995probability} for details.
    The 
    auto-correlation function \cite{dunn2017measurement} is defined by:
    \begin{equation}
    \text{acf}_{i,j} = 
    \E\left((\ZZ_i-\mu)(\ZZ_j-\mu)\right)
    \end{equation}
    \ie~without normalisation. For the auto-correlation coefficient 
    \cite{dunn2017measurement} one introduces 
    the following normalisation:
    \begin{equation}
    \text{acc}_{i,j} = \frac{\text{acf}_{i,j}}{\delta^2}
    \end{equation}
\end{example}

\begin{example}\label{ex::fulcorr}
    Suppose now, that $\ZZ_1, \ZZ_2, \ldots \ZZ_k$ are fully dependent. Assume 
    that 
    $z_1 = \pm \delta$ with equal probability, but:
    If $z_1 = \delta$ than $z_2 = \ldots = z_k = \delta$, and oppositely if 
    $z_1 = 
    -\delta$ than 
    $z_2 = \ldots = z_k = -\delta$. It follows that:
    \begin{equation}\label{eq::dependentrandvar}
    \sigma^2(\XXX) = \E\left(\XXX^{2}\right) = \E \left(\sum_{i = 
        1}^{k}\ZZ_{i}\right)^{2} = k^2 E \left(\ZZ_{1}\right)^{2} = k^2 
        \delta^2. 
    \end{equation}
    The standard deviation would be: $\sigma = k 
    \delta$~\ie~$\sigma \propto k$, and the PDF would be:
    \begin{equation}
    f_k(x) = \begin{cases} \frac{1}{2} \ \ \mbox{if } x =  k \delta \\
    \frac{1}{2} \ \ \mbox{if } x = - k \delta \\ 
    0 \ \ \mbox{elsewhere} \end{cases},
    \end{equation}
    that do not converge to a Gaussian distribution, see Figure~\ref{fig::cor}. 
    Importantly in this case we have $\sigma \propto k$ in contrary to the 
    independent 
    increments case, where $\sigma \propto \sqrt{k}$ see 
    Eq.~\eqref{eq::independentrandvar}. This 
    observation will be used further in an auto-correlation analysis.
\end{example}

\begin{example}\label{eg::garchexample}
    Consider the following Auto-Regressive 
    Conditional Heteroskedasticity (ARCH, GARCH) models 
    \cite{engle1982autoregressive,akgiray1989conditional, 
        bollerslev1987conditionally}. We have a random process defined by 
    increments 
    $\ZZ_1, \ldots, \ZZ_k$ with realisations $z_1, \ldots, z_k$, where
    \begin{equation}
    \ZZ_i \sim \mathcal{N}(0, \delta_{i}^2).
    \end{equation}
    The variance is a function of past realisations of increments, in the ARCH 
    model it is:
    \begin{equation}
    \delta_i^2 = \alpha_0+\sum_{j=1}^p\alpha_j \left(z_{(i-j)}\right)^2.
    \end{equation}
    In the Generalised ARCH model, called GARCH 
    \cite{bollerslev1986generalized, tim2010glossary} the variance is:
    \begin{equation}
    \delta_i^2 = \alpha_0+\sum_{j=1}^p\alpha_j 
    \left(z_{(i-j)}\right)^2+\sum_{j=1}^p\beta_j \delta_{(i-j)}
    \end{equation}
    where $\alpha_0 > 0, \alpha_{j}, \beta_j \geq 0$ and $p$ is the length 
    of the auto-correlation.
    Such models break conditions $1$ and $2$ of the Central Limit Theorem, 
    Theorem~\ref{t::clt}, hence a sum of increments
    \begin{equation}
    \XXX = \sum_{i = 1}^k \ZZ_i
    \end{equation} may not be 
    Gaussian 
    distributed. Such autoregressive models were developed to analyse financial 
    data \cite{engle1982autoregressive,akgiray1989conditional, 
        tim2010glossary} and rewarded by the Nobel prize in 2003. 
    Moreover, similar auto-regressive 
    models	appear in modelling computer network traffic 
    data~\cite{chandrasekaran2009survey, zhou2005network, 
        zhou2006traffic}. 
    Beside this non-linear GARCH/ARCH models are recently used to analyse 
    various data such 
    as weather data 
    \cite{taylor2004comparison}, EEG signal analysis \cite{mihandoost2015eeg}, 
    brain activity analysis~\cite{riccelli2017dynamical}, EMG 
    (electromyography) data analysis 
    \cite{amirmazlaghani2009emg}, speech signal 
    analysis~\cite{cohen2004modeling}, and
    sonar 
    imaging~\cite{mousazadeh2010anomaly}.
\end{example}

\subsection{Scaling approach}\label{sec::randomwalk}

In this subsection we concentrate on the specific ceases of auto-correlation of 
increments. For the need of this and the following subsections we introduce 
random variables as a sum of zero mean increments
\begin{equation}
\begin{split}
\XXX(k) &= \sum_{i=1}^k \ZZ_i, \\
\Delta \XXX(k, \Delta k) &= \sum_{i=k+1}^{k+ \Delta k} \ZZ_i.
\end{split}
\end{equation}
Suppose now that increments are independent and, for the sake of clarity of, 
all distributed according to $\ZZ_i 
\sim \NN(0,1)$. It is 
easy to show that 
\begin{equation}\label{eq::uncor}
\E(\XXX(k) \XXX(k+h)) = \E(\XXX(k) \XXX(k)) + \E(\XXX(k) \Delta \XXX(k, \Delta 
k)) = k.
\end{equation} 
Analogously the following auto-correlation is zero
\begin{equation}
\E(\Delta \XXX(k, \Delta k) \Delta \XXX(k+h, \Delta k)) = 0.
\end{equation}

Following \cite{vasconcelos2004guided}, consider now a random variable being a 
sum of not independent but zero mean increments $\ZZ_i$. Such random variable 
is assumed to fulfil
\begin{equation}\label{eq::colorscaling} 
\E\left(\XXX(k) \XXX(k+h)\right) = \frac{1}{2} \left(k^{2H} + 
{(k+h)}^{2H}  - 
h^{2H}\right), 
\end{equation}
where we have the scaling exponent $0 < H < 1$. Here we do not require 
explicitly $\ZZ_i$ to be Gaussian distributed. Given these, it can be shown 
from 
Eq.~\eqref{eq::colorscaling} that
\begin{equation}\label{eq::hex} 
\E\left(\XXX^{2}(k)\right) = k^{2H}, 
\end{equation}
and for $H = \frac{1}{2}$ we reproduce independent case in 
Eq.~\eqref{eq::uncor}. 

Following Eq.~\eqref{eq::colorscaling} and 
taking $h \gg \Delta k$ we measure the long-range auto-correlation. Using the 
expansion
\begin{equation}
(h\pm\Delta k)^{2H} = h^{2H} \pm 2H h^{2H -1} \Delta k + 
2H(2H-1) h^{2H -2} \frac{(\Delta k)^2}{2} \pm \ldots + \ldots
\end{equation} 
we obtain the following long range auto-correlation function of $\Delta \XXX$:
\begin{equation}\label{eq::corf}
\begin{split}
\E\left(\Delta \XXX(k, \Delta k) \Delta \XXX(k+h, \Delta k)\right) &=  
\E\big((\XXX(k+\Delta k) - \XXX(k))(\XXX(k+h+\Delta k) - \XXX(k+h))\big) \\ 
&= 
\frac{1}{2} 
\left((h+\Delta k)^{2H} + (h-\Delta k)^{2H} \right) - h^{2H} 
\\ &\approx H(2H-1)h^{2H-2} (\Delta k)^2.
\end{split}
\end{equation}
Now, from Eq.~\eqref{eq::colorscaling} we can show that
\begin{equation}\label{eq::hexi}
\E(\Delta \XXX^2(k, \Delta 
k)) = (\Delta k)^{2H}
\end{equation}
Importantly, Eq.~\eqref{eq::hexi} or its 
modification can be used to determine the scaling exponent. 
The auto-correlation coefficient (normalised 
auto-correlation function \cite{dunn2017measurement}) of $\Delta \XXX$ given 
the time lag $h \gg \Delta k$ is:
\begin{equation}
\text{acc}_{\Delta k, h} \approx \frac{H(2H-1)h^{2H-2} (\Delta k)^2}{(\Delta 
    k)^{2H}},
\end{equation}
which depends on
$\Delta k$. However, we are rather interested in the sign of the 
auto-correlation function for large $h$ that is:
\begin{equation}
\text{sign} \left( \text{acc}_{\Delta k, h}  \right) = \begin{cases} +1 & 
\text{if} \ 
\frac{1}{2} < H < 
1 \\ 0 & \text{if} \ H =  \frac{1}{2} \\ -1 & \text{if} \ H < \frac{1}{2}  
\end{cases}.
\end{equation}

Suppose now that despite the auto-correlation of elementary increments, $\Delta 
\XXX$ are Gaussian 
distributed, at least for some range of $\Delta k$. Then for $2 q \in 
\N$ we have
\begin{equation}\label{eq::gaussmoms}
\E\left(|\Delta \XXX(k, \Delta k)|^{2q}\right) = (\Delta k)^{2q H} 
\left(\frac{2^{q} 
    \Gamma\left(\frac{2q+1}{2}\right)}{\sqrt{\pi}}\right)
\end{equation}
that are moments of a Gaussian distribution with covariance $\sigma^2 = (\Delta 
t)^{2H}$. Hence:\begin{equation}
^{2q}\sqrt{\E\left(|\Delta \XXX(k, \Delta k))|^{2q}\right)} \propto (\Delta 
k)^{H},
\end{equation}
at lest for $2 q \in \N$. We have the same scaling exponent 
regardless which moment we analyse. 

On the other hand, if $\Delta \XXX$ is not Gaussian distributed due to 
auto-correlation or non-Gaussian distributions of elementary increments, we 
have the multi-fractal scaling~\cite{vasconcelos2004guided}
\begin{equation}
^{2q}\sqrt{\E\left(|\Delta \XXX(k, \Delta k)|^{2q}\right)} \propto (\Delta 
k)^{H(q)},
\end{equation}
where each $H(q)$ belongs to the spectrum of scaling exponents. 
In practice such multi-fractal analysis is extended as well on non-integer 
$2q$, negative $q$ and zero $q$.
For such multi-fractal scaling in financial data 
see~\cite{vasconcelos2004guided}, in computer 
network traffic data see~\cite{feldmann1998data, 
    nogueira2003modeling}, in image analysis 
see~\cite{blachowicz2016statistical} while in biomedical data 
see~\cite{marri2015multifractal}.
\begin{remark}\label{rem::tricacf}
    Consider that the auto-correlation coefficient (normalised 
    auto-correlation function \cite{dunn2017measurement}) 
    is a Pearson's correlation coefficient between $\Delta \XXX(k, \Delta k)$ 
    and $\Delta \XXX(k+h, \Delta k)$. As 
    discussed further in Chapter~\ref{chap::cops} and Section~\ref{sec::cumsg} 
    such Pearson's correlation carry a whole information about the 
    auto-correlation given a Gaussian distribution of $\Delta \XXX$.
    Otherwise, complementary to the multi-fractal approach, we may use 
    higher order cumulant's approach, as discussed in 
    \cite{manolakos2000systematic, 
        nikias1993signal}, see also Section~\ref{sec::cumsacf}. Exemplary
    tri-variate auto-correlation would be
    \begin{equation}
    \E\left(\Delta \XXX(k+h_1, \Delta k) \Delta \XXX(k+h_2, \Delta k) \Delta 
    \XXX(k, \Delta k)\right).
    \end{equation}
    Given $\Delta k =1$ we can use elementary increments. If the random process 
    is 
    stationary, 
    this tri-variate auto-correlation is $k$ independent. Given those it is 
    only 
    parametrised by two 
    lags parameters $h_1$ and $h_2$.
\end{remark}
\subsection{Practical applications of auto-correlation 
    analysis}\label{sec::autocorrelatedproces}

Let us state the following practical problem. We have a series of realisation 
of an univariate random variable $\mathfrak{Y}$, and want to analyse 
auto-correlation. Such analysis may be helpful to predict potential 
statistics of $\mathfrak{Y}$. Further real data are often a joint result of the
deterministic and random process, for example
\begin{equation}\label{eq::lang}
\Delta \mathfrak{Y}(k, \Delta k) =  f(k, \Delta k) + \Delta \XXX(k, \Delta k),
\end{equation}
were $f(k, \Delta k)$ is the deterministic term. To analyse the 
auto-correlation of $\Delta \XXX$, than can reveal a dynamics of 
interesting us system, one can use the Detrended 
Fluctuation Analysis (DFA), where the scaling of the divergence from the local  
trend is analysed~\cite{peng1994mosaic, 
    peng1995quantification}. In practice linear trend models are often used 
    here, 
but other non-linear models are also possible. 
In the multi-fractal DFA, one can determine
spectrum of scaling exponents $H(q)$ \cite{kantelhardt2002multifractal}. By 
taking $q=1$ we reproduce a mono-fractal approach.

The DFA is presented in Algorithm~\ref{alg::dfa}. Given data 
series $\R^t \ni Y = [y_1, \ldots y_t]^{\intercal}$, we divide it into $i$ 
non-overlapping 
sub-series of 
length $\tau = \lceil t /i \rceil$ each. In practice last two sub-series may 
overlap 
but this does not affect results.
For each sub-series we determine 
data model, 
usually by performing the linear regression from data in the sub-series. Values 
predicted by the model $M \in \R^{\tau}$.
Given the sub-series and the model we determine the measure of the noise 
divergence $\Delta_k$ - see line $9$ of Algorithm~\ref{alg::dfa}. Next we 
average 
this measure over sub--series. Finally we investigate its scaling versus the 
length of a sub-series.

\begin{algorithm}[t]
    \caption{Compute a Hurst Exponent using multifractal DFA.
        \label{alg::dfa}}
    \begin{algorithmic}[1]	
        \State \textbf{Input}: $Y \in \R^t$ - univariate time series, $q \in 
        \R$, a multi-fractal parameter, $N \ll t$ - Int, the algorithm 
        parameter. 	
        \State \textbf{Output:} $H(q) \in \R$ a Hurst exponent value.
        \Function{DFA}{$Y$, $q$}
        \For {$i \gets 1 \textrm{ to } N$}
        \State $\tau = \lceil \frac{t}{i} \rceil$
        \For {$k \gets 1 \textrm{ to } i$}
        \State $\R^{\tau} \ni Y' = [y_{(\tau-1) k + 1}, y_{(\tau-1) k + 2}, 
        \ldots,  
        y_{\left((\tau -1)k+ k\right)}]^{\intercal}$ 
        \State fit model $Y' \rightarrow M$ \Comment{model 
            predictions $M = [m_1, \ldots, 
            m_{\tau}]$}
        \State
        $\Delta_k = \frac{1}{\tau} \sum_{j = 1}^{\tau}\left(\left(y'_j - 
        m_j\right)^{2}\right)^q$ \Comment{$Y' = [y'_1, \ldots, y'_{\tau}]$}
        \EndFor
        \State	$\Delta_{(2q)}(\tau) = \frac{1}{i} \sum_{k=1}^i \Delta_k,$
        \EndFor
        \State find $H_q$ using linear regression of: 
        $\log(\Delta_{(2q)}(\tau)) 
        \propto 2q \cdot H(q) \log(\tau)$
        \State\Return $H_q$
        \EndFunction 
    \end{algorithmic}
\end{algorithm}

\begin{remark}
    Mention that the DFA is only one of the methods of determining the scaling 
    (Hurst) exponent 
    or exponents. One can 
    mention here an original 
    as 
    R/S method \cite{hurst1952nile}, wavelet methods \cite{jones1996wavelet, 
        simonsen1998determination, acharya2012application},
    multi-fractal wavelet methods \cite{riedi1999multifractal, 
    muzy1991wavelets}, 
    or the Detrended Cross Correlation Analysis 
    \cite{podobnik2008detrended} applicable to the auto-correlation analysis 
    in the multivariate data domain.
\end{remark}
Hurst exponent analysis has many applications. Let us mention here the analysis 
of financial 
data, since it reveals log periodic osculations that are present before the 
crash \cite{vasconcelos2004guided, vandewalle1997coherent, grech2008local,
    czarnecki2008comparison, domino2012use} or change in trend see 
\cite{domino2011use} and 
bibliography within. In~\cite{domino2016use}, the Hurst exponent was used 
to predict the crash on the Warsaw Stock Exchange, while applying higher order 
multivariate cumulants to determine safe investment portfolios 
for this crash, see Section~\ref{sec::als}.
Furthermore, the 
Hurst exponent can be used to analyse different types of auto-correlated data 
in signal analysis \cite{muzy1991wavelets}, biomedical data analysis 
\cite{acharya2012application, maner2006characterization}, computer network 
traffic analysis 
\cite{riedi1999multifractal, karagiannis2004long, chakraborty2004self}, or as 
initially introduced 
weather analysis \cite{hurst1952nile, bove2006complexity}. Those
analysis examples give evidence of long range 
auto-correlations of increments of many types of data. Hence such data may not 
be Gaussian distributed and more advanced methods are necessary for their 
statistical analysis. Such methods (copulas, higher order cumulants tensors and 
relative methods) will be discussed in further part of this 
manuscript.

\section{Probabilistic non-Gaussian models}\label{sec::nongausdist}

Having discussed stochastic processes suggesting non-Gaussian probabilistic 
models of data, we can discuss some of such non-Gaussian models
resulting form particular types of stochastic processes.

\subsection{L\'evy distribution}\label{ss::levy}

L\'evy distribution family contains stable distributions. If increments are 
L\'evy distributed the sum is L\'evy distributed as well. Given those the 
Gaussian distribution belong to the L\'evy distribution family as well.
\begin{definition}\label{def::sumlevy}
    Let $\ZZ_{1}, ..., \ZZ_{i}, ..., \ZZ_{k}$ be independent 
    random variables with identical frequency distributions $f'(z)$.
    The frequency distribution $f'(z)$ is stable if 
    \begin{equation}\forall_{k\geq 2} \  
    \exists_{a'_{k}, b'_{k}}: \ZZ_{1} + ...+ \ZZ_{k} = \mathfrak{X} \sim 
    f(x),
    \end{equation}
    where
    \begin{equation} f(x) = \frac{1}{a_k}f' \left( \frac{z - 
        b_k}{a_k} \right). \end{equation}
\end{definition}

\begin{example}
    If $f$ and $f'$ are zero mean Gaussian distribution
    \begin{equation} 
    f(x) = 
    \frac{1}{\sqrt{k}} f'\left( \frac{z}{\sqrt{k}} \right),
    \end{equation}
    then we have $a_k = \sqrt{k}$.
\end{example}

It was shown by Paul L\'evy \cite{levy1925calcul} that there exists a family 
containing all 
stable distributions, the L\'{e}vy distributions family. For 
simplicity we will discuss here a symmetric L\'{e}vy distributions with zero 
location parameter,
defined by a following characteristic function
\begin{equation} \varphi_{\alpha, \gamma}(v) = \exp\left(-|\gamma 
v|^{\alpha}\right), 
\end{equation}
where $0 < \alpha \leq 2 $ and $\gamma > 0$. 
Characteristic function is a Fourier transform of the PDF 
function, hence:
\begin{equation}\label{eq::charf}
{f_S}_{\alpha, \gamma}(x) = \frac{1}{2\pi} \int_{- \infty}^{\infty} 
\varphi_{\alpha, \gamma}(v) e ^{-ixv}dv, 
\end{equation}
here $S$ is a notation for the stable distribution. The analytical calculation 
of the 
integral in Eq.~\eqref{eq::charf} can be performed in a 
simple way only for $\alpha = 1, 2$ yielding
\begin{enumerate}
    \item Lorentz / Cauchy distribution
    ${f_S}_{\alpha = 1, \gamma}(x) = \frac{1}{\pi} 
    \frac{\gamma}{x^{2}+\gamma^{2}},$
    \item Gaussian distribution with standard deviation $\sigma = \gamma 
    \sqrt{2}$: \begin{equation}{f_S}_{\alpha = 2, \gamma}(x) = \frac{1}{\sqrt{4 
            \pi \gamma^2}} 
    e^{-\frac{x^{2}}{4 \gamma^2}}.\end{equation}
\end{enumerate}

\begin{remark}
    Given a zero mean symmetric L\'{e}vy distribution we have the following 
    scaling \cite{mandelbrot1960pareto}
    \begin{equation} 	
    {f_S}_{\alpha, \gamma}(x)= {f'_S}_{\alpha, \gamma'}(z), 
    \end{equation}
    where $\gamma = \gamma' k^{\frac{1}{\alpha}}$, and other notation is as in 
    Definition~\ref{def::sumlevy}.
\end{remark}
\begin{remark}
    Importantly it can be shown that for $\alpha < 2$ the L\'{e}vy distribution 
    does not have defined variance as well as higher moments and cumulants. 
    Hence 
    if 
    $\alpha < 2$ the sum of L\'{e}vy 
    distributed random variables with infinite variance gives a L\'{e}vy 
    distributed random variable, that is not a Gaussian 
    distribution. In the following two remarks we would take $\gamma = 1$ for 
    simplicity.
\end{remark} 
\begin{remark} For $0 < \alpha <2$, asymptotic behaviour of L\'evy PDF 
    ${f_S}_{\alpha}(x)$ can be approximated 
    by \cite{vasconcelos2004guided}
    \begin{equation} {f_S}_{\alpha}(x) \approx \frac{C_{\alpha}}{|x| 
        ^{1+\alpha}} \ \textbf{for} \ |x| \rightarrow \infty, \end{equation}
    where
    \begin{equation} C_{\alpha} = 
    \frac{\alpha}{\pi}\Gamma(1+\alpha)\sin\Big(\frac{\alpha\pi}{2}\Big), 
    \end{equation}
    and $\Gamma$ is the Euler Gamma function \cite{abramowitz1}.
\end{remark}
\begin{remark}
    There is a generalisation of the Central Limit Theorem for increments 
    modelled by the distribution with infinite variance \cite{voit2003random}. 
    For $0 < 
    \alpha <2$ the sum 
    of identically distributed random variables modelled by the distribution 
    with 
    a power law asymptotic behaviour
    \begin{equation}\label{eq::asymptoticlevy} {f_S}_{\alpha}(x) \propto 
    \frac{1}{|x| 
        ^{1+\alpha}} \ \textbf{for} \ |x| \rightarrow \infty,
    \end{equation}
    converges to the L\'{e}vy distribution.	
\end{remark}

\begin{remark}\label{rem::stabledist}
    The general `asymmetric' L\'evy distribution \cite{voit2013statistical} 
    characteristic function is given by
    \begin{equation}
    \varphi_{\alpha, \beta, \gamma, \delta}(v) = \exp\left(\mathrm{i} 
    t\delta-|\gamma v|^{\alpha}\left( 1 - \mathrm{i} \beta \text{sign}(v) \Phi 
    \right)\right),
    \end{equation}
    where 
    \begin{equation}
    \Phi = \begin{cases}\tan\left(\frac{\pi \alpha}{2} \right) &\text{ for } 
    \alpha \neq 1 \\ 
    -\frac{2}{\pi} \log(|v|) &\text{ otherwise}\end{cases},
    \end{equation}
    and $\mathrm{i}$ is an imaginary unit. Here $-1 \leq \beta \leq 1$, and in 
    case of $\beta \neq 0$ the L\'evy distribution is not symmetric. Further $c 
    \in (0, \infty)$ and $\mu \in (-\infty, \infty)$ are scale and and location 
    parameters. This general case of L\'evy distribution will be used further 
    in 
    Section~\ref{sec::samplarch} and Section~\ref{sec::nestedarch} while
    sampling L\'evy distribution in sampling algorithms of some copulas. Remind 
    that in 
    most cases its PDF 
    ${f_S}_{(\alpha, \beta, c, \mu)}$ and CDF ${F_S}_{(\alpha, \beta, c, \mu)}$ 
    do not have an analytical form, however they can be sampled, as presented in
    Algorithm~\ref{alg::levygen}~\cite{nolan2003stable}. Histograms of 
    exemplary L\'evy stable distributions are presented in 
    Figure~\ref{fig::symlevy}.
\end{remark}

\begin{algorithm}[t]
    \caption{Sample general L\'evy distribution.
        \label{alg::levygen}}
    \begin{algorithmic}[1]	
        \State \textbf{Input}: $0 < \alpha \leq 2$ - parameter, $-1 
        \leq \beta \leq 1$ - parameter. 	
        \State \textbf{Output:} $x \in \R$ a sample of the general L\'evy 
        distribution.
        \Function{levygen}{$\alpha$, $\beta$}
        \State sample $\theta \sim \text{Uniform}(-\frac{\pi}{2}, 
        \frac{\pi}{2})$
        \State sample $w \sim {F_{ex_1}}$ \Comment{exponential with scale 1: 
            ${f_{ex_1}}(x) = e^{-x}$}
        \If{$\alpha = 1$}
        \State\Return $\frac{2}{\pi}\left(\left(\frac{\pi}{2} +\beta \theta 
        \right) \tan(\theta) -\beta\log \left( \frac{\frac{\pi}{2}w 
            \cos(\theta)}{\frac{\pi}{2}+\beta \theta}\right) \right)$
        \Else
        \State $\theta_0 = \frac{\arctan(\beta \tan(\frac{\pi 
                \alpha}{2}))}{\alpha}$
        \State\Return $\frac{\sin(\alpha(\theta_0+\theta))}{(\cos(\alpha 
            \theta_0)\cos(\theta))^{\frac{1}{\alpha}}}\left( 
            \frac{\cos\left(\alpha 
            \theta_0 + 
            (\alpha-1)\theta\right)}{w}\right)^{\frac{1-\alpha}{\alpha}}$
        \EndIf
        \EndFunction 
    \end{algorithmic}
\end{algorithm}

\begin{figure}
    \subfigure[Symmetric case $\beta = 
    0$\label{fig::symlevy}]{\includegraphics{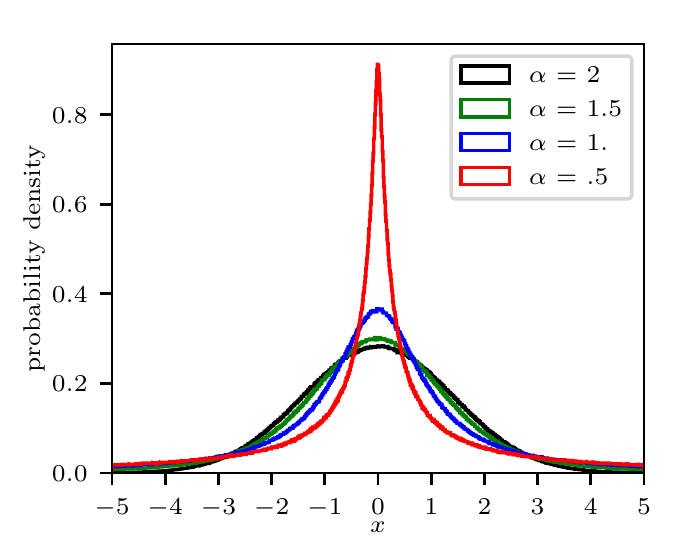}}
    \subfigure[Not symmetric case for $\alpha = 
    0.5$.\label{fig::asymlevy}]{\includegraphics{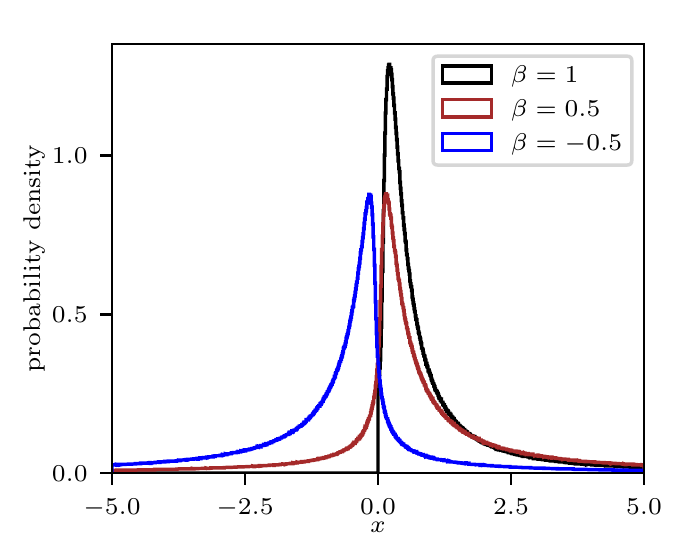}}
    \caption{Histograms of data generated by the L\'evy stable distribution 
        for $\gamma = 
        1$ and $\delta = 0$.}\label{fig::levy}
\end{figure}

Analogically to the Wiener process case, 
increments from distribution with asymptotic behaviour as in 
Eq.~\eqref{eq::asymptoticlevy}, 
would lead to another random process called the L\'{e}vy 
Flight~\cite{miranda2001truncated}. Such process models different types of data 
such as a 
network (internet) traffic~\cite{terdik2009levy}, 
biological data \cite{viswanathan1996levy, 
    viswanathan2010ecology, barthelemy2008levy} and is applicable in 
computer science optimisation and a global maximum search 
\cite{yang2010firefly, yang2010engineering}. Unfortunately the non-trivial 
L\'evy 
distribution do not have defined variance and higher order cumulants as 
well, hence it imposes the limitation for the use of higher order cumulants 
analysis of data that is discussed further in this book.

\subsection{Tsallis q-Gauss distribution}\label{sec::qgauss}

While discussing models of data generated from a random process it is worth to 
mention the Tsallis q-Gauss distribution that is a generalisation of a Gaussian 
distribution in a same 
sense as the Tsallis entropy is a generalization of the standard Shannon 
entropy \cite{tsallis2009nonadditive, gell2004nonextensive}. Sampling of the 
Tsallis q-Gauss 
distribution \cite{4385787} has an application in Generalised Simulated 
Annealing in computer science, where instead of standard Boltzman distribution 
\cite{landau2013course}
\begin{equation}\label{eq::boltzman}
p_i = \frac{\exp(-\beta H_i)}{\sum_j \exp(-\beta H_j)}
\end{equation}
in standard simulated 
annealing algorithm \cite{kirkpatrick1983optimization}, we sample 
\cite{tsallis1996generalized}
\begin{equation}\label{eq::qboltz}
p_i = \frac{\left(1-\beta(1-q)H_i\right)^{\left(\frac{1}{1-q}\right)}}{\sum_j 
    \left(1-\beta(1-q)H_j\right)^{\left(\frac{1}{1-q}\right)}}.
\end{equation}
Here $H_i$ is the generalised energy of the $i$\textsuperscript{th} 
configuration. 
If we have a 
continuous energy spectrum given by $H(x) = x^2$ and we replace sums in 
Eq.~\eqref{eq::boltzman} by an integral over $d x$, we reproduce the 
Gaussian distribution. Analogically 
from Eq.~\eqref{eq::qboltz} we reproduce the Tsallis q-Gauss 
distribution \cite{umarov2008q}
\begin{equation}\label{eq::pdfqg}
f_q (x) = \frac{\sqrt{\beta}}{C_{q}}\big(1-(1-q)\beta 
x^{2}\big)^{\left(\frac{1}{1-q}\right)}, \end{equation}
where the normalisation constant is
\begin{equation}
C_{q} = \begin{cases} 
\frac{2\Gamma\left(\frac{1}{1-q}\right)\sqrt{\pi}}{(3-q)\Gamma
    \left(\frac{3-q}{3(1-q)}\right)\sqrt{1-q}}
&\ \text{if} \ -\infty < q < 1
\\ \sqrt{\pi} &\ \text{if} \  q = 1
\\ 
\frac{\Gamma\left(\frac{3-q}{2(q-1)}\right)\sqrt{\pi}}{
    \Gamma\left(\frac{1}{q-1}\right)\sqrt{q-1}}
&\ \text{if} \  1 < q \leq 3
\end{cases}.
\end{equation}

Parameter $\beta$, in physics denoting the inverse of the
temperature~\cite{landau2013course}, controls the spread of samples and is 
decreased during the simulated annealing as we are moving toward a global 
solution. 
Recall that this $\beta$ here not comply with the $\beta$ parameter for a 
L\'evy distribution, where the scale parameter in literature is rather 
$\gamma$~\cite{nolan2003stable}, but they are not equivalent. 
\begin{remark}
    What is important here, Tsallis q-Gauss distribution parametrised by $1 
    \leq q 
    <3$ is an outcome of the generalisation of the Central Limit Theorem, 
    called 
    the $q$-Central Limit Theorem, where the independence condition is lifted 
    and 
    increments are allowed to be long-range correlated in some specific manner 
    \cite{umarov2008q}. Here we require the existence of the generalised 
    $q$-mean 
    and the generalised $(2q-1)$-variance \cite{umarov2008q} or the 
    distribution of increments. Recall that since we use some generalised for 
    of the mean and variance a standard scaling stability as in the L\'evy 
    distribution case may not hold. 
\end{remark}

The Tsallis q-Gauss distribution 
models better than a Gaussian distribution and many 
types of non-Gaussian distributed data~\eg~the 
financial data~\cite{devi2017financial, borland2002option, 	
    michael2003financial}. In a case of $q < 1$ the q-Gauss distributed random 
variable is limited to $x \in \left[-\frac{1}{\sqrt{(1-q) \beta}}: 
\frac{1}{\sqrt{(1-q) \beta}}\right]$ and 
can be applicable to analyse some biological 
systems~\cite{d2013bounded}.
\begin{figure}[ht!]
    \centering
    \includegraphics{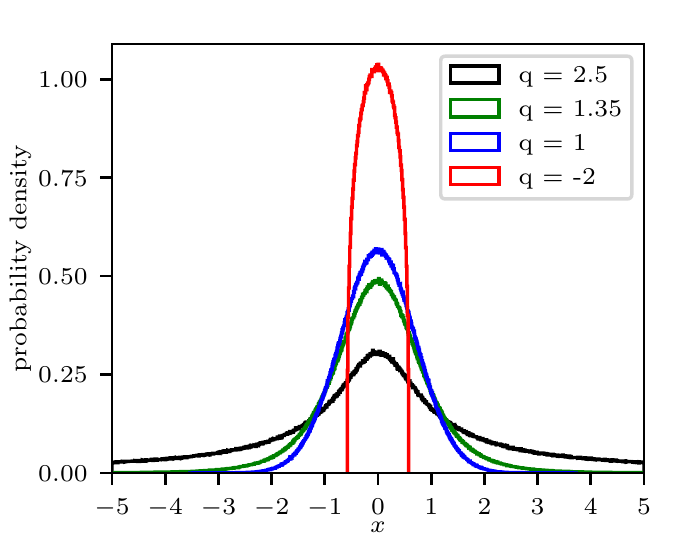}
    \caption{Zero mean q-Gaussian histograms for different $q$ values, $\beta = 
        1$. In a case of $q=1$ we have a Gaussian distribution with standard 
        deviation $\sigma = \sqrt{\frac{1}{2}}$.}
    \label{fig::tcalis}
\end{figure}

Since the CDF of the q-Gauss distribution, given in Eq.~\eqref{eq::pdfqg}, is 
rather complicated we use in general algorithm 
introduced in \cite{4385787} to sample univariate q-Gaussian distribution. We 
present this procedure in Algorithm~\ref{alg::qgaussgen}.

\begin{algorithm}[t]
    \caption{Generate data from q-Gauss distribution.
        \label{alg::qgaussgen}}
    \begin{algorithmic}[1]	
        \State \textbf{Input}: $-\infty < q < 3$ - parameter. 	
        \State \textbf{Output:} $x \in \R$ a sample of q-Gauss 
        distribution.
        \Function{q-gaussgen}{$q$}
        \State sample independently $u_1 \sim \text{Uniform}(0,1), u_2 \sim 
        \text{Uniform}(0,1)$
        \State $q' = \frac{q+1}{3-q}$
        \State $u_1 = \frac{u_1^{1-q'}-1}{1-q'}$ \Comment{$q'$-logarithm as 
            defined in \cite{4385787}}
        \State $\beta = 3-q$
        \State \Return $\frac{\sqrt{-2 u_1}\cos(2 \pi u_2)}{\sqrt{\beta}}$
        \EndFunction 
    \end{algorithmic}
\end{algorithm}

The correspondence between the q-Gaussian distribution and other distributions 
is as follow. For $q 
= 1$ case we reproduce the Gaussian distribution, for $q = \frac{\nu + 
    3}{\nu + 1}$ we reproduce the $t$-Student distribution with $\nu$ degrees 
    of 
freedom 
\cite{de1997student}, while for $q = 2$ we reproduce the Caushy distribution 
\cite{anteneodo2005non, de2017qgaussian}. The relation to the $t$-Student 
distribution gives the motivation to the discussion of the $t$-Student 
multivariate distribution and the $t$-Student copula further in this book.
Observe finally, that for $q < \frac{d+3}{d+1}$ the q-Gauss distribution has 
defined 
$d\textsuperscript{th}$ moment and the $d$\textsuperscript{th} cumulant 
as well. Given certain $q$ values it can be analysed by means of higher order 
cumulants.

\section{Ising model of data}\label{sec::ising}

The portfolio management problem discussed in further part of this book is 
computationally complex, especially if the number of assets in the 
portfolio is large. Hence, the quantum annealing on the D-Wave machine that 
uses the Ising model is promising due to potential 
computational time gain in comparison with classical computing in~\cite[Chapter 
$4$]{florescu2016handbook}, \cite[Section $6.1$]{srivastava2016commercial} and 
\cite{venturelli2018reverse, elsokkary2017financial}. 
    In general, the problem of solving the 
Ising model is NP-hard. However, there are attempts to reduce the computational 
complexity of such problem by means of the quantum 
computing~\cite{bian2010ising}. 

The classical approach to the portfolio optimisation takes the 
Markovitz model~\cite{markowitz1952portfolio} that uses the Gaussian 
distribution of financial data. As discussed in previous sections, 
real-life financial data are in most cases non-Gaussian 
distributed. It is why, to utilise quantum computing for financial data 
analysis, we focus rather on a stochastic model of investors behaviour 
as discussed in \cite{gligor2001econophysics}. Such Ising model is applicable 
to analyse investors driven complex dynamic of financial data. We hope, that 
one can use 
quantum annealing while searching for investors optimal 
behaviour given both their potential for aggregate or individual behaviours.

\subsection{The Ising model}

In this subsection, we discuss an example of a model of investors interactions 
that can be used to model financial data employing the Ising model.
Following \cite{gligor2001econophysics}, let us observe that the crash on the 
stock market occurs if a large group of investors simultaneously decide to sell 
their assets. In most cases, those agents do not know each other, do not 
arrange their actions, and do not follow any leader. For most of the trading 
time, they aggregately give a similar amount of buy and sell orders and hence 
the price of an asset is rather stable. However, sometimes they simultaneously 
sell stocks causing a crash. 

Let us model investor's interaction by means of the regular network of $N$ 
nodes (investors) with $\gamma$ neighbours each. Let us now assume, that each 
investor's
action is influenced both by its nearest surrounding and its individual 
preference. Investors are rather following each other. It is why the neighbour 
interaction will introduce order. Oppositely the individual preference will 
introduce disorder \cite{gligor2001econophysics}. 

Following~\cite{gligor2001econophysics}, to model investor's mutual 
interaction, we assign a number $s_{i} = \pm 1$ ($+1$ means sell, $-1$ means 
buy) to the $i$\textsuperscript{th} investor, and use the following `classical' 
Ising Hamiltonian
\begin{equation} H(\mathbf{s}) = - \sum_{i < j} \varepsilon_{ij}s_{i}s_{j}, 
\end{equation}
where the sum goes over the nearest neighbours set given the particular graph 
of investors interrelationship. Here $\mathbf{s} = [s_1, \ldots, 
s_N]$ is a configuration showing actions of each investor, and $ 
\varepsilon_{ij}$ 
is an interaction strength. Suppose for simplicity that the interaction is 
constant 
$\varepsilon_{ij} = \varepsilon > 0$ over the graph.
Given a regular graph of $N$ nodes (investors), and $\gamma$ neighbours for 
each node there are 
$\frac{\gamma N}{2} $ elements of the sum in 
Eq.~\eqref{eq::isingH},
\begin{equation}\label{eq::isingH}
H(\mathbf{s}) = - \varepsilon \sum_{ \{ij\} } s_{i}s_{j}. 
\end{equation}
Here $H(\mathbf{s})$ is a discrete function and we use the discrete form of the 
Boltzman distribution, as in Eq.~\eqref{eq::boltzman},
\begin{equation}\label{eq::ising} 
f_{\beta}(\mathbf{s}) =\frac{e^{-\beta H(\mathbf{s})}}{Z(\beta)}. 
\end{equation}
The normalisation 
constant  is called the partition function
\begin{equation}\label{eq::partition} Z(\beta) = \sum_{\mathbf{s}} 
e^{-\beta H 
    (\mathbf{s})},
\end{equation}
where the sum goes over the set of all possible configurations of $\mathbf{s}$ 
(all possible actions of investors). The constant $\beta$, which
in 
physics refers to the inverse of the temperature of the system, reflects 
the strength of the statistical noise~\ie~there is no noise if $\beta 
\rightarrow \infty$. In our model the noise correspond to the 
disorder imposed by personal preferences of investors. Hence, as such disorder 
falls, investors organise, and a crash occurs \cite{gligor2001econophysics}. 
Observe, however that presented is this subsection model is very simple, hence 
it can be exactly solved. 

\subsection{Simulated annealing on the D-Wave machine}

To make the model more realistic, we assume that the network of investors 
is not a regular lattice, but rather it reveals some hierarchical structure 
\cite{sornette1998hierarchical}.
Furthermore, we can introduce varying interaction strengths,
$\varepsilon_{ij}$, since different investors may be affected by others in 
different way. Finally, we can introduce $h_i$ representing interaction with 
some sort 
of the `external field', that we can model some 
external factors affecting agents or their particular group. This leads to more 
general Ising Hamiltonian
\begin{equation}\label{eq::ising_general} 
H(\mathbf{s}) = - \sum_{i < j} \varepsilon_{ij}s_{i}s_{j} - 
\sum_i h_i s_i.
\end{equation}
Unfortunately, such 
model becomes more complicated, and hard to simulate. To handle it, we can 
refer to the potentially fast quantum computation on the D-Wave machine 
\cite{cao2016solving, jones2013quantum}. The D-Wave machine is designed to 
solve a specific problem of 
finding a ground state of the following Hamiltonian
\begin{equation}\label{eq::ising_dwave} 
H = - \sum_{i < j} 
\varepsilon_{ij}\sigma_z^{(i)} \sigma_z^{(j)} - 
\sum_{i} h_i \sigma_z^{(i)}. 
\end{equation}
Here $\sigma_z^{(i)}$ is a z-Pauli matrix acting on the 
q-bit~\cite{nielsen2000quantum} located in 
the $i$\textsuperscript{th} node, in other words, on the quantum state of 
the $i$\textsuperscript{th} node. Formally,
\begin{equation}
\sigma_z^{(i)} = \underbrace{\mathbb{1} \otimes \ldots \otimes \mathbb{1}  
}_{i-1} \otimes \sigma_z\otimes \underbrace{\mathbb{1} \otimes \ldots \otimes 
    \mathbb{1}  
}_{N-i}.
\end{equation}
The quantum propagation on a D-Wave machine gives a minimal energy state of the 
Hamiltonian model presented in 
Eq.~\eqref{eq::ising_dwave}. This would correspond to the minimal energy 
solution of classical Hamiltonian presented 
in Eq.~\eqref{eq::ising_general}.

The Hamiltonian in Eq.~\eqref{eq::ising_dwave} has some limitation due to 
the hardware graph of the D-Wave computer chip which is full graph 
nor a regular lattice as in~\cite{gligor2001econophysics}. This processor chip 
is an intermediate case between two types of mentioned graphs, see 
Figure~\ref{fig::chim} \cite{boothby2016fast, king2017quantum, xu2018quantum}, 
and it is called a Chimera graph $\chi_{ij}$.  Such 
graph imposes
a hierarchical relation similar to this observed among investors 
\cite{sornette1998hierarchical}. In our case, investors may be 
organised in groups, that correspond to Chimera cells and 
$1$\textsuperscript{st} degree of the organisation. Those groups are organised 
in the network yielding the $2$\textsuperscript{nd} degree of the organisation.
Its hierarchical layout may be of 
advantage while modelling 
investors behaviour.

\begin{figure}
    \centering
    \includegraphics[scale = 0.8]{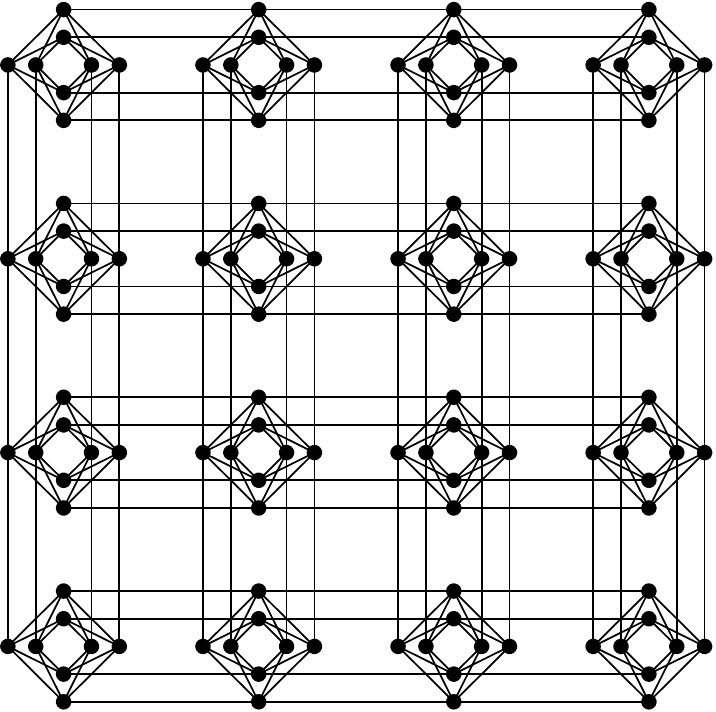}
    \caption{The Chimera graph used actually by the D-Wave 
        machine~\cite{glos2018impact}.}\label{fig::chim}
\end{figure}

Let us now discuss to some extend how the D-Wave machine 
works~\cite{cao2016solving}. It starts with the ground state of the $H_0$ 
Hamiltonian given~by
\begin{equation}
H_0 = - \sum_{i} h_i \sigma_x^{(i)}.
\end{equation} 
Next, it performs the adiabatic evolution of the time dependent 
Hamiltonian $H(t)$ given by
\begin{equation}\label{eq::dwave_evol}
H(t) = \left(1 - \frac{t}{\tau}\right)H_0 + \frac{t}{\tau}H,
\end{equation} 
where  $0 \leq t \leq \tau$ and $\tau$ is an annealing time parameter. In 
results we should obtain a ground state 
eigenvector $\ket{\psi^{(\tau)}}_0$ of the following eigen-equation
\begin{equation}
H(t) \ket{\psi^{(t)}}_n = E_n(t) \ket{\psi^{(t)}}_n,
\end{equation}
where $E_n \in \R$ is the $n$\textsuperscript{th} eigenvalue (energy) and 
$\ket{\psi}_n$ corresponding eigenvector. The following order is fulfilled $E_0 
\leq E_1 \leq \ldots 
\leq E_n$. The output configuration that is supposed to represent the ground 
state can be represented in the form of the vector
\begin{equation}
\mathbf{s} = [s_1, \ldots, s_N],
\end{equation}
where $s_i = \pm 1$. The D-Wave 
propagation is performed many times, and the result is 
probabilistic~\cite{cao2016solving}. As the result we obtain have a series 
$\mathbf{s}_1, 
\mathbf{s}_2, \ldots, \mathbf{s}_M$ all corresponding (in theory) to the ground 
state~\ie~the minimum of Eq.~\eqref{eq::ising_general}, since the D-Wave 
propagation~\cite{cao2016solving} is performed in (almost) $0$ Kelvin 
environment.

In our case, each $\mathbf{s}$ models the optimal state of 
investors for the given parameters of the network of, their mutual interaction 
and 
other external parameters.
This minimal energy Boltzmann distribution with (almost) no thermal noise. In 
the case of the financial data modelling,  
however, we can still be interested in a zero temperature solution to maximally 
`ordered' behaviour of investors to answer a question if a crisis is possible 
given initial settings of the model. To introduce a thermal noise, we can set 
the annealing time parameter 
$\tau$ very short~\cite{cao2016solving} resulting in the not full annealing 
yielding sometimes the ground state and sometimes one of the excited states 
with higher energy. In another approach, given 
an input Ising model see Eq.~\eqref{eq::ising_general} we can introduce a 
noise via a random factor in parameter $h_i$. Hence we can observe how the 
choice of statistics of the noise incorporated in $h_i$ would affect statistics 
of investors behaviour. Finally, statistics of investor behaviour can be 
retranslated on statistics of financial data. One can analyse how the initial 
set of parameters affects the financial data, and answer when they 
are Gaussian distributed and when not.
\chapter{Multivariate Gaussian models}\label{ch::mvn}

Gaussian models are widely used to analyse various types of real life 
data~\cite{duda2012pattern}. 
Unfortunately often Gaussian models are assumed a 
priorly without testing normal distribution on data. To test multivariate 
Gaussian distribution one can use the Marida's test~\cite{mardia1970measures}, 
or the BHEP test~\cite{baringhaus1988consistent}. In this book we will focus on 
higher order multivariate cumulants that are zero for multivariate Gaussian 
distributed data and hence can be used to test their distribution.
The advantage of the last comes form the fast programming 
implementation~\cite{domino2018efficient} especially applicable for large data 
sets.

Henceforth in this book we understand
$n$-variate data as a collection of realisations of the 
$n$-dimensional 
random vector $\XXX^{(n)} = (\XXX_1, \ldots, \XXX_n)$. Such random vector 
consists of $n$ interdependent univariate random variables $\XXX_i$. Given 
those, 
$t$ realisations of  
$n$-dimensional random vector can presented in the form of
matrix:
\begin{equation}\label{eq::variable}
\R^{t \times n} \ni \X = \left[ \begin{array}{ccc}
x_{1,1} & \dots & x_{1,n}  \\ 
\vdots & x_{j,i} & \vdots \\ 
x_{t,1} & \dots & x_{t,n}  \\ 
\end{array}   \right].
\end{equation}
The $i$\textsuperscript{th} column $\XXX_i$ of such matrix is the column vector 
of all 
realisations of the $i$\textsuperscript{th} marginal
\begin{equation}
\R^t \ni X_i = \left[x_{1,i}, \ldots, x_{j,i}, \ldots, 
x_{t,i}\right]^{\intercal},
\end{equation} 
and the single realisation of the random vector $\XXX^{(n)}$ is the row vector
\begin{equation}
\R^n \ni \mathbf{x}_j = [x_{j,1}, \ldots, x_{j,n}].
\end{equation}

\section{Multivariate Gaussian Distribution}\label{sec::mvn}

The Multivariate Gaussian Distribution is widely used to analyse multivariate 
data. The main argument for such probabilistic model comes from the results 
Multivariate 
Central Limit Theorem (MCLT) \cite{van1998asymptotic}.
\begin{theorem}\label{th::mclt}
Let ${\ZZ^{(n)}}_{1}, \ldots, {\ZZ^{(n)}}_{k}$ be $k$ independent
random vectors 
with identical distributions having a mean vector $\mu \in \R^n$ 
and 
a positive definite symmetric covariance matrix $\SSSS \in \R^{n 
\times n}$. The random vector
\begin{equation} \XXX^{(n)} = \frac{1}{\sqrt{k}}\sum_{i=1}^k 
\left({\ZZ^{(n)}}_i - \mu \right)
\end{equation}
has distribution that coverages as $k \rightarrow \infty$ to the 
multivariate Gaussian 
distribution with zero
mean vector and a covariance matrix $\SSSS$.
\end{theorem}
The multivariate Gaussian Probability Density Function (PDF) reads
\begin{equation}\label{eq::gaussmulti}
\rhog(\x) = 
\frac{1}{(2\pi)^{\frac{n}{2}}\det(\SSSS)^{\frac{1}{2}}} 
\exp\bigg(-\frac{1}{2}(\x - \mu)^{T}\SSSS^{-1}( \x - \mu) \bigg),
\end{equation}
and the corresponding multivariate Gaussian Cumulantive Distribution Function 
(CDF) is given by 
\begin{equation}\label{eq::gaussmulticum}
\Fg(\x) = P(\XXX_1 \leq x_1 \wedge \ldots \wedge \XXX_n \leq x_n) =  
\int_{-\infty}^{x_1} \ldots 
\int_{-\infty}^{x_n} 
\rhog(y_1, \ldots, 
y_n) d y_1 \ldots d y_n.
\end{equation}
Gaussian distributed random vector is: $\XXX^{(n)} 
\sim \mathcal{N}(\mu, \SSSS)$. The covariance matrix parameter $\SSSS \in \R^{n 
\times n}$ is symmetric and positive definite. Its elements  $s_{i_1, 
i_2}$ are either diagonal $s_{i,i} = 
\sigma_i^2$, being a variance of the $i$\textsuperscript{th} marginal, or 
off-diagonal $s_{i_1 \neq i_2}$, being a covariance between $i_1$ and 
$i_2$ marginals. These elements yield the Pearson cross-correlation 
coefficient 
\cite{pal2012multiple}
\begin{equation}\label{eq::corrR}
\text{cor}(\XXX_{i_1}, \XXX_{i_2}) = r_{i_1, i_2} = 
\frac{s_{i_1,i_2}}{\sqrt{s_{i_1,i_1} 
s_{i_2,i_2}}}.
\end{equation} 
The Pearson correlation
matrix (or simply the correlation matrix) $\mathbf{R} \in \R^{n \times n}$ is 
positive definite, and symmetric matrix with ones on a diagonal $r_{i,i} 
= 
1$. Since $\mathbf{R}$ is positive definite we get $-1 < r_{i_1,i_2} 
< 1$.

If we have data distributed according to multivariate Gaussian model, there are 
some statistical features important in analysing such data. Let us mention a 
few that may not be fulfilled if data are non-Gaussian distributed as discussed 
in the further part of this book. 
\begin{remark}\label{rem::buvcorsg}
	Given $\XXX^{(n)} \sim \NN(\mu, 
	\SSSS)$ the cross-correlation 
	between marginals is included in $\mathbf{R}$ matrix with elements 
	$r_{i_1,i_2}$. Hence the pairwise 
	analysis of the cross-correlation of Gaussian distributed data would 
	extract all information about marginals interdependency. The observation 
	presented in Remark~\ref{rem::buvcorsg} is not true in 
	general. In Chapter~\ref{ch::hostats}, we discuss higher order cumulants of 
	multivariate data that carries information about the simultaneous 
	interdependence of more than $2$ marginals, and are zero for Gaussian 
	distributed multivariate data.
\end{remark}

\begin{remark} Given $\XXX^{(n)} \sim \NN(\mu, 
	\SSSS)$ we can use the 
Singular Value Decomposition (SVD) of the covariance matrix to obtain
project data on independent directions. The new random vector
\begin{equation}\label{eq::pca}
	\mathfrak{Y}^{(n)} = \XXX^{(n)} \mathbf{A}
\end{equation}
consists of independent marginals, that can be ordered with respect to 
information significance~\cite{duda2012pattern}.

The symmetric and 
positive definite covariance matrix $\SSSS$ can be 
decomposed as
\begin{equation}\label{eq::cvd}
\SSSS = \mathbf{A} \SSSS_{(\text{d})} \mathbf{A}^{\intercal},
\end{equation}
where $\SSd$ is a diagonal matrix 
real and positive eigenvalues of $\SSSS$ notes by ${s_{(d)}}_{i,i}$, and 
$\mathbf{A} \in \R^{n \times n}$ is an unitary matrix of eigenvectors. We 
assume eigenvectors are normalised to $1$. The 
ordering ${s_{(d)}}_{i_1,i_1} \geq 
{s_{(d)}}_{i_2,i_2}$ if $i_1 > i_2$ results in ordering of $\mathfrak{Y}_i$ 
with respect to their variability and hence information they are carrying.
\end{remark}
Such approach is the core of the Principal Component Analysis (PCA) 
\cite{jolliffe2002principal} used to extract information from data. If 
$\XXX^{(n)}$ is not multivariate Gaussian distributed, here information may be 
tied to higher order cross-correlation between marginals. Hence, 
Eq.~\eqref{eq::pca} gives marginals $\mathfrak{Y}_1, 
\ldots, \mathfrak{Y}_n$ that are not independent and their informative 
hierarchy is not straightforward. 

To show non-Gaussian distribution of data that are modelled as a sum of many 
increments consider financial data, where conditions of the MCLT may not hold. 
Take as example multivariate ARCH and GARCH  
models~\cite{stelzer2007multivariate, stelzer2010multivariate, 
lindner2009continuous} that break increments independence and identical 
distributions conditions. To discuss the practical application of the SVD in 
financial data analysis observe the following example.
\begin{example}
	In classical financial engineering one uses the SVD of the covariance 
	matrix in the classical Value at Risk 
	(VaR) procedure~\cite{best2000implementing}, determining
	eigenvectors that correspond to 
	low eigenvalues. These eigenvectors are supposed to give linear 
	combinations 
	of marginals with low variance corresponding to safe portfolio of risky 
	assets, if one tie the variance with the variability and risk. Negative 
	weights correspond to so called `short sale'. In practice, financial data 
	are often non-Gaussian distributed and the 
	covariance analysis may fails 
	to~anticipate cross-correlated extreme events, appearing for many assets 
	simultaneously \cite{cherubini2004copula, yao2017study}. This typically 
	happens during the crisis, and as non-predicted by the Gaussian model is 
	potentially dangerous for equity holders.
\end{example}

In next section we are going to discuss a family of non-Gaussian multivariate 
distributions that can be easily transformed to the Gaussian one.

\section{Gaussian copula}\label{sec::Gcop}

By the Sklar's theorem~\cite{sklar1959fonctions} each multivariate distribution 
can be split onto marginal distributions and the copula, being a `core' of the 
distribution, responsible for the interdependence of marginals. Formally the 
Sklar's theorem is following.
\begin{theorem}
	Every multivariate CDF $\mathbf{F}(x_1, \ldots, x_n)$ of a random vector 
	$\XXX^{(n)} = (\XXX_1, \ldots, \XXX_n)$ can be expressed in terms of its 
	marginal univariate CDFs $F_i(x_i)$ and the copula $\C: [0,1]^n  
	\rightarrow [0,1]$. The relation reads
	\begin{equation}
	\mathbf{F}(x_1, \ldots, x_n) = \C(F_1(x_1), \ldots, F_n(x_n)).
	\end{equation} 
	The Sklar's theorem also implies that if $F_1(x_1), \ldots, F_n(x_n)$ are 
	continuous, the copula function $\C$
	is uniquely defined by $\mathbf{F}$.
\end{theorem}
Henceforth we will discuss the case of continuous univariate PDFs. 
To construct Gaussian 
multivariate distribution we~need the Gaussian copula and the Gaussian 
marginals. 
Consider multivariate random vector $\XXX^{(n)} \sim 
\mathcal{N}(\mu, \SSSS)$. Its marginals are 
univariate Gaussian distributions $\XXX_{i} \sim \mathcal{N}(\mu_i, 
\sigma^2_i)$, where $\sigma^2_i 
= s_{i,i}$ which is the $i$\textsuperscript{th} diagonal element of $\SSSS$. 
Let 
$F_{\mathcal{N}(\mu_i, \sigma_i^2)}(x_i)$ be marginal 
univariate Cumulantve Distribution Function (CDF) of those marginals.
Now define a new $n$-variate random vector $\mathfrak{U}^{(n)} = (\UU_1, 
\ldots, \UU_n)$ fulfilling
\begin{equation}\mathfrak{U}_i = F_{\mathcal{N}(\mu_i, 
\sigma_i^2)}(\XXX_i) \text{ or } \XXX_i = F^{-1}_{\mathcal{N}(\mu_i, 
\sigma_i^2)}(\mathfrak{U}_i).
\end{equation} 
Each marginal $\mathfrak{U}_i$ is uniformly distributed
on a line segment $[0,1]$, \ie~$\mathfrak{U}_i \sim 
\text{Uniform}\left([0,1]\right)$, see Table~\ref{tab::symbols}.
A realisations matrix and the single 
realisation vector of $\mathfrak{U}^{(n)}$ random vector are
\begin{equation}
\U \in [0,1]^{t \times n} \text{ and }\uu \in [0,1]^{n},
\end{equation}
where $t$ is a number of realisations. We can introduce the Gaussian 
copula.
\begin{definition}\label{d::gcop}
	A Gaussian copula~\cite{nelsen1999introduction} $\C_{\mathbf{R}}: [0,1]^n 
	\rightarrow 
	[0,1]$ is the multivariate CDF of $\UU^{(n)}$ with uniform marginals on 
	$[0,1]$ segment, such that
\begin{equation}\label{eq::gauscop1}
\begin{split}
	[0,1] \ni \C_{\mathbf{R}}(\uu) =& \Fg\left(F_{\mathcal{N}(\mu_i, 
		\sigma_1^2)}^{-1}(u_1), \ldots, F_{\mathcal{N}(\mu_i, 
	\sigma_n^2)}^{-1}(u_n)\right) \\ =& 
	\mathbf{F}_{\mathcal{N}(0,\mathbf{R})}\left(F_{\mathcal{N}(0,1)}^{-1}(u_1), 
	\ldots, F_{\mathcal{N}(0,1)}^{-1}(u_n)\right).
\end{split}
\end{equation}
\end{definition}
The Gaussian copula models interdependency between marginals of $\UU^{(n)}$ by 
means of 
multivariate Gaussian distribution. Following 
Eq.~\eqref{eq::corrR} $\mathbf{R} \in 
\R^{n \times n}$ is the correlation matrix, and $F_{\NN(0,1)}(x)$ 
is the standard Gaussian univariate CDF (with zero mean and $1$ variance). 
Suppose, we have univariate continuous CDFs $F_1, \ldots, F_n$. The  
multivariate distribution
	\begin{equation}\label{eq::gcoparbmarg}
	\mathbf{F}(\mathbf{y}) = \mathbf{C}_{\mathbf{R}}(F_1(y_1), \ldots,  
	F_n(y_n))
	\end{equation}
has Gaussian copula and $F_1, \ldots, F_n$ marginals. To show the 
importance of the Gaussian copula consider the following theorem.

\begin{theorem}\label{the::gcoptomv}
	Let us consider a random vector $\mathfrak{Y}^{(n)}$ distributed according 
	to $\mathbf{F}$, that is certain multivariate CDF, with 
	continuous univariate marginal CDFs $F_1, \ldots, F_n$, such that $F_i 
	:(-\infty, 
	\infty) \rightarrow [0,1]$. 
	There exist a transformation of marginals changing $\mathfrak{Y}^{(n)} 
	$ to a multivariate Gaussian distributed $\XXX^{(n)}$  iff 
	$\mathbf{F}$ (if and only if) has the Gaussian copula.
\end{theorem}
\begin{proof}
	We transform the random vector $\mathfrak{Y}^{(n)} \rightarrow \UU^{(u)}$, 
	by the following 
	transformation of all marginals
	\begin{equation}
 	\UU_i = F_i(\mathfrak{Y}_i).
	\end{equation}
	Since $F_i$ is continuous, $\mathfrak{U}_i$ is continuous and uniformity 
	distributed on 
	$[0,1]$. Now we can use the inverse univariate Gaussian CDFs (quantile 
	function) $F_{\NN{(\mu_i, 
	\sigma_i^2)}}^{-1}$ and transform each 
	marginal further as
	\begin{equation}
	\XXX_i =  F_{\mathcal{N}(\mu_i, \sigma_i^2)}^{-1}(\UU_i).
	\end{equation} 
	Parameters $\mu_i$ 
	and $\sigma_i^2$ are arbitrary parameter of an univariate Gaussian 
	distribution. The transformation $\mathfrak{Y}^{(n)} 
	\rightarrow \UU^{(u)} \rightarrow \XXX^{(n)}$ affects only marginals. Hence 
	it requires Gaussian copula of $\UU^{(u)}$ to produce multivariate Gaussian 
	distribution of $\XXX^{(n)}$. Finally Gaussian copula of $\UU^{(u)}$ yields 
	a Gaussian copula of $\mathfrak{Y}^{(n)}$.
\end{proof}

\begin{remark}\label{rem::savepca}	
Observe that such transformation 
$\mathfrak{Y}^{(n)} 
\rightarrow \XXX^{(n)}$ preserves rank 
cross-correlation between marginals. 
As previously, $F_i$ is univariate continuous CDFs, hence 
it is increasing. The probability of $y$ being a sample of $F_i$ equals 
$f_i(y) = \frac{d F_i(y)}{d y}$, hence we can exclude such ranges of $y$ where 
$F_i(y)$ is constant. Based on these two observations the relation $u = F_i(y)$ 
is order 
preserving. Observe as well,
that $F_{\mathcal{N}(\mu_i, \sigma_i^2)}^{-1}$ is strictly increasing and hence 
order preserving. Finally the relation 
\begin{equation}
x = F_{\mathcal{N}(\mu_i, 
\sigma_i^2)}^{-1}\left(F_i (y)\right)
\end{equation} is order preserving as well.
Referring to Theorem~\ref{the::gcoptomv}, if data have Gaussian copula, they 
may 
be transformed to multivariate 
Gaussian distribution in a rank cross-correlation preserving method. Then 
methods of 
Gaussian data analysis such as PCA can be applied. Nevertheless, the Gaussian 
copula 
assumption is necessary here. 
\end{remark}

Referring to Remark~\ref{rem::buvcorsg} a Gaussian multivariate distribution 
introduces the cross-correlation between pairs of marginals by means of the 
correlation matrix $\mathbf{R}$. The same is the case for a Gaussian copula. 
Hence it is convenient to use bivariate Gaussian sub-copula while discuss 
interdependencies between marginals. The sub-copula notation is not often used 
in the copula literature, however we used it here to distinguish between 
bivariate and $d$-variate (as $d > 2$) interdependency measures. The later are 
included in $d$\textsuperscript{th} order multivariate cumulants discussed 
further in this book.
\begin{definition}\label{def::gsubcop}
	Suppose we have the Gaussian copula $\C_{\mathbf{R}}(\uu)$, 
	where $\uu = [u_1, \ldots, u_n]$ is single realisation of the random vector 
	$\UU_{(n)}$ with following marginals $\UU_i \sim 
	\text{Uniform}([0,1])$, the probability of $\UU_i \leq 1$ equals one. The 
	bivariate Gaussian sub-copula for marginals $i_1$ and $i_2$ is defined as 
	\begin{equation}
	\begin{split}
	\C_{\mathbf{R}'}(u_{i_1}, u_{i_2}) &= P(\UU_{i_1} \leq u_{i_1} \wedge 
	\UU_{i_2} \leq u_{i_2}) \\ &= 
	P(\UU_{i_1} \leq u_{i_1} \wedge \UU_{i_2} \leq u_{i_2} \wedge_{k \neq 
	i_1,i_2} \UU_k \leq 1) \\ 
	&= 
	\C_{\mathbf{R}}(1, 
	\ldots, 1, u_{i_1}, 1,\ldots 1, u_{i_2}, 1, 
	\ldots, 1),
	\end{split}
	\end{equation}
	where $\mathbf{R}' = 
	\begin{bmatrix} 1 & r_{i_1,i_2} \\ r_{i_2,i_1} & 1 \end{bmatrix}$, and
	$r_{i_1,i_2} = 
	r_{i_2,i_1}$ are corresponding elements of the original correlation matrix 
	$\mathbf{R}$.
\end{definition}

It is easy to conclude that bivariate sub-copula $\C_{\mathbf{R'}}$ models 
bivariate measures of interdependency between marginals $i_1$ and $i_2$. 
Besides standard Pearson cross-correlation (equal to $r_{i_1,i_2}$ given 
Gaussian marginals) another bivariate interdependency 
measures are bivariate `tail' dependencies.
\begin{definition}\label{deg::lambda_l_g}
	The lower `tail' dependency~\cite{nelsen1999introduction} is defined as 
	\begin{equation}
	\lambda_l = \lim_{u \rightarrow 0^+} P(\UU_1 \leq u | \UU_2 \leq 
	u) = \lim_{u \rightarrow 0^+} \frac{P(\UU_1 \leq u, 
		\UU_2 \leq u)}{P(\UU_2 \leq u)} = 
	\lim_{u \rightarrow 0^+} \frac{\mathbf{C}(u, u)}{u},
	\end{equation}
	the last equity results form the fact, that $P(\UU \leq u) = u$ since $\UU 
	\sim \text{Uniform}([0,1])$. 
	Here $u_i \rightarrow 0$ corresponds to very low value of 
	$x_i = F_i^{-1}(u_1)$, and $\mathbf{C}(u, u)$ is a copula function, for 
	 this Chapter it is bivariate Gaussian copula $\C_{\mathbf{R'}}$.
\end{definition}

\begin{definition}\label{deg::lambda_u_g}
	Analogically to definition~\ref{deg::lambda_l_g}, the upper `tail' 
	dependency~\cite{nelsen1999introduction} is defined as
	\begin{equation}
	\begin{split}
	\lambda_u &= \lim_{u \rightarrow 1^-}P(\UU_1 > u | \UU_2 > u)\\& = \lim_{u 
		\rightarrow 1^-} \frac{P(\UU_1 > u, \UU_2 > u)}{P(\UU_2 > u)} \\ &= 
		\lim_{u 
		\rightarrow 1^-} \frac{1 - P(\UU_1 \leq u) - P( \UU_2 \leq u) + P(\UU_1 
		\leq u, 
		\UU_2 \leq u)}{P(\UU_2 > u)}\\& = \lim_{u \rightarrow 1^-} \frac{1-2u+ 
		\mathbf{C}(u,u)}{1-u}.
	\end{split}
	\end{equation}
\end{definition}

The lower `tail' dependency gives a probability of the appearance of
an extreme low value of one marginal given an extreme low value of 
another marginal. Analogically upper `tail' dependency gives a probability of 
the extreme high 
value of one marginal given an extreme high value of an another. 
Importantly, it is easy to show, that 
for Gaussian copula $\lambda_l = \lambda_u = 0$
\cite{nelsen1999introduction}.  
To model data where simultaneous extreme values of many marginals are 
possible, we need other non-Gaussian copula. Some of them are introduced in the
next Chapter.

\chapter{Copulas}\label{chap::cops}

In the previous chapter we introduced the Gaussian copula that is a specific 
case of the wide family of copulas discussed in this chapter. For many 
real-life data non-Gaussian copulas are better models than Gaussian one. This 
is the case especially if simultaneous extreme event in many marginals are 
recorded. Such copulas are used in financial data 
analysis~\cite{cherubini2004copula}, 
reliability analysis~\cite{pham2003recent, 
	wu2014construction}, Civil 
Engineering~\cite{kilgore2011estimating, zhang2015long}, time series analysis 
\cite{eban2013dynamic}, neuroscience~\cite{onken2009analyzing}, 
hydrology~\cite{laux2009modelling}, climate 
research~\cite{domino2014meteo, schoelzel2008multivariate} and random data 
generation~\cite{bandara2011characteristics, strelen2009tools}, to mention a 
few.

In this chapter we discuss multivariate models of data constructed from the 
non-Gaussian copula and continuous univariate marginals. The formal definition 
of the copula function \cite{nelsen1999introduction, sklar1959fonctions} states 
as follows.
\begin{definition}\label{d::copdef} The function $\C: [0,1]^n \rightarrow 
[0,1]$ 
is the $n$-dimensional copula if
	\begin{enumerate}
		\item $\C(u_1, \ldots u_{k-1}, 0, u_{k+1}, \ldots u_n) = 
		0$,
		\item $\C(1, \ldots, 1, u_k, 1, \ldots, 1) = u_k$,
		\item $\C(\mathbf{u})$ is $n$ non-decreasing in the following manner
		\begin{equation}
		\forall_{B \subseteq [0,1]^n} \ \int_{B} d \C(\uu) \geq 0,
		\end{equation}
	\end{enumerate}
		given $\uu = [u_1, \ldots, 
		u_n] \in [0,1]^n$, and
		$B =\prod_{i=1}^n[u_i, v_i]
		\subseteq 
		[0,1]^n$, where $0 \leq u_i \leq v_i \leq 1$.
	
It can be easily shown that the Gaussian copula, from
Definition~\ref{d::gcop} fulfils those conditions.
\end{definition}

\begin{remark} Following Definition~\ref{d::copdef} copula is a 
multivariate CDF of the random 
vector $\UU^{(n)} = (\UU_1, \ldots, \UU_n)$ with uniform marginals~\ie~$\UU_i 
\sim 
\text{Uniform}([0,1])$
\begin{equation}\label{eq::copprob}
[0,1] \ni \C(u_1, \ldots, u_n) = P(\UU_1 \leq u_1 \wedge \ldots \wedge \UU_n 
\leq u_n).
\end{equation}
The corresponding multivariate PDF of $\UU^{(n)}$ is called 
the copula density
\begin{equation}\label{eq::subcop}
\mathbf{c}(\uu) = \frac{\partial ^n}{\partial u_1 \ldots \partial u_n} 
C(\uu).
\end{equation}
\end{remark}

Analogically to the Gaussian copula case in Eq.~\eqref{eq::gcoparbmarg}, 
each copula $\C$ can be used to construct variety of multivariate 
distributions. The multivariate 
CDF of the $n$-variate random 
vector $\XXX^{(n)}$ modelled by a copula $\C$ and marginal CDFs $F_i$ is
\begin{equation}\label{eq::Co}
\mathbf{F}(\x) = P(\XXX_1 \leq x_1 \wedge \ldots \wedge \XXX_n \leq x_n) = 
\C\left(F_1(x_1), \ldots, F_n(x_n)\right).
\end{equation}
As $u_i = F_i(x_i)$ and $\frac{d u_i}{d x_i} = \frac{d}{d x_i} F_i(x_i) = 
f_i(x_i)$, corresponding multivariate PDF is
\begin{equation}\label{eq::co}
\mathbf{f}(\x) = \mathbf{c}\left(F_1(x_1), \ldots, F_n(x_n)\right) \prod_{i = 
	1}^n f_i(x_i).
\end{equation}

We will use here the sub-copula notation in order to discuss various measures 
of the interdependency between marginals. This includes both bivariate 
measures and $d$-variate measures imposed by $d$\textsuperscript{th} order 
multivariate cumulants discussed in following chapters. For the significance of 
bivariate sub-copulas let us refer to 
\cite{erdely2016subcopula}, where the measures of dependencies between
pairs of marginals of a given multivariate random vector are discussed by means 
of such bivariate sub-copulas. However, 
such approach is limited, as discussed further in Chapter~\ref{ch::hostats}.
\begin{definition}\label{d::subcop}
	Let $\textbf{r} = (r_1, \ldots, r_k) \subset (1,2, \ldots, n)$, 
	$k < n$, and let $\C:[0,1]^n \rightarrow [0,1]$ be the $n$-variate copula. 
	The multivariate CDF od the marginal's subset indexed by $\textbf{r}$ is a 
	sub-copula $\C^{(s)}: [0,1]^k 
	\rightarrow [0,1]$, where
	\begin{equation}
	\C^{(s)}(u_{r_1}, \ldots, u_{r_k}) = \C(w_1, \ldots, w_{n}), \text{ where } 
	w_i = 
	\begin{cases} u_i\ \ \text{if} \ i \in \mathbf{r} \\
	1\ \ \ \text{otherwise}
	\end{cases}.
	\end{equation}
	It can be shown, that such $k$-dimensional sub-copula fulfils conditions 
	from Definition~\ref{d::copdef} and is a copula itself 
	\cite{cherubini2004copula}. Recall that in Definition~\ref{def::gsubcop} we 
	have introduced a Gaussian bivariate 
	sub-copula in similar manner.
\end{definition}
 
 Suppose we have a bivariate random vector 
$\UU^{(2)}$ with marginals $\UU_i \sim \text{Uniform}({0,1})$ modelled by a 
bivariate copula $\C(u_1, u_2)$. The covariance between marginals is
\begin{equation}
\text{cov}(\UU_1, \UU_2) = \int_{0}^1 \int_0^1\left(u_1 - 
\frac{1}{2}\right)\left(u_2 - 
\frac{1}{2}\right)d \C(u_1, u_2),
\end{equation}
or equivalently \cite{schweizer1981nonparametric} for continuous copulas
	\begin{equation}\label{eq::copcor}
	\text{cov}(\UU_1, \UU_2) = \int_{0}^1 \int_0^1 \left(\C(u_{1}, u_{2}) 
	-u_{1}u_{2}\right) 
	du_1 d u_2 = \int_{0}^1 \int_0^1 
	\C(u_{1}, 
	u_{2})du_{1} du_{2} - 3.
	\end{equation}
The variance of an uniform distribution on $[0,1]$ is $\frac{1}{12}$, hence the 
Pearson cross-correlation between marginals is
\begin{equation}\label{eq::coprho}
\text{cor}(\UU_1, \UU_2) = 12 \text{cov}(\UU_1, \UU_2).
\end{equation}
The Spearman's rank cross-correlation 
coefficient \cite{schweizer1981nonparametric} is
\begin{equation}
\rho(\UU_1, \UU_2) := \text{cor}(\text{rank}(\UU_1), 
\text{rank}(\UU_2)) = \text{cor}(\UU_1, \UU_2).
\end{equation}
The second equity results from the fact, that univariate 
marginals are uniformly distributed on $[0,1]$.
Kendall's rank cross-correlation coefficient $-1 \leq \tau \leq 1$ equals 
to~\cite{trivedi2007copula}
	\begin{equation}\label{eq::coptau}
	\tau(\UU_1, \UU_2) = 4 \int_{0}^1 \int_0^1  \C(u_1, u_2) d \C(u_1, u_2) -1.
	\end{equation} 
\begin{remark}
	Suppose we have random vectors $\XXX^{(n)}$ and $\UU^{(n)}$ with the same 
	copula. $\UU^{(n)}$ have all uniform marginals on $[0,1]$, and $\XXX^{(n)}$ 
	has continuous univariate marginal CDFs $F_i$, hence $\UU_i = F_i(\XXX_i)$. 
	Following 
	Remark~\ref{rem::savepca} such 
	transformations does not change the rank of the given realisation of the 
	given marginal. Hence it does not change Spearman's and Kendall's 
	rank cross-correlation sa well. We have
	\begin{equation}
		\rho(\XXX_1, \XXX_2) = \rho(\UU_1, \UU_2), \quad
	\tau(\XXX_1, \XXX_2) = \tau(\UU_1, \UU_2).
	\end{equation}
	On the other hand, this does not hold for Pearson's cross-correlation 
	\cite{schweizer1981nonparametric}
\begin{equation}
\text{cov}(\XXX_1, \XXX_2) = \int_{0}^1 \int_0^1 \left(\C(u_{1}, u_{2}) 
-u_{1}u_{2}\right) dF_1^{-1}(u_1) d F_2^{-1}(u_2),
\end{equation}
and in general
\begin{equation}
\text{cor}(\XXX_1, \XXX_2) \neq \text{cor}(\UU_1, \UU_2).
\end{equation}	
\end{remark}

Another important bivariate measures of interdependency between marginals are 
bivariate `tail' dependencies defined as \cite{nelsen1999introduction}
\begin{equation}\label{eq::llambda}
\lambda_l = \lim_{u \rightarrow 0^+} \frac{\C(u, u)}{u},
\end{equation}
and 
\begin{equation}\label{eq::ulambda}
\lambda_u = \lim_{u \rightarrow 1^-} \frac{1- 2u + \C(u, u)}{1-u}.
\end{equation}
In the Gaussian copula case those would correspond to 
Definitions~\ref{deg::lambda_l_g}, and~\ref{deg::lambda_u_g}. As discussed 
further, many non-Gaussian copulas have non-zero `tail' 
dependencies and hence can be used to model simultaneous extreme events 
appearing in two (or many in general case) marginals. 

\section{Elliptical copulas}\label{sec::elcops}
The Gaussian copula mentioned in Chapter~\ref{ch::mvn} is a member of wider 
family of elliptical copulas derived from elliptical multivariate 
distributions, see Definition $5.1$ in~\cite{embrechts2001modelling}. 

\begin{definition} The random vector 
	$\XXX^{(n)} \sim 
	\mathbf{E}(\mu, \mathbf{\Sigma}, \psi)$ is modelled by an elliptical 
	distribution if the probability distribution function of 
	$\mathfrak{Y}^{(n)} = 
	\XXX^{(n)} - \mu$ 
has 
the following characteristic function 
\begin{equation}
\phi_{\mathfrak{Y}^{(n)}}(\mathbf{\tau}) 
= \psi(\tau^{\intercal} \mathbf{\Sigma} \tau).
\end{equation}
Here $\mu
\in \R^n$, and $\psi$ is some scalar function~\cite{cambanis1981theory}, and
$\mathbf{\Sigma} \in \R^{n 
\times n}$ is positive semi-definite symmetric matrix. 
\end{definition} 
Gaussian multivariate distribution is an elliptical 
one with $\psi(t) = e^{-\frac{t}{2}}$~\cite{embrechts2001modelling}.
\begin{definition}\label{def::elcop}
	The elliptical 
	copula $\C_{\mathbf{E}}$ is given by
	\begin{equation}
	\C_{\mathbf{E}}(u_1, \ldots, u_n) = \mathbf{E}_{\mu, \mathbf{\Sigma, 
			\psi}}(E_1^{-1}(u_1), \ldots, 
	E_n^{-1}(u_n)).
	\end{equation}
	Here $\mathbf{E}_{\mu, \mathbf{\Sigma}, \psi}$  is the multivariate 
	Cumulative Distribution Function (CDF) 
	of the elliptical distribution, and $E_i$ are corresponding univariate 
	marginal CDFs. For simplicity, we can use $\mu = 0$ 
	both in 
	$\mathbf{E}_{\mu, \mathbf{\Sigma}, \psi}$ and all $E_i$. 
\end{definition}

\begin{remark}\label{rem::elcopgen}
Suppose we want to generate $t$ realisations of $n$-variate elliptical copula, 
derived from $\mathbf{E}(\mathbf{\Sigma}, \psi)$, we can perform it using the 
following steps:
\begin{enumerate}
	\item Sample $\X \in \R^{t \times n}$ with elements $x_{j,i}$ from 
	$\mathbf{E}(\mathbf{\Sigma}, \psi)$.
	\item Transform $\X \rightarrow \mathbf{U} \in 
	[0,1]^{t \times n}$ by means of $u_{j,i} = E_i(x_{j,i})$.
\end{enumerate}
For the special 
case of $t$-Student distribution see Algorithms $5.2$ 
in~\cite{embrechts2001modelling}.
\end{remark}

It can be shown, that a $t$-Student multivariate distribution is an elliptical 
one~\cite{kotz2004multivariate}, its multivariate PDF reads
\begin{equation}
\mathbf{t}_{\mathbf{R}, \nu}(\x) = \frac{\Gamma(\frac{\nu + 
		n}{2})}{\Gamma(\frac{\nu}{2})\sqrt{\nu^{n} \pi 
	^{n}\det(\mathbf{R})}}\left(1+\frac{\x
	\mathbf{R}^{-1}\x^{\intercal}}{\nu}\right)^{-\frac{\nu+n}{2}}.
\end{equation}
In this case, the function $\psi$ has complicated analytical form, and is 
parametrised by an integer parameter $\nu \in\N^+$ 
~\cite{kotz2004multivariate}. Further $\mathbf{R} \in \R^{n \times n}$ is 
symmetric positive definite matrix with ones on a diagonal. The
integer parameter $\nu \in\N^+$ is an number of degrees of freedom parameter.
Obviously, the multivariate $t$-Student CDF would be
\begin{equation}
[0,1] \ni \mathbf{T}_{\nu, \mathbf{R}}(\x) = \int_{-\infty}^{x_1} \ldots 
\int_{-\infty}^{x_n} \mathbf{t}_{\nu, \mathbf{R}}(y_1, \ldots, y_n) dy_1 \ldots 
dy_n.
\end{equation}
The marginal univariate $t$-Student PDF is
\begin{equation}\label{eq::tmarg}
t_{\nu}(x) = \frac{\Gamma(\frac{\nu + 
		n}{2})}{\Gamma(\frac{\nu}{2})\sqrt{\nu 
		\pi}}\left(1+\frac{x^2}{\nu}\right)^{-\frac{\nu+n}{2}},
\end{equation}
which can be integrated to marginal CDF $T_{\nu}(x) \in [0,1]$.
Following Definition~\ref{def::elcop}, the $t$-Student 
copula $\C_{\mathbf{R}, \nu}: [0,1]^n \rightarrow [0,1]$ is
\begin{equation}
\C_{\mathbf{R}, \nu}(\uu) = \mathbf{T}_{\mathbf{R}, 
\nu}\left(T_{\nu}^{-1}(u_1), 
\ldots, T_{\nu}^{-1}(u_n)\right).
\end{equation}
It is used mainly to model financial data 
\cite{de2004measuring, szego2002measures, semenov2017portfolio}. However it is 
applied also in a machine learning, for example in colour texture 
classification \cite{el2011color} or as a kernel 
function of the Support Vector Data Description \cite{belghith2013change}. 

\begin{remark} Take bivariate $t$-Student copula parametrised by $\mathbf{R} = 
\begin{bmatrix} 1 & r \\ r & 1 \end{bmatrix}$ where $-1 < 
r < 1$ and $\nu \in\N^+$, Kendall's rank 
cross-correlation~\cite{demarta2005t} equals
\begin{equation}\label{eq::kencorr}
\tau = \frac{2}{\pi}\arcsin(r),
\end{equation}
and both `tail' dependencies~\cite{nelsen1999introduction} are
\begin{equation}\label{eq::t_tail}
\lambda_l = \lambda_u = 
2t_{\nu+1}\left(-\sqrt{\nu+1}\left(\frac{\sqrt{1-r}}{\sqrt{1+r}}\right)\right),
\end{equation}
where $t_{\nu}$ is introduced in Eq.~\eqref{eq::tmarg}.
\end{remark}
As discussed in~\cite{matthias2017simulating}, Eq.~\eqref{eq::kencorr} 
holds in general for elliptical copulas, including the Gaussian one. To see the 
Gaussian copula case take $\nu \rightarrow \infty$.
Below we present some interesting example of the $t$-Student copula.

\begin{example}\label{eg::svdfaliststudent}
	In Figure~\ref{fig::t_cop} we present a scatter-plot and univariate 
	histograms of data sampled from the $t$-Student bivariate copula 
	parametrised 
	by 
	$\nu = 
	1$ and 
	$\mathbf{R} = \begin{bmatrix} 1 & 0.75 \\ 0.75 & 1 \end{bmatrix}$ - left 
	panel, and $\nu = 
	1$ and $\mathbf{R} = \begin{bmatrix} 1 & 0 \\ 0 & 1 \end{bmatrix}$ - 
	right panel.
	\begin{figure}
			\subfigure[$r = 
			0$.\label{fig::t_not_cor}]{\includegraphics[width=0.47\textwidth]{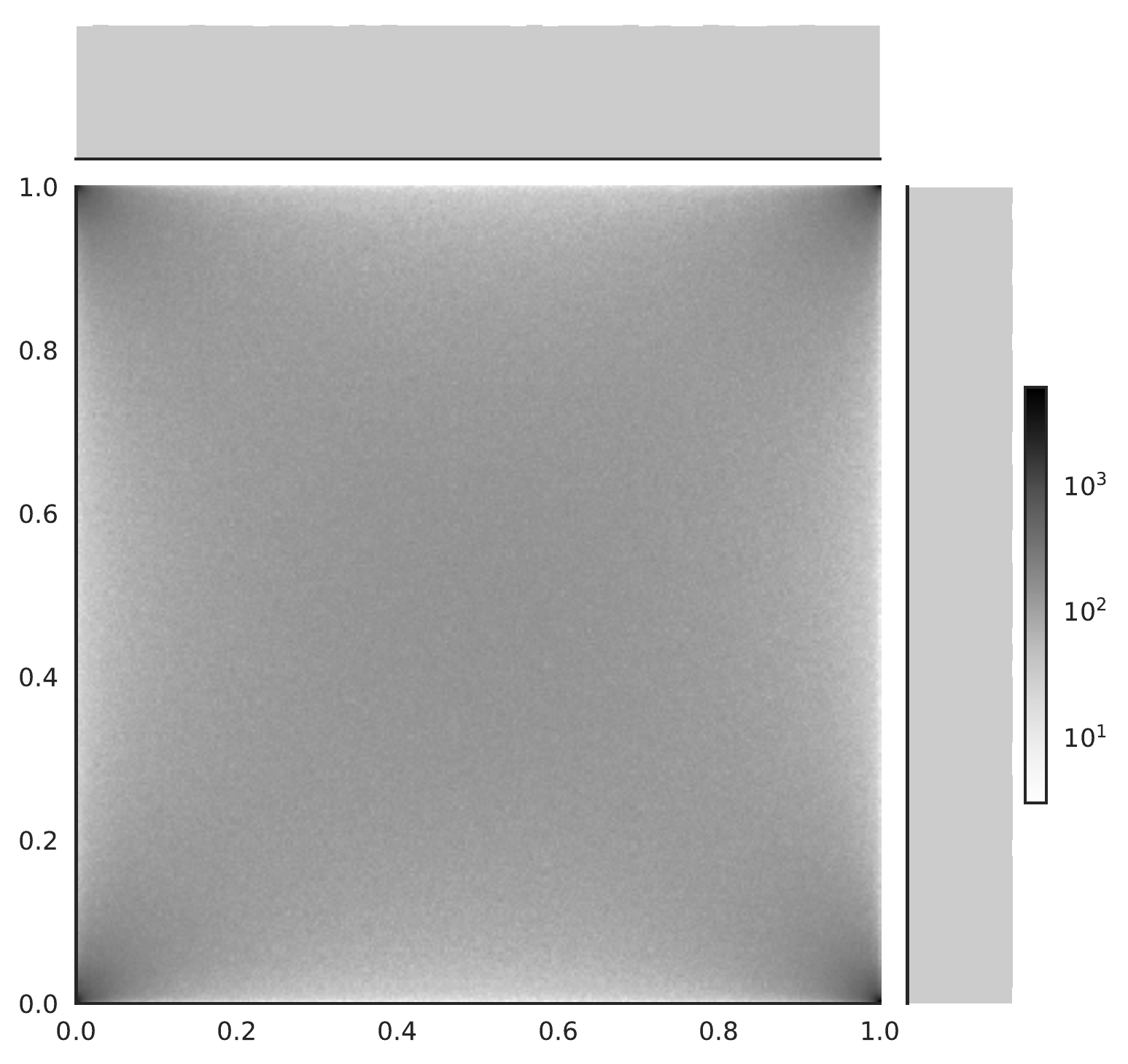}}
		\subfigure[$r = 
		0.75$\label{fig::t_cor}]{\includegraphics[width=0.47\textwidth]{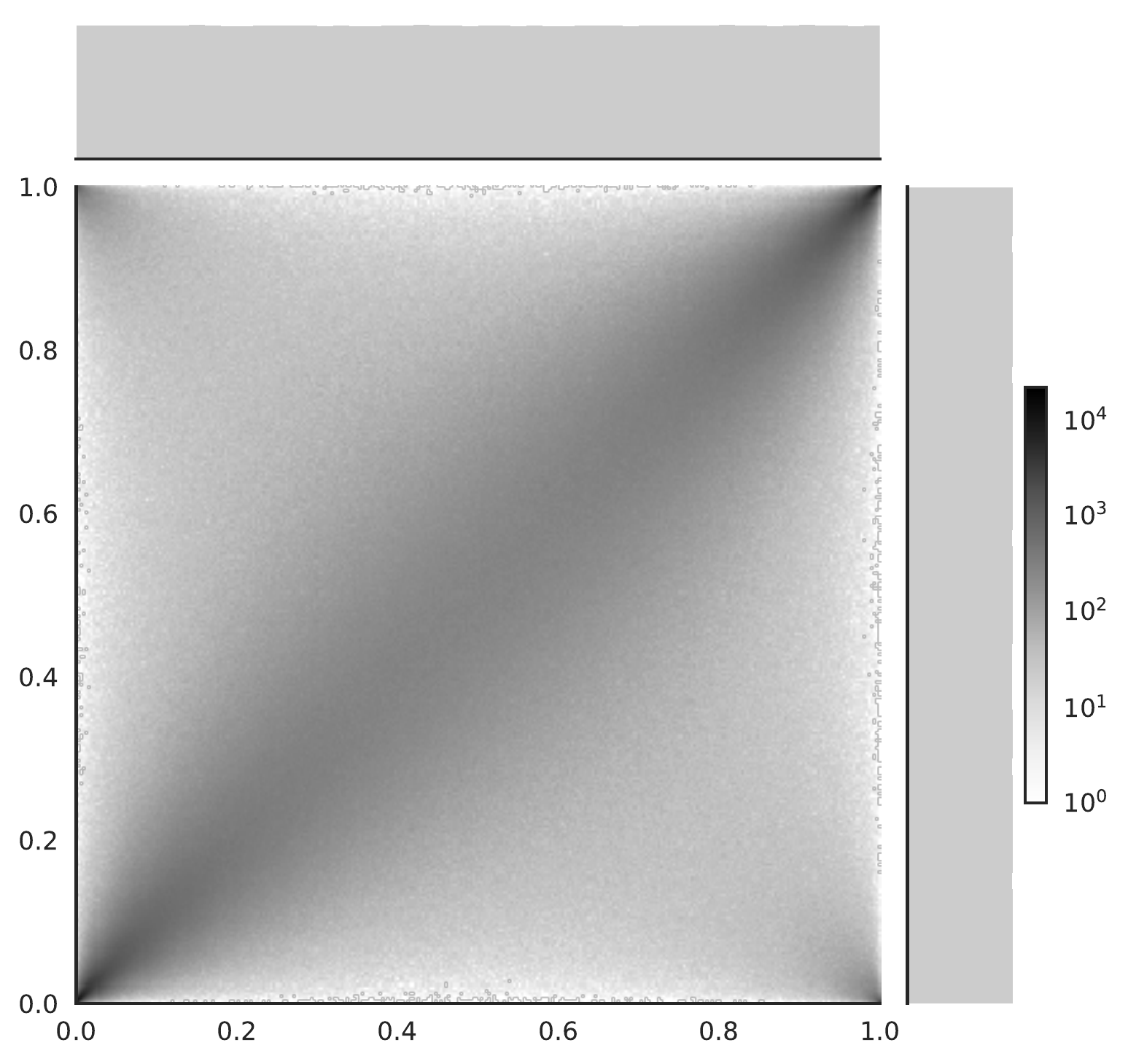}}+
		\caption{A scatter plot and histograms of data sampled from a bivariate 
			$t$-Student copula with $\nu = 1$.}\label{fig::t_cop}
	\end{figure}
Figure~\ref{fig::t_not_cor} refers to zero Kendall's, Spearman's 
and Pearson's cross-correlation between 
marginals, but a data pattern is not simple uniform, but with higher 
concentration on corners.
\end{example}

\begin{example}\label{rem::pcafailure}
	Consider bivariate random vector $\UU^{(2)}$ modelled by a $t$-Student 
	copula 
	with $\mathbf{R} = \begin{bmatrix} 1 & 0 \\ 0 & 1 \end{bmatrix}$ and $\nu = 
	1$, see Figure~\ref{fig::t_not_cor}. Referring to 
	Eq.~\eqref{eq::t_tail}, for this model we have 
	$\lambda_{l} = \lambda_{u} = \frac{1}{4}$. The 
	covariance matrix for such data model would be
 \begin{equation}
 \SSSS = \begin{bmatrix} \frac{1}{12} & 0 \\ 0 & \frac{1}{12} \end{bmatrix},
 \end{equation}
 which is diagonal. Hence, the columns of $ \1 = \begin{bmatrix} 1 & 0 \\ 0 
 & 1 \end{bmatrix}$ are its eigenvectors. From this point of view marginal 
 variables $\UU_1$ 
 and $\UU_2$ should be independent. Despite the fact, that Pearson's, 
 Spearman's and Kendall's cross-correlation measures are zero their `tail' 
 dependencies, modelling simultaneous extreme events, are not. Hence the 
 Singular Value 
 Decomposition of the 
 covariance matrix $\SSSS$ do not reduce risk of simultaneous extreme events. 
\end{example}

For further discussion on elliptical distributions and elliptical copulas see 
Chapter $4$ in~\cite{matthias2017simulating}, where elliptical distributions 
are discussed in Section $4.2$, while elliptical copulas in Section $4.5$.

\section{Upper and lower limit, Fr\'echet families}\label{sec::frechcops}

The Fr\'echet copula's family is derived from limits imposed on copulas 
functions. The $n$-dimensional copula function can be limited from above and 
from below by following Fr\'echet-Hoeffding copula bounds 
\cite{nelsen1999introduction},
\begin{equation}\label{eq::coplims}
\max \left(1-n+\sum_{i = 1}^n u_i, 0\right) \leq \C(u_1, 
\ldots, 
u_n) \leq 
\min\left(u_1, \ldots, 
u_n\right).
\end{equation}

\subsection{Maximal copula}
The upper limit,
\begin{equation}
\C_{\max}(u_1, \ldots, u_n) = \min\left(u_1, \ldots, u_n\right),
\end{equation}
is a copula (fulfils conditions in Definition~\ref{d::copdef}) for each $n$. 
It corresponds to perfectly positively correlated marginal variables. It is 
easy to show that if $\UU^{(n)}$ is a random vector modelled by the maximal 
copula for each pair of marginals all bivariate measures of the 
cross-correlation, Kendall's, Spearman's or Pearson's ones equal to $1$.
Given bivariate maximal copula its tail dependencies are $\lambda_l = \lambda_r 
= 1$. Finally observe, that since $0 \leq u_i, u_j \leq 1$, a 
bivariate maximal sub-copula 
is given by
\begin{equation}
	\C_{\max}^{(s)}(u_i, u_j) = \min\left(\{1,\ldots, u_i, \ldots, 1, \ldots 
	u_j, \ldots\, 1\right) \}= \min\left(u_i, u_j \right).
\end{equation}

\subsection{Minimal copula}\label{s::mincopula}

The lower limit form Eq.~\eqref{eq::coplims},
\begin{equation}\label{eq::lowerlim}
\max \left(1-n+\sum_{i = 1}^n u_i, 0  \right) = \mathbf{L}(u_1, \ldots, u_n),
\end{equation}
in only a copula in the bivariate case, this is called the minimal copula
\begin{equation}\label{eq::mincopula}
\C_{\min}(u_1, u_2) = \max \left(u_1 + u_2 -1, 0  \right).
\end{equation}
For $n > 2$, the condition $(3)$ form Definition~\ref{d::copdef} in not 
fulfilled \cite{nelsen1999introduction}. To show it intuitively, observe 
that in bivariate case $\C_{\min}(u_1, u_2)$ corresponds to 
perfectly negative cross-correlated marginals with elements $u_{t,1}$ and 
$u_{t,2} = 1-u_{t,1}$. Obviously, it is not 
possible to have more that $2$ univariate real valued marginals that are 
all pairwise perfectly negatively cross-correlated. As such, it is easy to 
show, that 
given a bivariate minimal copula Kendall's, Spearman's and Pearson's  
cross-correlation coefficients equal to $-1$. Finally, both 
`tail' dependencies of the minimal copula are $\lambda_u = \lambda_l = 0$.
Both maximal and minimal copulas are examples of the wider Fr\'echet family 
of copulas.
\subsection{Independent copula and a Fr\'echet family copula}

The last component to the Fr\'echet family of copulas is the independent copula.

\begin{definition}
	Independent copula of the $n$-variate random vector $\UU^{(n)}$ is simply a 
	product of arguments,
	\begin{equation}
		 \C_{\perp}(u_1, \ldots, u_n) = \prod_{i=1}^n u_i.
	\end{equation}
	The independent copula correspond to the random vector with all independent 
	marginals. It is easy to show that $\C_{\perp}$ fulfils all conditions 
	from Definition~\ref{d::copdef}. 
\end{definition}
Obviously, a bivariate sub-copula would be
	\begin{equation}
	\C_{\perp}^{(s)}(u_i, u_j) = 1 \cdot \ldots \cdot 1 \cdot u_i \cdot 1 
	\cdot \ldots \cdot 1
	\cdot u_j \cdot 1 \ldots  = u_i \cdot u_j.
\end{equation}
Given a independent copula all measures of the cross-correlation 
(\eg~Kendall's, Spearman's, Pearson's or `tail' dependencies) equal to $0$.

Given maximal, minimal and independent copulas we can integrate them 
into the Fr\'echet copulas. In the case of $n=2$ we can use $\C_{\max}$, 
$\C_{\perp}$ and $\C_{\min}$, while in the case of $n > 2$ we can use 
$\C_{\max}$ and $\C_{\perp}$, since $\C_{\min}$ is not a copula any-more.

\begin{definition}\label{def::frech2}
	The $2$-parameters Fr\'echet copula \cite{nelsen1999introduction} of a 
	bivariate random vector $\UU^{(2)}$ is given by
	\begin{equation}
	\C_{\alpha, \beta}(u_1, u_2) = \alpha \C_{\max}(u_1, u_2) + 
	\beta\C_{\min}(u_1, u_2) + \left(1-\alpha - \beta\right)\C_{\perp}(u_1, 
	u_2),
	\end{equation}
	where $0 \leq \alpha \leq 1$, $0 \leq \beta \leq 1$ and $\alpha + \beta 
	\leq 1$. These parameters indicate portions of maximal, minimal and 
	independent copula accordingly. It is easy to show that $\C_{\alpha, 
	\beta}$ is a copula by Definition~\ref{d::copdef}.
\end{definition}
 
The impact to `tail' dependencies of $\C_{\alpha, 
\beta}$ will come from the 
maximal copula only, and thus
\begin{equation}\label{eq::frechlambdas}
\lambda_l = \lambda_u = \alpha.
\end{equation}
Analogically, one can show that the Spearman's cross-correlation of the
$2$-parameters Fr\'echet copula is
\begin{equation}\label{eq::frachsp}
\rho = \alpha - \beta.
\end{equation}

Referring to Eq.~\eqref{eq::frechlambdas} and 
Eq.~\eqref{eq::frachsp} there may appear a suggestion that by varying 
parameters $\alpha$ and $\beta$ we can easily fit bivariate Fr\'echet copula 
to given Spearman cross-correlation and `tail' dependencies, if lower and 
upper `tail' dependencies are expected to be similar. However formally, a 
$2$-parameters Fr\'echet copula refers to the simple stochastic process where 
two elements of a bivariate realisation are either fully correlated, or 
uncorrelated 
(independent), or fully 
anti-correlated. Based on these, the simple procedure of sampling 
$2$-parameters Fr\'echet copula is presented in 
Algorithm~\ref{alg::frechet2dgen}. However, stochastic processes of 
real life 
data modelling are more complex in most cases, hence we will discuss other 
families of copulas in next section. Now we will discuss the $n$-variate 
generalisation of the Fr\'echet copula, where its flexibility is reduced, since 
we have to drop the minimal copula. Hence, the $n$-variate Fr\'echet copula 
refers to even simpler stochastic process where all elements of a realisation 
are either fully correlated, or uncorrelated (independent), we 
present a sampling procedure in Algorithm~\ref{alg::frechetndgen}.

\begin{algorithm}[t]
	\caption{Sample bivariate Fr\'echet copula}
		\label{alg::frechet2dgen}
	\begin{algorithmic}[1]	
		\State \textbf{Input}: $t$ -- Int, number of samples: 
		$0 \leq \alpha, \beta \leq 1$ - parameters: $\alpha + \beta \leq 1$. 	
		\State \textbf{Output:} $\U \in [0,1]^{t\times 2}$ -- samples. 
		\Function{frechetcopulagen}{$t$, $\alpha$, $\beta$}
		\For {$j \gets 1 \textrm{ to } t$}
		\State sample $v \sim \text{Uniform}([0,1])$ 
		\If{$v \leq \alpha$} 
		\State \text{sample} $u_{j,1} \sim \text{Uniform}([0,1])$
		\State $u_{j,2} = u_{j,1}$
		\ElsIf{$\alpha < v \leq \alpha+\beta$} \State \text{sample} $u_{j,1} 
		\sim \text{Uniform}([0,1])$
		\State $u_{j,2} = 1- u_{j,1}$
		\Else
		\State \text{Sample independently} $u_{j,1} \sim \text{Uniform}([0,1]), 
		u_{j,2} \sim \text{Uniform}([0,1])$
		\EndIf
		\EndFor
		\State\Return $\U$
		\EndFunction 
	\end{algorithmic}
\end{algorithm}	

\begin{algorithm}[t]
	\caption{Sample $n$-variate Fr\'echet copula}
	\label{alg::frechetndgen}
	\begin{algorithmic}[1]	
		\State \textbf{Input}: $t$ -- Int, number of samples, $n$ -- Int, 
		number of marginals, $0 \leq \alpha 
		\leq 1$ -- parameter. 	
		\State \textbf{Output:} $\U \in [0,1]^{t\times n}$ - samples. 
		\Function{frechetcopulagen}{$t$, $n$, $\alpha$}
		\For {$j \gets 1 \textrm{ to } t$}
		\State sample $v \sim \text{Uniform}([0,1])$ 
		\If{$v \leq \alpha$} 
		\State \text{sample} $u_{j,1} \sim \text{Uniform}([0,1])$
		\State $u_{j,n} = u_{j,n-1} =  \ldots = u_{j,2} = u_{j,1}$
		\Else
		\State \text{Sample indep.} $u_{j,1} \sim \text{Uniform}([0,1]), 
		\ldots,  
		u_{j,n} \sim \text{Uniform}([0,1])$
		\EndIf
		\EndFor
		\State\Return $\U$
		\EndFunction 
	\end{algorithmic}
\end{algorithm}

\begin{definition}\label{def::frechcop}
	The $1$-parameter Fr\'echet copula is given by
	\begin{equation}
	\C_{\alpha}(u_1,\ldots,  u_n) = \alpha \C_{\max}(u_1, \ldots,  u_n) +  
	\left(1-\alpha\right)\C_{\perp}(u_1, \ldots, 
	u_n),
	\end{equation}
	where $0 \leq \alpha \leq 1$.
	Here $\alpha$ is a portion of the maximal copula, while $1-\alpha$ is a 
	portion of the independent copula.
\end{definition}
It is easy to show, that both Spearman's cross-correlation and `tail' 
dependencies for each pair of marginals modelled by the $1$-parameter Fr\'echet 
copula are
\begin{equation}\label{eq::frechetrho}
\rho = \lambda_l = \lambda_u = \alpha.
\end{equation}

\begin{example}\label{rem::pcafailsagain}
Let us take a $2$-parameter bivariate Fr\'echet copula with $\alpha = \beta$, 
fulfilling $\alpha + \beta \leq 1$. Referring to 
Eq.~\eqref{eq::frachsp} the Spearman's cross-correlation between marginals 
will be $0$. Analogically to Example~\ref{rem::pcafailure} the correlation 
matrix will be
\begin{equation}
\SSSS = \begin{bmatrix} \frac{1}{12} & 0 \\ 0 & \frac{1}{12} \end{bmatrix},
\end{equation}
which eigenvectors being $[1 \ 0]^{\intercal}$, $[0 \ 1]^{\intercal}$. Hence, 
analogically to 
Example~\ref{rem::pcafailure} the Singular Value Decomposition 
of the 
covariance matrix $\SSSS$ would not reduce the risk of simultaneous extreme 
events. 
\end{example}

\section{Archimedean copulas}\label{sec::archcops}

The motivation for introducing Archimedean copulas family comes from their wide 
application in real data analysis concerning: modelling financial data (shares 
prices, credit assets and risk analysis) \cite{embrechts2001modelling, 
	naifar2011modelling, domino2014use, domino2015use}, modelling computer 
	networks traffic \cite{dong2015copula, 
		dong2017copula}, wireless communication analysis
	\cite{peters2014communications}, hydrological research 
	\cite{zhang2012application, 
	tsakiris2015flood, cong2012interdependence, domino2014meteo, ozga2016snow}, 
	signal processing 
\cite{zeng2014copulas}, and neuroimage data analysis 
\cite{silva2014statistically}. Following \cite{mcneil2009multivariate}
let us move to the following definition. 
\begin{definition}\label{def::archgen}
The Archimedean copula generator $\psi_\theta:[0,\infty)\to[0,1]$ is the
continuous function 
 parametrised by $\theta$ such that
 \begin{enumerate}
  \item $\psi_\theta(0) = 1$,
  \item $\psi_\theta(\infty) = 0$,
  \item $\psi_\theta$ is strictly decreasing on $[0, \inf\{v:\psi_\theta(v) = 
  0\}]$.
 \end{enumerate}
We can define an inverse $\psi_\theta^{-1}$ on $(0,1]$. 
Following~\cite{mcneil2009multivariate} we take $\psi^{-1}(0) = 
\inf\{v:\psi_\theta(v) = 0\}$ introducing a pseudo-inverse defined on $[0,1]$. 
Further one can show that 
$\psi^{-1}_{\theta}(1) = 0$.
\end{definition}

\begin{definition}\label{def::bivarchcop}
Given the convex Archimedean copula generator as in 
Definition~\ref{def::archgen}, the 
bivariate Archimedean copula is given by \cite{mcneil2009multivariate}
 \begin{equation}\label{eq::bivarchcop}
 [0,1] \ni \C_{\psi_{\theta}}(u_1, u_2) = \psi_{\theta}\left(
 \psi_{\theta}^{-1}(u_1) +\psi_{\theta}^{-1}(u_2)  \right).
 \end{equation}
 It can be shown, that such function fulfils conditions of a copula from 
 Definition~\ref{d::copdef}.
\end{definition}
 The $n$-variate generalisation needs further 
conditions to be fulfilled by the generator function 
\cite{mcneil2009multivariate}.
\begin{definition}\label{def::nvarchcop}
If the Archimedean copula generator 
$\psi_\theta:[0,\infty)\to[0,1]$ fulfils conditions in 
Definition~\ref{def::archgen}, and additionally is 
$n$-monotone on $[0,\infty)$ in a sense of Definition $2.3$ in 
\cite{mcneil2009multivariate}, \ie
\begin{equation}\label{eq::mcneecondition}
\forall_{v \in [0, \infty)}(-1)^k \frac{d^k}{d v^k} \psi(v) \geq 0 \ 
\text{where} \ 
k = 0,1,\ldots n-2,
\end{equation}
and
\begin{equation}\label{eq::mcneecondition2}
(-1)^{n-2} \frac{d^{n-2}}{d v^{n-2}} \psi(v) \ 
\text{is nonincreasing and convex on } [0, \infty),
\end{equation}
it defines the $n$-variate Archimedean copula
 \begin{equation}\label{eq::archcop}
 [0,1] \ni \C_{\psi_{\theta}}(\uu) = \psi_{\theta}\left (\sum_{i=1}^{n} 
 \psi_{\theta}^{-1}(u_i) \right ).
 \end{equation}
\end{definition}
In a case of $n = 2$ the $n$-monotone condition from
Definition~\ref{def::nvarchcop} gives convexity conditions in 
Definition~\ref{def::bivarchcop}. Furthermore, it can 
be shown that if the function is $n$-monotone, it is $(n-1)$-monotone as well, 
however the opposite may not be true. Hence if 
$\psi_{\theta}$ is the generator of the $n$-variate copula is a 
generator of the $(n-1)$-variate copula~\cite{mcneil2009multivariate}, while
the opposite may not be true.

 \begin{remark}\label{rem::subcopgen}
 	To show a necessity of the $d$-monotone condition, 
 	following~\cite{mcneil2009multivariate} let us 
 	take a generator 
 	\begin{equation}\label{eq::genmincop}
 	\psi(v) = \max(1-v, 0),
 	\end{equation}
 	fulfiling all conditions in 
 	Definition~\ref{def::bivarchcop} but is not $n$-monotone and $n 
 	> 
 	2$. Its pseudo-inverse is $\psi^{-1}(x) = 1-x$. From 
 	Eq.~\eqref{eq::archcop} we have Eq.~\eqref{eq::lowerlim}, that 
 	represents the 
 	minimal 
 	copula that is well defined only in the bivariate case. 
 \end{remark}
\begin{remark}
Following Definition~\ref{d::subcop}, using $\psi^{-1}_{\theta}(1) = 0$, and 
taking $n$-variate Archimedean copula, we can 
introduce a sub-copula, by setting to $1$ those $u_i$ arguments we 
are not interested in. Such sub-copula would be a copula as well.
\end{remark}
For bivariate Archimedean 
copula the Kendall's 
cross-correlation can be computed in the following manner 
\cite{genest1986copules}
\begin{equation}
\tau = 1-4 \int_{0}^{\psi_{\theta}^{-1}(0)} v \left( \frac{d 
\psi_{\theta}(v)}{d v} 
\right)^2 dv.
\end{equation}

\subsection{Archimedean copulas examples}

In this subsection we discuss well known one parameter Archimedean copulas, 
defined by following generator functions~\cite{mcneil2008sampling, 
hofert2008sampling}.
\begin{enumerate}
	\item Gumbel copula with generator
	\begin{equation}
		 \psi_{\theta}(v) = \exp \left(-t^{\frac{1}{\theta}} \right) \ \ 
		 \text{and} \ \  \psi_{\theta}^{-1}(x) = \left( - \log(x) 
		 \right)^{\theta}.
	\end{equation}
	\item Clayton copula~\cite{clayton1978model} with generator
	\begin{equation}
		\psi_{\theta}(v) = \max\left((1+\theta v)^{{-1}/{\theta}}, 0 \right) \ 
		\ 
		\text{and} \ 
		\psi_{\theta}^{-1}(x) = \frac{x^{-\theta}-1}{\theta}.
	\end{equation}
	\item Frank copula with generator
	\begin{equation}
		\psi_{\theta}(v) = -\frac{1}{\theta} \log \left(1 + e^{-v} 
		\left(e^{-\theta} -1 \right) \right) \ 
		\text{and} \ 
		\psi_{\theta}^{-1}(x) = -\log \left(\frac{e^{-\theta 
		x}-1}{e^{-\theta}-1} \right).
	\end{equation}
	\item Ali-Mikhail-Haq (AMH) copula \cite{ali1978class} with generator
	\begin{equation}
		\psi_{\theta}(v) = \frac{1-\theta}{\exp(v)-\theta}  \ 
		\text{and} \ \  \psi_{\theta}^{-1}(x) = 
		\log\left(\frac{1-\theta(1-x)}{x} \right).
	\end{equation}
\end{enumerate}
There are some limitations on copula parameter $\theta$ value necessary to 
fulfil conditions in
Definition~\ref{def::bivarchcop} in bivariate case, or 
in Definition~\ref{def::nvarchcop} in $n$-variate
case, where $n > 2$. In the bivariate case these limitations are presented in 
Table~\ref{tab::thetasbiv}, see 
\cite{nelsen1999introduction} or~\cite{kumar2010probability} for the AMH copula.
For some measures of cross-correlations between marginals given these copulas, 
see 
Table~\ref{tab::depbivarch}. 

\begin{table}[h]
	\centering
	\begin{tabular}{ccccc}
		 & Gumbel & Clayton & Frank & AMH \\ \hline
		$\theta \in$ & $[1, \infty)$ & $[-1, 0)\cup (0, \infty)$ & $(-\infty, 
		0) \cup (0, 
		\infty)$ & $[-1,1]$ \\ \hline
	\end{tabular}
	\caption{Parameters limitations for bivariate Archimedean 
	copulas.}\label{tab::thetasbiv}
\end{table}

\begin{table}[h]
	\centering
	\begin{tabular}{ccccc}
		& Gumbel & Clayton & Frank  & AMH  \\ \hline
		$\tau$ & $1-\frac{1}{\theta}$ & $\frac{\theta}{\theta + 2}$ & 
		$1-\frac{4}{\theta}\left(1-\frac{1}{\theta}\int_{0}^{\theta}\frac{t}{e^t-1}dt\right)$
		 & $\frac{3 \theta -2}{3 \theta} - \frac{2(1-\theta)^2\log(1-\theta)}{3 
		 \theta^2}$ \\ 
		 \hline
		$\lambda_u$ & $2-2^{\frac{1}{\theta}}$ & $0$ & $0$ & $0$ \\ \hline
		$\lambda_l$ & $0$ & $2^{-\frac{1}{\theta}}$ & $0$ & $= \begin{cases}
		0.5 &\text{ for } \theta = 1 \\ 0 &\text{ for } \theta < 1
		\end{cases}$ \\ \hline
	\end{tabular}
	\caption{Kendall's cross-correlation and `tail' dependencies for 
	bivariate Archimedean copulas.}\label{tab::depbivarch}
\end{table}

In $n$-variate case $\theta$ limitations are stronger and 
$n$ dependent. In Table~\ref{tab::thetasnv} we present limitations on $\theta$ 
sufficient for arbitrary $n$~\cite{hofert2008sampling}. It was 
shown, that if these limitations are fulfilled we can sample arbitrary 
$n$-variate copula using the Marshall-Olkin 
algorithm~\cite{marshall1988families}, we require here 
the Archimedean copula generator to be completely monotone,~\ie~to be 
$n$-monotone according to Definition~\ref{def::nvarchcop} for 
each~$n$~\cite{hofert2008sampling}.
\begin{example}\label{eg::claytonnvar}
Limitations on $\theta$ parameter presented in Table~\ref{tab::thetasnv} 
may be weaken in some cases, still leading to properly defined Archimedean 
copula. For example for $n$-variate Clayton copula to fulfil conditions of 
Definition~\ref{def::nvarchcop} we require $\theta \in \left[\frac{-1}{n-1}, 
0\right) \cup (0, \infty)$ \cite{mcneil2009multivariate}. However in this 
case we can not use the Marshall-Olkin sampling algorithm. 
\end{example}

\begin{table}[h]
	\centering
	\begin{tabular}{ccccc}
		& Gumbel & Clayton & Frank  & AMH  \\ \hline
		$\theta \in$ & $[1, \infty)$ & $(0, \infty)$ & $(0, 
		\infty)$ & $[0,1)$ \\ \hline
	\end{tabular}
	\caption{Parameters limitations for multivariate Archimedean 
	copulas.}\label{tab::thetasnv}
\end{table}

\subsection{Sampling Archimedean copulas}\label{sec::samplarch}

A basic method of sampling data form $n$-variate Archimedean copula 
is to use the Marshall-Olkin algorithm~\cite{marshall1988families, 
hofert2008sampling}. It assumes, that 
Archimedean copula generator $\psi_{\theta}$ is the 
Laplace–Stieltjes transform of certain continuous univariate CDF function 
$\F(v)$ or discrete 
probability mass function $\f_k$. It was shown, 
that this assumption holds, if 
$\psi_{\theta}$ is 
completely monotone according to 
Definition~\ref{def::nvarchcop}, and $\psi_{\theta}(0) = 
1$, see \cite{feller1971william}.
For Gumbel or Clayton copula we have the continuous case
\begin{equation}\label{eq::intlstrphicont}
\psi_{\theta}(s) = \int_{-\infty}^{\infty} e^{-s v} d \F(v) = 
\mathbb{E}\left( e^{-s 
\mathfrak{V}} \right) \text{ for } s \in [0, \infty),
\end{equation}
where $\mathbb{E}$ is the expecting value operator and $\mathfrak{V}$ the 
corresponding random variable. In the Clayton 
copula case, $\F(v)$ 
is the CDF of 
the
Gamma 
distribution $\Gamma_{(\frac{1}{\theta}, 1)}$, while in the Gumbel copula case 
$\F(v)$ is the CDF of the
 L\'evy  general distribution ${F_S}_{\left(\frac{1}{\theta}, 1, 
 \left(\cos\left( \frac{\pi}{2 \theta} \right) \right)^{\theta}, 0\right)}$, 
 see Remark~\ref{rem::stabledist} and Table 
$1$ in~\cite{hofert2008sampling}.
 For the Frank or AMH copulas we have a discrete case, 
\begin{equation}\label{eq::intlstrphidisc}
\psi_{\theta}(s) = \sum_{k=0}^{\infty} \f_k e^{-v_k s} \text{ where } 0 <
v_0 < 
v_1 < \ldots \text{ and } s \in [0, \infty).
\end{equation}
The probability mass function $\f_k$ is the geometric 
series expansion for AMH copula, and the logarithmic 
series expansion for Frank copula, see Table $1$ 
in~\cite{hofert2008sampling}.
After performing the difficult step of sampling the inverse Laplace-Stieltjes 
transform 
of the Archimedean copula generator, the Marshall-Olkin 
algorithm~\cite{marshall1988families, 
hofert2008sampling} is straightforward, see 
Algorithm~\ref{alg:arch_sampler}. For the proof see Algorithm $2.1$ and its 
discussion see~\cite{matthias2017simulating}. In~\cite{matthias2017simulating} 
the Marshall-Olkin algorithm is 
presented in slightly different way and compared with the analogical algorithm  
that uses the Williamson $n$-transform instead of the inverse Laplace-Stieltjes 
transform, see Algorithm $2.3$ therein.

\begin{algorithm}[t]
\caption{Sampling an $n$-variate Archimedean copula.}
\label{alg:arch_sampler}
\begin{algorithmic}[1]	
	\State \textbf{Input}: $t$ -- Int, number of samples, $n$ -- Int, 
	number of marginals, $\psi_{\theta}$ -- completely 
	monotone~\cite{hofert2008sampling} Archimedean generator function.
	\State \textbf{Output:} $\U \in [0,1]^{t \times n}$ - samples. 
	\Function{archcopulagen}{$t$, $n$, $\psi_{\theta}$}
	\For{$j=1,\ldots,t$}
	\State sample $v \sim \F$ fulfilling $\psi_{\theta}(s) = \begin{cases} 
	\int_{-\infty}^{\infty} 
	e^{-s v} d \F(v) \\ \sum_{k=0}^{\infty} \f_k e^{-x_k s}\end{cases}$
	\State Sample independently $x_1 \sim \text{Uniform}([0,1]), \ldots, x_n 
	\sim \text{Uniform}([0,1]).$ 
	\For{$i=1,\ldots,n$}
	\State $u_{j, i} \leftarrow \psi_\theta\left ( \frac{-\log 
		(x_{i})}{v}\right ) $
	\EndFor	
	\EndFor	
	\State\Return $\U$	
	\EndFunction 
\end{algorithmic}
\end{algorithm}

As mentioned before, for valid ranges of copulas parameters for the 
Marshall-Olkin algorithm, 
see Table~\ref{tab::thetasnv}. This may limit the vast part of the 
parameters range, specially in the bivariate case of the Clayton, Frank and 
AMH copulas, see Table~\ref{tab::thetasbiv}. 
To overcome this problem one can use more straightforward sampling method that 
is relatively simple for bivariate copulas \cite{nelsen1999introduction}, that 
can be represented as
\begin{equation}\label{eq::bivcopdecomp}
\C(u_1, u_2) = P(\UU_1 \leq u_1 \wedge \UU_2 \leq u_2) = P(\UU_2 \leq u_2 | 
\UU_1 \leq u_1) P(\UU_1 \leq u_1) = \C(u_2 | u_1) u_1
\end{equation}
since $u_1$ is uniformly distributed on 
$[0,1]$ and one can write $ P(\UU_1 \leq u_1) = u_1$. Thanks to the conditional 
copula
\begin{equation}
\C(u_2 | u_1) = \frac{\partial \C(u_1, u_2)}{\partial u_1},
\end{equation} 
and given a product in Eq.~\eqref{eq::bivcopdecomp} we can now sample 
$u_1$ 
from 
$\text{Unirom}\left([0,1]\right)$ and $u_2$ from $\C(u_2 | u_1)$. This 
algorithm was 
introduced in practice in~\cite{nelsen1999introduction} for Clayton and Frank 
copulas, and in~\cite{kumar2010probability} for the AMH copula, see also 
Algorithms $2.5$ and $2.6$ in~\cite{matthias2017simulating}. It is 
presented in the book as Algorithm~\ref{alg:arch_2dsampler}. 

\begin{algorithm}[t]
	\caption{Sampling bivariate Archimedean copula for $n = 2$.}
	\label{alg:arch_2dsampler}
	\begin{algorithmic}[1]	
		\State \textbf{Input}: $t$ -- Int, $\psi_{\theta}$ -- convex 
		Archimedean 
		copula generator.
		\State \textbf{Output:} $\U \in [0,1]^{t \times 2}$ -- samples. 
		\Function{archcopulagen}{$t$, $\psi_{\theta}$}
		\For{$j=1,\ldots,t$}
		\State sample $u_{j, 1} \sim \text{Uniform}([0,1])$
		\State sample $u_{j,2} \sim \C_{\psi_{\theta}}(u_2 | 
		u_1)$ 
		\EndFor	
		\State\Return $\U$	
		\EndFunction 
	\end{algorithmic}
\end{algorithm}

\begin{example}
	Suppose we want to sample the bivariate Frank copula
	\begin{equation}
	\C_F(u_1, u_2) = 
	-\frac{1}{\theta}\log\left(1+\frac{(e^{-u_1\theta}-1)(e^{-u_2\theta}-1)}{e^{-\theta}-1}
	 \right).
	\end{equation}
	We sample $u_1 \sim \text{Uniform}(0,1)$ in the first step. We use $u_1$ to 
	determine $\C(u_2 | u_1)$ and sample it in the second step. We can treat 
	$\C(u_2 | u_1)$ as an univariate CDF of $u_2$ parametrised by $u_1$. Its 
	inverse (a quantile function) would be
	\begin{equation}
	Q(w)\Big|_{u_1} = -\frac{1}{\theta}\log\left( 
	\frac{1+w(1-e^{-\theta})}{w(e^{-\theta u_1}-1)-e^{-\theta u_1}}
	 \right).
	\end{equation}
	Now we can sample $w \sim \text{Uniform}([0,1])$, and use $u_2 = 
	Q(w)\Big|_{u_1}$.
\end{example}
\begin{remark}
	Analogically one can sample the $3$-variate Archimedean
	copula using rule~\cite{nelsen1999introduction}
	\begin{equation}\label{eq::chaincopula}
	\C(u_1, u_2, u_3) = \C\left(u_3 |u_2, u_1\right) \C(u_2 | u_1) u_1
	\end{equation}
	however the computation of such chain especially if $n$ is large is  
	computationally complicated in comparison with the bivariate 
	case. Further the larger the $n$, the parameter range gain would be 
	smaller. For the copula to be correctly defined the generator 
	$\psi_{\theta}$ has 
	to be $n$-monotone, see Definition~\ref{def::nvarchcop}. For the 
	Algorithm~\ref{alg:arch_sampler} to be applicable the generator has to be 
	completely 
	monotone, \ie~$n$-monotone, as $n \rightarrow \infty$. 
	To demonstrate how the $n$-monotone condition reflect the parameter range 
	of the Clayton copula see example~\ref{eg::claytonnvar}, and compare it 
	with Table~\ref{tab::thetasnv} where the parameter range is given for the 
	completely monotone $\psi_{\theta}$.
\end{remark}

\begin{example}
	In Figure~\ref{fig::a_cop} we present a scatter-plot and univariate 
	histograms sampled from the Archimedean copula with the Kendall's $\tau$ 
	cross-correlation equal to $\pm 0.6$. Compare  
 the  positive 
	$\tau$ and $\theta$ case in Figure~\ref{fig::cl_1biv} that can be sampled 
	either by Algorithm~\ref{alg:arch_sampler} or by 
	Algorithm~\ref{alg:arch_2dsampler}, with the negative $\tau$ and $\theta$ 
	case in Figure~\ref{fig::cl_-0.5biv} sampled 
	by 
	Algorithm~\ref{alg:arch_2dsampler}. One can observe 
	completely different patterns, especially in the low tail region. 
	\begin{figure}
			\subfigure[$\theta = 
		3$.\label{fig::cl_1biv}]{\includegraphics[width=0.47\textwidth]{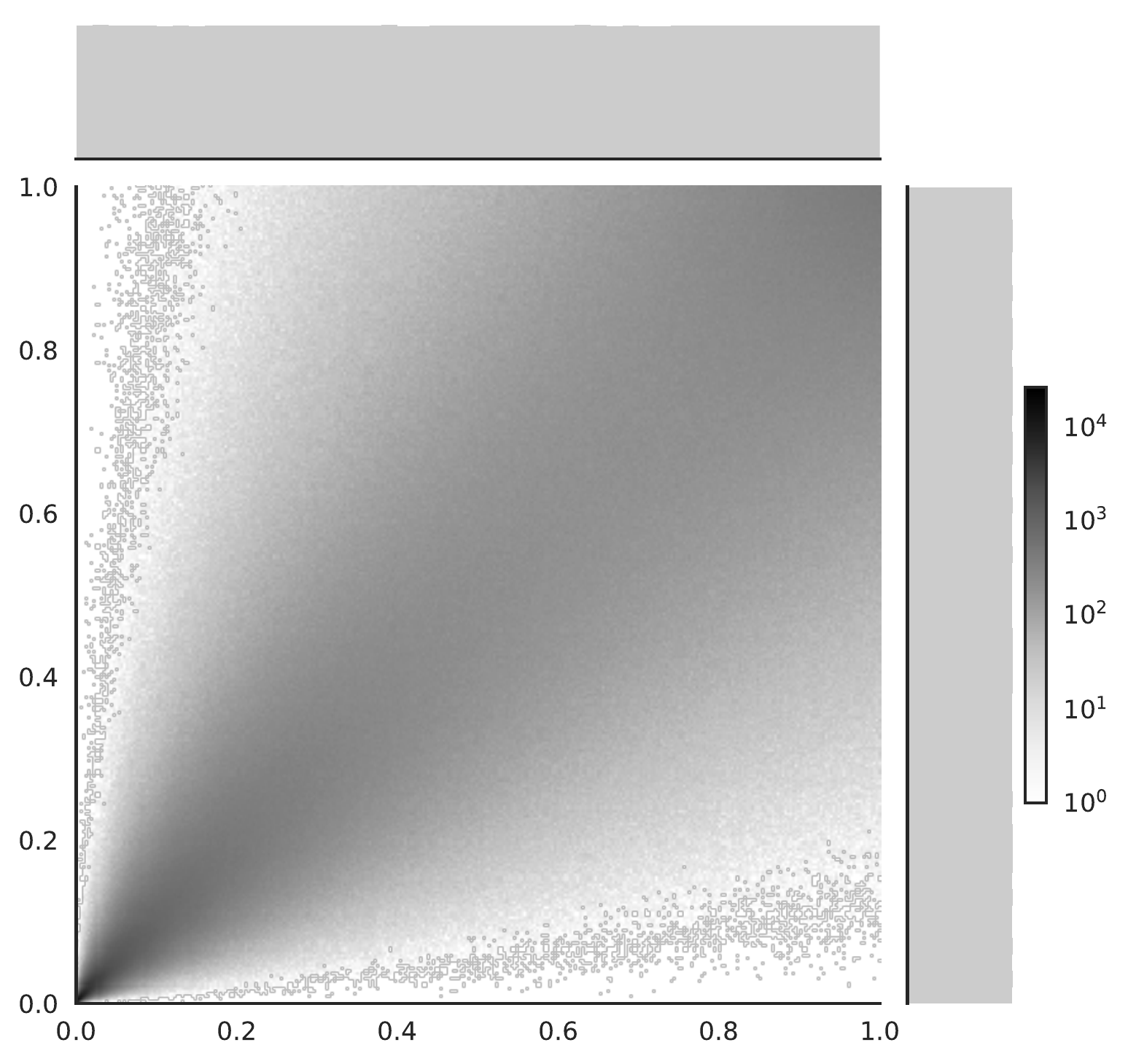}}
		\subfigure[$\theta = 
		-0.75$\label{fig::cl_-0.5biv}]{\includegraphics[width=0.47\textwidth]{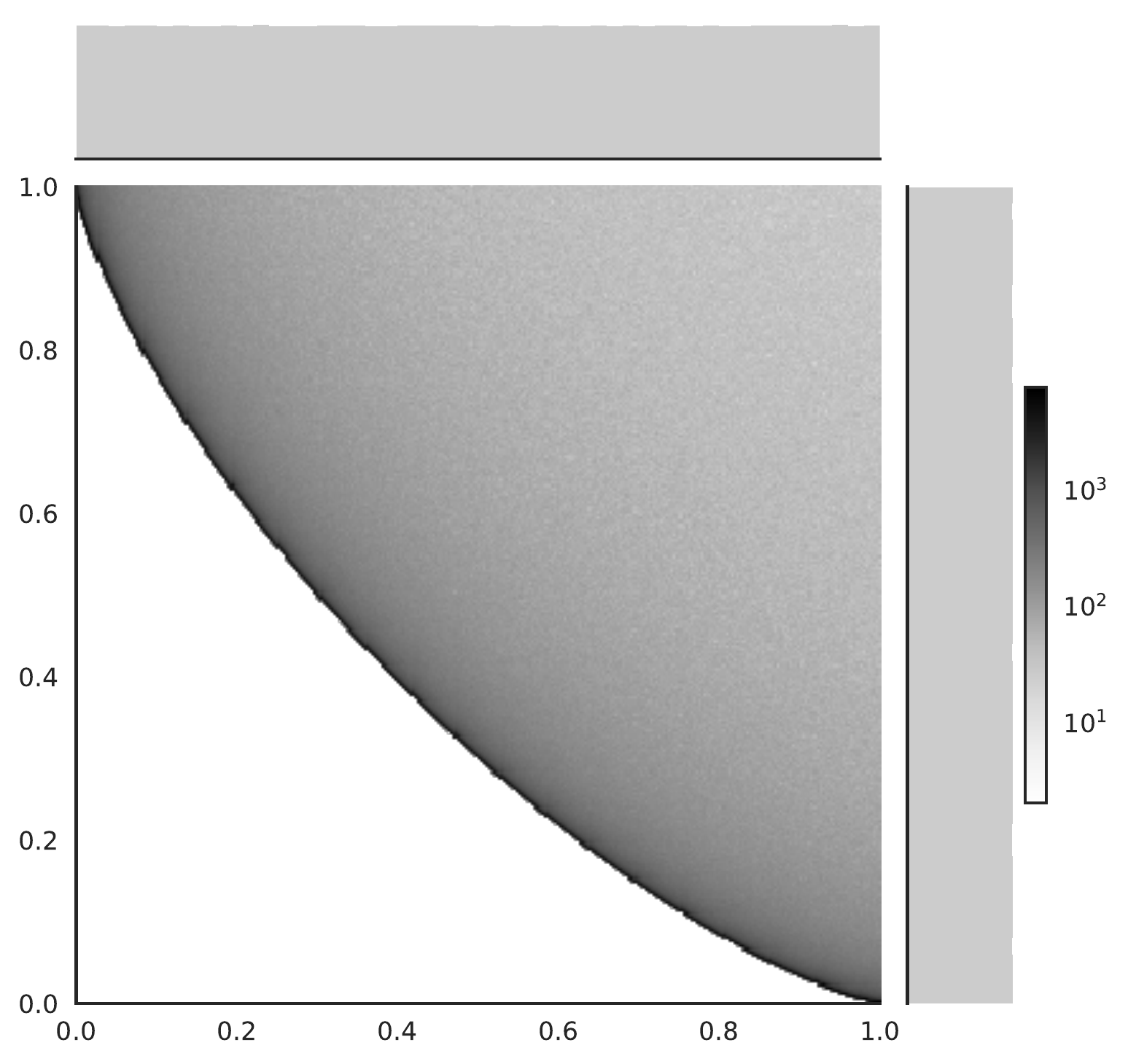}}
		\caption{A scatter plot and histograms of the Clayton 
		copula for positive and negative Kendall's cross-correlation $\tau = 
		\pm 
		0.6$.}\label{fig::a_cop}
	\end{figure}
\end{example}

\subsection{Nested Archimedean copula}\label{sec::nestedarch}

It is easy to observe, that one parameters Archimedean copulas 
such as the Gumbel, Clayton, Frank 
or AMH ones have the same cross-correlation between all pairs of marginals, 
derived from their scalar $\theta$ parameter. Unfortunately, this is far 
from modelling the real live data, especially if we deal with large number of 
marginals-$n$.
To overcome this problem, one can use nested Archimedean 
copulas \cite{hofert2013densities} that model data in such a way, that chosen 
subsets of marginals may have higher cross-correlations in comparison 
with an overall cross-correlation between marginals. 
\begin{definition}\label{d::nestedcop}
	Nested copula \cite{hofert2013densities} is defined 
	as follows
	\begin{equation}\label{eq::nescop}
	\C\left(\mathbf{u}) = \C_0(\C_{1}(\mathbf{u}_{\mathbf{r}_1}\right), \ldots, 
	\C_{\iota}\left(\mathbf{u}_{\mathbf{r}_{\iota}}\right)
	\ldots, 
	\C_{s}\left(\mathbf{u}_{\mathbf{r}_s}), \mathbf{u}_{\mathbf{r}'}\right)
	\end{equation}
	where $\mathbf{u} = (u_1, \ldots, u_n)$ is a vector of arguments of the 
	copula, and
	$\mathbf{u}_{\mathbf{r}_{\iota}} = \left(u_{r_{\iota_1}}, 
	u_{r_{\iota_2}}, \ldots\right)$ is the sub-vector of $\mathbf{r}$ indexed 
	by $\mathbf{r}_{\iota}$, that correspond to the 
	$\iota$\textsuperscript{th} child 
	copula $\C_{\iota}$. Here $\{\mathbf{r}_1, \ldots, \mathbf{r}_{\iota}, 
	\ldots, 
	\mathbf{r}_s, \mathbf{r}'\}$ is the set partition of $(1:n)$ that is a 
	vector $[1,2,\ldots, n]$. In this notation $\C_0$ is the 
	parent copula. Obviously, we can have such set partition, that $\mathbf{r}' 
	= \emptyset$ and each marginal is modelled both by the parent 
	copula and some child copula. 
\end{definition}

Interestingly, nested copula leads to different sub-copulas given by
Definition~\ref{d::subcop}, depending on a 
subset of marginals we are interested in, it will be shown in following 
examples.
\begin{example}\label{r::cumchildcop}
	Suppose we have nested copula as in 
	Definition~\ref{d::nestedcop}.
	Suppose we are interested in the subset $\mathbf{u}_{\mathbf{r}''}$ whole 
	belonging the single child copula, say 
	$\C_{\iota}$, \ie~$\mathbf{r}'' \subseteq \mathbf{r}_{\iota}$. We set 
	$u_i = 1 
	\text{ if } i \notin \mathbf{r}''$ and we have the following sub-copula
	\begin{equation}
	\C^{(s)}({\mathbf{u}_{\mathbf{r''}}}) = \C_0(1, \ldots, 1, 
	\C_{\iota}({\mathbf{u}_{\mathbf{r''}}}), 1, \ldots, 1),
	\end{equation}
	since $\C_{i \neq \iota}(1, \ldots, 1) = 1$. The 
	${\iota}$\textsuperscript{th} child copula will return some $v \in [0,1]$. 
	By Definition~\ref{d::copdef} the outcome is a $\iota$\textsuperscript{th} 
	child 
	copula if $\mathbf{r}_{\iota} = \mathbf{r}''$ or its sub-copula if 
	$\mathbf{r}_{\iota} \subset \mathbf{r}''$
	\begin{equation}
	\C^{(s)}\left({\mathbf{u}_{\mathbf{r''}}}\right) = \C_0\left(1, \ldots, 1, 
	v, 1, \ldots 
	1\right) = v = \C^{(s)}_{\iota}\left({\mathbf{u}_{\mathbf{r''}}}\right).
	\end{equation}
\end{example}

\begin{example}\label{r::cumparcop}
	Suppose we have $\mathbf{r}'' \subset (1:n)$, obtained by 
	taking one marginal from each child copula
	\begin{equation}
	\forall_{{\iota}} \  \# \left(\mathbf{r}'' \cap \mathbf{r}_{\iota} \right) 
	= 1,
	\end{equation} 
and all marginals modelled only by a 
	parent copula
	\begin{equation}
	\mathbf{r}'' \setminus \left(\cup_{\iota} \mathbf{r}_{\iota} \right) =  
	\mathbf{r'}.
	\end{equation}
	As in the previous example, we set $u_i \neq 1$ only if $i \in 
	\mathbf{r}''$.
	Each child copula will return
	\begin{equation}
	\C_{\iota}(\mathbf{u}_{\iota}) = \C_{\iota}(1,\ldots, 1, u_i, 1, \ldots, 1) 
	= u_i.
	\end{equation}
 Given these we have only the parent copula
	\begin{equation}
	\C^{(s)}\left(\mathbf{u}_{\mathbf{r''}}\right)  = 
	\C_0\left(\mathbf{u}_{\mathbf{r''}}\right).
	\end{equation}
\end{example}
\begin{remark}
In an intermediate case, there will remain some nesting structure. The use of 
the sub-copula notation allows making a simple observation, that given the 
bivariate sub-copula, the intermediate case is not possible. The sub-copula 
would refer either one child copula or to the parent copula. Given nested 
copula, any bivariate dependency measures of marginals would refer either to 
the parent copula or to the child one. A more interesting intermediate case may 
be recorded while measuring interdependency of $d$ marginals simultaneously, 
where $d > 2$. 
\end{remark}

Although an Archimedean copula is introduced by the single copula generator, 
nested Archimedean copula can be introduced by the series of generators.
\begin{definition}\label{def::nestedarch}
The nested Archimedean copula \cite{hofert2013densities}, is defined as 
\begin{equation}
\C_{\psi_{\theta_0}, \psi^1, \ldots, \psi^s}(u_1, \ldots, u_n) = 
\C_{\psi_{\theta_0}}\left(\C_{\psi^1}(\mathbf{u}_{\mathbf{r}_1}), \ldots, 
\C_{\psi^s}(\mathbf{u}_{\mathbf{r}_s}), \mathbf{u}_{\mathbf{r}'} \right),
\end{equation}
we use notation as in Definition~\ref{d::nestedcop}.
The $\C_{\psi_{\theta_0}}$ is the parents copula with generator 
$\psi_{\theta_0}$ parametrised by 
$\theta_0$ and $\C_{\psi^{\iota}}$ are children copula with generators 
$\psi^{\iota}$.
\end{definition}
There are some further generalisations of nested Archimedean copulas discussed 
for example in~\cite{hofert2013densities}, however not included in this book 
for the simplicity of presentation. One of those is the use of a higher degree 
of nesting by using the parent copula, children copulas, grandchildren copulas 
etc. Another generalisation comes from the fact that parents and children 
copulas may be from different copulas families, as long as they fulfil some 
conditions called sufficient nesting conditions \cite{hofert2011efficiently, 
hofert2013densities}. Nevertheless, for the clarity of presentation, we will 
concentrate in this book on the case where a parent and children copulas are 
from the same Archimedean family.

To sample nested Archimedean copulas one can use the advanced version of the 
Marshall-Olkin algorithm. 
We start with the parent copula generator $\psi_{\theta_0}$ parametrised by 
$\theta_0$ and sample 
$v_0$ as in Algorithm~\ref{alg:arch_sampler} (where it is called $v$). Given 
$v_0$ we have following children copulas generators each, parametrised by 
global 
$\theta_0$ and $v_0$ and individual $\theta_{\iota}$
\cite{hofert2011efficiently, mcneil2008sampling},
\begin{equation}\label{eq::nestedgen}
\psi_{v_0, \theta_0, \theta_{\iota}}(t) = \exp\left(-v_0  
{\psi_{\theta}}_0^{-1}\left(\psi_{\theta_{\iota}}(t)\right)\right) = 
\left(\exp\left( 
-{\psi_{\theta}}_0^{-1}\left(\psi_{\theta_{\iota}}(t)\right)\right)\right)^{v_0}.
\end{equation}
Such child copula generator must have completely monotone derivative 
\cite{mcneil2008sampling} in a 
sense of Definition $2.3$ 
in \cite{mcneil2009multivariate}. This condition is called the sufficient 
nesting condition.

In our case it is equivalent to
 \begin{equation}\label{eq::sufnestcond}
\forall_{{\iota}} \ \theta_0 \leq \theta_{\iota},
\end{equation}
see~\cite{hofert2011efficiently, mcneil2008sampling} for justification.
Further parameter $\theta_0$ has to fulfil conditions in 
Table~\ref{tab::thetasnv}, in a particular case of the AMH nested copula 
$\theta_{\iota} < 1$ condition is required \cite{hofert2011efficiently} 
for 
each ${\iota}$.

Following~\cite{hofert2011efficiently}, next step concern sampling an inverse 
of 
the Laplace–Stieltjes 
transform of the child copula generator given in 
Eq.~\eqref{eq::nestedgen}, analogically to the Archimedean copula 
generator in Eq.~\eqref{eq::intlstrphicont}. In the continuous case, for 
Clayton 
and Gumbel nested
copula we have
\begin{equation}\label{eq::trnestecont}
\psi_{v_0, \theta_0, \theta_{\iota}}(s) = \int_{-\infty}^{\infty} e^{-s v}   
d \F_{v_0, \theta_0, \theta_{\iota}}(v).  
\end{equation}
In the discrete case, for Frank and AMH 
nested copula we have
\begin{equation}\label{eq::trnestedisc}
\psi_{v_0, \theta_0, \theta_{\iota}}(s) = \sum_{k=0}^{\infty} {(\f_{v_0, 
\theta_0, 
\theta_{\iota}})}_k e^{-v_k s},
\end{equation}
where $0 < v_0 < v_1 < \ldots$ and $s \in [0, \infty)$.

Given these, sampling of nested Archimedean copulas is summarised in 
Algorithm~\ref{alg::nestedarchgen}. Generators of children copulas are as 
follows.
\begin{enumerate}
	\item Gumbel copula (continuous case) we sample the L\'evy general 	
	distribution parametrised by $\alpha = \frac{\theta_0}{\theta_{\iota}}$, 
	$\beta = 
	1$, $\gamma = \left(\cos\left(\frac{\pi \theta_0}{2 
		\theta_{\iota}}\right)v_0\right)^{\frac{\theta_{\iota}}{\theta_0}}$ and 
		$\delta = 
		v_0$, the 
		${F_S}_{(\alpha, \beta, \gamma, \delta)}$, see 
		Remark~\ref{rem::stabledist}. To simplify we can change parameter 
		$\gamma \rightarrow \gamma' = \left(\cos\left(\frac{\pi \theta_0}{2 
			\theta_{\iota}}\right)\right)^{\frac{\theta_{\iota}}{\theta_0}}$ 
			and 
			$\delta 
			\rightarrow \delta' = 0$, and sample 
	${F_S}_{(\alpha, \beta, \gamma' \delta')}$, and use in next step 
	$\tilde{\psi}_{\theta_0, \theta_{\iota}}(t) 
	= \exp\left(-t^{\frac{\theta_0}{\theta_{\iota}}}\right)$, instead of 
	$\psi_{v_0, \theta_0, \theta_{\iota}}(t) 
	= \exp\left(-v_0 t^{\frac{\theta_0}{\theta_{\iota}}}\right)$. Given 
	such 
	approach discussed in \cite{mcneil2008sampling} the $v_0$ 
	will 
	cancel out.
	\item Clayton copula (continuous case) we sample the exponentially 
	tilted 
	L\'evy general distribution \cite{mcneil2008sampling} with the PDF
	\begin{equation}
	f_{\left(v_0, \frac{\theta_0}{\theta_{\iota}}\right)}\left(x\right) =  
	\frac{{f_S}_{(\alpha, \beta, \gamma' \delta')}\left(x\right)}{e^{-v_0}} 
	e^{-x  v_0^{\frac{\theta_0}{\theta_{\iota}}}}.
	\end{equation} 
	\item AMH copula (discrete case): we sample $y$ from negative 
	binomial distribution with parameters $\left(v_0, 
	\frac{1-\theta_{\iota}}{1-\theta_0}\right)$ and return $v_0+y$.
	\item Frank copula (discrete case): there is a complicated representation 
	of 
	$\psi_{v_o, \theta_0, \theta_{\iota}}$ and its Laplace-Stieltjes 
	transform, hence we use a sampling scheme described step by step in
	\cite{hofert2011efficiently}. 
\end{enumerate}

\begin{algorithm}[t]
	\caption{Generate data using nested Archimedean copula 
	\cite{mcneil2008sampling}. 
		\label{alg::nestedarchgen}}
	\begin{algorithmic}[1]	
		\State \textbf{Input}: $t$ - Int, samples number; $n_1, \ldots, 
		n_{\iota}, \ldots, n_s$ - Ints marginals' numbers for each child 
		copula, 
		$\theta_1, \ldots, \theta_{\iota}, \ldots, \theta_s$ - 
		Floats children copulas parameters, $\theta_0$ - Float parent copula 
		parameter.
		\State \textbf{Output:} $\U \in [0,1]^{t \times k}$ - samples. \Comment 
		where 
		$k = \sum_{\iota} n_{\iota}$
		\Function{nestedarchcopulagen}{$t$, $n_1, \ldots, n_s$, $\theta_1, 
		\ldots, \theta_s$, $\theta_0$}
		\For {$j \gets 1 \textrm{ to } t$}
		\State sample $v_0 \sim \F$ \Comment{see line $5$ in 
		Algorithm~\ref{alg:arch_sampler}}
		\State sample indep.  $x_1 \sim 
		\text{Uniform}([0,1]), \ldots, x_k \sim 
		\text{Uniform}([0,1])$ 
		
		\For {${\iota} \gets 1 \textrm{ to } s$} \Comment 
		over children 
		\State Sample $v_{\iota} \sim \F_{v_0,\theta_0,\theta_{\iota}}$ or 
		$\f_{v_0, 
		\theta_0, \theta_{\iota}}$ \Comment see 
		Eq.~s~\eqref{eq::trnestecont}~\eqref{eq::trnestedisc}
		\For {$i \gets \sum_{l = 1}^{{\iota}-1}n_l+1 \textrm{ to } \sum_{l= 
		1}^{\iota} 
		n_l$} 
		\State  $x_{j, i} = \psi_{v_0, \theta_0, 
		\theta_{\iota}}\left(\frac{-\log(x_{j, 
		i})}{v_{\iota}}\right)$
		\EndFor
		\EndFor
		\For {$i \gets 1 \textrm{ to } k$}
		\State $u_{j, i}  = 
		\psi_{\theta}\left(\frac{-\log(x_{j, i})}{v_0}\right)$ 
		\EndFor
		\EndFor
		\State\Return $\U$
		\EndFunction 
	\end{algorithmic}
\end{algorithm}

\section{Data generation for features detection}\label{sec::datagen}

We use introduced copulas, to prepare artificial data in 
analysing cumulants based features selection and features extraction methods 
discussed in Chapter~\ref{cha::cumml}. Given tests on artificially generated 
data for which statistical features are known, discussed further features 
selection and features extraction methods can be applicable for the real 
life non-Gaussian distributed data analysis. 

Since we are interested in information tied to non-Gaussian joint 
distribution of features,  we transform multivariate Gaussian distributed data 
by introducing higher order cross-correlations into chosen subset of marginals 
using copula. Ideally the procedure should be performed in 
such a way that univariate statistics and standard cross-correlation measures  
are affected as little as possible. Given these, changed subset of marginals 
should be hard to detect using simple methods. This section is the evaluation 
of~\cite{domino2018hiding}  where such algorithm was 
introduced by means of Archimedean copulas and~\cite{domino2018use} for the 
$t$-Student copula case. To prepare data for 
further analysis, we start with 
multivariate data $\X \in \R ^{t \times n}$ sampled from $\mathcal{N}(\mu, 
\SSSS)$, and we 
transform these data set $\X \rightarrow \X' \in \R^{t \times n}$, in such a 
way that univariate statistics, the covariance matrix and the rank 
cross-correlation between marginals are as little affected 
as possible.

\begin{remark}\label{rem::datatransf}
Let us introduce the particular data transformation $\X \rightarrow \X'$, where 
$\X, \X' \in \R ^{t \times n}$ and $\X \sim \mathcal{N(\mu, \SSSS)}$ by 
fulfilling.
\begin{enumerate}
\item All univariate marginal distributions are the same for $\X$ and $\X'$.
\item The given subset of marginals $\mathbf{r} = (r_1, \ldots, 
r_k) \subset (1:n)$ of $\X'$ is modelled by the non-Gaussian copula 
$\C_{\text{ng}}$.
\item Rank cross-correlations between marginals of $\X'$ belonging to the 
$\mathbf{r}$ subset are 
similar to rank cross-correlations of analogical marginals of $\X$,
\begin{equation}
\rho(X_{a}, X_{b}) \approx \rho(X'_{a}, X'_{b}) \ \ \forall_{a, b \in 
\mathbf{r}}.
\end{equation}
By point $1$ the Pearson's 
cross-correlations would be similar as well.
\item For remaining marginals subset $(1:n) \setminus \mathbf{r}$, we have the 
same Gaussian copula for both case, hence 
referring to point $1$ we have
\begin{equation}
\text{cor}(X_a, X_b) = \text{cor}(X_a, X_b) \ \  \forall_{a, b \in (1:n) 
\setminus \mathbf{r}}.
\end{equation}
\item Cross-correlation between subsets $(1:n) \setminus \mathbf{r}$ and 
$\mathbf{r}$ in $\X'$ is similar to those in $\X$
\begin{equation}
\text{cor}(X_a, X_b) \approx \text{cor}(X'_a, X'_b) \ \  \forall_{a \in 
\mathbf{r}, b \in (1:n) \setminus \mathbf{r}}.
\end{equation}
 Finally, from points $3$ and $4$ we 
have
\begin{equation} 
	\text{cor}(\X) \approx \text{cor}(\X').
\end{equation}
From point $1$ we have
\begin{equation} 
\text{cov}(\X) \approx \text{cov}(\X').
\end{equation}
\end{enumerate}
\end{remark}

To fulfil point $1$, we need to use the 
copula approach to separate off univariate marginal distributions. To fulfil 
point $2$, the
$\mathbf{r}$ subset of marginals need  to be modelled by the non-Gaussian 
copula. To fulfil point $3$, we need to determine 
such copula parameter 
or parameters that would produce rank cross-correlations similar to
those of the corresponding marginals of original data. To fulfil point $4$, we 
simply leave unchanged marginals $(1:n) \setminus \mathbf{r}$ while 
transforming $\X \rightarrow \X'$. Fulfilling point $5$ is 
not simple. If we generate randomly marginals $X'_{r_1}, \ldots X'_{r_k}$ using 
the non-Gaussian copula fulfilling points $1-4$ we would have little 
cross-correlation between the subset $\mathbf{r}$ and remaining $(1:n) 
\setminus 
\mathbf{r}$. We call such approach the naive one that fragments the 
covariance 
matrix od $\X'$ into two blocks and gives the covariance matrix based 
algorithms advantage 
to find the subset $\mathbf{r}$. As discussed in~\cite{domino2018hiding}, to 
maintain the 
correlation between changed and non-changed subsets on the other hand we have to
transfer
$X_{r_1}, \ldots X_{r_k} \rightarrow X'_{r_1}, \ldots X'_{r_k}$ in information 
preserving way.

\subsection{$t$-Student copula case}\label{sec::tstudentsubcop}

Let us discuss the case, where we use the $t$-Student copula to modify the 
$\mathbf{r}$ subset 
of marginals according to Remark~\ref{rem::datatransf} by means of the 
relatively simple approach discussed in~\cite{domino2018use}. This approach is 
based on the fact, that there is well known algorithm of transforming 
multivariate Gaussian distributed data into $t$-Student distributed ones 
\cite{kotz2004multivariate}. 
Further we need only to transform back $t$-Student univariate marginals into 
the 
Gaussian ones.

To transform the $\mathbf{r}$ subset of marginals of multivariate Gaussian 
distributed data, with zero means and variances one each, into $t$-Student 
multivariate 
distributed one we need to perform the following transformation for each 
realisation of multivariate data, see \cite{kotz2004multivariate}. First, 
sample 
independently the scalar ${v_0}$ from $\chi^2(\nu)$ distribution, where the 
parameter $\nu$ is the scalar parameter of a $t$-Student copula. Next, for the
$\mathbf{r}$ subset of marginals multiple each element of the subset by 
$\sqrt{\frac{\nu}{v_0}}$. This transformation, as performed on the subset of 
marginals $\mathbf{r}$, affects slightly a cross-correlation between this 
subset and the 
reminding 
$(1:n) \setminus \mathbf{r}$ subset. This effect diminishes as $\nu$ 
rises, since the higher the $\nu$ value, the lower the spread of the 
distribution of $\sqrt{\frac{\nu}{v_0}}$. In next subsection we will 
discuss more complex cases where the subset is modelled by the Fr\'echet or 
the Archimedean copula.

%\begin{algorithm}[t]
%	\caption{Convert a subset of Gaussian distributed multivariate data to such 
%	modelled by 
%	the $t$-Student sub-copula.
%		\label{alg::g2tcop}}
%	\begin{algorithmic}[1]	
%		\State \textbf{Input}: $\X \in \R^{t\times n}$ - $t$ realisations of 
%		$n$-variate $\mathcal{N}(\mu, \SSSS)$, $\nu$ -- $t$-Student 
%		copula parameter, $\mathbf{r} \subset (1:n)$ - the subset 
%		of marginals. 	
%		\State \textbf{Output:} $\X' \in \R^{t\times n}$ - transformed data.
%		\Function{gcop2tstudent}{$\X$, $\nu$}
%		\State $\X' = \X$
%		\For {$i \gets 1 \textrm{ to } t$}
%		\State Sample $v_0 \sim \chi^2(\nu)$ \Comment{$\chi^2$ 
%			distribution parametrised by $\nu$}
%		\For {$j \textrm{ in } \mathbf{r}$}
%		\State $x'_{i,j} = \frac{x'_{i,j}-\mu_j}{\sigma_j}$ 
%		\Comment{normalise, $\sigma_j = \sqrt{s_{j,j}}$}
%		\State $x'_{i,j}= x'_{i,j}\sqrt{\frac{\nu}{v_0}}$ 
%		\Comment{change to t-Student multivariate dist.}
%		\State $u = T_{\nu}\left(x'_{i,j}\right)$
%		\Comment{change to t-Student copula, see Eq.~\eqref{eq::tmarg}}
%		\State $x'_{i,j} = F^{-1}_{\mathcal{N}(\mu_j, \sigma_{j}^2)}(u)$ 
%		\Comment{convert back univariate marginals}
%		\EndFor
%		\EndFor
%		\State\Return $\X'$
%		\EndFunction 
%	\end{algorithmic}
%\end{algorithm}	

\subsection{Fr\'echet copula case}

We can transform the $\mathbf{r}$ subset of marginals of multivariate Gaussian 
distributed data with zero means and variances one into such modelled by a 
Fr\'echet copula according to Remark~\ref{rem::datatransf} where conditions 
$3$ and $5$ are fulfilled at least approximately. We can perform it by 
modifying Algorithm~\ref{alg::frechetndgen} 
to such that transform $k$ independent vectors $U_1, \ldots, 
U_k$ all uniformly distributed on $[0,1]$ to such that are still uniformly 
distributed on $[0,1]$, but cross-correlated according to the Fr\'echet 
copula. This 
modification is presented in Algorithm~\ref{alg::unif2frech}, where last column 
of $\U$ is 
distinguished especially for large $\alpha$, see loop starting in line $8$. 
Hence this column appear to be a candidate 
to carry an information about
an overall correlation between the subset $\mathbf{r}$ and the subset $(1:n) 
\setminus \mathbf{r}$. This observation will be used in data transformation.

\begin{algorithm}[t]
	\caption{Transform independent uniformly distributed on $[0,1]$ data to 
	such modelled by multivariate Fr\'echet copula parametrised by $0 \leq 
	\alpha \leq 1$.}
	\label{alg::unif2frech}
	\begin{algorithmic}[1]	
		\State \textbf{Input}: $\U \in [0,1]^{t\times k}$ - uniformly 
		distributed data with independent marginals, $0 \leq \alpha 
		\leq 1$ - parameter of Fr\'echet copula. 	
		\State \textbf{Output:} $\U' \in [0,1]^{t\times k}$ - data modelled by 
		a Fr\'echet copula. 
		\Function{frechet}{$\U$, $\alpha$}
		\State $\U' = \U$
		\For {$j \gets 1 \textrm{ to } t$}
		\State sample $v \sim \text{Uniform}([0,1])$ 
		\If{$v \leq \alpha$} 
		\For {$i \gets 1 \textrm{ to } k-1$}
		\State $u'_{j,i} = u'_{j,k}$
		\EndFor
		\EndIf
		\EndFor
		\State\Return $\U' \in [0,1]^{t \times k}$
		\EndFunction 
	\end{algorithmic}
\end{algorithm}	

Following~\cite{domino2018hiding} we observe, that it is easy to 
transform Gaussian distributed $k$-variate $\X \in \R^{t \times k}$ into $\U 
\in [0,1]^{t 
\times k}$ where each columns $U_1, \ldots, U_i, \ldots,  U_k$ are uniformly 
distributed on $[0,1]$ 
independent on each other and carrying information about original $\X$ in such 
a way that the higher $i$ the more meaningful the information is. We make 
this transformation 
by performing the eigenvalue, eigenvector decomposition of the covariance 
matrix, 
sorting eigenvalues in increasing order, transforming data using corresponding 
eigenvectors and finally transforming univariate marginals to uniformly 
distributed on 
$[0,1]$ see Algorithm~\ref{alg::core}. Despite the fact that columns of the 
output $U_1, \ldots U_k$ are independent on each other they 
still carry the information, that in our case concerns the cross-correlation 
between the subset $\mathbf{r}$ and the subset $(1:n) \setminus \mathbf{r}$.

\begin{algorithm}[t]
	\caption{Transform multivariate normally distributed data with $0$ means 
	and variance $1$ into independent 
	marginals uniformly distributed on $[0,1]$}
	\label{alg::core}
	\begin{algorithmic}[1]	
		\State \textbf{Input}: $\mathbf X\in\R^{t \times k}$ - $t$ 
		realization of $k$-variate  $\mathcal{N}(0, \mathbf{R})$
		\State \textbf{Output:} $ \mathbf \U \in [0,1]^{t \times k}$ -- $t$ 
		realisations of $k$ independent $\UU_i \sim \text{Uniform}(0,1)$. 
		\Function{norm2unif}{$\mathbf X$}
		\State $\mathbf{R}' = \textsc{cor}(\mathbf X)$  
		\Comment{$\mathbf{R}' \approx \mathbf{R}$}
		\State $\lambda, U = \textsc{eigenvals}(\mathbf{R}'), 
		\textsc{eigenvec}(\mathbf{R}')$ 
		\Comment{$\lambda_1 < \ldots < \lambda_{k}$}
		\For{$j \leftarrow 1 \text{ to } t$}
		\State $\mathbf{x}_j\leftarrow \mathbf{x}_j\cdot U$ 
		\Comment{$\mathbf{x}_j = [x_{j,1}, \ldots, x_{j,k}]$}
		\For{$i \leftarrow 1 \text{ to } k$}
		\State $u_{j,i} =  
		F_{\mathcal{N}(0,\lambda_i)}(x_{j,i})$ 
		\Comment{convert to uniform univ. marginals}
		\EndFor 
		\EndFor 
		\State\Return $\U$		
		\EndFunction 
	\end{algorithmic}	
\end{algorithm}

Finally given the subset of multivariate Gaussian distributed marginals denoted 
by $\mathbf{r}$ we can input them into Algorithm~\ref{alg::core} to 
achieve independent uniformly distributed marginals that are carrying 
interesting us information. Next we can input them into 
Algorithm~\ref{alg::unif2frech} to 
achieve the subset of marginals modelled by the Fr\'echet copula and finally we 
can transform back their univariate distributions into the Gaussian one. 
Obviously the Fr\'echet copula gives the same cross-correlation between all 
marginals, hence
the point $3$ in 
Remark~\ref{rem::datatransf} will be fulfilled only 
approximately. This is due 
to the fact that the Spearman's cross-correlation 
inside the subset of marginals modelled by the Fr\'echet copula will be 
constant and equal to the copula's 
parameter $\alpha$, see Eq.~\eqref{eq::frechetrho}.
Further, observe that due to the resampling scheme in 
Algorithm~\ref{alg::unif2frech} 
some information about the cross-correlation between the subset $\mathbf{r}$ 
and the subset $(1:n) \setminus \mathbf{r}$ may be lost, especially for low 
$\alpha$. In the case of the Archimedean copulas this two problems can be 
resolved at least to some extend, what is discussed in the next subsection.

%\begin{algorithm}[t]
%	\caption{Transform the $\mathbf{r}$ subset of marginals of multivariate 
%	normally distributed data into such distributed according to the Fr\'echet 
%	one parameter sub-copula - marginal univariate distributions are unchanged.}
%	\label{alg::gcop2frechet}
%	\begin{algorithmic}[1]	
%		\State \textbf{Input}: $\R^{t \times n} \ni \mathbf X$ - $t$ 
%		realisations of $
%		\mathcal{N}(\mu, \SSSS)$, $\mathbf{r} \subset (1:n)$ -  
%		subset of marginals.
%		\State \textbf{Output:} $\mathbf X'\in\R^{t \times n}$ - changed data. 
%		\Function{gcop2frechet}{$\mathbf X$}
%		\State $\X'$ = $\X$
%		\State $\alpha$ = mean$(\rho(X_{r_1}, \ldots, X_{r_k}))$ \Comment{
%		mean Spearman cross-correlation}
%		\State $\U = \textsc{norm2unif}([X_{r_1}, \ldots, X_{r_k}])$
%		\State $\U = \textsc{frechet}(\alpha, \U)$
%		\For{$i \leftarrow 1 \text{ to } k$}
%		\State 
%		$X'_{r_i} =  
%		F_{\mathcal{N}(\mu_{r_i},\sigma^2)}^{-1}(U_i)$ \ \Comment{where
%		$\sigma^2 = s_{r_i,r_i}$}
%		\EndFor 
%		\State\Return $\X'$		
%		\EndFunction 
%	\end{algorithmic}	
%\end{algorithm}

\subsection{Archimedean copula case}\label{sec::archsamp}

We can transform the $\mathbf{r}$ subset of marginals of multivariate Gaussian 
distributed to such modelled by the Archimedean copula using an analogical 
approach as in the Fr\'echet copula case, see 
also~\cite{domino2018use}. The basic difference is that to obtain a sample of 
$k$-variate 
Archimedean copula, by means 
of the Algorithm~\ref{alg:arch_sampler}, we need $k+1$ independent samples from 
uniform distribution 
on $[0,1]$ - one additional sample is required. In other words, following 
Algorithm~\ref{alg:arch_sampler}, $k$ samples 
will be used in line $6$ while the additional can be transformed to 
$v$ via the quantile function $\F^{-1}$ of the inverse Laplace-Stieltjes 
transform of the corresponding Archimedean copula generator
in line $5$. Following~\cite{domino2018hiding}, recall that this generator 
$\psi_{\theta}$ is strictly decreasing 
by Definition~\ref{def::archgen}. Hence the function
\begin{equation}\label{eq::genls}
 f_i(u, x_i) = \psi_\theta\left ( \frac{-\log 
 	(x_i)}{\F^{-1}(u)}\right ),
 \end{equation} 
used in line $8$ of Algorithm~\ref{alg:arch_sampler} will be 
strictly increasing in $x_i$ for constant $u$. Moreover for constant 
 $x_i \in (0,1)$  it will be strictly increasing in $u$ if $\F$ is continuous 
 (the Gumbel and Clayton copula case), or non-decreasing if $\F$ is discrete 
 (the AMH and Frank copula case). 
  In 
 data transformation scheme such $u$ can be used to carry a general information 
 into 
transformed data. Hence we will take it from the most informative (last) column 
of the 
 output of Algorithm~\ref{alg::core}. On the other hand $x_1, 
 \ldots, x_k$ can 
 be used to carry individual informations. We will take them from other 
 columns of the 
 output of Algorithm~\ref{alg::core}. Due 
 to the monotonicity of Eq.~\eqref{eq::genls} the rank of $u$ will affect 
 the rank of the realisation of all marginals (general information), while the 
 rank of $x_i$ would affect the rank of the $i$\textsuperscript{th} marginal 
 (individual information). Here the general information affects all 
 realisations, not only part of them as in the Fr\'echet copula case. It is why 
 the outcome given an Archimedean copula should be better than given the 
 Fr\'echet copula. To produce $k+1$ independent uniformly distributed data be 
 means of the 
Algorithm~\ref{alg::core}, from $\X \in \R^{t \times k}$ we sample 
independently $t$ samples from univariate $\mathcal{N}(0,1)$ and add them as 
another marginal to $\X$ before inputting data to Algorithm~\ref{alg::core}.

One can note that the same function as in Eq.~\eqref{eq::genls} is used in the 
nested case - see line $14$ of Algorithm~\ref{alg::nestedarchgen}. Hence the 
nested Archimedean copula generator can be used as well to transform data. 
Observe as well that sampling or data input to 
lines $5$ and $6$ in both Algorithm~\ref{alg:arch_sampler} 
and~\ref{alg::nestedarchgen} is the same. 
Further in line $14$ of
Algorithm~\ref{alg::nestedarchgen} we have the same monotone transformation as 
in line $8$ of Algorithm~\ref{alg:arch_sampler}. There are some differences 
as well. Data inputted to 
$x_1, \ldots, x_k$ in Algorithm~\ref{alg::nestedarchgen} have to 
pass through line $10$ of the algorithm, where we have a child copula 
generator $\psi_{v_0, 
\theta_0, \theta_{\iota}}$ parametrised by $v_0$, see 
Eq.~\eqref{eq::nestedgen}. Although this generator, in analogy to 
Eq.~\eqref{eq::genls}, gives the function $g_i(x_i) = \psi_{v_0, 
\theta_0, \theta_{\iota}}\left(\frac{-\log(x_{i})}{v_{\iota}}\right)$ that is 
expected 
to be strictly increasing in $x_i$, we have there randomly 
generated $v_{\iota}$, see line $8$ in 
Algorithm~\ref{alg::nestedarchgen}. Unfortunately, this random factor will 
worsen to some extend the cross-correlation between 
subsets $\mathbf{r}$ and $(1:n) \setminus \mathbf{r}$. On the other hand, the 
advantage of 
the use of the nested Archimedean copula is the fact that we do not have a 
constant cross-correlation within $\mathbf r$ subset but we can have subsets 
$\mathbf{r}_{\iota} \subset \mathbf{r}$ with higher cross-correlations 
$\rho_{\iota}$ 
in 
comparison with an overall cross-correlation $\rho_0$ within $\mathbf{r}$. 
Finally, we require $\rho_{\iota} > \rho_0$ to fulfil the 
sufficient nesting condition, see Eq.~\eqref{eq::sufnestcond}, since we 
assume that parents and children copulas are from the same Archimedean family.

The procedure of determination of  $\{\mathbf{r}_1,
\ldots, \mathbf{r}_{\iota}, \ldots,  \mathbf{r}_s, \mathbf{r}'' \}$, being a 
set partition of $(1:n)$, is summarised in Algorithm~\ref{alg::corfrag}. Here 
each $\mathbf{r}_{\iota}$ determines such subset of marginals, where a 
cross-correlation $\rho_{\iota}$ is higher than an overall cross-correlation 
$\rho_0$. The input to the Algorithm~\ref{alg::corfrag} is the 
matrix of the Spearman's cross-correlations of size $n \times n$. The 
Algorithm~\ref{alg::corfrag} returns an 
overall cross correlation $\rho_0$ that is used to determine parameter 
$\theta_0$ and cross-correlations $\rho_{\iota}$ corresponding to higher 
cross-correlations 
regions of children copulas parametrised by $\theta_{\iota}$ and modelling 
marginals subsets denoted by $\mathbf{r}_{\iota}$. 

\begin{algorithm}[t]
	\caption{Determinig subsets of marginals with higher cross-correlation}
	\label{alg::corfrag}
	\begin{algorithmic}[1]	
		\State \textbf{Input}: $\mathbf{C} \in\R^{k \times k}$ -- 
		''correlation'' matrix 
		\State \textbf{Output:} $\mathbf r_1, \ldots, 
		\mathbf{r}_s$ -- such that $\{\mathbf r_1, \ldots, 
		\mathbf{r}_s, \mathbf{r}''\}$ is the set partition of $(1:n)$, $\rho_1, 
		\ldots, 
		\rho_s$ -- correlations inside each subset, $\rho_0$ -- general 
		correlation. 
		\Function{getcors}{$\mathbf{C}$}
		\For {$\mathbf r_1, \ldots, \mathbf r_s \leftarrow $set 
		partitions$(1:k)$} \Comment{chose as $\forall_{\iota} 
		|\mathbf{r}_{\iota}| \geq 2$ 
		and $s > 1$}
		\For {$\iota \leftarrow 1:s$}
		\State $\rho_{\iota} = \text{mean}_{a,b}(c_{a,b})$ \Comment{mean over 
		$a,b \in 
		\mathbf{r}_{\iota} \text{ such that } a \neq b$}
	\EndFor
		\State $\rho_0 = \text{mean}_{a_0,b_0}(c_{a_0,b_0})$ \Comment{$a_0 
		\in \mathbf{r}_{\iota}$, $b_0 \in \mathbf{r}_{\iota'}: \iota \neq 
		\iota' \wedge a_0 
		\neq b_0$}
		\State test sufficient 
		nesting $\forall_{\iota} \ \rho_{\iota} > \rho_0$ 
		\State $\left(c_{\text{theor}}\right)_{a,b} = \begin{cases}
		1 &\text{ if } a = b \\
		\rho_{\iota} &\text{ if } a,b \in \mathbf{r}_{\iota} \\
		\rho_0 &\text{ elswhere}.
		\end{cases}$
		\State penalty$ = \|\mathbf{C} -  \mathbf{C}_{\text{theor}}\|$ 
		\Comment{Frobenius 
		norm}
		\EndFor
		\State\Return $[\mathbf{r}_1, \ldots \mathbf{r}_s]$, $[\rho_1, \ldots 
		\rho_s]$, $\rho_0$ such that penalty is minimal	
		\EndFunction 
	\end{algorithmic}	
\end{algorithm}

\section{Implementation and experiments}\label{sec::datamalfexp}

\subsection{Implementation}

Discussed copula-based data generation algorithms were implemented in the 
Julia programming language \cite{bezanson2012julia, 
bezanson2014julia}. The Julia is a modern, open source and high-level 
programming language. As it is open source, the code can be accessed, reviewed 
and developed online by scientists, making Julia the proper tool for 
scientific computation. Apart from this, the main advantage of the 
language is its solid performance.  Linear operations and random sampling 
operations implemented in this language require significantly less processor 
time than similar operations implemented in other well-known programming 
languages, see \cite{bezanson2014julia}. Furthermore, Julia is 
specialised in multidimensional arrays (tensors) calculation 
\cite{bezanson2014array}, which is an advantage for cumulants tensors 
calculations. 

Discussed here copula sampling algorithms are available on a GitHub repository 
\cite{cop} as the \texttt{DatagenCopulaBased.jl} module. Using implemented 
there functions one 
can sample Gaussian, $t$-Student, Fr\'echet and Archimedean (Gumbel, Frank, 
Clayton, AMH) copulas. In the case of the Fr\'echet and Archimedean copulas 
bivariate case is distinguished, since a wider range of parameters is 
available. The nested Archimedean copulas are supported as well. Instead of 
Archimedean copula parameter or parameters values, one can insert expected 
Spearman's or Kendall's cross-correlation coefficient or coefficients in a 
nested case. 

Apart from these copulas, in module \texttt{DatagenCopulaBased.jl} the sampler 
for 
Marshall-Olkin copula is introduced. For definition and features of this copula 
see~\cite{schmidt2007coping}. In the bivariate case, such copula is introduced 
by the stochastic process modelling exponentially distributed extinctions 
parametrised by three parameters $\lambda_1$, $\lambda_2$ and $\lambda_{1,2}$.  
In the $n$ variate case we have many parameters concerning single 
marginals, all combinations of two marginals, all combinations of three 
marginals etc., hence their number grows rapidly with $n$. The 
investigation of the relationship between higher order cumulants and parameters 
of the Marshall-Olkin copula seems to be interesting but complicated.

Module \texttt{DatagenCopulaBased.jl} provides functions that change a subset 
of marginals of multivariate Gaussian distributed data into those modelled by 
the $t$-Student, Fr\'echet or Archimedean copulas. This is important for future 
experiments. In the last case, both a nested and non-nested cases are 
supported. We have implemented as well the Marshal-Olkin copula case, but only 
for the subset of marginals of size $2$, since for this copula, the number of 
free parameters rises rapidly with the size of the subset. Due to this 
observation, the Marshal-Olkin copula is rather not applicable for experiments 
discussed in the next part of this book.  Finally, there are many methods of 
random correlation matrix generation for tests, see~\cite{domino2018hiding} for 
details. In experiments in the next subsection we use the following methods of 
the covariance matrix generation, all implemented in~\cite{cop}.
\begin{enumerate}
\item Constant $\mathbf{R}_{\text{const}, \alpha} \in\R^{n \times n}$, we set 
simply each off-diagonal 
element to the constant value $0 < 
\alpha <1$.
\begin{lstlisting}
julia> cormatgen_constant(3, 0.5)
3×3 Array{Float64,2}:
1.0  0.5  0.5
0.5  1.0  0.5
0.5  0.5  1.0
\end{lstlisting}

\item Constant noised $\mathbf{R}_{\text{const}, \alpha, \epsilon} \in\R^{n 
\times n}$, to each element of the constant matrix we add the random value 
multiplied by the parameter $\epsilon$. By default 
$\epsilon = \frac{1-\alpha}{2}$. 

\begin{lstlisting}
julia> Random.seed!(42);

julia> cormatgen_constant_noised(3, 0.5)
3×3 Array{Float64,2}:
1.0       0.392423  0.515595
0.392423  1.0       0.649938
0.515595  0.649938  1.0  
\end{lstlisting}

\item Random, wy sample independently elements $\mathbf{A} \in \R^{n \times n}$ 
from the $\text{Uniform}([0,1])$. Next for normalisation we use $\mathbf{D_A} = 
\text{diagm}(\mathbf{A} \mathbf{A}^{\intercal})$ and return 
\begin{equation}
\R^{n \times n} \ni \mathbf{\mathbf{R}}_{\text{rand}} = 
\mathbf{D_A}^{-\frac{1}{2}} 
\mathbf{A} 
\mathbf{A}^{\intercal} \mathbf{D_A}^{-\frac{1}{2}}.
\end{equation}
This method returns rather high positive correlations.

\begin{lstlisting}
julia> Random.seed!(42);

julia> cormatgen_rand(3)
3×3 Array{Float64,2}:
1.0       0.673166  0.538302
0.673166  1.0       0.756497
0.538302  0.756497  1.0  
\end{lstlisting}

\item Toeplitz, given a parameter $0 < \rho < 1$, we return 
$\mathbf{R}_{\text{T}, \rho} \in \R^{n \times n}$ with off diagonal elements 
$r_{i_1, i_2} = 
\rho^{|i_1 - i_2|}$

\begin{lstlisting}
julia> cormatgen_toeplitz(3, 0.5)
3×3 Array{Float64,2}:
1.0   0.5  0.25
0.5   1.0  0.5 
0.25  0.5  1.0 
\end{lstlisting}

\end{enumerate}

\subsection{Experiments}\label{sec::experimentmats}

In this subsection we analyse the following experiment. We generate at random 
element-wise positive symmetric and positive definite covariance matrix 
$\mathbf{R}$ with ones on a diagonal. The generation methods are discussed in 
~\cite{domino2018hiding}.
In this book, we use following methods, \ie~random, constant noised, and 
Toeplitz. In noised examples we use the default noise 
parameter $\epsilon = \frac{1-\alpha}{2}$. In 
Figure~\ref{fig::cors} we present chosen elements of the $\mathbf{R}$ matrix 
generated by different methods. Obviously $r_{i_1, i_2} = 1$ if $i_1 = i_2$ and 
$r_{i_1, i_2} < 1$ otherwise. For random method correlations are rather high, 
for Toeplitz method we have a specific correlation pattern, while for constant 
noised matrix the mean correlation is determined by the $\alpha$ parameter.
\begin{figure}
	\subfigure[Random method 
	\label{fig::rand_cor}]{\includegraphics{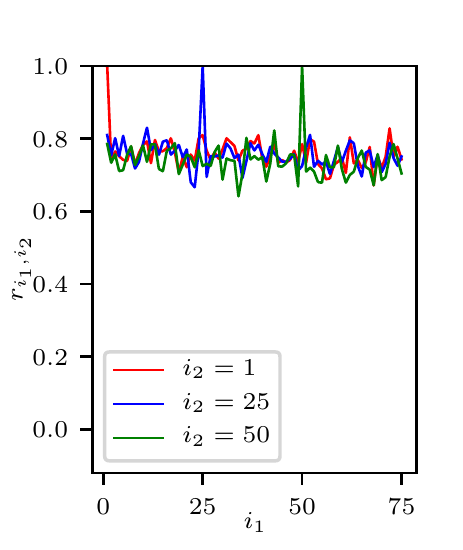}}
	\subfigure[Constant noised method with parameter $\alpha = 
	0.5$\label{fig::c_n}]{\includegraphics{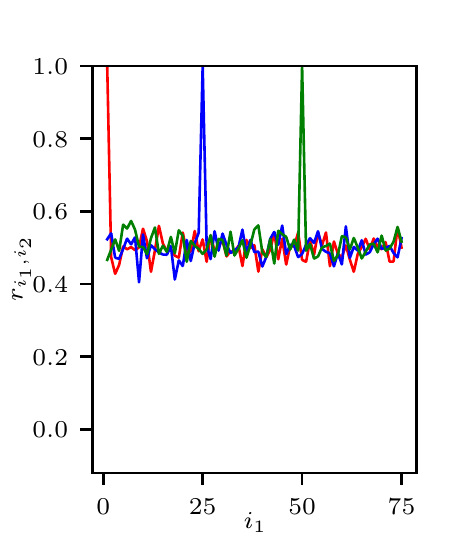}}
		\subfigure[Toeplitz method with parameter $\rho = 
	0.95$\label{fig::top}]{\includegraphics{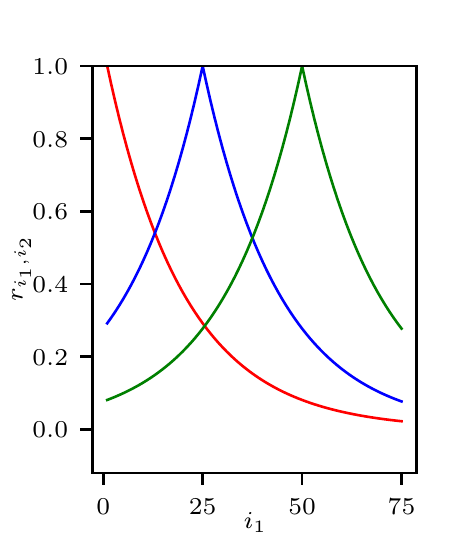}}
	\caption{Values of chosen elements of the correlation matrix given 
	different methods of generation, $n = 75$. We chose such methods and 
	parametrisations to have different patterns of the cross-correlation 
	between marginals.}\label{fig::cors}
\end{figure}

Given the correlation matrix, we sample $\R^{t 
\times n} \ni \X \sim \mathcal{N}(0, \mathbf{R})$ (with standard normal 
marginals) and perform the transformation $\X \rightarrow \X'$ \ie~we change 
randomly chosen subsets of marginals $\mathbf{r}$ of given length $|\mathbf{r}| 
= k$ by means of $t$-Student, Fr\'echet and Archimedean copulas. We use 
following parameters, number of marginals $n = 
75$, number of changed marginals $k = 10$.
For each experiment we compute the measure of the change of the covariance 
matrix due to the transformation of data.
\begin{equation}
\delta(\X, \X') = \frac{\|\text{cov}(\X) - 
\text{cov}(\X')\|}{\|\text{cov}(\X)\|},
\end{equation}
where we use the Frobenius norm of the difference between covariance matrices  
of $\X$ and $\X'$. The lower the $\delta(\X, \X')$ the better the 
transformation is.

\begin{figure}
    \subfigure[Random correlation matrix. 
    \label{fig::trand}]{\includegraphics{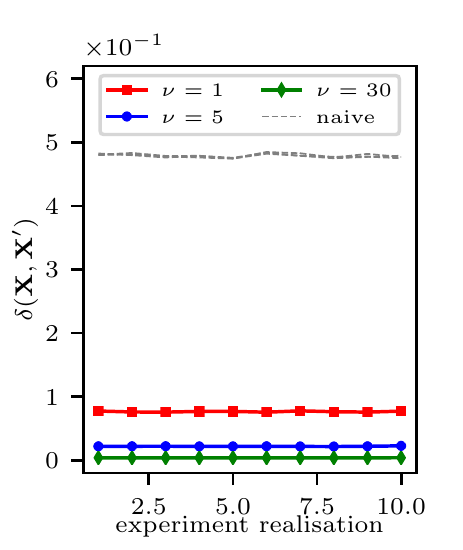}}
    \subfigure[Constant noised correlation matrix with $\alpha 
    =0.5$.\label{fig::tcn}]{\includegraphics{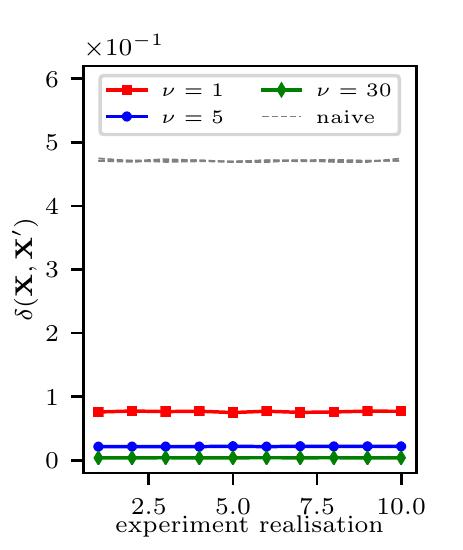}}
    \subfigure[Toeplitz correlation matrix with $\rho = 
    0.95$.\label{fig::ttoep}]{\includegraphics{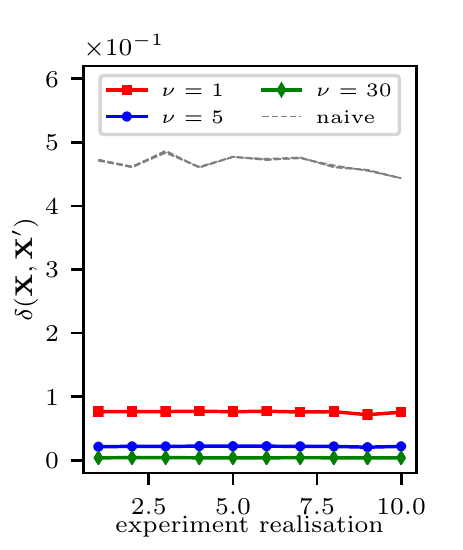}}
    %    \caption{The measure of the difference of the correlation matrix due 
    %to 
    %        data transformation for the $t$-Student copula, $k = 10, n = 
    %        75$.}
    %    \label{fig::tcop}
    
    \subfigure[Random correlation matrix. 
    \label{fig::frand}]{\includegraphics{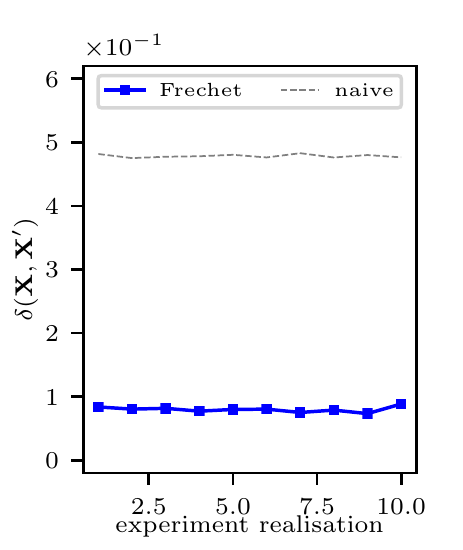}}
    \subfigure[Constant noised correlation matrix with $\alpha 
    =0.5$.\label{fig::fcn}]{\includegraphics{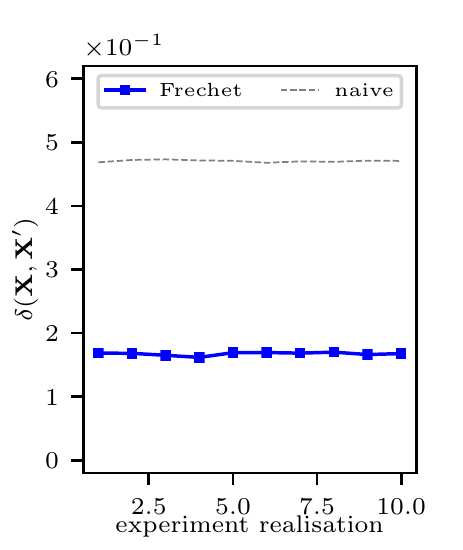}}
    \subfigure[Toeplitz correlation matrix with $\rho = 
    0.95$.\label{fig::ftoep}]{\includegraphics{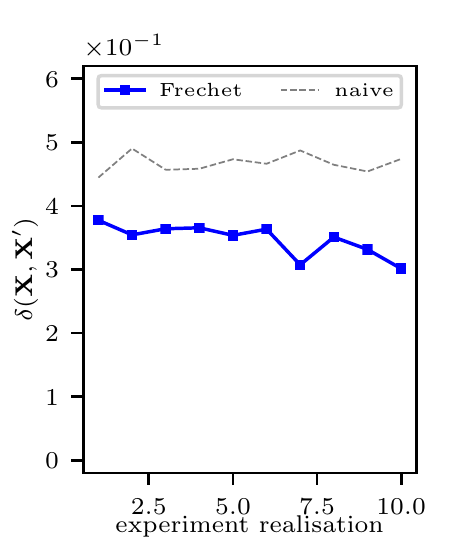}}

    %    \label{fig::fcop}
    
    \subfigure[Random correlation matrix. 
    \label{fig::arand}]{\includegraphics{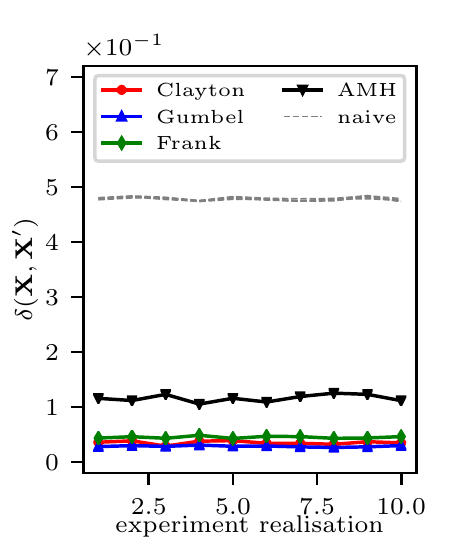}}
    \subfigure[Constant noised correlation matrix with $\alpha 
    =0.5$.\label{fig::acn}]{\includegraphics{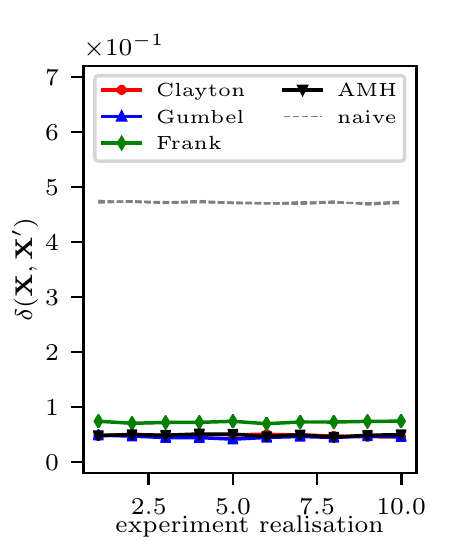}}
    \subfigure[Toeplitz correlation matrix with $\rho = 
    0.95$.\label{fig::atoep}]{\includegraphics{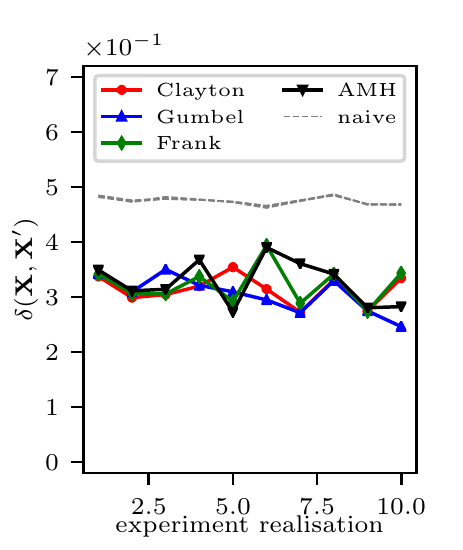}}
    
    \caption{The measure of the difference of the correlation matrix due to 
        data transformation with $k = 
        10$ and $n = 75$ for the $t$-Student copula \subref{fig::trand}, 
        \subref{fig::tcn} 
        and \subref{fig::ttoep}, the Fr\'echet copula 
        \subref{fig::frand}, \subref{fig::fcn} and 
        \subref{fig::ftoep}, and Archimedean not nested copulas 
        \subref{fig::arand}, \subref{fig::acn} and \subref{fig::atoep}.}
    
    %	\caption{The measure of the difference of the correlation matrix due to 
    %		data transformation for Archimedean not nested copulas, $k = 10, n 
    %= 
    %		75$.}\label{fig::acop}
\end{figure}

%\begin{figure}
%		\subfigure[Random correlation matrix. 
%		\label{fig::trand}]{\includegraphics{tstudent1.pdf}}
%		\subfigure[Constant noised correlation matrix with $\alpha 
%		=0.5$.\label{fig::tcn}]{\includegraphics{tstudent05.pdf}}
%		\subfigure[Toeplitz correlation matrix with $\rho = 
%		0.95$.\label{fig::ttoep}]{\includegraphics{tstudent2.pdf}}
%		\caption{The measure of the difference of the correlation matrix due to 
%		data transformation for the $t$-Student copula, $k = 10, n = 
%		75$.}\label{fig::tcop}
%\end{figure}

In Figures~\ref{fig::trand}-\subref{fig::ttoep} we present $\delta(\X, \X')$ 
values for 
the $t$-Student copula, different $\nu$ parameters and different 
correlation matrices. As proposed in Subsection~\ref{sec::tstudentsubcop}
the higher the $\nu$ parameter value the lower $\delta(\X, \X')$.
What is important the $\delta(\X, \X')$ is low in comparison with the naive 
approach regardless the correlation matrix.

%\begin{figure}
%	\subfigure[Random correlation matrix. 
%	\label{fig::frand}]{\includegraphics{frech1.pdf}}
%	\subfigure[Constant noised correlation matrix with $\alpha 
%	=0.5$.\label{fig::fcn}]{\includegraphics{frech05.pdf}}
%	\subfigure[Toeplitz correlation matrix with $\rho = 
%	0.95$.\label{fig::ftoep}]{\includegraphics{frech2.pdf}}
%	\caption{The measure of the difference of the correlation matrix due to 
%		data transformation for the Fr\'echet copula, $k = 10, n = 
%		75$.}\label{fig::fcop}
%\end{figure}

In Figures~\ref{fig::frand}-\subref{fig::ftoep} we present $\delta(\X, \X')$ 
values for 
the Fr\'echet copula and different correlation matrices. Results are much worse 
than in the $t$-Student copula. They are still on the comparable level in the 
random case, see Figure~\ref{fig::frand}, where the correlation between 
marginals is high. Here the little information about the correlation between 
changed and unchanged subset of marginals is lost during the sampling procedure 
in Algorithm~\ref{alg::unif2frech}. For the Toeplitz case, see 
Figure~\ref{fig::ftoep}, we have poor results almost as bad as the naive 
algorithm outcome.

%\begin{figure}
%	\subfigure[Random correlation matrix. 
%	\label{fig::arand}]{\includegraphics{arch1.pdf}}
%	\subfigure[Constant noised correlation matrix with $\alpha 
%	=0.5$.\label{fig::acn}]{\includegraphics{arch05.pdf}}
%	\subfigure[Toeplitz correlation matrix with $\rho = 
%	0.95$.\label{fig::atoep}]{\includegraphics{arch2.pdf}}
%	\caption{The measure of the difference of the correlation matrix due to 
%		data transformation for Archimedean not nested copulas, $k = 10, n = 
%		75$.}\label{fig::acop}
%\end{figure}

\begin{figure}
	\subfigure[Random correlation, Gumbel copula. 
	\label{fig::ngrand}]{\includegraphics{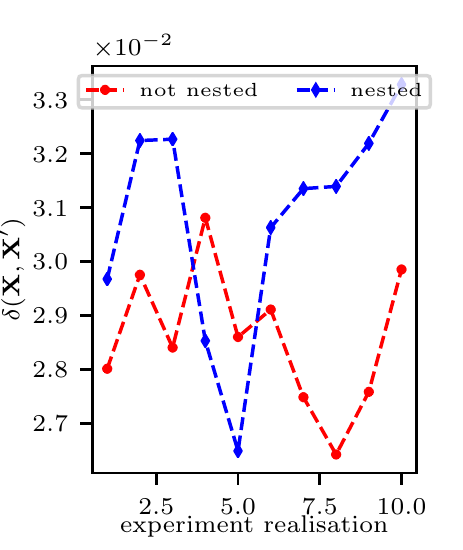}}
	\subfigure[Constant noised correlation $\alpha 
	=0.5$, Gumbel 
	copula.\label{fig::ngcn}]{\includegraphics{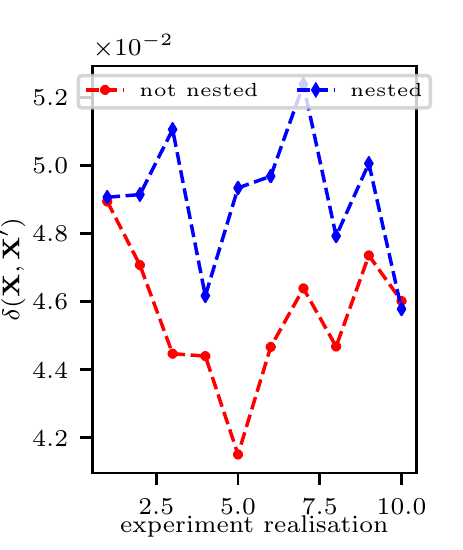}}
	\subfigure[Toeplitz correlation $\rho = 
	0.95$, Gumbel 
	copula.\label{fig::ngtoep}]{\includegraphics{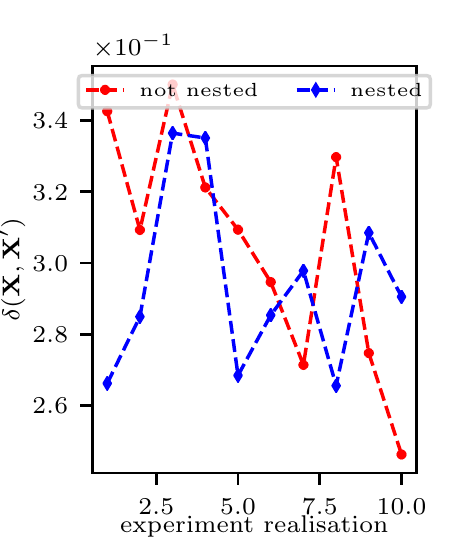}} 
	\\
		\subfigure[Random correlation, Frank copula. 
	\label{fig::nfrand}]{\includegraphics{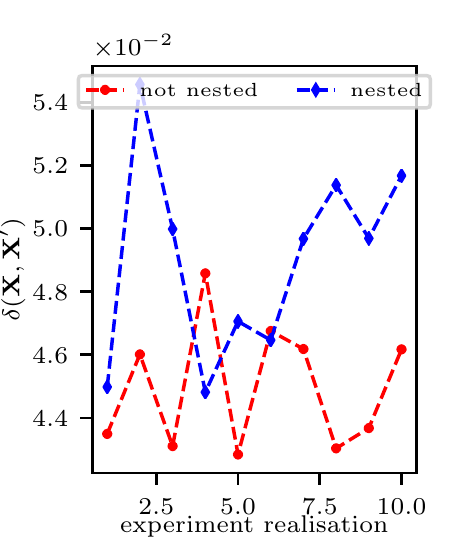}}
	\subfigure[Constant noised correlation $\alpha 
	=0.5$, Frank 
	copula.\label{fig::nfcn}]{\includegraphics{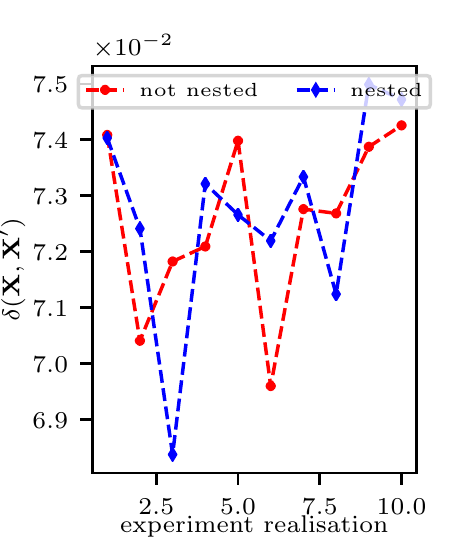}}
	\subfigure[Toeplitz correlation $\rho = 
	0.95$, Frank 
	copula.\label{fig::nftoep}]{\includegraphics{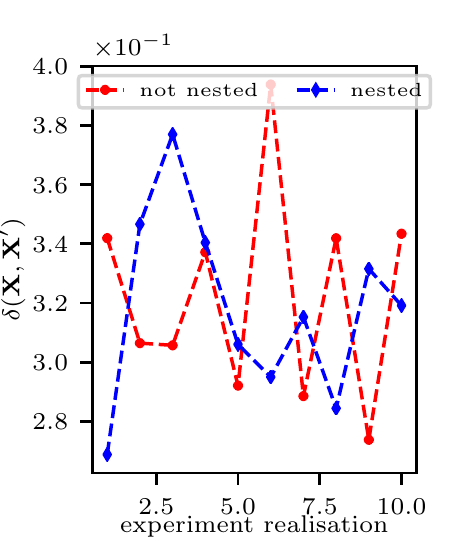}}
	\caption{The measure of the difference of the correlation matrix due to 
		data transformation. The comparison of nested nod not-nested 
		Archimedean copulas, mind different scales on y 
		axis.}\label{fig::nested}
\end{figure}

In Figures~\ref{fig::arand}-\subref{fig::atoep} we present $\delta(\X, \X')$ 
values for Archimedean not-nested copulas. In the case of high correlations the 
outcome of the AMH copula is rather poor since this copula has a limit for the 
Persons coefficient equal to $0.5$. Further the sampling of the Frank and AMH 
copula included discrete inverse Laplace-Stieltjes transform of the copula 
generator.  
Hence the function from Eq.~\eqref{eq::genls} is not strictly increasing 
and may loose some information about the cross-correlation between subsets 
$\mathbf{r}$ and $(1:n)\setminus \mathbf{r}$. For the Toeplitz matrix, see 
Figure~\ref{fig::atoep} we have poor results similar to those of the Fr\'echet 
copula. Hence for such correlation matrix the transformation discussed in 
Algorithm~\ref{alg::core} is ineffective while concerning the preservation of 
the information between the changed and unchanged data subset.

In Figure~\ref{fig::nested} we compare nested and not-nested Archimedean 
copulas on the example of the Gumbel and the Frank ones. In general, outcomes 
are similar. This is due to two competing issues. On the one hand, nesting 
allows for non-uniform correlation matrix inside changed data subset. Such 
non-uniform correlation is determined by means of 
Algorithm~\ref{alg::corfrag}.  On the other hand, however, as discussed in 
Subsection~\ref{sec::archsamp} the transfer of information about the 
correlation between changed and unchanged subsets of marginals through 
Algorithm~\ref{alg::nestedarchgen} (nested case) may be less efficient than 
through Algorithm~\ref{alg:arch_sampler} (not nested case). This is due to the 
fact that Algorithm~\ref{alg:arch_sampler} is simpler than 
Algorithm~\ref{alg::nestedarchgen}. Nevertheless, the nested case is more 
interesting since marginals can be grouped inside a changed subset.

Concluding the $t$-Student copula's algorithm is effective for all investigated 
here correlations matrices, especially for large $\nu$. For other copulas we 
use Algorithm~\ref{alg::core} that requires correlation matrices with rather 
moderate variations in elements values, as in Figure~\ref{fig::rand_cor} and 
Figure~\ref{fig::c_n} for the information preserving transformation. Given such 
correlations, Clayton and Gumbel copulas are comparable with the $t$-Student 
one and even may outperform the second for some values of its parameter $\nu$.  
Apart from this, the AMH copula requires additionally low correlations (due to 
the parameter's limitations), while Frank and Fr\'echet copulas require rather 
high correlations. Hence given the random correlation matrix, as in 
Figure~\ref{fig::rand_cor}, all copulas apart from the AMH have their optimal 
results. It is why such a correlation matrix will be used for experiments in 
Section~\ref{cha::cumml}. Finally the performance of nested and not-nested 
Archimedean copulas are similar, but the first case is more interesting.

Finally for the graphical presentation on results in Figure~\ref{fig::g} we 
present univariate and bivariate histograms of $\X \sim \NN(0, \mathbf{R})$ 
with high cross-correlations. In Figures~\ref{fig::archandg} we present 
similarly univariate and bivariate histograms of transformed data $\X'$ in such 
a way, that marginals number $1$ and $2$ have been changed employing the 
Clayton copula, while marginals number $3$ and $4$ remind unchanged. Observe 
the intermediate case of marginals $1$ versus $3$ and $2$ versus $4$. Given 
such insight into multivariate non-Gaussian distributions we can move in next 
chapter to higher order multivariate statistics and higher order multivariate 
cumulants that can be used to measure higher-order correlations introduced by 
non-Gaussian distributions as in Figures~\ref{fig::ag12}~\ref{fig::ag13} 
and~\ref{fig::ag24}.

\begin{figure}[t!]
	\subfigure[Marginals $1$ and 
	$2$\label{fig::g12}]{\includegraphics[width=0.47\textwidth]{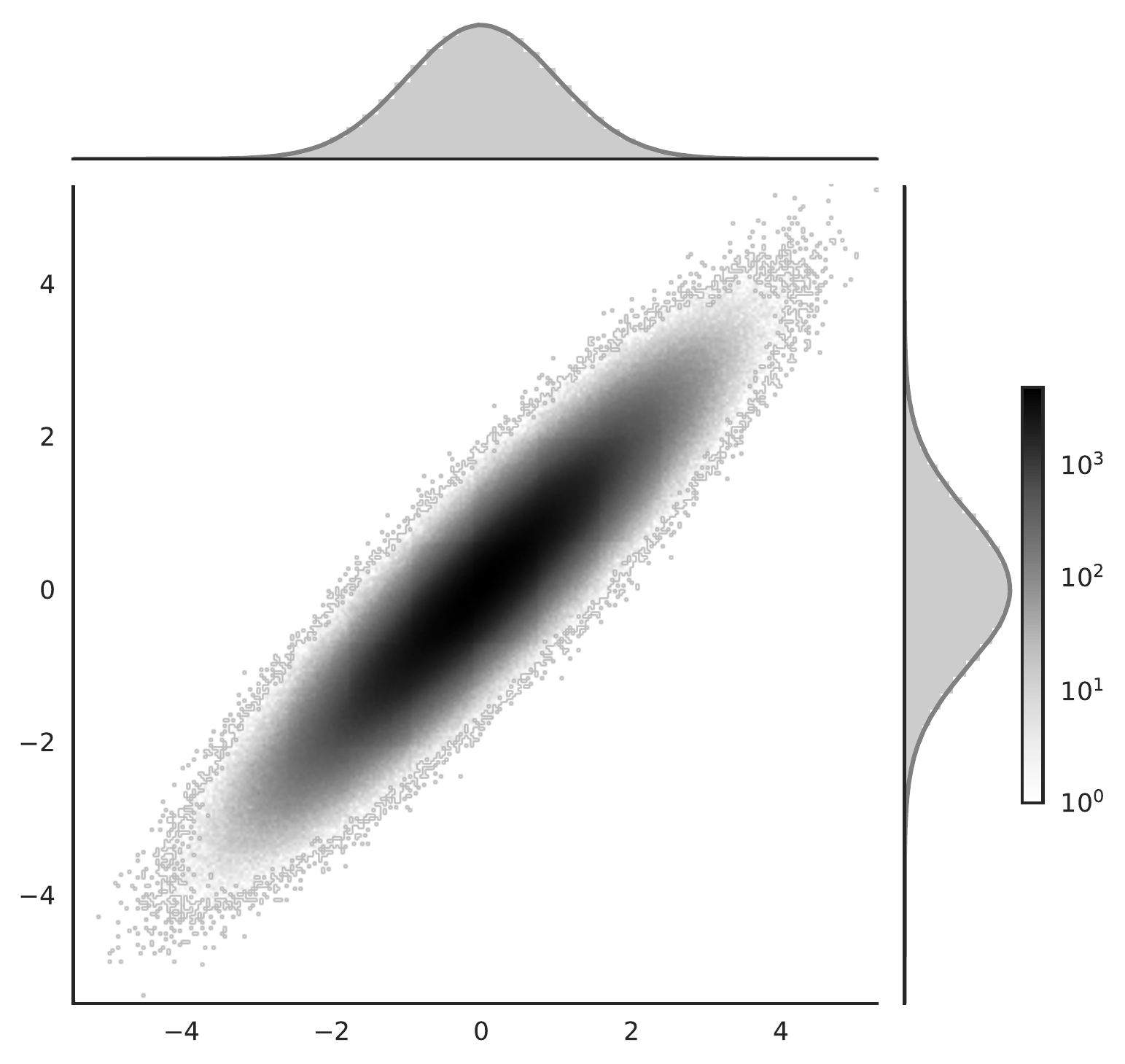}}
	\subfigure[Marginals $3$ and 
	$4$\label{fig::g34}]{\includegraphics[width=0.47\textwidth]{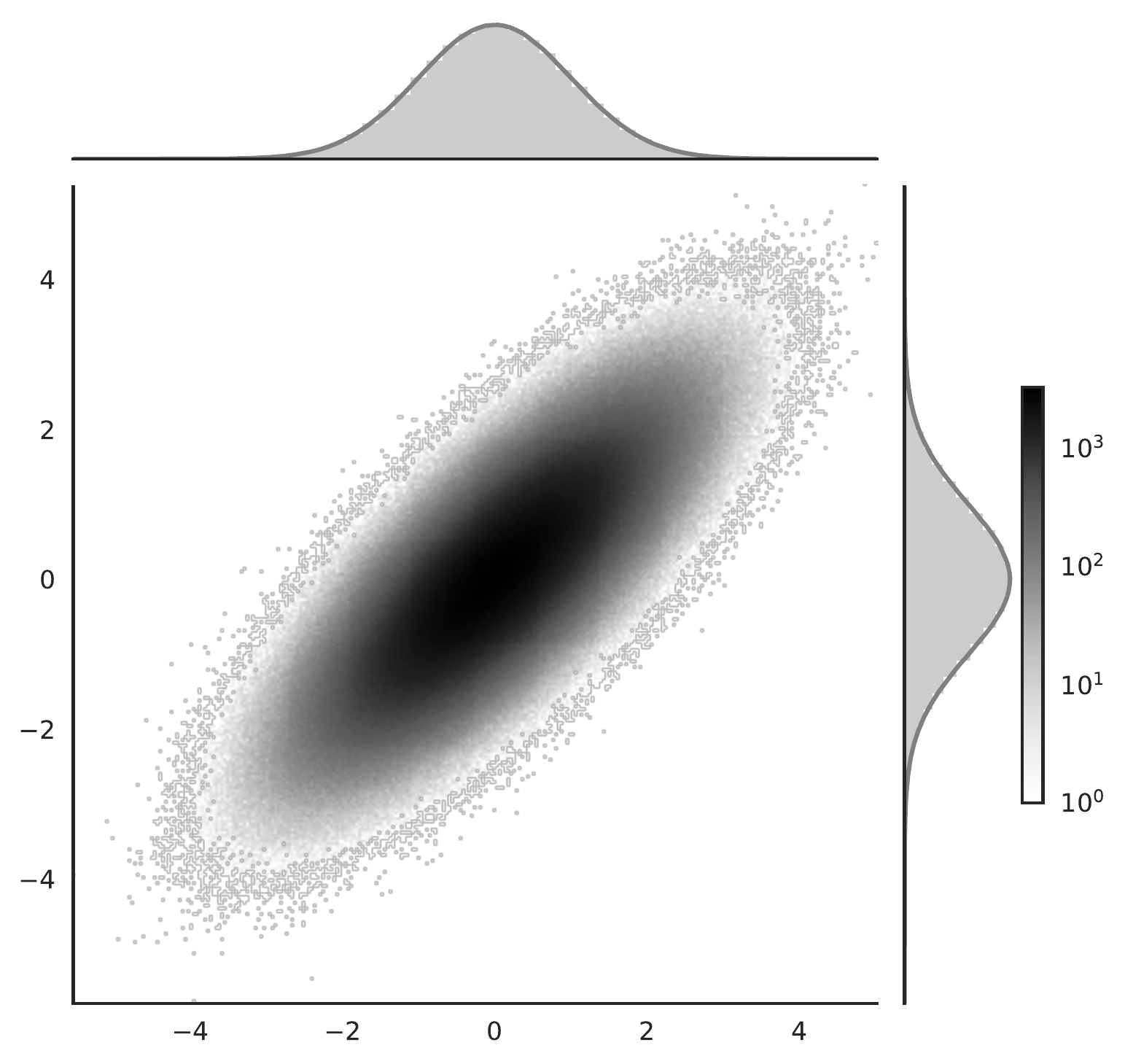}}
	\\ 
	\subfigure[Marginals $1$ and 
	$3$\label{fig::g13}]{\includegraphics[width=0.47\textwidth]{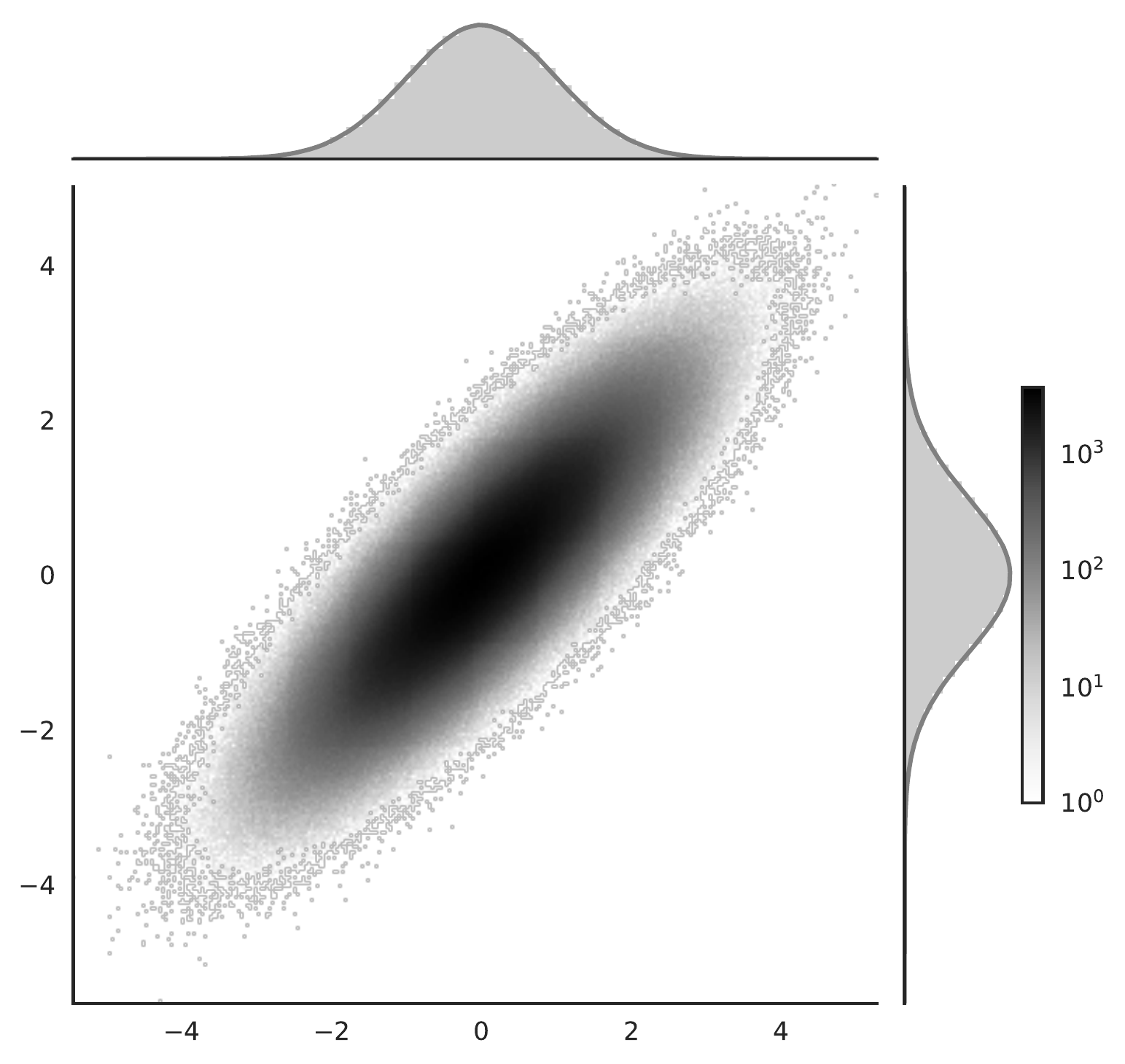}}
	\subfigure[Marginals $2$ and 
	$4$\label{fig::g24}]{\includegraphics[width=0.47\textwidth]{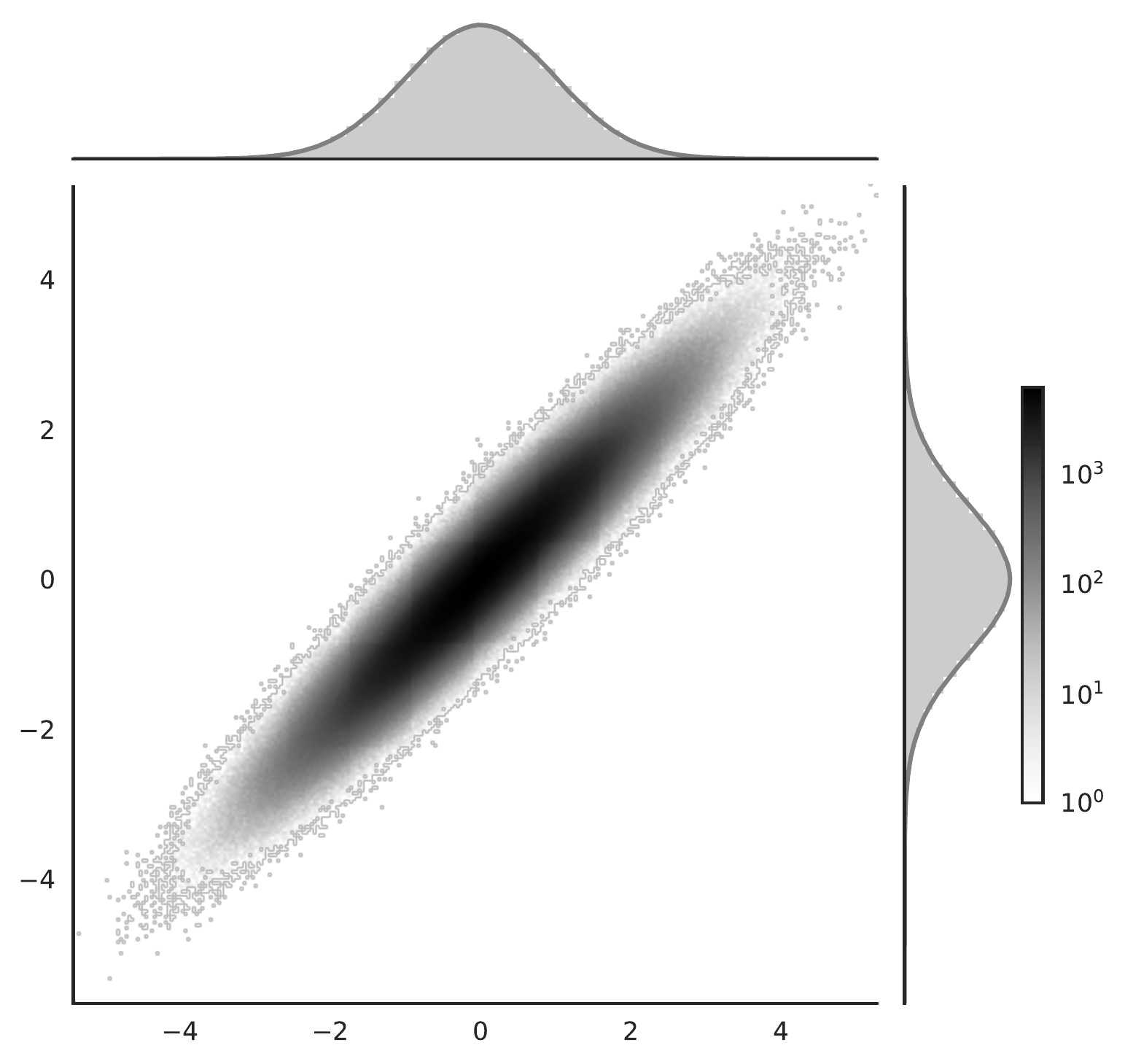}}
	\caption{Univariate and bivariate 
		histograms of the original data $\X$}\label{fig::g}
\end{figure}

\begin{figure}
	\subfigure[Marginals $1$ and 
	$2$\label{fig::ag12}]{\includegraphics[width=0.47\textwidth]{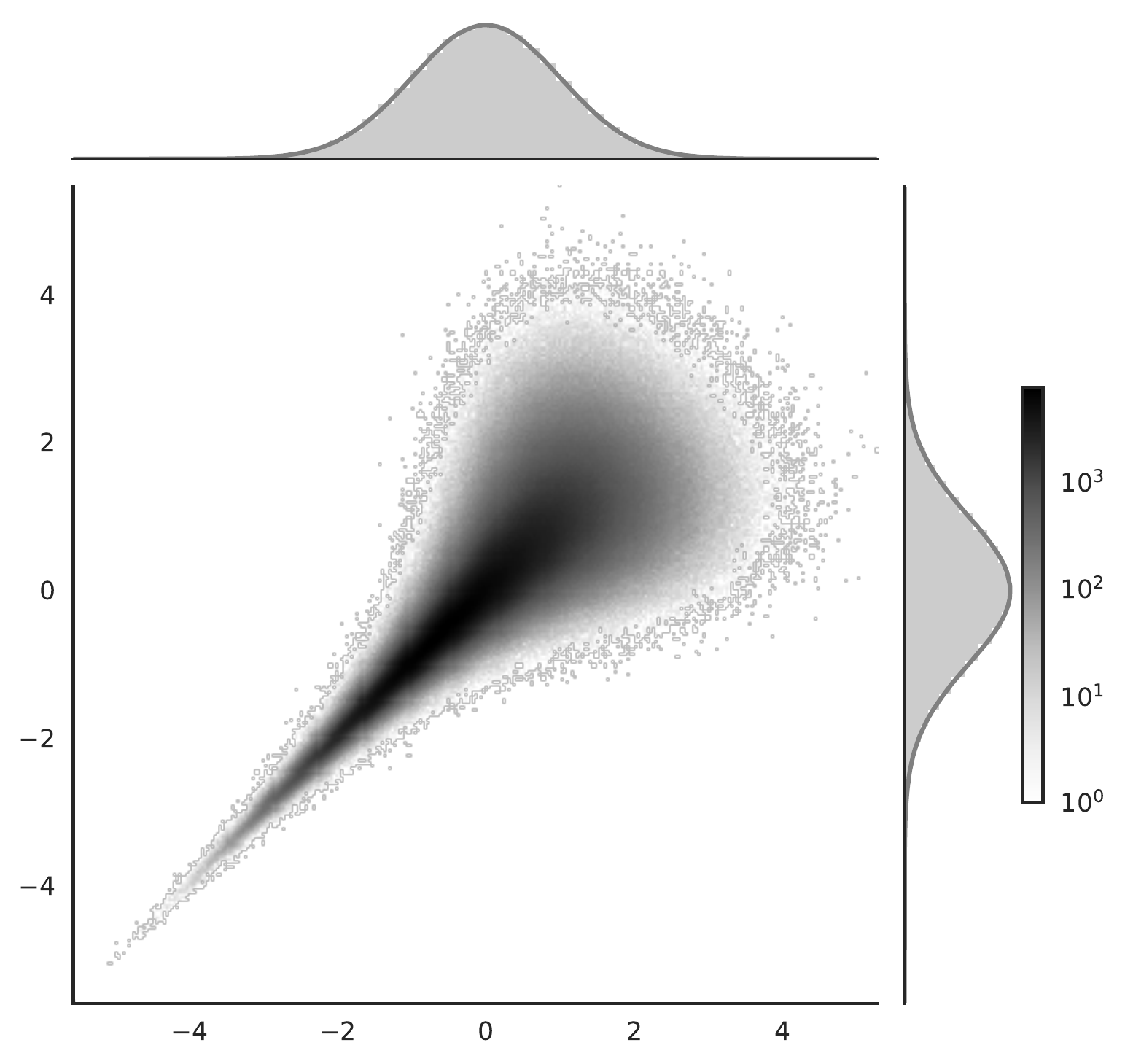}}
	\subfigure[Marginals $3$ and 
	$4$\label{fig::ag34}]{\includegraphics[width=0.47\textwidth]{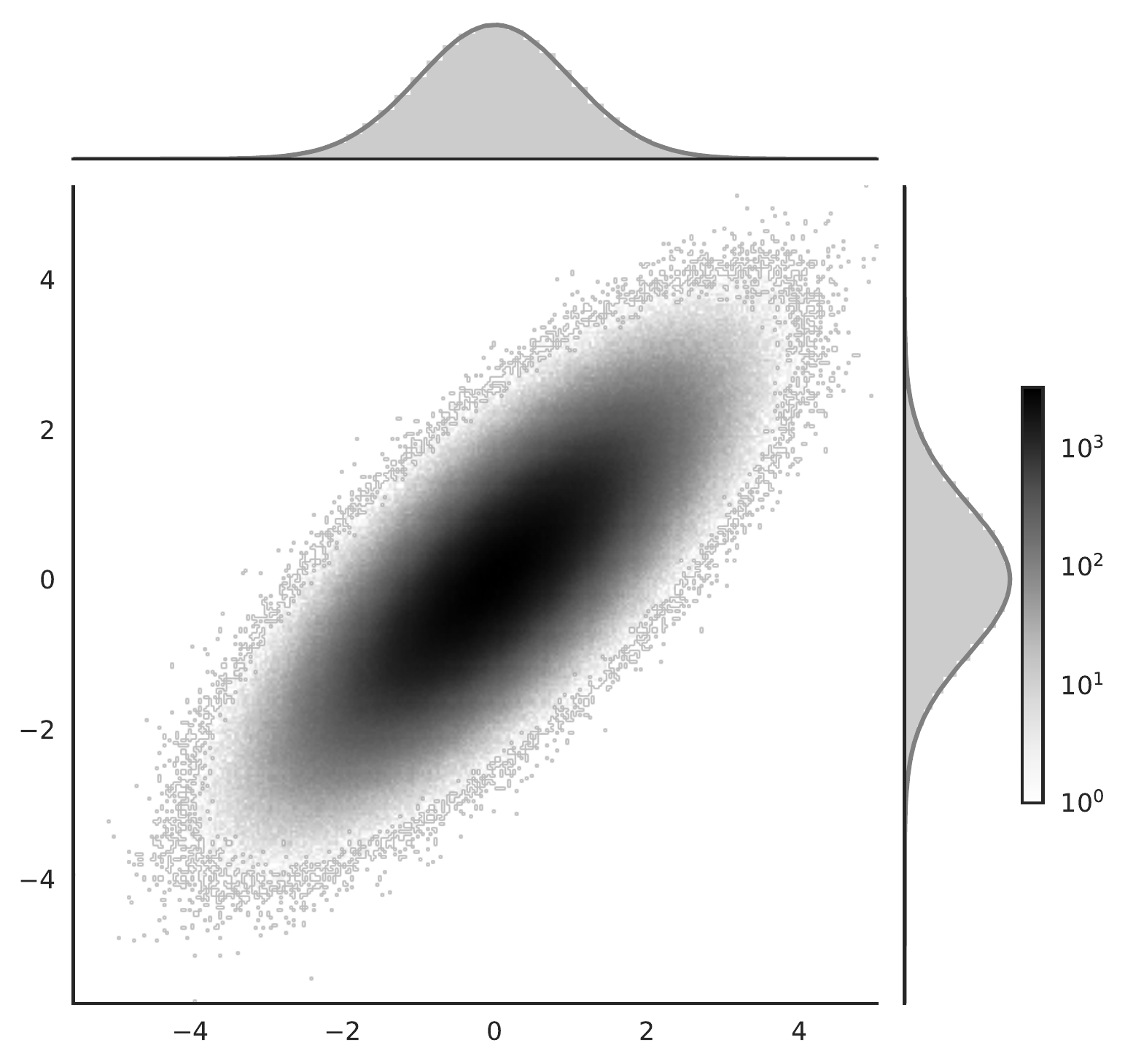}}
	\\ 
	\subfigure[Marginals $1$ and 
	$3$\label{fig::ag13}]{\includegraphics[width=0.47\textwidth]{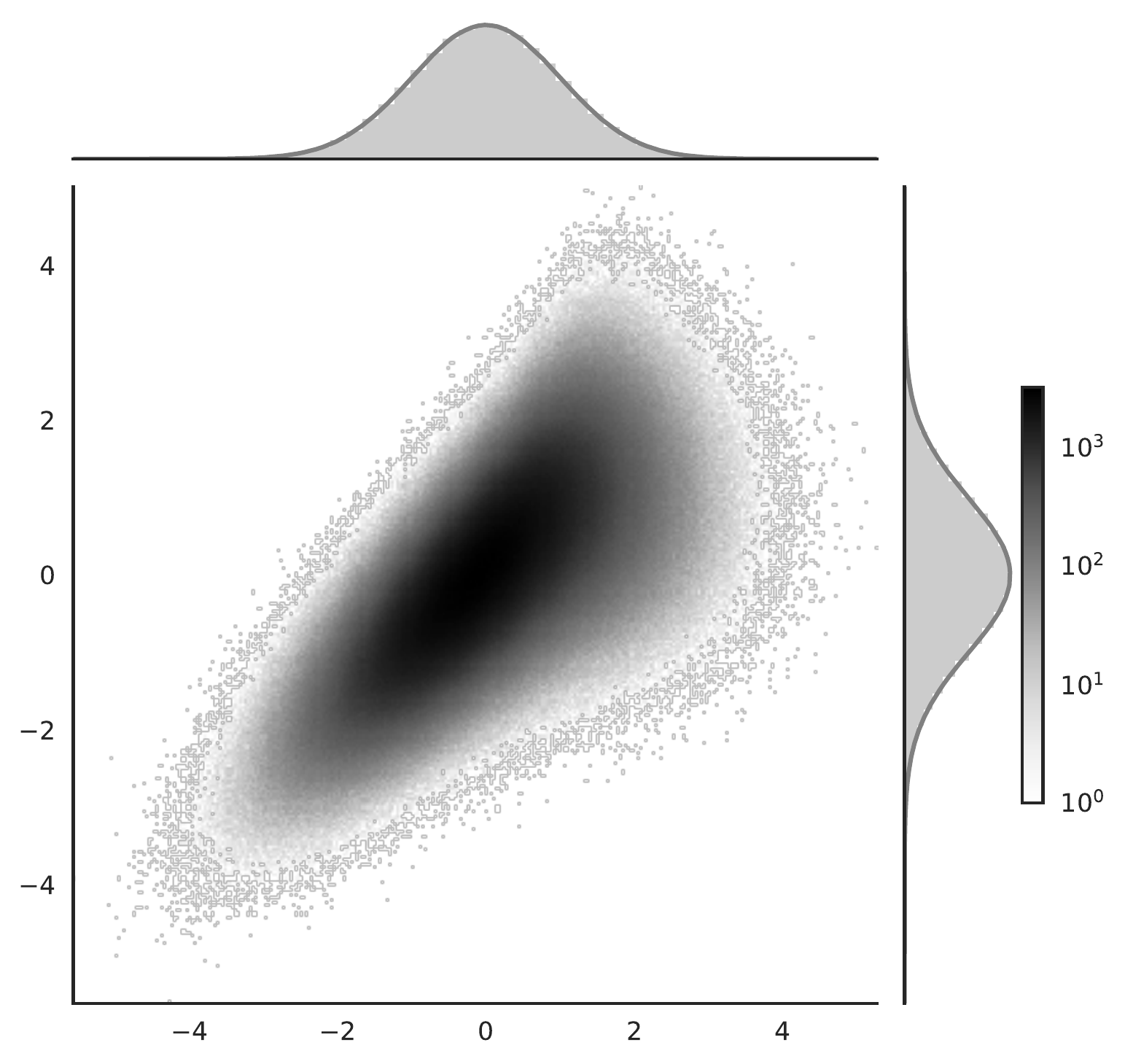}}
		\subfigure[Marginals $2$ and 
	$4$\label{fig::ag24}]{\includegraphics[width=0.47\textwidth]{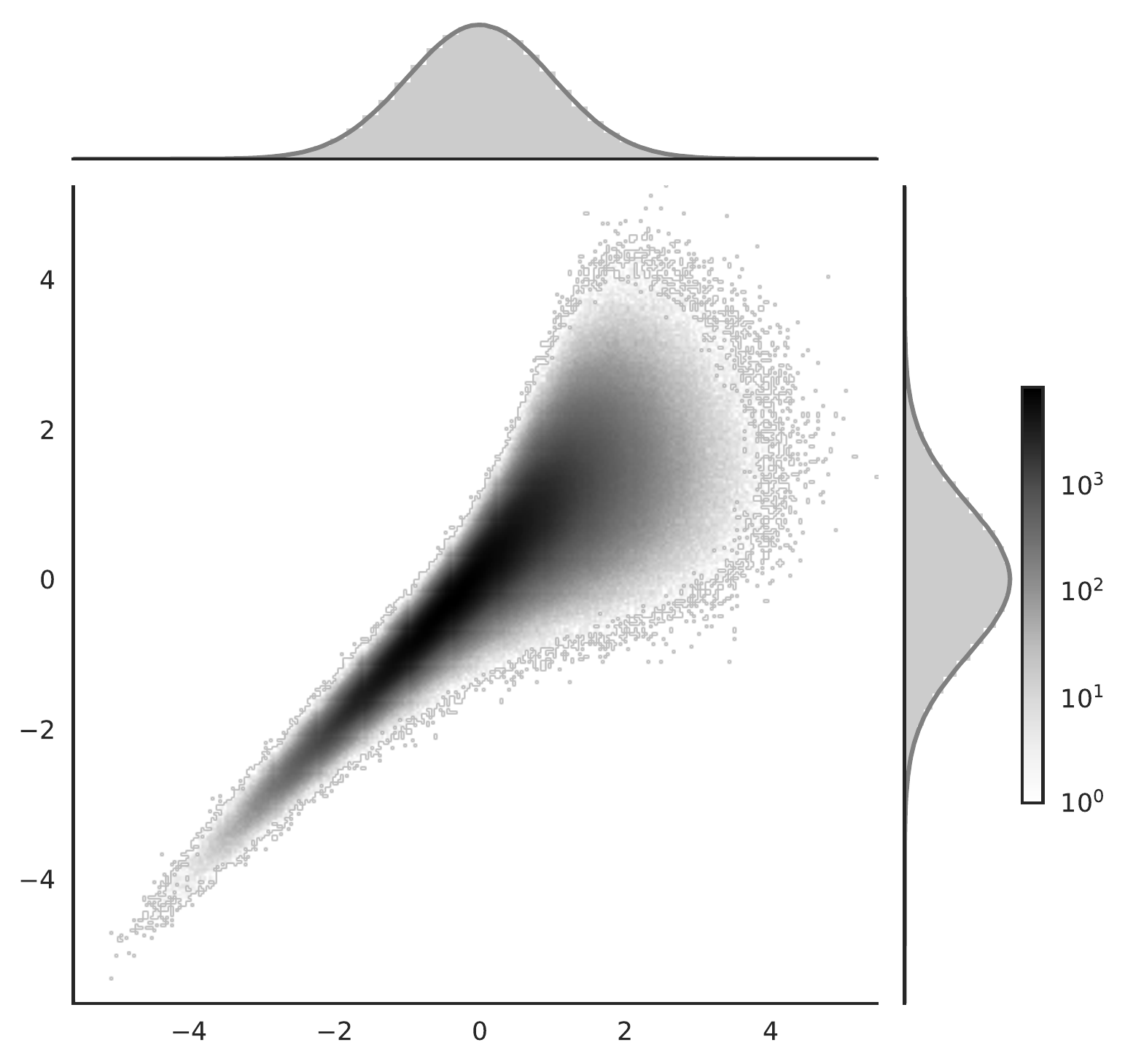}}
	\caption{Univariate and bivariate 
		histograms of transformed data $\X'$ by means of the Clayton 
		copula.}\label{fig::archandg}
\end{figure}

\chapter{Higher order statistics of multivariate data}\label{ch::hostats}

In this chapter we discuss higher order statistics of multivariate data, 
\ie~statistics of order $d$, where $d > 2$. Classical example of univariate 
higher 
order statistics are asymmetry ($d=3$) and kurtosis ($d=4$). These are 
normalised univariate higher order cumulants.  
In this chapter, we will concentrate on multivariate higher order cumulants, 
applicable 
to analyse non-Gaussian distributed multivariate data. For their 
practical application in multivariate non-Gaussian data analysis refer to  
signals analysis, for example in 
signal filtering~\cite{geng2011research, latimer2003cumulant}, finding the 
direction of received signals
\cite{porat1991direction, moulines1991second, cardoso1995asymptotic,
	liang2009joint} and signal auto-correlation
analysis~\cite{manolakos2000systematic}. Furthermore, those cumulants are used 
in
hyper-spectral image analysis~\cite{geng2015joint}, financial data
analysis~\cite{arismendi2014monte,jondeau2015moment} and neuroimage
analysis~\cite{birot2011localization, becker2014eeg}. Additionally univariate 
higher order cumulants are used in quantum noise 
investigation~\cite{gabelli2009high}, and computer network traffic analysis 
~\cite{bai2012modeling, salagean2010anomaly}. We focus on the 
multivariate data case, which can be easily simplified to the univariate~case.

Let us start first with some technical definitions including 
tensors~\cite{kolda2009tensor} and 
their super-symmetry~\cite{schatz2014exploiting}, since we use
super-symmetric tensors to store 
higher order cumulants. In our approach, the $d$-dimensional tensor is simply 
the 
$d$-dimensional array.
\begin{definition}\label{def::tensors}
Given the multi-index $\ii = (i_1, \ldots, i_d)$, where $i_k \in 
(1:I_k)$, see Table~\ref{tab::symbols} for notation explanation, and $|\ii| = 
d$, the $d$ mode tensor is the following 
$d$-dimensional array of data \cite{kolda2009tensor}
		\begin{equation}\label{equ:tensor}
		\mathcal{T}=\{t_{\ii}\}_{\ii=1}^{I_1,\ldots ,I_d}
		\in \R^{I_1 \times \ldots \times I_d} 
		\end{equation} 
 where the $k$\textsuperscript{th} element of the multi-index - $i_k$ 
 corresponds to the mode $k$. 
\end{definition}
The $3$ mode tensor is simply the $3$ dimensional data box, see 
\cite{kolda2009tensor} and Figure $1.1$ within.
\begin{definition} 
Let $\pi$ be a permutation of the multi-index $\ii$, for $d=3$ one of the 
permutations is
\begin{equation}
	\pi(i_1, i_2, i_3) = (i_2, i_1, i_3).
\end{equation}
Let $\Pi_d$ be a set of all permutations of the set $(1,2,\ldots d)$.
Tensor $\mathcal{T} \in \R^{\underbrace{n \times \dots \times n}_{d}}$ 
is 
	super-symmetric \cite{schatz2014exploiting} iff
\begin{equation} \forall_{\pi \in \Pi_d} \ \forall_{\ii} \ t_{\ii} 
	= t_{\pi(\ii)}.
\end{equation}
To be consistent with \cite{domino2018efficient} we use the following notation 
for the super-symmetric tensor
\begin{equation}
	\mathcal{T} \in \R^{[n, d]}.
\end{equation}
\end{definition}

\begin{definition}\label{def::cumsdef}
Following \cite{kendall1946advanced,lukacs1970characteristics}, let 
$\mathbf{f}(\mathbf{v}): \R^n \rightarrow \R$ be the
continuous $n$-variate Probability Density Function (PDF), its characteristic 
function is given by
\begin{equation}\label{eq::char}
	\varphi :\mathbb{R}^{n} \to 
	\mathbb{C} \ \ \  \varphi(\vv) = \int_{\R^n} 
	\exp(	\mathrm{i}\vv^{\intercal} \x) \mathbf{f}(\x) d \x,
\end{equation} 
where $\mathrm{i}$ is an imaginary unit. The $d$\textsuperscript{th} cumulant 
element of $\mathbf{f}$, indexed by $i_1, \ldots, i_d$ is given by
\begin{equation}\label{eq::cumdef}
	c_{i_1, \ldots, i_d} = (-\mathrm{i})^d \frac{\partial^d}{\partial 
		v_{i_1}, 
		\ldots, \partial v_{i_d}} \log(\varphi(\vv)) 
	\bigg{|}_{\vv = 0}.
\end{equation}
The element
$c_{i_1, \ldots, i_d}$ may by considered as a part of the $d$-mode 
super-symmetric tensor since differentiation in Eq.~\eqref{eq::cumdef} is 
commutative. Hereafter, for the $d$\textsuperscript{th} order 
cumulant's tensor we use the following notation
\begin{equation}
\CC_{d} \in \R^{[n,d]}.
\end{equation}
\end{definition}
Remark, that sometimes in literature one uses the cumulant generation function
\begin{equation}
\kappa :\mathbb{R}^{n} \to 
\mathbb{C} \ \  \kappa(\vv) = \log(\varphi(\vv)).
\end{equation}

\begin{remark}\label{rem::zeromeancums}
	It is easy to show from Eq.~\eqref{eq::char}, that if we use 
	$\mathbf{f}'(\x) = \mathbf{f}(\x - 
	\mathbf{a})$ where 
	$\mathbf{a} \in \R^n$ is a vector of constants, corresponding 
	characteristic function 
	 $\varphi'$ is
	\begin{equation} 
	\varphi'(\vv) = \varphi(\vv) \exp\left(-\mathrm{i} \vv^{\intercal} 
	\mathbf{a}\right)
	\end{equation}
	yielding
	\begin{equation}\label{eq::kammaprim}
		\kappa'(\vv) = \kappa(\vv) -\mathrm{i} \vv^{\intercal} 
		\mathbf{a}.
	\end{equation}
	What is important the term $-\mathrm{i} \vv^{\intercal} 
	\mathbf{a}$ would 
	vanish after double differentiation over $\vv$ in 
	Eq.~\eqref{eq::cumdef}. 
	Hence cumulants tensors of order $d \geq 2$ of $\mathbf{f}$ and 
	$\mathbf{f}'$ would be the same. We can note it formally by $c_{\ii} = 
	c'_{\ii}$ if $|\ii| \geq 2$. Hereafter as $\mathbf{f}'$ we will use 
	the zero mean frequency 
	distribution:
	$\mathbf{f}'(\x) = \mathbf{f}(\x - \mu)$, where $\mu = [\mu_1, \ldots, 
	\mu_n]$ is a mean vector of $\mathbf{f}$. Such approach will make the 
	cumulant 
	computation simpler.
\end{remark}

\begin{definition}\label{def::momdef}
	Analogically to Definition~\ref{def::cumsdef} the $d$\textsuperscript{th} 
	moment element is
	\begin{equation}\label{eq::momdiff}
	m_{i_1, \ldots, i_d} = (-\mathrm{i})^d \frac{\partial^d}{\partial 
		v_{i_1}, 
		\ldots, \partial v_{i_d}} \varphi (\vv)
	\Big{|}_{\vv = 0},
	\end{equation}
	for $\varphi (\vv)$ see Eq.~\eqref{eq::char}. As in the cumulant's 
	case, the above is an 
	element of the super-symmetric moment tensor
	\begin{equation}
	\MM_{d} \in \R^{[n,d]}.
	\end{equation}
	By performing differentiation of Eq.~\ref{eq::momdiff} we have
	\begin{equation}\label{eq::momelsf}
		m_{i_1, \ldots, i_d} = \int_{\R^n} x_{i_1} \cdot \ldots \cdot x_{i_d} 
		\mathbf{f}(\x) d \x.
	\end{equation}
\end{definition}

Having introduced cumulant's and moment's tensors from $\mathbf{f}(\x)$ that is 
the PDF function, 
observe from 
Eq.~\eqref{eq::co} that $\mathbf{f}(\x)$ can be 
split
onto the copula density $\mathbf{c}$ and
univariate marginal PDFs $f_i$. The later 
correspond with univariate marginal
CDFs $F_i$. From Eq.~\eqref{eq::momelsf} we have
\begin{equation}\label{eq::cumsmoms}
	m_{i_1, \ldots, i_d} = \int_{\R^n} x_{i_1} \cdot \ldots \cdot x_{i_d} \cdot 
	\mathbf{c}\left(F_1(x_1), \ldots, F_n(x_n)\right) \prod_{i = 
		1}^n f_i(x_i) d x_i.
\end{equation} 
Concluding, we have an impact on the moment's tensors elements 
both from a copula and from marginal distributions, what is rather complicated 
and will be discussed further in this book.

\begin{remark}\label{rem::univstats}
	We can shown that the super-diagonal element of $\MM_d$,~\ie~such 
	indexed by $i_1 = i_2 = \ldots = i_d = i$, corresponds to 
	the
	$d$\textsuperscript{th} moment of the $i$\textsuperscript{th} marginal 
	distribution and is copula independent. For the proof, let us start with
	\begin{equation}\label{eq::momcum}
	m_{i, \ldots, i} =\int_{\R^n} (x_{i})^d \mathbf{c}\left(F_1(x_1), \ldots, 
	F_i(x_i), \ldots, 
	F_n(x_n)\right) \prod_{i' = 1}^n f_{i'}(x_{i'}) d x_{i'}.
	\end{equation}
	 Using $ \ u_{i'} = F_{i'}(x_{i'})$ and $d u_{i'} = f(x_{i'}) dx_{i'}$, 
	where $u_{i'} \in [0,1]$, we can perform following integrations 
	\begin{equation}\label{eq::copints}
	\begin{split}
	\int_{[0,1]^{n-1}} &\mathbf{c}\left(u_1, \ldots, u_{i-1}, u_i, u_{i+1}, 
	\ldots u_n\right) d u_1\cdots 
	du_{i-1} du_{i+1} \cdots du_n \\ &= \frac{\partial \C(1,\ldots,1,u_i, 1, 
		\ldots, 1)}{\partial u_i} = \frac{\partial u_i}{\partial u_i} = 1.
	\end{split}
	\end{equation}
	We use here the relation between a copula density and a copula in 
	Eq.~\eqref{eq::subcop}, and the point $2$ of Definition~\ref{d::copdef}. 
	Next, using $x_i = F_i^{-1}(u_i)$ we have 
	\begin{equation}\label{eq::momwocums}
		m_{i, \ldots, i} =\int_{0}^1 \left(F_i^{-1}(u_{i})\right)^d du_i,
	\end{equation}
	and converting back to original marginal variable:
	\begin{equation}
	m_{i, \ldots i} =\int_{\R} (x_i)^d f_i(x_i) dx_i.
	\end{equation}
\end{remark}

\section{Cumulants of multivariate Gaussian distribution}\label{sec::cumsg}

In this section we conclude the meaning of higher order cumulants tensors in 
non-Gaussian data analysis, by showing that they are zero if data are 
multivariate Gaussian distributed. For this purpose, consider the following 
quadratic cumulant generation function
\begin{equation}\label{eq::h1}
\kappa(\vv) = \mathrm{i}\vv^{\intercal}\mu -\frac{1}{2}\vv^{\intercal} 
\SSSS \vv.
\end{equation}
where $\mu = [\mu_1, \ldots, \mu_n]$.
The first cumulant's element would be
\begin{equation}\label{eq:c1n}
	c_{i} = -\mathrm{i}\frac{\partial}{\partial v_{i}}
    \kappa(\vv) \Big|_{\vv = 0} = \mu_i,
\end{equation}
and the second cumulant's element
\begin{equation}\label{eq:c2n}
	 c_{i_1, i_2} = \frac{\partial^2}{\partial v_{i_1} \partial 
 	v_{i_2}} \kappa(\vv) \Big|_{\vv = 0} = s_{i_1, i_2},
\end{equation}
where $s_{i_1, i_2}$ is an element of $\SSSS$.
Obviously as $\kappa$ is quadratic in $\vv$, for $d > 2$ we have
\begin{equation}\label{eq:cdn}
	\forall_{d > 2} \ \  c_{i_1,\ldots, i_d} = \frac{\partial^d}{\partial 
	v_{i_1} 
	\ldots \partial 
	v_{i_d}} \kappa(\vv) \Big|_{\vv = 0} = 0.
\end{equation}
From $\kappa(\vv)$ given in Eq.~\eqref{eq::h1}, we can reconstruct the 
following characteristic function
\begin{equation}
\varphi(\vv) = \exp \left( \mathrm{i}\vv^{\intercal}\mu 
-\frac{1}{2}\vv^{\intercal} 
\SSSS \vv \right).
\end{equation}
This is a characteristic function of multivariate 
Gaussian distribution, see 
Eq.~\eqref{eq::gaussmulti} and~\cite{gut2009transforms}.

Oppositely, if the multivariate frequency distribution is non--Gaussian its 
characteristic function will have different form and higher than quadratic 
terms may appear yielding non-zero higher order cumulants. It is why the strong 
suggestion appears, that we can 
use higher order cumulants to extract information about the frequency 
distribution that is non-Gaussian. Given these, higher order cumulants have 
specific 
meaning as they carries information about the divergence form Gaussian model. 
And finally, higher order cumulants have an advantage over higher order 
moments, 
that are non-necessary non-zero for multivariate Gaussian distributed data.

\section{Tensors and tensor networks - quantum mechanics inspired 
tools}\label{sec::cumstn}

In this section we use the tensor and graphical notation to discuss relation 
between higher order cumulants and higher order moments. As such the approach 
is inspired on the quantum mechanics. 

\subsection{Moments tensors}\label{sec::momenttens}

Suppose we have the random vector $\XXX^{(n)}$, and its
$t$ realisations are stored in the matrix form $\X \in \R^{t \times 
n}$, here the $j$\textsuperscript{th} realisation is $\x_j = 
[x_{j,1}, \ldots, x_{j,n}]$. Suppose now, each such realisation is equal 
probable, what is a fair 
assumption for large $t$ - such that extreme events are sampled properly. Given 
the probability of each realisation as equal to $\frac{1}{t}$ moment's tensor's 
elements can be estimated as
way
\begin{equation}\label{eq::momel}
	m_{\ii}(\X) = \E(X_{i_1} \cdot \ldots \cdot X_{i_d}) = \frac{1}{t} 
	\sum_{j = 1}^t 
	\prod_{k=1}^d 
	x_{j, i_k}
\end{equation}
Here $X_i = [x_{1,i}, \ldots, x_{t,i}]^{\intercal}$ is a vector of all 
realisations of the $i$\textsuperscript{th} marginal and $\E$ is the 
expectation value operator. We can present Eq.~\eqref{eq::momel} using tensor 
operation and corresponding graphical notation.

\begin{definition}\label{def::tmmodemult}
Suppose we have $d$-mode tensor $\mathcal{T} \in \R^{t \times \ldots \times t}$ 
and matrix $\mathbf{A} \in \R^{n \times t}$. We defined the matrix times tensor
multiple in all modes \cite{kolda2009tensor},
\begin{equation}
	\R^{n \times \ldots \times n}  \ni \mathcal{T}' = \mathbf{A} 
	\times_{1,\ldots, d} \mathcal{T},
\end{equation} 
by the following multiple contraction
\begin{equation}
{t'}_{i_1, \ldots, i_d} = \sum_{j_1, \ldots, j_d} a_{i_1, j_1} \cdot \ldots 
\cdot a_{i_d,j_d} \cdot t_{j_1, 
	\ldots, j_d} .
\end{equation}
Obviously if $\mathcal{T}$ is super-symmetric, tensor
$\mathcal{T}'$ would be super-symmetric as well.
\end{definition}

\begin{definition}{\label{d::md1}}
	Let us define the identity tensor $\mathbb{1}_{(d)} \in \R^{[t, d]}$ as the 
	$n$'th mode tensor with $1$ on the super-diagonal~\cite{kolda2009tensor} 
	and zeros elsewhere
	\begin{equation}
	\left(\mathbb{1}_{(d)}\right)_{\ii} = \begin{cases} 1 \text{ if } i_1 = i_2 
	= \ldots = i_d \\ 0 \text{ otherwise}\end{cases}.
	\end{equation}
	For an exemplary $\mathbb{1}_{(3)}$ see Figure~\ref{fig:13}.
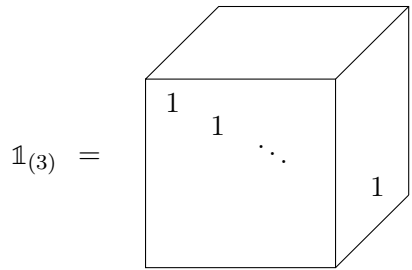
\begin{figure}[!h]
	\centering
	\begin{tikzpicture}
	\draw (2.5,0,0) coordinate (x) |- (0,2.5,0) coordinate [midway] (h) 
	coordinate 
	(y) 
	-- (0,2.5,2.5) coordinate (a) -- (0,0,2.5) coordinate (z) -- (2.5,0,2.5) 
	edge 
	(x) -- 
	(2.5,2.5,2.5) coordinate (v) edge (h)
	-- (a)  ;\node [label=45:1] at (-0.2,1.6,2) {};
	\node [label=45:1] at (0.2,1.1,1.5) {};
	\node [label=45:1] at (1.8,-0.2,0.2) {};
	\node [label=45:$\ddots$] at (0.6,0.5,1) {};
	\node [label=45:$\mathbb{1}_{(3)}$] at (-3,0,0) {};
	\node at (-1.7,0.5,0){$=$};
	\end{tikzpicture}
	\caption{The $\mathbb{1}_{(3)}$ example.}
	\label{fig:13}
\end{figure}
\end{definition}

\begin{remark}
	By taking $\R^{n \times t} \ni \mathbf{A} = \mathbf{X}^{\intercal}$ and 
	$\mathbb{1}_{(d)} \in \R^{[t, d]}$, Eq.~\eqref{eq::momel} can be 
	rewritten into the following form using tensor notation
	\begin{equation}\label{eq::momenttensornot}
	\MM_{d}(\X) = \frac{1}{t} \X^{\intercal} 
	\times_{1,\ldots, d} \1_{(d)}.
	\end{equation}
\end{remark}

\begin{definition}\label{def::matmatgraph} The matrix $\mathbf{A} \in \R^{n_a 
\times t}$ is the 
$2$ 
mode tensor. As such, it can be presented using 
	the graphic notation
	\begin{center}
		\begin{tikzpicture}
		\node[shape=rectangle,draw=black] (C) at (1,0) {$\mathbf{A}$};
		\node at (0.46,0.2) {\scriptsize{$(1)$}};
		\node at (1.6,0.2) {\scriptsize{$(2)$}};
		\path [-] (C) edge node[left] {} (0,0);
		\path [-] (C) edge node[right] {} (2,0);
		\end{tikzpicture}
	\end{center}
	where $(1)$ and $(2)$ enumerates its modes \cite{kolda2009tensor}. 
	Furthermore, 
	let $\mathbf{B} \in \R^{t \times n_b}$ be another matrix. Their 
	multiplication 
	\begin{equation}
		\R^{n_a \times n_b} \ni \mathbf{A}' = \mathbf{A} \mathbf{B}
	\end{equation} 
	is the contraction
	\begin{equation}\label{eq::simplecontract}
		a'_{i,j} = \sum_k a_{i,k} b_{k,j}
	\end{equation} 
	that can be represented by means of the graphic notation
	\begin{center}
		\begin{tikzpicture}
			\node[shape=rectangle,draw=black] (A) at (4,0) 
			{$\mathbf{A}$};
			\node[shape=rectangle,draw=black] (B) at (5.8,0) {$\mathbf{B}$};
			\node[shape=rectangle,draw=black] (C) at (1,0) {$\mathbf{A}'$};
			\node at (0.44,0.2) {\scriptsize{$(1)$}};
			\node at (1.6,0.2) {\scriptsize{$(2)$}};
			\node at (3.44,0.2) {\scriptsize{$(1)$}};
			\node at (4.55,0.2) {\scriptsize{$(2)$}};
			\node at (5.3,0.2) {\scriptsize{$(1)$}};
			\node at (6.32,0.2) {\scriptsize{$(2)$}};
			\path [-] (C) edge node[left] {} (0,0);
			\path [-] (C) edge node[right] {} (2,0);
			\node at (2.5,0){$=$};
			\path [-] (A) edge node[left] {} (B);
			\path [-] (A) edge node[left] {} (3.0,0);
			\path [-] (B) edge node[left] {} (6.8,0);
		\end{tikzpicture}
	\end{center}
	The connection of node $(2)$ in $\mathbf{A}$ with node $(1)$ in 
	$\mathbf{B}$ 
	means the contraction of mode $(2)$ of $\mathbf{A}$ with mode $(1)$ of 
	$\mathbf{B}$, as in Eq.~\eqref{eq::simplecontract}.
\end{definition}

\begin{definition}\label{def::graphicmoms}
	The $d$ mode tensor $\mathcal{T}$ can be represented graphically as
\begin{center}
	\begin{tikzpicture}
	\node[shape=rectangle,draw=black] (C) at (0,0) {$\mathcal{T}$};
	\path [-] (C) edge node[left] {} (-1.1,1);
	\path [-] (C) edge node[left] {} (-1.1,-1);
	\node at (0,1){$\cdots$};
	\node at (0.75,0.75){$\ddots$};
	\node at (0.75,-0.75){\reflectbox{$\ddots$}};
	\node at (0,-1){$\cdots$};
	\node at (-0.27,0.51) {\scriptsize{$(1)$}};
	\node at (-0.65,-0.25) {\scriptsize{$(d)$}};
	\node at (1.2,0.1) {$\vdots$};
	\end{tikzpicture}
\end{center}
In this representation, particular node of the graph corresponds to particular 
mode of the tensor. If a 
tensor (or a 
matrix) is 
super-symmetric we drop nodes numeration because it does not matter.
Suppose we have the super-symmetric tensor $\mathcal{T} \in \R^{[n,d]}$, its 
multiplication with a matrix in all modes $\mathcal{T}' = \mathcal{T} 
\times_{1,\ldots, d} \mathbf{A}$ can be represented graphically:
\begin{center}
	\begin{tikzpicture}
	\node[shape=rectangle,draw=black] (C) at (0,0) {$\mathcal{T}'$};
	\node at (1.7,0){$=$};
	\path [-] (C) edge node[left] {} (-1.1,1);
	\path [-] (C) edge node[left] {} (1.1,1);
	\path [-] (C) edge node[left] {} (-1.1,-1);
	\path [-] (C) edge node[left] {} (1.1,-1);
	\node at (0,1){$\cdots$};
	\node at (0,-1){$\cdots$};
	\node at (1.2,0.1) {$\vdots$};
	\node[shape=rectangle,draw=black] (D) at (4,0) {$\mathcal{T}$};
	\node at (4,1) {$\cdots$};
	\node[shape=rectangle,draw=black] (C) at (2.8,1) {$\mathbf{A}$};
	\node at (3.4, .76) {\scriptsize{$(2)$}};
	\node at (2.63, 1.45) {\scriptsize{$(1)$}};
	\node at (3.0, -.55) {\scriptsize{$(2)$}};
	\node at (2.23, -1.25) {\scriptsize{$(1)$}};
	\node at (5.2,0.1) {$\vdots$};
	\node at (4.6, .78) {\scriptsize{$(2)$}};
	\node at (5.4, 1.5) {\scriptsize{$(1)$}};
	\node at (5.05, -.55) {\scriptsize{$(2)$}};
	\node at (5.8, -1.25) {\scriptsize{$(1)$}};
	\node at (5.2,0.1) {$\vdots$};
	\node[shape=rectangle,draw=black] (E) at (5.2,1) {$\mathbf{A}$};
	\node at (4,-1) {$\cdots$};
	\node[shape=rectangle,draw=black] (G) at (2.8,-1) {$\mathbf{A}$};
	\node[shape=rectangle,draw=black] (I) at (5.2,-1) {$\mathbf{A}$};
	
	\path [-] (D) edge node[left] {} (E);
	\path [-] (D) edge node[left] {} (G);
	\path [-] (D) edge node[left] {} (C);
	\path [-] (D) edge node[left] {} (I);
	\path [-] (E) edge node[left] {} (5.9,1.7);
	\path [-] (C) edge node[left] {} (2.1,1.7);
	\path [-] (I) edge node[left] {} (5.9,-1.7);
	\path [-] (G) edge node[left] {} (2.1,-1.7);
	\end{tikzpicture}
\end{center}
Note that only the numeration of modes for the matrix matters.	
\end{definition}

Having introduced the graphic notation, referring to 
Eq.~\eqref{eq::momenttensornot} the 
$d$\textsuperscript{th} moment tensor can be represented in the graphical form:
\begin{center}
	\centering
	\begin{tikzpicture}
	\node[shape=rectangle,draw=black] (C) at (0,0) {$\MM_d(\X)$};
	\node at (1.7,0){$=$};
	\path [-] (C) edge node[left] {} (-1.1,1);
	\path [-] (C) edge node[left] {} (1.1,1);
	\path [-] (C) edge node[left] {} (-1.1,-1);
	\path [-] (C) edge node[left] {} (1.1,-1);
	\node at (0,1){$\cdots$};
	\node at (0,-1){$\cdots$};
	\node at (-1,1.5) {$(1)$};
	\node at (-1,-1.5) {$(d)$};
	\node at (1.2,0.1) {$\vdots$};
	\node[shape=rectangle,draw=black] (D) at (4,0) {$\frac{1}{t}\cdot 
	\1_{(d)}$};
	\node at (4,1) {$\cdots$};
	\node[shape=rectangle,draw=black] (C) at (2.8,1) {$\X$};
	\node at (3.4, .76) {\scriptsize{$(1)$}};
	\node at (2.63, 1.45) {\scriptsize{$(2)$}};
	\node at (3.0, -.55) {\scriptsize{$(1)$}};
	\node at (2.23, -1.25) {\scriptsize{$(2)$}};
	\node at (5.2,0.1) {$\vdots$};
	\node at (4.6, .78) {\scriptsize{$(1)$}};
	\node at (5.4, 1.5) {\scriptsize{$(2)$}};
	\node at (5.05, -.55) {\scriptsize{$(1)$}};
	\node at (5.8, -1.25) {\scriptsize{$(2)$}};
	\node at (5.2,0.1) {$\vdots$};
	\node[shape=rectangle,draw=black] (E) at (5.2,1) {$\X$};
	\node at (4,-1) {$\cdots$};
	\node[shape=rectangle,draw=black] (G) at (2.8,-1) {$\X$};
	\node[shape=rectangle,draw=black] (I) at (5.2,-1) {$\X$};
	
	\path [-] (D) edge node[left] {} (E);
	\path [-] (D) edge node[left] {} (G);
	\path [-] (D) edge node[left] {} (C);
	\path [-] (D) edge node[left] {} (I);
	\path [-] (E) edge node[left] {} (5.9,1.7);
	\path [-] (C) edge node[left] {} (2.1,1.7);
	\path [-] (I) edge node[left] {} (5.9,-1.7);
	\path [-] (G) edge node[left] {} (2.1,-1.7);
	\end{tikzpicture}
\end{center}
operation $\frac{1}{t}\cdot 
\1_{(d)}$ is simply a scalar element-wise multiple of $\1_{(d)}$ by 
$\frac{1}{t}$.

\begin{remark}\label{rem::cmom}
	If we use zero mean frequency 
	distribution~\ie~such that $\mathbf{f}'(\x) = \mathbf{f}(\x-\mu)$, we can 
	calculate central moments $m'_{\ii}$, see Eq.s~\eqref{eq::char} 
	and~\eqref{eq::momdiff}. Obviously, the first central moment would be zero.
	For large $t$, we can use following central moments estimators
	\begin{equation}\label{eq::cmomel}
	m'_{\ii}(\X) = m_{\ii}(\tilde{\X}) = \E(\tilde{X}_{i_1} \cdot \ldots \cdot 
	\tilde{X}_{i_d}) = \frac{1}{t} 
	\sum_{j = 1}^t 
	\prod_{k=1}^d 
	\left(x_{j, i_k} - \mu_{i_k}\right),
	\end{equation}
	where $\tilde{X}_{i}$ is a vector with zero mean, and $\tilde{\X}$  
	column (marginal) wisely centred data. Elements of $\tilde{\X}$ given 
	original data 
	$\X$, are:
	\begin{equation}
	\tilde{x}_{j,i} = x_{j,i} - \mu_i,
	\end{equation}
	where:
	\begin{equation}
	\mu_i = m_i(\X) = \E(X_i) = \frac{1}{t} \sum_{j=1}^t x_{j,i}.
	\end{equation}
	What is important, 
	following~\cite{domino2018efficient}, we use central 
	moments to compute cumulants since 
	such approach allows to simplify formulas. For
	graphic representation of central moment we can simply to replace $\X$ by 
	$\tilde{\X}$. 
\end{remark}

\subsection{Cumulants tensors}\label{sec::cumtens}

Having discussed higher order moments tensors estimators we can move to 
cumulant's calculation formulas that uses corresponding moments. Formulas for 
calculation cumulants of order $1-3$ are relatively simple, hence to give some 
examples we can derive them step by step from Definition~\ref{def::cumsdef}, 
see also Appendix $B$ in~\cite{domino2018efficient}.
\begin{example}
To compute $1$\textsuperscript{st} cumulant let us single differentiate 
Eq.~\eqref{eq::cumdef}
\begin{equation}\label{eq::c1}
c_{i} = -\mathrm{i} \frac{\partial}{\partial 
	v_{i}} \log(\varphi(\vv)) \bigg{|}_{\vv = 0} = -\mathrm{i}  
	\frac{\frac{\partial}{\partial 
		v_{i}} \varphi(\vv) }{\varphi(\vv)}\bigg{|}_{\vv = 0} = -\mathrm{i} 
		\frac{\partial}{\partial 
	v_{i}} \varphi(\vv) \Big{|}_{\vv = 0} = m_i,
\end{equation}
we used Definition~\ref{def::momdef} and the fact that according to 
Eq.~\eqref{eq::char} we have
 \begin{equation}
\varphi(\vv = 0)  = \int_{\R^n} 
\exp(0) \mathbf{f}(\x) d \x = \int_{\R^n} 
\mathbf{f}(\x) d \x = 1,
\end{equation}
due to the normalisation of the PDF function. 
\end{example}
\begin{example}
The second cumulant is
\begin{equation}\label{eq::c2}
\begin{split}
c_{i_1,i_2} =& -\frac{\partial^2}{\partial 
	v_{i_1} \partial v_{i_2}} \log(\varphi(\vv)) \bigg{|}_{\vv = 0} = 
	-\left( \frac{ \varphi(\vv)
	\frac{\partial^2}{\partial 
		v_{i_1} \partial v_{i_2}} \varphi(\vv) -  
		\left(\frac{\partial}{\partial 
		v_{i_1}}\varphi(\vv) \right) \left( \frac{\partial}{\partial 
		v_{i_2}} \varphi(\vv) \right) 
		}{\left(\varphi(\vv)\right)^2}\right)\bigg{|}_{\vv = 0} \\ =& 
		-\frac{
			\frac{\partial^2}{\partial 
				v_{i_1} \partial v_{i_2}} \varphi(\vv)}{\varphi(\vv)} 
				\bigg{|}_{\vv = 0} + \frac{\left(\frac{\partial}{\partial 
				v_{i_1}}\varphi(\vv) \right) \left( \frac{\partial}{\partial 
				v_{i_2}} \varphi(\vv) \right)
		}{\left(\varphi(\vv)\right)^2}\bigg{|}_{\vv = 0}
		= m_{i_1,i_2} - m_{i_1} 
		m_{i_2}.
\end{split}
\end{equation}
\end{example}
\begin{example}
The third cumulant requires another 
differentiation of Eq.~\eqref{eq::c2}.
Hereafter, for clarity we use $\varphi$ for $\varphi(\vv)$ and 
$\frac{\partial \varphi}{\partial v_i}$ for $\left(\frac{\partial}{\partial 
v_i} \varphi(\vv)\right)$ etc. We have
\begin{equation}\label{eq::c3dif}
\begin{split}
c_{i_1,i_2, i_3} &= \mathrm{i}\frac{\partial^3}{\partial 
	v_{i_1} \partial v_{i_2} \partial v_{i_3}} \log(\varphi) \bigg{|}_{\vv 
	= 0} \\ &= \mathrm{i} \left(\frac{
	\varphi\frac{\partial^3 \varphi}{ \partial 
		v_{i_1} \partial v_{i_2} \partial v_{i_3}} - 
		\frac{\partial^2 \varphi }{ 
		\partial v_{i_1} \partial v_{i_2}}  
		\frac{\partial \varphi }{ \partial v_{i_3}} }{\varphi^2} \right.\\
  &-
\left.\frac{\varphi^2\left(\frac{\partial^2 \varphi }{\partial 
		v_{i_1} \partial 
		v_{i_3}} \frac{\partial \varphi }{\partial 
		v_{i_2}}  + \frac{\partial \varphi}{\partial 
		v_{i_1}} \frac{\partial^2 \varphi }{\partial 
		v_{i_2}  \partial v_{i_3}} \right) - 2 \varphi \frac{\partial \varphi 
		}{\partial 
	v_{i_1}} \frac{\partial \varphi }{\partial 
	v_{i_2}}  \frac{\partial \varphi }{\partial 
	v_{i_3}}}{\varphi^4}\right)_{\vv = 0} \\
&= \mathrm{i} \frac{
	\frac{\partial^3 \varphi}{ \partial 
		v_{i_1} \partial v_{i_2} \partial v_{i_3}}}{\varphi} 
		\bigg{|}_{\vv = 0} - \mathrm{i} \frac{\frac{\partial^2 
		\varphi}{\partial 
				v_{i_1} \partial 
				v_{i_2}} \frac{\partial \varphi}{\partial 
				v_{i_3}}  + \frac{\partial^2 \varphi}{\partial 
				v_{i_1} \partial 
				v_{i_3}} \frac{\partial \varphi}{\partial 
				v_{i_2}}  + \frac{\partial \varphi}{\partial 
				v_{i_1}} \frac{\partial^2 \varphi}{\partial 
				v_{i_2}  \partial 
				v_{i_3}} 
			}{\varphi^2} \bigg{|}_{\vv = 0} \\ &+ 2 \mathrm{i} \frac{\varphi 
			\frac{\partial \varphi}{\partial 
			v_{i_1}}\frac{\partial \varphi}{\partial 
			v_{i_2}} \frac{\partial \varphi}{\partial 
			v_{i_3}} }{\varphi^4} \bigg{|}_{\vv = 0}
\end{split}
\end{equation}
what leads to
\begin{equation}\label{eq::c3}
c_{i_1, i_2, i_3} = m_{i_1, i_2, i_3} - m_{i_1, i_2} m_{i_3} - m_{i_1, i_3} 
m_{i_2} - m_{i_2, i_3} m_{i_1} + 2 m_{i_1} m_{i_2} m_{i_3}.
\end{equation}
Referring to Eqs.~\eqref{eq::c1} and \eqref{eq::c2} we can simplify this as
\begin{equation}\label{eq::c3e}
m_{i_1, i_2, i_3} = c_{i_1, i_2, i_3} + c_{i_1, i_2} c_{i_3} + c_{i_1, i_3} 
c_{i_2} + c_{i_2, i_3} c_{i_1} + c_{i_1} c_{i_2} c_{i_3}.
\end{equation}
\end{example}
To simplify this relation further let us introduce $c'_{\ii}$, the element of 
the 
cumulant's tensor of the zero mean frequency distribution $\mathbf{f}'(\x) = 
\mathbf{f}(\x - \mu)$. By Remark~\ref{rem::zeromeancums} we have
\begin{equation}
c'_{\ii} = \begin{cases} 0 &\text{ if } |\ii| = 1 \\
 c_{\ii} &\text{ if } |\ii| \geq 2 \end{cases}.
\end{equation}
Following~\cite{domino2018efficient} such approach simplifies cumulant's 
calculation formulas.

\begin{remark}
In the case of the $2$\textsuperscript{nd} cumulant, we can simplify 
Eq.~\eqref{eq::c2}
\begin{equation}
c_{i_1, i_2} = c'_{i_1, i_2} =  m'_{i_1, i_2} - m'_{i_1} m'_{i_2} = m'_{i_1, 
i_2}
\end{equation}
in a graphic tensor notation, using Definition~\ref{def::matmatgraph} 
we have
	\begin{center}
		\begin{tikzpicture}
		\node[shape=rectangle,draw=black] (A) at (4,0) 
		{$\tilde{\X}$};
		\node[shape=rectangle,draw=black] (B) at (5.5,0) {$\frac{1}{t} \cdot 
		\1$};
		\node[shape=rectangle,draw=black] (D) at (7.,0) {$\tilde{\X}$};
		\node[shape=rectangle,draw=black] (C) at (1,0) {$\CC_2$};
		\node at (6.45,0.18) {\scriptsize{$(1)$}};
		\node at (7.55,0.18) {\scriptsize{$(2)$}};
		\node at (4.55,0.18) {\scriptsize{$(1)$}};
		\node at (3.45,0.18) {\scriptsize{$(2)$}};
		\path [-] (C) edge node[left] {} (0,0);
		\path [-] (C) edge node[right] {} (2,0);
		\path [-] (B) edge node[left] {} (D);
		\node at (2.5,0){$=$};
		\path [-] (A) edge node[left] {} (B);
		\path [-] (A) edge node[left] {} (3.0,0);
		\path [-] (D) edge node[left] {} (8.,0);
		\node at (8.4,0){$=$};
		\node at (8.85,0){$\frac{1}{t} \cdot$};
		\node[shape=rectangle,draw=black] (E) at (10.,0.) {$\tilde{\X}$};
		\node[shape=rectangle,draw=black] (F) at (12.,0.) {$\tilde{\X}$};
		\path [-] (E) edge node[left] {} (F);
		\path [-] (E) edge node[right] {} (9.1,0);
		\path [-] (F) edge node[left] {} (13.,0);
		\node at (11.45,0.18) {\scriptsize{$(1)$}};
		\node at (12.55,0.18) {\scriptsize{$(2)$}};
		\node at (10.55,0.18) {\scriptsize{$(1)$}};
		\node at (9.45,0.18) {\scriptsize{$(2)$}};
		\end{tikzpicture}
	\end{center}
Where $\1$ is simply the identity matrix. In the matrix notation we have
	\begin{equation}
		\CC_2(\mathbf{X}) = \frac{1}{t} \tilde{\mathbf{X}}^{\intercal} 
		\tilde{\mathbf{X}}
	\end{equation}
\end{remark}
\begin{remark}\label{rem::cedef}
In the case of the $3$\textsuperscript{rd} cumulant, we can simplify 
Eq.~\eqref{eq::c3e}
\begin{equation}
c_{i_1, i_2, i_3} = c'_{i_1, i_2, i_3} =  m'_{i_1, i_2, i_3} - c'_{i_1, i_2} 
c'_{i_3} - c'_{i_1, i_3} c'_{i_2}  - c'_{i_2, i_3} c'_{i_1} - 
c'_{i_1}c'_{i_2}c'_{i_3} = m'_{i_1, i_2, i_3},
\end{equation}
in a graphic notation, see Definition~\ref{def::graphicmoms}, we have:
\begin{center}
	\begin{tikzpicture}
	\node[shape=rectangle,draw=black] (C) at (0,0) {$C_3$};
	\node at (1.6,0){$=$};
	\path [-] (C) edge node[left] {} (0,1);
	\path [-] (C) edge node[left] {} (-1,-1);
	\path [-] (C) edge node[left] {} (1,-1);
	\node[shape=rectangle,draw=black] (D) at (4,0.1) {$\frac{1}{t} \cdot 
	\mathbb{1}_{(3)}$};
	\node[shape=rectangle,draw=black] (E) at (4,1.3) {$\tilde{\X}$};
	\node[shape=rectangle,draw=black] (F) at (2.9,-1.1) {$\tilde{\X}$};
	\node[shape=rectangle,draw=black] (G) at (5.1,-1.1) {$\tilde{\X}$};
	\path [-] (D) edge node[left] {} (E);
	\path [-] (D) edge node[left] {} (F);
	\path [-] (D) edge node[left] {} (G);
	\path [-] (E) edge node[left] {} (4,2.2);
	\path [-] (F) edge node[left] {} (2.3, -1.8);
	\path [-] (G) edge node[left] {} (5.75, -1.8);
	\node at (4.2,1.8){\scriptsize{$(2)$}};
	\node at (4.2,0.75){\scriptsize{$(1)$}};
	\node at (5.68,-1.46){\scriptsize{$(2)$}};
	\node at (4.95,-0.62){\scriptsize{$(1)$}};
	\node at (2.3,-1.46){\scriptsize{$(2)$}};
	\node at (3.1,-0.62){\scriptsize{$(1)$}};
	\end{tikzpicture}
\end{center}
\end{remark}

For the the $4$\textsuperscript{th} cumulant its relation with central moments 
becomes more complex. To understand this observe 
that if we differentiate Eq.~\eqref{eq::c3dif} again besides
$\frac{1}{{\varphi} }\frac{\partial^4 \varphi}{ \partial 
	v_{i_1} \partial v_{i_2} \partial v_{i_3} \partial v_{i_4}}
\big{|}_{\vv = 0}$ corresponding to $4$\textsuperscript{th} moment, we have
\begin{equation}\label{eq::mixed4}
-\frac{1}{\varphi^2} \left(\frac{\partial^2 \varphi}{\partial 
		v_{i_1}  
		v_{i_2}} \frac{\partial^2 \varphi}{\partial 
		v_{i_3} 
		v_{i_4}}  + \frac{\partial^2 \varphi}{\partial 
		v_{i_1} 
		v_{i_3}} \frac{\partial^2 \varphi}{\partial 
		v_{i_2} 
		v_{i_4}}   + \frac{\partial^2 \varphi}{\partial 
		v_{i_1} 
		v_{i_4}}\frac{\partial^2 \varphi}{\partial 
		v_{i_2}   
		v_{i_3}}  \right) \bigg{|}_{\vv = 0},
\end{equation}
corresponding to the combinatorics symmetrizing sum of the product of second 
moments, which do not vanish while using central cumulants and moments. Now, 
one can 
show that $4$\textsuperscript{th} cumulant elements 
are~\cite{domino2018efficient}
\begin{equation}
c_{i_1, i_2, i_3, i_4} = c'_{i_1, i_2, i_3, i_4} =  m'_{i_1, i_2, i_3, i_4} - 
m'_{i_1, i_2} m'_{i_3, i_4} - m'_{i_1, i_3} m'_{i_2, i_4} - m'_{i_1, i_4} 
m'_{i_2, i_3},
\end{equation}
or in the equivalently
\begin{equation}\label{eq::c4elementwise}
c_{i_1, i_2, i_3, i_4} =  m'_{i_1, i_2, i_3, i_4} - 
c_{i_1, i_2} c_{i_3, i_4} - c_{i_1, i_3} c_{i_2, i_4} - c_{i_1, i_4} 
c_{i_2, i_3}.
\end{equation}
For the graphic notation, we need a few additional definitions.

\begin{definition}\label{def::symsummats}
The outer product of symmetric matrix $\mathbf{A} \in \R^{[n,2]}$ 
by itself in the $4$-mode tensor representation is $\R^{n \times n \times n 
\times 
n} \ni  \mathcal{T} = \mathbf{A} 
\otimes \mathbf{A}$,
with elements
\begin{equation}
t_{i_1, i_2, i_3, i_4} = a_{i_1, i_2} a_{i_3, i_4}.
\end{equation}
In a graphic 
notation we have
\begin{center}
		\begin{tikzpicture}
		\node[shape=rectangle,draw=black] (C2A) at (8,0.5) {$\mathbf{A}$};
		\node[shape=rectangle,draw=black] (C2B) at (8,-0.5) {$\mathbf{A}$};
		\path [-] (C2A) edge node[left] {} (8.6,1.7);
		\path [-] (C2A) edge node[left] {} (7.4,1.7);
		\path [-] (C2B) edge node[left] {} (7.4, -1.7);
		\path [-] (C2B) edge node[left] {} (8.6, -1.7);
		(14,-2.2);
		\node at (5.05,0.55){\scriptsize{$(1)$}};
		\node at (5.95,0.55){\scriptsize{$(2)$}};
		\node at (5.05,-0.55){\scriptsize{$(4)$}};
		\node at (5.95,-0.55){\scriptsize{$(3)$}};
		\node[shape=rectangle,draw=black] (A) at (5.5, 0) {$\mathbf{A} 
		\otimes \mathbf{A}$};
		\node at (7., 0) {$=$};
		\path [-] (A) edge node[left] {} (6.1,1.7);
		\path [-] (A) edge node[left] {} (4.9,1.7);
		\path [-] (A) edge node[left] {} (4.9, -1.7);
		\path [-] (A) edge node[left] {} (6.1, -1.7);
		\end{tikzpicture}
\end{center}
However, such tensor is not super-symmetric.
To overcome this problem refer to \cite{cardoso1990eigen}, where 
the idea of symmetries of tensor products of symmetrical matrices is discussed.
To get the super-symmetric tensor we need the symmetrising sum of such outer 
products $ 
\mathcal{T}_{\text{sym}} = \sum_{\text{sym}}
\mathbf{A} 
\otimes \mathbf{A}$,
which is defined element wisely:
\begin{equation}
\left(t_{sym}\right)_{i_1, i_2, i_3, i_4} = a_{i_1, i_2} a_{i_3, i_4} + 
a_{i_1, i_3} a_{i_2, i_4} + a_{i_1, i_4} a_{i_2, i_3}.
\end{equation}
 In a graphic 
notation we have
	
\begin{center}
	\begin{tikzpicture}
	\node[shape=rectangle,draw=black] (A) at (4.5, 0) {$\sum_{\text{sym}} 
	\mathbf{A} 
				\otimes \mathbf{A}$};
	\node at (6.2, 0) {$=$};
	\path [-] (A) edge node[left] {} (5.1,1.7);
	\path [-] (A) edge node[left] {} (3.9,1.7);
	\path [-] (A) edge node[left] {} (3.9, -1.7);
	\path [-] (A) edge node[left] {} (5.1, -1.7);
	\draw [decorate,decoration={brace,amplitude=10pt}](7,-2) -- (7,2);
	\node[shape=rectangle,draw=black] (C2A) at (8,0.5) {$\mathbf{A}$};
	\node[shape=rectangle,draw=black] (C2B) at (8,-0.5) {$\mathbf{A}$};
	\path [-] (C2A) edge node[left] {} (8.6,1.7);
	\path [-] (C2A) edge node[left] {} (7.4,1.7);
	\path [-] (C2B) edge node[left] {} (7.4, -1.7);
	\path [-] (C2B) edge node[left] {} (8.6, -1.7);
	(14,-2.2);
	\node at (9.4, 0){$+$};
	\node[shape=rectangle,draw=black] (C2C) at (11,0.5) {$\mathbf{A}$};
	\node[shape=rectangle,draw=black] (C2D) at (11,-0.5) {$\mathbf{A}$};
	\path [-] (C2C) edge node[left] {} (11.6,1.7);
	\path [-] (C2D) edge node[left] {} (11.6, -1.7);
	\draw (10.2,1.7) .. controls  (10.6,-1.4) .. (10.8,-0.82);
	\draw (10.2,-1.7) .. controls  (10.6,1.4) .. (10.8,0.82);
	\node at (12, 0){$+$};
	\node[shape=rectangle,draw=black] (C2E) at (13,0.5) {$\mathbf{A}$};
	\node[shape=rectangle,draw=black] (C2F) at (13,-0.5) {$\mathbf{A}$};
	\path [-] (C2E) edge node[left] {} (12.4,1.7);
	\path [-] (C2F) edge node[left] {} (12.4, -1.7);
	\draw (13.8,1.7) .. controls  (13.4,-1.4) .. (13.2,-0.82);
	\draw (13.8,-1.7) .. controls  (13.4,1.4) .. (13.2,0.82);
	\draw [decorate,decoration={brace,mirror,amplitude=10pt}](14.5,-2) -- 
	(14.5,2);
	\end{tikzpicture}
\end{center}
\end{definition}

\begin{remark}\label{rem::c4graph}
Referring to Eq.~\eqref{eq::c4elementwise}, the 
$4\textsuperscript{th}$ cumulant in a tensor notation would be:
\begin{equation}
\CC_4(\X) = \MM_4(\tilde{\X}) - \sum_{\text{sym}} \CC_2 \otimes \CC_2.
\end{equation}
In a graphic notation we have
\begin{center}
	\begin{tikzpicture}
	\node[shape=rectangle,draw=black] (C) at (1,0) {$\CC_4$};
	\node at (2.4,0){$=$};
	\path [-] (C) edge node[left] {} (0,1);
	\path [-] (C) edge node[left] {} (2,1);
	\path [-] (C) edge node[left] {} (0,-1);
	\path [-] (C) edge node[left] {} (2,-1);	
	\node[shape=rectangle,draw=black] (D) at (4,0) {$\1_{(4)}$};
	\node[shape=rectangle,draw=black] (C) at (3,1) {$\tilde{\X}$};
	\node[shape=rectangle,draw=black] (E) at (5,1) {$\tilde{\X}$};
	\node[shape=rectangle,draw=black] (F) at (3,-1) {$\tilde{\X}$};
	\node[shape=rectangle,draw=black] (G) at (5,-1) {$\tilde{\X}$};
	\path [-] (D) edge node[left] {} (E);
	\path [-] (D) edge node[left] {} (F);
	\path [-] (D) edge node[left] {} (G);
	\path [-] (D) edge node[left] {} (C);
	\path [-] (E) edge node[left] {} (5.6,1.7);
	\path [-] (C) edge node[left] {} (2.4,1.7);
	\path [-] (F) edge node[left] {} (2.4, -1.7);
	\path [-] (G) edge node[left] {} (5.6, -1.7);
	
	\node at (6.3,0){$-$};
	\draw [decorate,decoration={brace,amplitude=10pt}](7,-2) -- (7,2);
	
	\node[shape=rectangle,draw=black] (C2A) at (8,0.5) {$\CC_2$};
	\node[shape=rectangle,draw=black] (C2B) at (8,-0.5) {$\CC_2$};
	\path [-] (C2A) edge node[left] {} (8.6,1.7);
	\path [-] (C2A) edge node[left] {} (7.4,1.7);
	\path [-] (C2B) edge node[left] {} (7.4, -1.7);
	\path [-] (C2B) edge node[left] {} (8.6, -1.7);
	(14,-2.2);
	\node at (9.4, 0){$+$};
	\node[shape=rectangle,draw=black] (C2C) at (11,0.5) {$\CC_2$};
	\node[shape=rectangle,draw=black] (C2D) at (11,-0.5) {$\CC_2$};
	\path [-] (C2C) edge node[left] {} (11.6,1.7);
	\path [-] (C2D) edge node[left] {} (11.6, -1.7);
	\draw (10.2,1.7) .. controls  (10.6,-1.4) .. (10.8,-0.82);
	\draw (10.2,-1.7) .. controls  (10.6,1.4) .. (10.8,0.82);
	\node at (12, 0){$+$};
	\node[shape=rectangle,draw=black] (C2E) at (13,0.5) {$\CC_2$};
	\node[shape=rectangle,draw=black] (C2F) at (13,-0.5) {$\CC_2$};
	\path [-] (C2E) edge node[left] {} (12.4,1.7);
	\path [-] (C2F) edge node[left] {} (12.4, -1.7);
	\draw (13.8,1.7) .. controls  (13.4,-1.4) .. (13.2,-0.82);
	\draw (13.8,-1.7) .. controls  (13.4,1.4) .. (13.2,0.82);
	\node at (3.58, .69) {\scriptsize{$(1)$}};
	\node at (2.85, 1.5) {\scriptsize{$(2)$}};
	\node at (3.2, -.5) {\scriptsize{$(1)$}};
	\node at (2.45, -1.31) {\scriptsize{$(2)$}};
	\node at (4.41, .69) {\scriptsize{$(1)$}};
	\node at (5.15, 1.5) {\scriptsize{$(2)$}};
	\node at (4.8, -.5) {\scriptsize{$(1)$}};
	\node at (5.55, -1.35) {\scriptsize{$(2)$}};	
	\draw [decorate,decoration={brace,mirror,amplitude=10pt}](14.2,-2) -- 
	(14.2,2);
	\end{tikzpicture}
\end{center}
\end{remark}

In the case of the $5$\textsuperscript{th} cumulant double differentiation of 
Eq.~\eqref{eq::c3dif} produces a term $\frac{1}{{\varphi} 
}\frac{\partial^5 \varphi}{ \partial 
	v_{i_1} \partial v_{i_2} \partial v_{i_3} \partial v_{i_4} \partial v_{i_5}}
\big{|}_{\vv = 0}$ and additional $10$ terms
$-\frac{1}{\varphi^2}\left(\frac{\partial^2 \varphi}{\partial 
	v_{i_1} \partial v_{i_2}} \frac{\partial^3 \varphi}{\partial 
	v_{i_3}  \partial v_{i_4} \partial v_{i_5}} \right) \big{|}_{\vv = 0}$, 
	$-\frac{1}{\varphi^2}\left(\frac{\partial^2 \varphi}{\partial 
		v_{i_1} \partial v_{i_5}} \frac{\partial^3 \varphi}{\partial 
		v_{i_2}  \partial v_{i_3} \partial v_{i_4}} \right) \big{|}_{\vv = 0}, 
		\ldots$ that do not vanish. Reminding terms will have a single 
		derivative of $\varphi$ 
		leading to first moments that vanishes due to the centring. 
Given these, and the fact, that cumulants of order $2$ and $3$ equals to 
corresponding
central moments, the $5$\textsuperscript{th} cumulant element is given by
\begin{equation}\label{eq::c5els}
c_{i_1, \ldots i_5} = m'_{i_1, \ldots i_5} - \underbrace{c_{i_1, i_2} c_{i_3, 
i_4, i_5} - 
c_{i_1, i_3} c_{i_2, i_4, i_5} - c_{i_1, i_4} c_{i_2, i_3, i_5} - 
\ldots}_{\times 10}.
\end{equation}
Here again we have a symmetrising sum of outer products of second
cumulant's matrix and third cumulant's tensor, that 
have $10$ terms.

\begin{definition}{The symmetrising sum of tensors' outer products.}
	Suppose we have $d_1$ and $d_2$ modes super-symmetric tensors 
	$\mathcal{T}_{d_1}$ and $\mathcal{T}_{d_2}$, the symmetrised sum of outer 
	products of these tensors in a graphic notation is given by
\begin{center}
	\begin{tikzpicture}
	\node[shape=rectangle,draw=black] (A) at (4.5, 0) {$\sum_{\text{sym}} 
		\mathcal{T}_{d_1} 
		\otimes \mathcal{T}_{d_2}$};
	\node at (6.3, 0) {$=$};
	\path [-] (A) edge node[left] {} (5.5,1.7);
	\path [-] (A) edge node[left] {} (3.5,1.7);
	\path [-] (A) edge node[left] {} (3.5, -1.7);
	\path [-] (A) edge node[left] {} (5.5, -1.7);	
	\node at (4.5,1.7){$\cdots$};
	\node at (4.5,-1.7){$\cdots$};
	\draw [decorate,decoration={brace,amplitude=10pt}](7,-2) -- (7,2);	
	\node[shape=rectangle,draw=black] (C2A) at (8,0.5) {$\mathcal{T}_{d_1}$};
	\node[shape=rectangle,draw=black] (C2B) at (8,-0.5) {$\mathcal{T}_{d_2}$};
	\path [-] (C2A) edge node[left] {} (9,1.7);
	\path [-] (C2A) edge node[left] {} (7,1.7);
	\node at (8,1.7){$\cdots$};
	\path [-] (C2B) edge node[left] {} (7.1, -1.7);
	\path [-] (C2B) edge node[left] {} (9.1, -1.7);
	\node at (8,-1.7){$\cdots$};
	\node at (9.4, 0){$+$};
	\node[shape=rectangle,draw=black] (C2C) at (11,0.5) {$\mathcal{T}_{d_1}$};
	\node[shape=rectangle,draw=black] (C2D) at (11,-0.5) {$\mathcal{T}_{d_2}$};
	\path [-] (C2C) edge node[left] {} (12,1.7);
	\node at (11,1.7){$\cdots$};
	\node at (11,-1.7){$\cdots$};
	\path [-] (C2D) edge node[left] {} (12.1, -1.7);
	\draw (10,1.7) .. controls  (10.6,-1.4) .. (10.8,-0.82);
	\draw (10.2,-1.7) .. controls  (10.6,1.4) .. (10.8,0.82);
	\draw [decorate,decoration={brace,mirror,amplitude=10pt}](14.5,-2) -- 
	(14.5,2);
	\draw [decorate,decoration={brace,mirror,amplitude=10pt}](7.5,-2.2) 
	-- 
	(14,-2.2);
	\node at (10.5, -2.8){$\frac{d!}{d_1! d_2! s!} \times$};
	\end{tikzpicture}
\end{center}
	If $d_1 = d_2$ we have the outer product of the $\mathcal{T}_{d_1}$ by 
	itself and set $s = 2$, otherwise we set $s = 1$.
\end{definition}	

\begin{remark}\label{rem::c5graph}
Referring to Eq.~\eqref{eq::c5els}, in the tensor notation the 
$5$\textsuperscript{th} cumulant is given by
\begin{equation}
\CC_5(\X) = \MM_5(\tilde{\X}) - \sum_{\text{sym}} \CC_2(\X) \otimes \CC_3(\X),
\end{equation}
and in a graphic notation by
\begin{center}
\includegraphics[width=\textwidth]{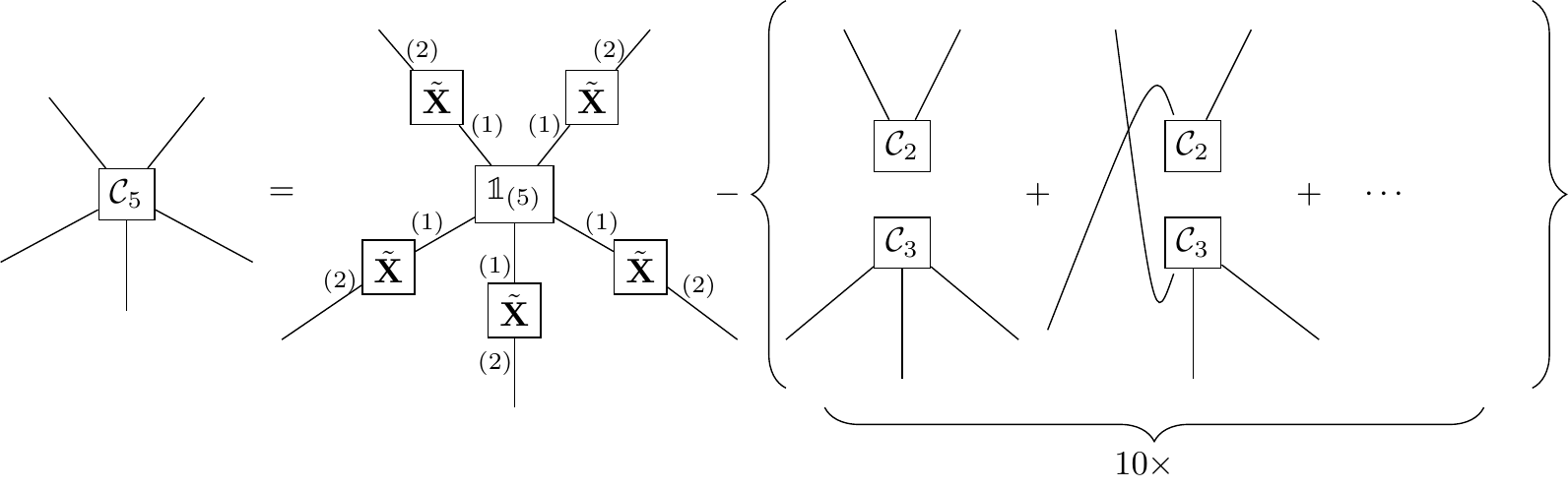}
\end{center}

\end{remark}

\begin{remark}\label{rem::c6graph}
The $6$\textsuperscript{th} cumulant's tensor formula is much more complicated. 
For the general formula in the tensor notation we can refer to 
\cite{domino2018efficient}, yielding:
\begin{equation}
\CC_6(\X) = \MM_6(\tilde{\X}) - \sum_{\text{sym}} \CC_4(\X) \otimes \CC_2(\X) - 
\sum_{\text{sym}} \CC_3(\X) \otimes \CC_3(\X) - \sum_{\text{sym}} \CC_2(\X) 
\otimes \CC_2(\X) \otimes \CC_2(\X)
\end{equation}
The last term is the symmetrising sum of the outer produce of $3$ symmetric 
matrices. $\R^{[n,6]} \in \mathcal{T}_{\text{sym}} = \sum_{\text{sym}} 
\mathbf{A} \otimes \mathbf{A} \otimes \mathbf{A}$ where $\mathbf{A} \in \R^{[n, 
2]}$ what can be represented by the following element wise notation:
\begin{equation}
\left(t_{\text{sym}}\right)_{i_1, \ldots i_6} = \underbrace{a_{i_1, i_2}a_{i_3, 
i_4}a_{i_5, 
i_6} + a_{i_1, i_3}a_{i_2, i_4}a_{i_5, i_6} + a_{i_1, i_4}a_{i_2, i_3}a_{i_5, 
i_6} + \ldots}_{\times 15}.
\end{equation}
Finally in a graphic notation we have
\begin{center}
	\includegraphics[width=\textwidth]{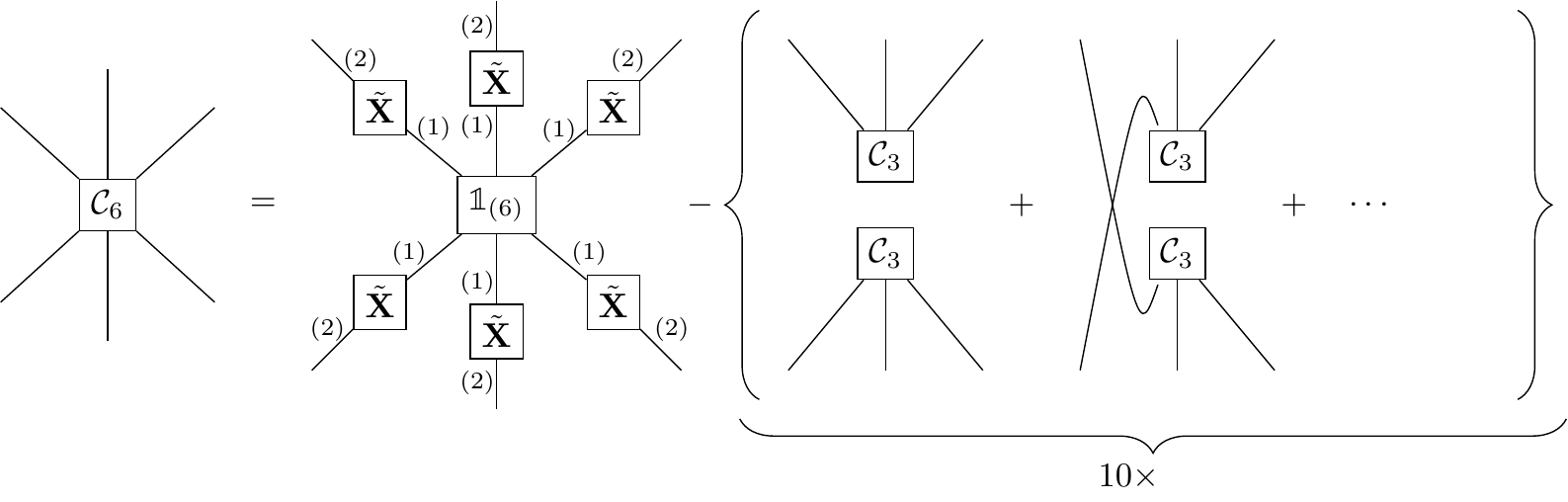}\\	
	\includegraphics[width=\textwidth]{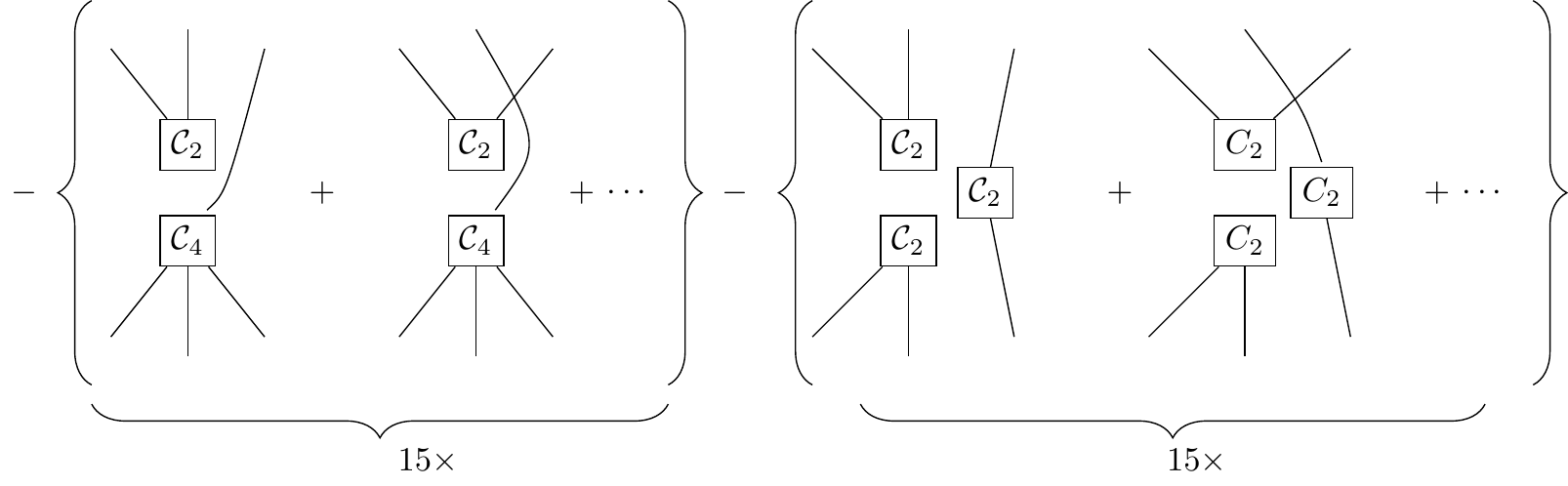}
\end{center}
\end{remark}
\subsection{Calculation and programming implementation}\label{sec::blockimpl}

As discussed is Section~\ref{sec::momenttens} and Section~\ref{sec::cumtens} 
both moments tensors and cumulants tensors are super-symmetric. Hence it is 
sufficient to store and calculate only one hyper-pyramid part of such tensor. 
Such storage scheme was discussed in~\cite{domino2018efficient}. Let us mention 
here only a symmetric matrix example that can be stored in blocks in the 
following form
\begin{equation}
\R^{[n,2]} \ni \mathbf{A} = \left[ \begin{array}{cccc}
({\mathbf{A}})_{11} & ({\mathbf{A}})_{12} & \cdots & ({\mathbf{A}})_{1\bar{n}} 
\\ 
\text{Void} & ({\mathbf{A}})_{22} & \cdots &  ({\mathbf{A}})_{2\bar{n}} \\ 
\vdots & \vdots & \ddots & \vdots \\
\text{Void} & \text{Void} & \cdots & (\mathbf{A})_{\bar{n} \bar{n}} \\ 
\end{array}   \right].
\end{equation}
The use of blocks makes a computation implementation simpler and more 
efficient, see also \cite{schatz2014exploiting}. As discussed 
in~\cite{domino2018efficient}, given a block size parameter $b$ and $\bar{n} = 
\frac{n}{b}$ and assuming that $b | n$, the super-symmetric tensor 
have $\binom{\bar{n}+d-1}{\bar{n}}$ blocks, hence we store $b^d 
\binom{\bar{n}+d-1}{\bar{n}}$ elements out of $n^d$ in an naive approach. For 
$n \gg b$ 
we save up to the $\frac{1}{d !}$ computer memory and computational power.
Those block storage scheme was implemented in the Julia programming 
language in the \texttt{SymmetricTensors.jl} module, see \cite{st} 
GitHub repository. Algorithms for the calculation of moment's and cumulant's 
were implemented as well in \texttt{Cumulants.jl} module, see
\cite{cum} GitHub repository. 

\section{Cumulants of copulas}\label{sec::cumcops}

Having introduced higher order cumulants we can discuss their 
significance in the copula determination of multivariate data. Given data 
probabilistic model including marginals and the copula, multivariate moments 
tensors can be computed by Eq.~\eqref{eq::cumsmoms}. These moments tensors are 
the base to compute multivariate cumulants tensors. We are interested 
in this section rather on the impact of 
copulas on multivariate cumulants, since non-Gaussian univariate marginal 
distributions 
are easy to detect and analyse by means of univariate statistics. 
Following~\cite{domino2018hiding} and inspired by the fact, 
that univariate Gaussian distribution results in zero higher order univariate 
cumulants, we present in this section examples of cumulants of 
many non-Gaussian copulas with univariate standard Gaussian marginals hoping
to reveal in cumulants characteristics of copulas. For each experiment we use 
the sample of 
size $t = 5 \times 10^6$ to ensure accurate approximation of cumulants.

Further in this section we have measured the simultaneous interdependence of 
$3$ or $4$ marginals by means of cumulants. What is important, we have shown by 
experiments, that these interdependences are meaningful.
Alternative approach to the use of univariate Gaussian marginals is to use 
uniform 
marginals on $[0,1]$. However such marginals 
have its own impact on even order higher order cumulants, due to highly 
negative kurtosis of the uniform distribution. 
Examples, where uniform marginals have a dominant impact on 
elements of
the $4$\textsuperscript{th} cumulant's tensor are presented at the end of this 
section.

\subsection{Archimedean copulas}\label{sec::ac}

\begin{figure}
	\subfigure[Clayton 
	copula\label{fig::cl_c3}]{\includegraphics{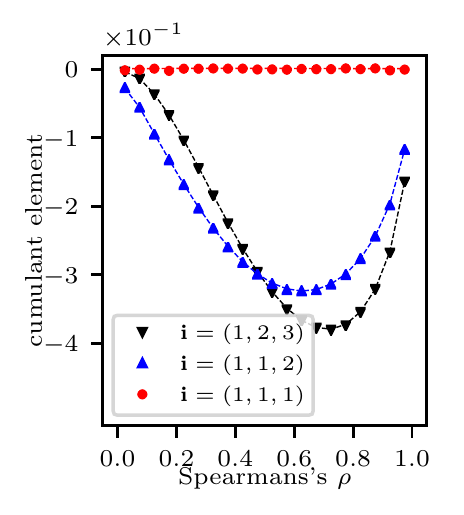}}
	\subfigure[Gumbel 
	copula\label{fig::gu_c3}]{\includegraphics{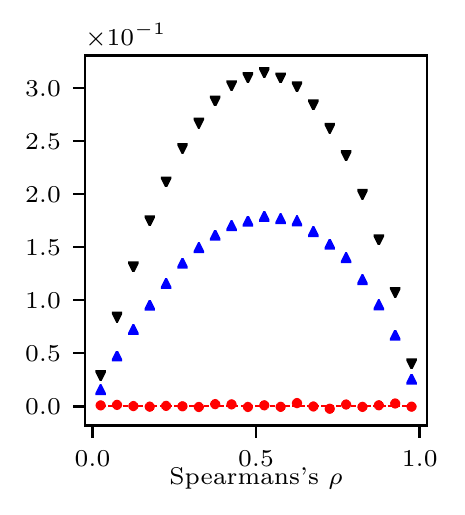}}
	\subfigure[Frank 
	copula\label{fig::fr_c3}]{\includegraphics{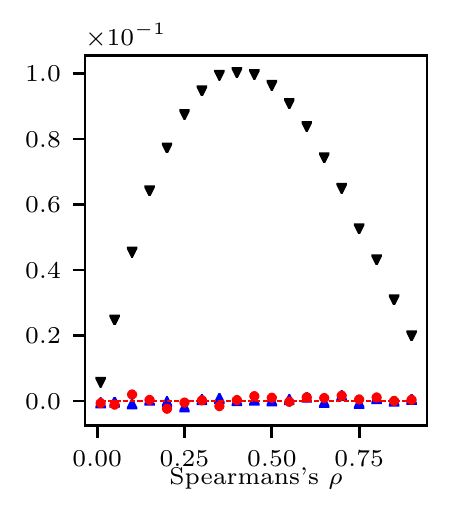}}
	\caption{$3$\textsuperscript{th} cumulant's elements of Archimedean 
		copulas, and standard Gaussian marginals. Solid lines represent 
		theoretical outcomes, while points represent outcomes form generated 
		data.}\label{fig::archc3}
\end{figure}
\begin{figure}
		\subfigure[Clayton 
		copula\label{fig::cl_c4}]{\includegraphics{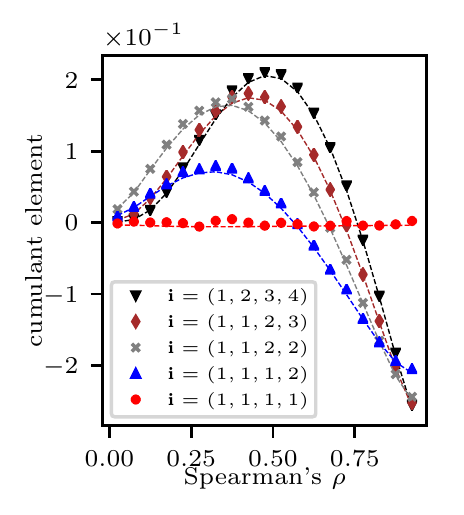}}
		\subfigure[Gumbel 
		copula\label{fig::gu_c4}]{\includegraphics{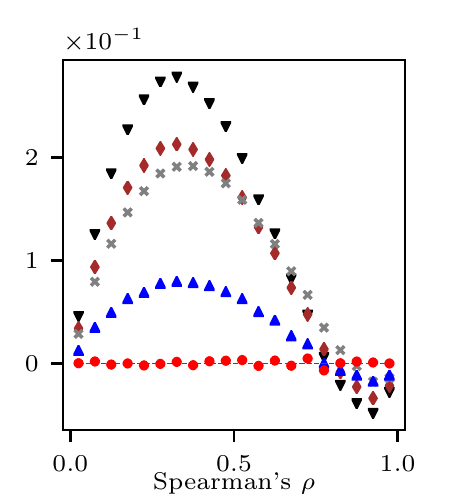}}
		\subfigure[Frank 
		copula\label{fig::fr_c4}]{\includegraphics{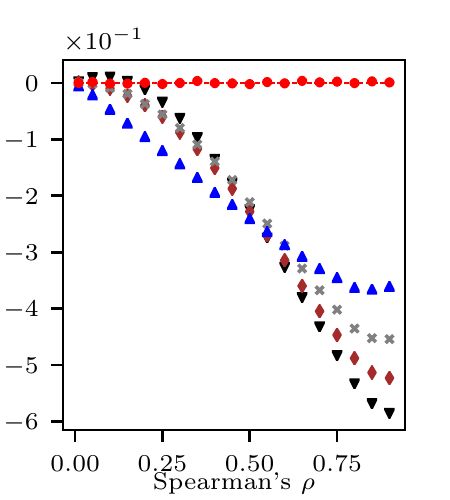}}
		\caption{$4$\textsuperscript{th} cumulant's elements of Archimedean 
		copulas, and standard Gaussian marginals. Solid lines represent 
		theoretical outcomes, while points represent outcomes form generated 
		data.}\label{fig::archc4}
\end{figure}

We start with the (non nested) Archimedean copula, as introduced in 
Definition~\ref{def::nvarchcop}. It is easy to show that such copula is 
unchanged under any permutation of
marginals,
\begin{equation}\label{eq::symarch}
\forall_{i, j \in (1:n)}\C_{\psi_{\theta}}(u_1, \ldots, 
u_i, \ldots u_j, \ldots u_n) = \C_{\psi_{\theta}}(u_1, \ldots, u_j, \ldots u_i, 
\ldots u_n),
\end{equation}
and we call such feature the symmetry of the Archimedean copula. 

The simplest approach to investigate distinct elements of 
$3$\textsuperscript{rd} cumulant tensor of the Archimedean copula 
with standard Gaussian marginals will be to take the tri-variate Archimedean 
copula. The probabilistic model is given by:
\begin{equation}
\C_{\psi_{\theta}}(u_1, u_2, u_3) \text{ and } x_i = 
F_{\mathcal{N}(0,1)}^{-1}(u_i).
\end{equation}
Following~\cite{domino2018hiding}, we have here $3$ distinct 
$3$\textsuperscript{rd}
cumulant's elements:
\begin{itemize}
\item $c_{1,1,1} = c_{2,2,2} = c_{3,3,3} = 0$ -- diagonal elements, zero due to 
Gaussian marginal distributions;
\item $c_{1,1,2} = c_{1,2,2} = c_{1,3,3} = c_{3,3,1} \ldots$ -- partially 
diagonal elements, equal due to copula symmetry, identical marginal 
distributions, and cumulants super-symmetry;
\item $c_{1,2,3} = c_{2,1,3} = c_{3,1,2} = \ldots $ - off-diagonal elements, 
equal due to cumulants super-symmetry.
\end{itemize}
In Figure~\ref{fig::archc3} we present elements of the $3$\textsuperscript{rd} 
cumulant calculated for data sampled from Gumbel, Clayton or Frank Archimedean 
copulas, for a further discussion see~\cite{domino2018hiding}. 
In the case of the 
Clayton copula~\ref{fig::cl_c3} we present in addition 
theoretical values obtained by integrating numerically 
Eq.~\eqref{eq::cumsmoms}. 
For all Archimedean copulas, we determine the 
copula parameter $\theta$ from the Spearman's $\rho$ correlation coefficient, 
using Eq.~\eqref{eq::copcor}.

Analogically we have $5$ distinct elements of the $4$\textsuperscript{th} 
cumulant's tensor given $4$-variate Archimedean copula (with standard Gaussian 
marginals in our case). Values of these elements computed from data sampled 
from these copula are presented in Figure~\ref{fig::archc4}. 

From Figures~\ref{fig::archc3} and~\ref{fig::archc4} one can conclude that 
higher 
order cumulants of different Archimedean copulas with Gaussian marginals have 
different patterns, hence such cumulants can be used to 
distinguish between copulas. 
Moreover, from Figure~\ref{fig::fr_c3}, one can conclude that given the 
Frank copula, the $3$\textsuperscript{rd} cumulant has only non-zero 
off-diagonal elements. These elements are 
calculated from three distinct marginals. This is a meaningful example of the 
tri-variate $3$\textsuperscript{rd} order cross-correlation measure that is 
non-zero given zero bivariate $3$\textsuperscript{rd} order cross-correlation 
measures, such as $c_{1,1,2}$. For the practical discussion of
cumulants based $d$-variate interdependence measures (where $d > 2$) in data 
analysis see~\cite{glombband2018}. There cumulants are used to analyse 
information from the hyper-spectral data.

\subsection{Fr\'echet copula}

In this subsection, we discuss cumulants of the $n$-variate $1$-parameter 
Fr\'echet copula with 
standard Gaussian marginals, see
Definition~\ref{def::frechcop}. It is 
easy to show that this copula has the 
same symmetry as an Archimedean one, see 
Eq.~\eqref{eq::symarch}. Further 
it can be demonstrated, at least numerically, that higher order odd 
cumulants of such copula with symmetric marginals (such as Gaussian 
marginals) should be 
zero, see Figure~\ref{fig::norms} for justification.

For experiments results, in Figure~\ref{fig::fre4} we present distinguishable
$4$\textsuperscript{th} cumulant's elements of discussed here Fr\'echet copula. 
One can observe that the pattern of the 
$4$\textsuperscript{th} cumulant differs from patterns of Archimedean copulas, 
presented in 
Figure~\ref{fig::archc4}. Importantly for each Spearman's $\rho$ value, 
off-diagonal cumulant's tensor's elements of the Fr\'echet copula, are 
distinguishable (due to their higher values) in comparison with 
partially diagonal ones. These off-diagonal elements (such as $c_{1,2,3,4}$)
correspond to the $4$\textsuperscript{th} order $4$-variate cross-correlation 
measure between marginals. 

\begin{figure}
	\subfigure[$t$-Student copula $\nu = 
	5$\label{fig::t5}]{\includegraphics{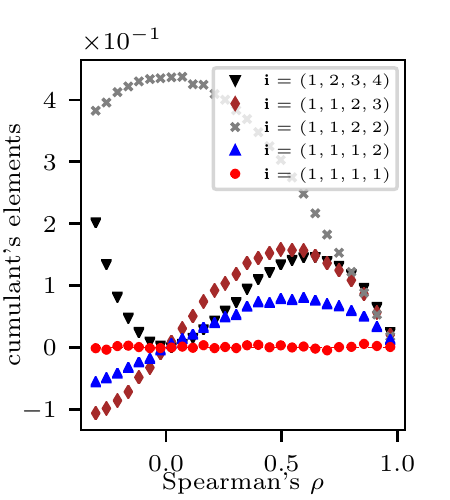}}
	\subfigure[$t$-Student copula $\nu = 
	20$\label{fig::t20}]{\includegraphics{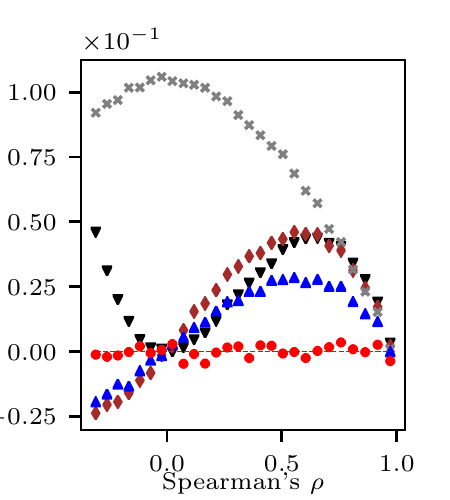}}
	\subfigure[Fr\'echet 
	copula\label{fig::fre4}]{\includegraphics{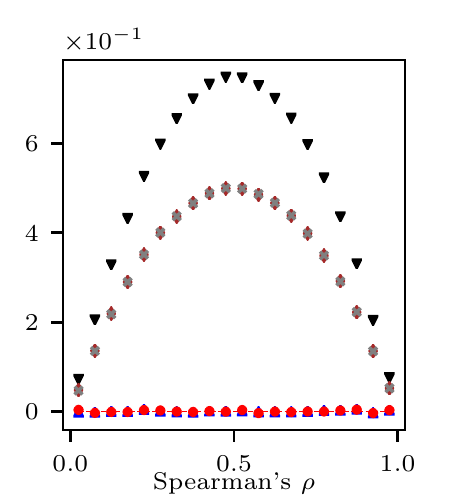}}
	\caption{$4$\textsuperscript{th} cumulant's elements of $t$-Student and 
	Fr\'echet 
		copula, and standard Gaussian marginals.}\label{fig::t}
\end{figure}

\subsection{$t$-Student copula}

Next, we analyse the $t$-Student copula with standard Gaussian marginals
parametrised by an integer 
parameter $\nu$ and the matrix $\mathbf{R}$ with constant off-diagonal elements 
$r_{1,2} = r_{1,3} = \ldots = r$ - see the 
constant correlation matrix in~\cite{domino2018hiding}. Given such 
parametrisation, similar marginals exchange symmetry appears to 
Archimedean and Fr\'echet copulas cases, hence similar presentation of results 
is possible. For experiments we use the $4$-variate $t$-Student copula 
parametrised by $r$ and $\nu$ and standard Gaussian marginals. We use the 
following limitation for the $r$ parameter
$-\frac{1}{3} < r < 1$ to ensure the matrix $\mathbf{R} \in \R^{[4,2]}$ 
to be positive definite given $1$ on its diagonal.
 In 
Figures~\ref{fig::t5}~\ref{fig::t20} we present results of our 
experiments concerning the $4$\textsuperscript{th} cumulant's tensors. For 
further analysis of cumulants of $t$-Student copula see~\cite{domino2018use}. 
The pattern of the $t$-Student copula differs from patterns of other 
copulas 
(Archimedean and Fr\'echet), what gives an advantage for 
$4$\textsuperscript{th} order cumulant in 
copula detection.
For the $t$-Student copula, we discuss $4$\textsuperscript{th} cumulant, since 
 in the case of the $t$-Student copula with 
symmetric marginals higher 
order odd cumulants are supposed to be zero as it is in a case of $t$-Student 
multivariate 
distribution~\cite{kotz2004multivariate}. See also a discussion on elliptical 
copulas (including $t$-Student one) in~\cite{matthias2017simulating}.

To demonstrate the significance of univariate Gaussian marginals to reveal in 
cumulants features of copulas, refer to Figure~\ref{fig::umar}, where we 
present values of elements of the 
$4$\textsuperscript{th} cumulants tensors of various cumulants with uniform 
marginals on $[0,1]$. Here patterns are similar, for different copulas, and the 
proper copula's 
determination on the base of these cumulants is difficult. This is due a high 
negative kurtosis of the uniform distribution and its impact on multivariate 
cumulant of order $4$.

\begin{figure}
	\subfigure[Gumbel 
	copula\label{fig::4gu}]{\includegraphics{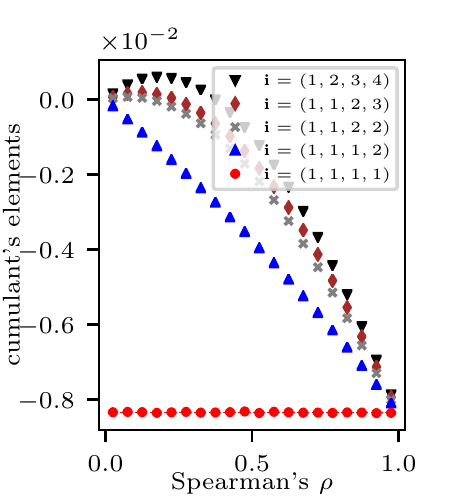}}
	\subfigure[$t$-Student copula $\nu = 
	5$\label{fig::4t}]{\includegraphics{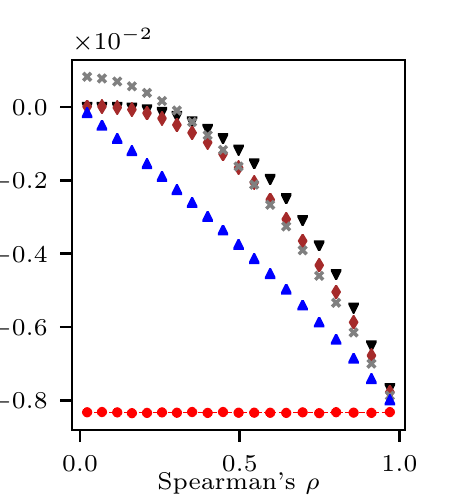}}
	\subfigure[Fr\'echet copula 	
	\label{fig::4f}]{\includegraphics{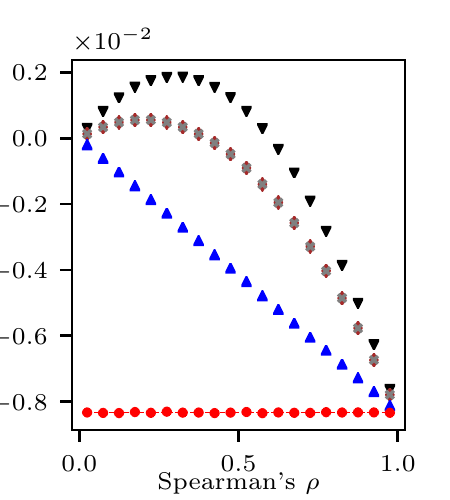}}
	\caption{$4$\textsuperscript{th} cumulant's elements given uniform 
	marginals on $[0,1]$.}\label{fig::umar}
\end{figure}

\section{Auto-correlation function and cumulants}\label{sec::cumsacf}

Having discussed the use of cumulants to analyse interdependency between 
marginals of multivariate data we can mention their potential use to analyse 
auto-correlation of univariate data. What is interesting, the combination of 
both approaches is possible as well.

Suppose we have a random process represented by a series of univariate random 
variables 
$\ZZ_1, \ldots \ZZ_i, \ldots, \ZZ_j, \ldots, \ZZ_k$ and want to analyse the
auto-correlation inside such series. The simplest approach would be to use the 
Pearson's correlation based auto-correlation coefficient 
\cite{dunn2017measurement}:
\begin{equation}
\text{acc}_{i,j} = \frac{\text{acf}_{i,j}}{\sigma_i \sigma_j} = \frac{\E((\ZZ_i 
	- 
	\mu_i)(\ZZ_j-\mu_j))}{\sigma_i \sigma_j}
\end{equation}
here $\mu_i$ and $\sigma_i$ are the mean and the standard deviation of $\ZZ_i$, 
and $\E$ is expecting value operator. 
It is easy to show, that the 
auto-correlation function is simply an element of the 
$2$\textsuperscript{nd} cumulant of $k$ - variate random vector $\ZZ^{(k)}$ 
with marginals $\ZZ_1, \ldots \ZZ_k$. If $\ZZ^{(k)} \sim \mathcal{N}(\mu, 
\SSSS)$ the covariance matrix carries all information about the 
auto-correlation, and $\text{acf}_{i,j} = s_{i,j}$. Given some assumptions 
about statistics of these variables, further simplification of the 
auto-correlation function is possible. 

\begin{remark} If all random variables $\ZZ_i$ 
are from the 
same distribution we have: $\sigma_i = \sigma_j = \sigma$ and $\mu_i = 
\mu_j = \mu$. Further if the auto-correlation function $\text{acf}_{i,j}$ is 
only the function 
of $\tau = |i-j|$ we have a stationary process and we can parametrize the 
auto-correlation 
function or coefficient by one lag parameter $\tau$. Suppose we have 
$\tilde{Z} \in \R^t$ being a vector of $t$ realisations of zero mean univariate 
random 
variable. Given stationary of $\tilde{Z}$ the 
auto-correlation function can be estimated by~\cite{dunn2017measurement}
\begin{equation}
\text{acf}_{\tau}(Z) = \frac{1}{t-\tau} \sum_{i = 1}^{t-\tau}\tilde{z}_i 
\tilde{z}_{i+\tau}
\end{equation}
\end{remark}

If our series is represented by $\ZZ^{(k)}$ that is not distributed according 
to multivariate Gaussian,
we may want to consider higher order cumulants approach in place of standard 
auto-correlation function, see \cite{nikias1993signal, manolakos2000systematic}.
\begin{example}
	Let us present some examples of cumulants based auto-correlation functions, 
	namely
	\begin{enumerate}
		\item $3$\textsuperscript{rd} order in analogy to 
		$3$\textsuperscript{rd} cumulant
	\begin{equation}\label{eq::acc3}
	\text{acf}_{\tau_1, \tau_2}(Z) = \frac{1}{t-\max(\tau_1, \tau_2)} \sum_{i = 
	1}^{t-\max(\tau_1, \tau_2)} \tilde{z}_i 
	\tilde{z}_{i+\tau_1}\tilde{z}_{i+\tau_2}
	\end{equation}
	\item $4$\textsuperscript{th} order in analogy to $4$\textsuperscript{th} 
	cumulant
	\begin{equation}
	\begin{split}
			&\text{acf}_{\tau_1, \tau_2, \tau_3}(Z) = \frac{1}{t-\tau} \sum_{i 
			= 
				1}^{t-\tau}\tilde{z}_i 
			\tilde{z}_{i+\tau_1}\tilde{z}_{i+\tau_2}\tilde{z}_{i+\tau_3} - 
			\left(\frac{1}{t-\tau}\right)^2 \\  &\left(\sum_{i = 
			 	1}^{t-\tau}\tilde{z}_i\tilde{z}_{i+\tau_1} \sum_{i = 
			 	1}^{t-\tau}\tilde{z}_{i+\tau_2}\tilde{z}_{i+\tau_3} + \sum_{i = 
			 	1}^{t-\tau}\tilde{z}_i\tilde{z}_{i+\tau_2} \sum_{i = 
			 	1}^{t-\tau}\tilde{z}_{i+\tau_1}\tilde{z}_{i+\tau_3}  + \sum_{i 
			 	= 
			 	1}^{t-\tau}\tilde{z}_i\tilde{z}_{i+\tau_3} \sum_{i = 
			 	1}^{t-\tau}\tilde{z}_{i+\tau_1}\tilde{z}_{i+\tau_2}   \right) 
			 	\\ 
			 	& = \frac{1}{t-\tau} \sum_{i = 
			 		1}^{t-\tau}\tilde{z}_i 
			 	\tilde{z}_{i+\tau_1}\tilde{z}_{i+\tau_2}\tilde{z}_{i+\tau_3} - 
			\text{acf}_{\tau_1} \text{acf}_{|\tau_3 - \tau_2|} - 
			\text{acf}_{\tau_2} \text{acf}_{|\tau_3 - \tau_1|} - 
			\text{acf}_{\tau_3} \text{acf}_{|\tau_2 - \tau_1|},
			\end{split}
	\end{equation}
	where $\tau = \max(\tau_1, 
	\tau_2, \tau_3)$. Here products of one parameter auto-correlation functions 
	ensures that $\text{acf}_{\tau_1, \tau_2, \tau_3}(Z)$ carries additional 
	information included in those one parameter auto-correlation functions.
		\end{enumerate}
\end{example}

To show a significance of higher order cumulants in auto-correlation analysis 
in 
the computer science
refer to \cite{manolakos2000systematic}, where the 
parallel computer architecture for the calculation of $4$'th order cumulant 
that measures auto-correlation is discussed. Further in \cite{qasim2008fpga} 
the novel algorithm for the FPGA (Field Programmable 
Gate Array) based architecture was used to calculate the third order 
auto-correlation cumulant. Let us present further potential applications of 
cumulant based auto-correlations in real life biomedical data analysis.

\begin{example}
While analysing biomedical signals one is often interested in a correlation of 
two signals in a frequency domain. For justification of such approach 
consider~\cite{muthuswamy1998spectral} where the spectral distance analysis of 
the neurological EEG signals was introduced. To determine the correlation of 
two signals as a function of frequency the
coherence function~\cite{pereda2005nonlinear} is used. Such coherence is the 
normalised 
Fourier transform 
of auto-correlation function between two signals. In details, suppose we have 
two zero 
mean signals 
$\tilde{Z} \in \R^t$ and $\tilde{W} \in \R^t$. First we compute their 
auto-correlation function:
\begin{equation}\label{eq::crossac}
	\text{acc}_{\tau}(Z,W) = \frac{1}{\sigma(Z) \sigma(W)(t-\tau)} \sum_{i = 
	1}^{t-\tau}\tilde{z}_i 
	\tilde{w}_{i+\tau},
\end{equation}
where $\sigma(Z) \sigma(W)$ are corresponding standard deviations of signals. 
Next we perform the Fourier transform of the result over 
$\tau$, moving to frequency domain since $\tau$ for biomedical data indicates 
often a time lag. Finally the coherence $-1 \leq \kappa^2_{Z,W}(f) \leq 1$ is 
normalised Fourier transform of the auto-correlation function as in 
Eq.~\eqref{eq::crossac}.

For application of the coherence function in biomedical data analysis 
see~\cite{rabotti2010characterization, 
domino2018uterine, domino2016linear}. 
However consider for example~\cite{domino2017biomathematical} where $3$ 
biomedical signals are analysed and the pairwise 
coherence was computed for those signals. Such approach can be generalised 
using the $3$\textsuperscript{rd}
cumulants~\ie~computing following auto-correlation
\begin{equation}\label{eq::crossac3}
\text{acc}_{\tau}(Z,W,V) = \frac{1}{\sigma(Z) \sigma(W) \sigma(V)(t-\tau)} 
\sum_{i = 1}^{t-\tau}\tilde{z}_i 
\tilde{w}_{i+\tau_1}\tilde{v}_{i+\tau_2},
\end{equation}
	where $\tau = \max(\tau_1, 
	\tau_2)$. Next by the double Fourier transform and corresponding 
	normalisation we can have the double frequencies coherence that can be 
	compared with pairwise results.

\end{example}

\chapter{Cumulants in machine learning}\label{cha::cumml}

In this chapter, we consider multivariate data where marginal variables are 
features. Next, we discuss algorithms applying cumulants that can be used for 
feature selection or feature extraction when new features are linear 
combinations of original ones. We focus on examples, where interesting 
information is associated with the joint non-Gaussian distribution of a subset 
of features. In general, such an approach is applicable when searching for 
outliers that represent extreme events often resulting from non-usual dynamics 
of the complex (physical) system producing data, \eg~\cite{domino2014meteo, 
gligor2001econophysics}. Such dynamics may result in non-Gaussian distribution. 
In other words, outliers lie in the tail of the multivariate distribution of 
data. Hence outliers may be modelled by non-Gaussian models.  Apart from this, 
the majority of algorithms presented in this chapter can be modified in such a 
way, that they can be used to search for/find subsets of Gaussian-distributed 
features.  These can then be used in machine learning algorithms that rely on 
assumption that data are normally distributed. 

For the practical example, consider financial data analysis in the safe 
investment portfolio determination. We are searching for such a portfolio (the 
linear combination of assets being features) where potential extreme events, 
being high drops in the portfolio's value are unlikely. The standard approach 
that uses Gaussian multivariate distribution~\cite{best2000implementing}, 
although simple applicable, fails to anticipate extreme events.  To overcome 
this problem, one can use higher order cumulants tensors, that reflect higher 
order cross-correlations between marginals and hence can be used to anticipate 
simultaneous extreme events occurring for many 
marginals~\cite{morton2009algebraic}. The practical use of higher order 
cumulants tensors in the safe investment portfolios determination, during the 
crisis on the Warsaw Stock Exchange, is discussed in~\cite{domino2016use}. 

As another example, let us consider non-Gaussian distributed data collected 
utilising the hyper-spectral imaging (HSI).  Analysis of such images is 
complicated due to their high dimensionality, lack of training examples for 
classification and the presence of feature redundancy phenomena 
\cite{ghamisi2017advances}. 

\begin{example}
	In this example we discuss data collected by means of the hyper-spectral 
	camera. Such camera outputs for each pixel (a position on 2D plane) a 
	vector of $n$ features being intensities of light (reflectance) recorded 
	for $n$ different wavelengths: $\lambda_{1}, \ldots, \lambda_{n}$.  
	Hyper-spectral images are represented by a $3$ mode tensor $\mathcal{X} \in 
	\R^{x \times y \times n}$. One can analyse hyper-spectral data in the 
	spectral approach or the spatial approach~\cite{benediktsson2015spectral}. 
	In the spectral approach, each feature is a reflectance vector tied to the 
	given wavelength. Here pixels index realisations. Hence, we can represent 
	data in the form of matrix $\textbf{X} \in \R^{t \times n}$. Such 
	representation gives naturally $n$ features and $t = xy$ realisations. The 
	disadvantage of such representation is the loose of the spatial 
	information, what is the cost of the use of the spectral approach. 
	Alternatively, in the spatial approach, we examine $n$ different monochrome 
	images of data that correspond to $n$ different wavelengths.
    
    In~\cite{glomb2018application} authors discussed the detection of the 
    gunshot residue (GSR) on various fabrics types utilizing a hyper-spectral 
    imagining. For this purpose, the authors used advanced machine learning 
    algorithms. Such an approach is potentially applicable to support the 
    forensic analyst in the objective evidence collection and follows the 
    active research in forensic science of the GSR identification and 
    analysis~\cite{dalby2010analysis}.  In more details, 
    in~\cite{glomb2018application}, authors used spectral approach together 
    with unsupervised and supervised machine learning algorithms, the first 
    being the (RX) Reed-Xiaoli Detector~\cite{reed1990adaptive}, and the second 
    the (SVM) Support Vector Machine classifier~\cite{smola1998learning}.  
    Authors focused on the anomaly detection scenario, where the GSR pattern 
    occupies a small portion of pixels. For the purpose of experiments, authors 
    have created the dataset by annotating, by hand, some of the pixels 
    reflecting the fabric itself (background) and GSR (outliers).  Even though 
    GSR patterns have a district spatial characteristics: spherical or 
    elliptical, the spectral approach applied here gives a straightforward 
    application of probabilistic models in detection.

    The detection performance in both unsupervised RX detector and supervised 
    SVM classification varied between materials. In majority cases most of the 
    annotated outliers have been detected. However there was a notable number 
    of false positives. To demonstrate unsupervised detection, the 
    visualisation of both background and GSR data is presented in 
    Figure~\ref{fig:rxnorm}. Here the PCA projection of features on two most 
    significant components is presented.  Further, some realisations are 
    annotated as background data or outliers. The ellipse represents the 
    approximate threshold of the multivariate Gaussian model at $95\%$ 
    confidence. Hence, if data were multivariate Gaussian distributed (with the 
    same covariance matrix for data and outliers), the ellipse should refer to 
    the detection threshold. Given those,  most of the outliers are supposed to 
    lay outside the ellipse.  Analogically most ($95\%$) of background data are 
    supposed to lay inside the ellipse. In Figure~\ref{fig:rxmoro} we 
    demonstrate the result, where most of the outliers are detected.
    
    Nevertheless, we have a significant number of false detections due to the 
    non-Gaussian structure of the background. In Figure~\ref{fig:rxwhite}, 
    although outliers are on the side of the background data, the vast number 
    of them lie within the background model.  One can see here the non-Gaussian 
    structure of data as well, however different than in 
    Figure~\ref{fig:rxmoro} case. Given those, in~\cite{glomb2018application} 
    authors have concluded, that data pattern suggests the non-Gaussian 
    distribution of outliers, as well as of background data given some (mainly 
    not uniform) materials.
	
		\begin{figure*}[!t]
		\centering
		\subfigure[A mixture of cotton and polyester (a military uniform). 
		\label{fig:rxmoro}]{
			\includegraphics[width=0.48\linewidth]{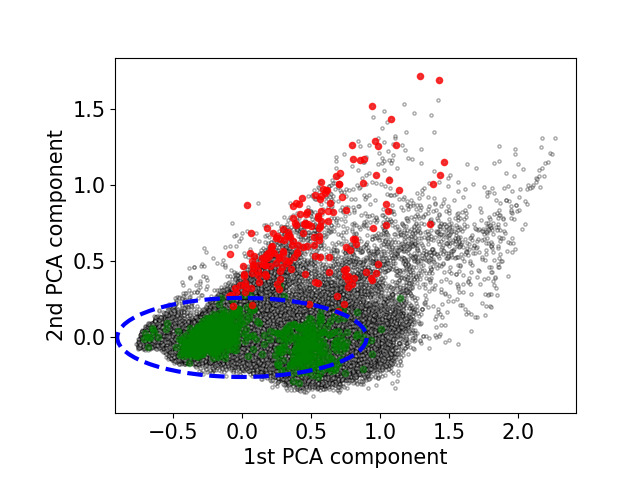}}
		\subfigure[White cotton (a shirt).\label{fig:rxwhite}]{	
			\includegraphics[width=0.48\linewidth]{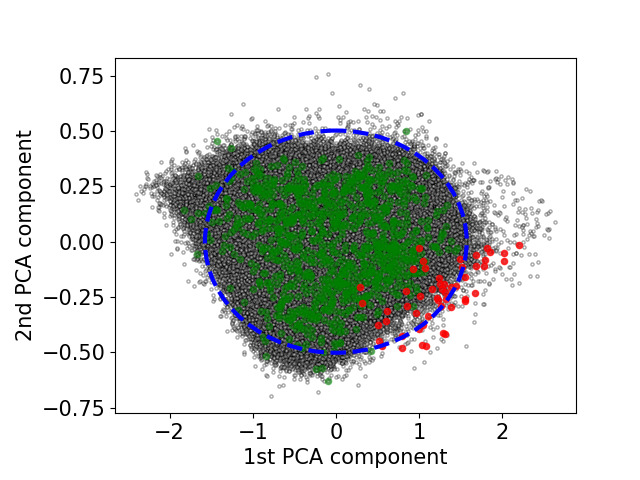}}
		\caption{Each figure presents a PCA projection on two most significant 
			components. Dark grey dots mark projections of not annotated 
			hyper-spectral vectors, light grey dots mark projections of vectors 
			annotated as 
			background, while black dots mark projections of vectors annotated 
			as 
			GSR.}
		\label{fig:rxnorm}
	\end{figure*}

	For the supervised approach the detection 
	rate is consistent or slightly better than in the unsupervised case 
	and the number of false positives is lower. What is important from the 
	application point of view, authors  have shown that the hyper-spectral 
	imaging in general outperforms the RGB (red green blue) imaging in term of 
	the GSR detection accuracy.

Promising directions of further research include the development of background 
model, possibly using copulas, to build background's non-Gaussian probabilistic 
model. Further one can evolve features selection algorithm to find such 
features that have highly non-Gaussian joint distribution and carry meaningful 
information about outliers. Alternatively, one can find such linear combination 
of features that carries most information about non-Gaussian joint distribution 
of features in the features extraction procedure. Both of above can be 
performed by means of higher order cumulants tensors. 
\end{example}

From such examples, one can conclude that it would be beneficial to have the 
measure determining 
how far the multivariate distribution of data or its subset (a background data, 
outliers, etc...) is from the multivariate Gaussian one.
Such measure may be introduced, by means of higher 
order cumulants tensors that are zero for multivariate Gaussian 
distributed data, see Section~\ref{sec::cumsg}. As discussed 
in~\cite{domino2018sliding} the simple approach, to construct such measure, 
would be to compute the Frobenius norm of higher order cumulants tensors. 
For an arbitrary tensor $\mathcal{T} \in \R^{d_1 
\times ,\ldots, \times d_n}$ the Frobenius norm is given by
\begin{equation}\label{eq::cumnorm}
	\| \mathcal{T} \| = \sqrt{\sum_{i_1, \ldots i_d} \left(t_{i_1, \ldots, 
	i_d}\right)^2}.
\end{equation}
The norm in Eq.~\eqref{eq::cumnorm} is easily implementable for 
the super-symmetric tensor (such as cumulant's tensor), that uses blocks to 
store only its meaningful part, see Section~\ref{sec::blockimpl}. For the
algorithmic implementation see~\cite{domino2018sliding} and Algorithm $5$ 
within, while for implementation in Julia programming language 
see~\cite{cupdat} and \texttt{norm()} function within. Following 
discussion in~\cite{domino2018sliding} we 
use the normalised 
norm
\begin{equation}
h_{\text{norm,d}}(\CC_2, \CC_d) = \frac{\| \CC_d \|}{\|\CC_2 \|^{d/2}} \text{ 
for } d \geq 3,
\end{equation}
to measure how much the 
probabilistic model of data diverges from the multivariate Gaussian 
distribution. 
The motivation for such normalisation in denominator
comes from 
the fact that in an univariate 
case,~\ie~$n=1$ we reproduce the module of the asymmetry for $d = 3$ 
and the module of the kurtosis for $d = 4$. 

In Figure~\ref{fig::norms} we 
present 
$h_{\text{norm,d}}$ for artificial data modelled by the $t$-Student copula 
(Figure~\ref{fig::t1norm}) and the Fr\'echet copula (Figure~\ref{fig::fnorm}). 
The $t$-Student copula is parametrised by the scalar parameter $\nu = 1$ and 
the matrix parameter 
$\mathbf{R}$ with constant off-diagonal elements equal to $r=0.53$ and ones on 
the diagonal. This reflect constant cross-correlation between marginals. The 
Fr\'echet copula is parametrised by $\alpha = 0.5$. The cross-correlation 
between marginals  in the $t$-Student copula case is similar to the 
cross-correlation in the 
Fr\'echet copula case. Further, in both 
cases we use standard Gaussian marginals (with zero kurtosis), to reveal only 
copulas properties, see Section~\ref{sec::cumcops} to justify such approach. 
We performed all computations for $t = 5 \times 10^6$ 
realisations of $n = 30$ features. 

\begin{figure}
    \begin{center}
	\subfigure[$t$-Student copula $\nu = 1, r = 	
	0.53$\label{fig::t1norm}]{\includegraphics{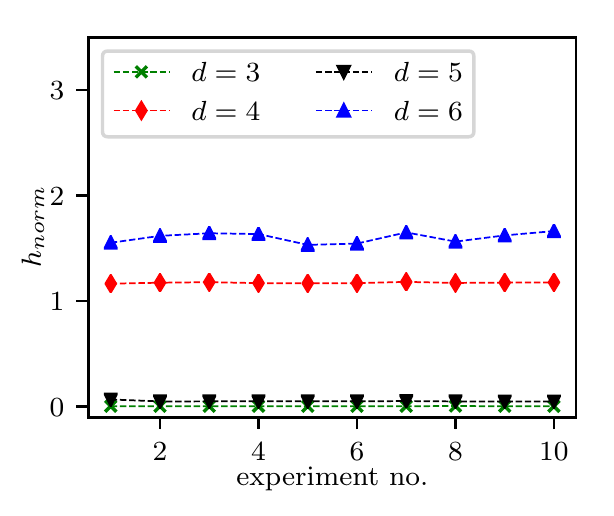}}
	\subfigure[Fr\'echet 
	copula $\alpha = 0.5$ 
	\label{fig::fnorm}]{\includegraphics{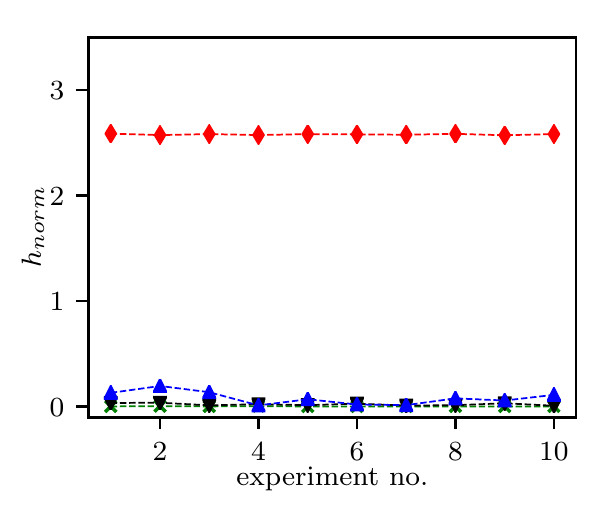}}
	\caption{$h_{\text{norm,d}}$ of higher order cumulants for $t$-Student and 
		$1$-parameter Fr\'echet 
		copulas with standard Gaussian marginals and $n = 30$ 
		features. Given such parametrisation both copulas have similar 
		cross-correlation 
		between marginals. For $t$-Student and Fr\'echet 
		copulas with standard Gaussian marginals we have a strong suggestion of 
		zero 
		odd order cumulants. Observe the 
		characteristic case of the $1$-parameter Fr\'echet 
		copula, where only the $4$\textsuperscript{th} cumulant's tensor is 
		non-zero, while other cumulants (of order $3,5,6$) appear to be 
		zero.}\label{fig::norms}
    \end{center}
\end{figure}

Analysing Figure~\ref{fig::norms}, we can conclude that for both 
copulas, there is a strong suggestion that odd order higher order cumulants are 
zero. 
In the case 
of the $t$-Student copula it results form the specific symmetry of the 
$t$-Student multivariate distribution~\cite{kotz2004multivariate} that is the 
base of the $t$-Student copula. In the case of the $1$-parameter Fr\'echet 
copula there is a 
strong suggestion that cumulant of order $6$ is zero as well, while cumulant 
of order $4$ is not zero. Interestingly the $6$\textsuperscript{th} cumulant is 
not zero for bivariate $2$-parameters Fr\'echet copula introduced by 
Definition~\ref{def::frech2}, however such copula is not applicable for $n > 2$.

Concluding, the $h_{\text{norm}, d}$ can be used in 
features selection given various copulas models. Given Gaussian univariate 
marginals we have
\begin{enumerate}
	\item $h_{\text{norm, 4}} = h_{\text{norm, 6}} = 0$ for Gaussian copula,
	\item $h_{\text{norm, 4}} \neq 0$ and $h_{\text{norm, 6}} \approx 0$ for 
	Fr\'echet copula,
	\item $h_{\text{norm, 4}} \neq 0$ and $h_{\text{norm, 6}} \neq 0$ for 
	$t$-Student copula.
\end{enumerate}
While using other copulas families (such as Archimedean) the $h_{\text{norm}, 
d}$ for odd $d$ would be non-zero as well, see Subsection~\ref{sec::ac}. 
Hence more selection scenarios are possible.
Apart from this, norms of higher order cumulants can be used to analyse 
real-life data, as in following example.

\begin{example}\label{eg::pigs}

Consider the analysis of real-life biomedical data 
analysis \cite{dominodomino2018} where electromyographic (EMG) 
\cite{eswaran2004prediction} signals recorded on the porcine uterine walls, 
during the pre-fertilisation and early pregnancy periods. The motivation for 
such research comes from the fact, that pigs uterine activity has been used to 
model those of humans due to their close similarities 
\cite{fanchin2009uterine}. Refer to researches on the role of the uterine 
activity in promotion of the fertilization~\cite{sammali2018feasibility}.  
In~\cite{dominodomino2018} the novelty of the study comes from the fact that 
the analytic tools there are higher order cumulants tensors.

In practice, the EMG signal organises in the form of bursts reflecting the 
electrical activity in the uterus. In~\cite{dominodomino2018}, these bursts are 
investigated employing commonly used features in the EMG signal analysis. 
Features in the time domain \cite{devedeux1993uterine, gajewski2004oviductal} 
are Duration of the burst, Pause -- length of the pause between bursts, 
Amplitude of the burst,  and RMS -- root mean square of the burst.  Features in 
the frequency domain \cite{oczeretko2007uterine} are MaxP -- maximum power ob 
the burst, MinP -- minimum power of the burst, and DF -- and dominant frequency 
of the burst.  The motivation for the second approach comes from the fact, that 
frequency domain features were suggested to be more effective than the 
conventional time-dependent features while analysing the uterine 
dynamics~\cite{garfield2007biophysical}. Given those, quantitative multivariate 
data in~\cite{dominodomino2018} were discussed here features of the EMG signal 
recorded in different part of the porcine uterus and different days of the 
pre-fertilisation and early pregnancy periods. These locations are: the tip of 
the right uterine horn (channel $1$), the middle of the right uterine horn 
(channel $2$), and the corpus uteri (channel $3$), see 
Figure~\ref{fig::uterus}. Data were collected from $8$ Polish Landrace sows and 
aggregated together for a given location and a given day of the experimental 
period.

\begin{figure}
	\includegraphics[width=\textwidth]{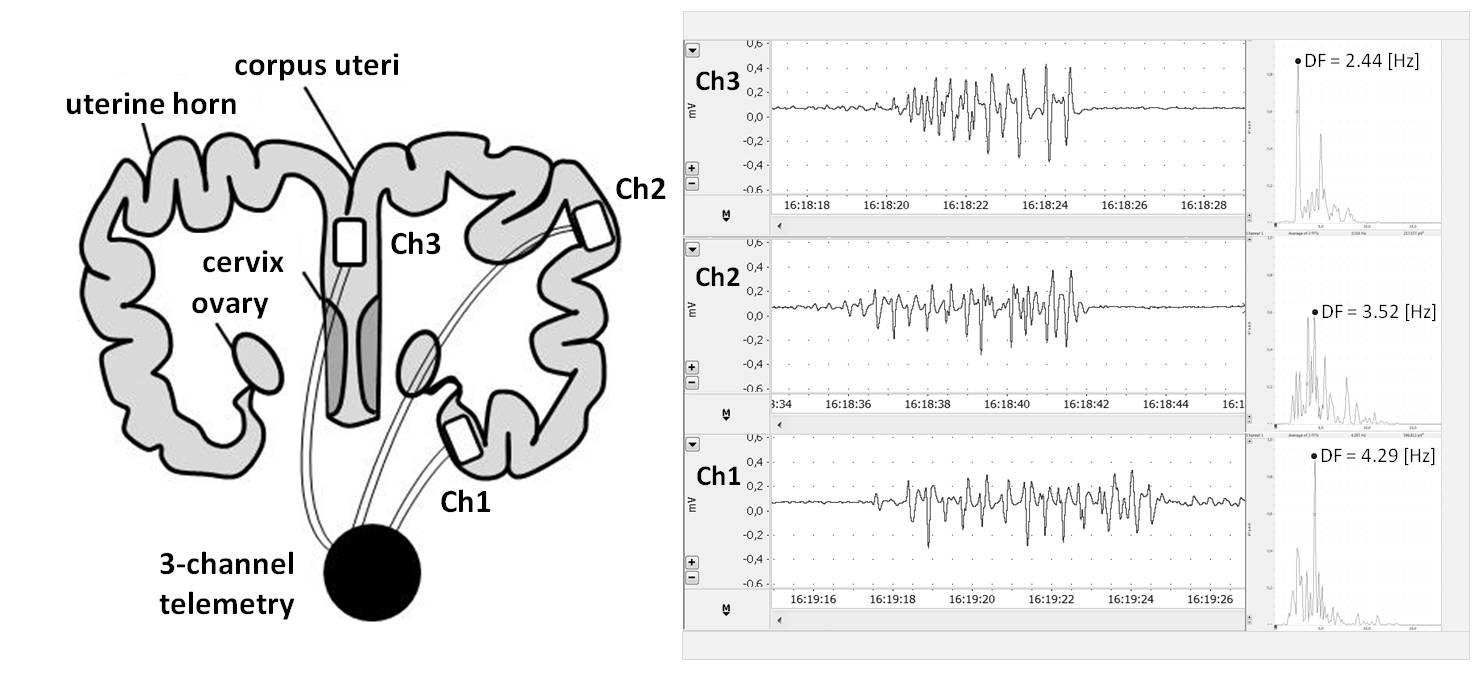}
	\caption{The experimental setting used for data acquisition 
	in~\cite{dominodomino2018}. Here channels correspond to various locations 
	of probes on the uterus on the left panel. Example of rough EMG data 
	recorded are presented on the right panel.}\label{fig::uterus}
\end{figure}
Due to complex  
biochemical process associated with the uterine 
behaviour~\cite{rabotti2015propagation}, authors expect non-Gaussian joint 
distribution of the features. For each day of the pre-fertilisation and early 
pregnancy periods and each channel, features are represented in the form of 
matrix $\mathbf{X} \in \R^{t \times n}$. Each time we have $t=400$ 
realisations and $n=7$ features (Duration; Pause; Amplitude; RMS; MaxP; 
MinP;  DF). Since values of different features (marginals) differs by many 
orders of magnitude, authors normalise each column vector of $\mathbf{X}$ 
by its standard deviation. Normalised data are represented in the form of 
matrix 
$\mathbf{\hat{X}}\in \R^{t \times n}$, see 
Eq.~$(4)$ in~\cite{dominodomino2018}. Such normalisation facilitates the 
analysis of the joint distribution and higher order cross-correlations between 
features.

\begin{figure}[ht!]
	\subfigure[channel 
	$3$ - corpus 
	uteri\label{fig::bioc33}]{\includegraphics{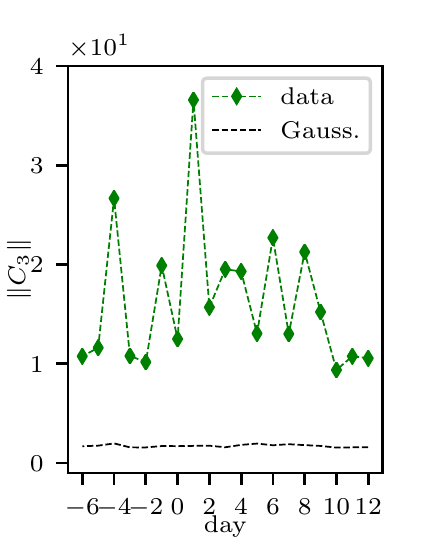}}
	\subfigure[channel 
	$2$ - middle of the right 
	horn of 
	the uterus\label{fig::bioc32}]{\includegraphics{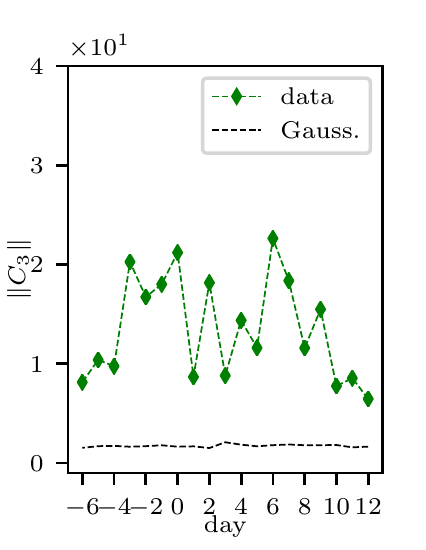}}
	\subfigure[channel $1$ - tip of the right horn of the uterus
	\label{fig::bioc31}]{\includegraphics{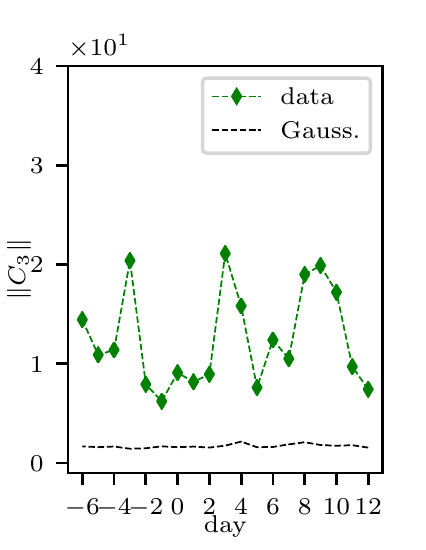}} \\
		\subfigure[channel 
	$3$ - corpus 
	uteri\label{fig::bioc43}]{\includegraphics{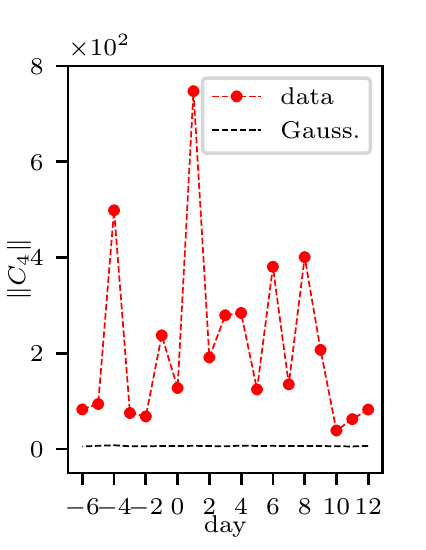}}
	\subfigure[channel 
	$2$ - middle of the right 
	horn\label{fig::bioc42}]{\includegraphics{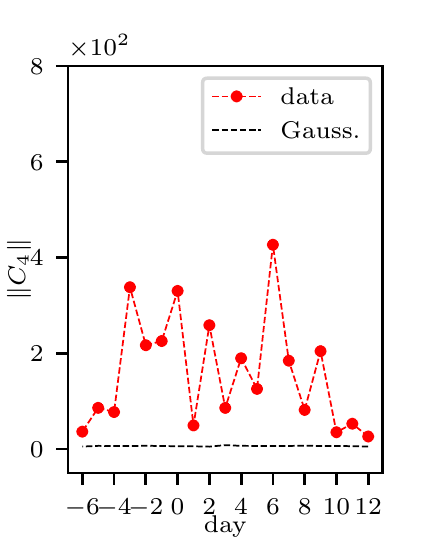}}
	\subfigure[Chanel $1$ - tip of the right horn
	\label{fig::bioc41}]{\includegraphics{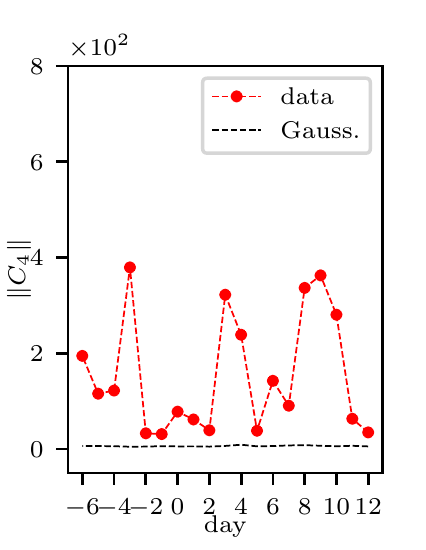}}
	\caption{Frebenious norms of higher order cumulants tensors calculated for 
$\mathbf{\hat{X}}$ - features of EMG signals, collected on 
	pigs uterus. Channels indicates position on the uterus where data are 
	collected. On the horizontal scale we have days of the per-fertilisation 
	and early pregnancy cycle. Dashed black line indicates $0.975$ percentile 
	of $\| \CC_d\|$ given hypothetical Gaussian model. 
	Day $0$ corresponds to the Artificial Insemination (AI), 
	see~\cite{dominodomino2018}.}\label{fig::bio}
\end{figure}

The goal of \cite{dominodomino2018} was 
to apply higher order cumulants tensors to analyse probabilistic models of data 
for distinct days of the pre-fertilization period and the early pregnancy. To 
measure how much the probabilistic model of features diverged form the Gaussian 
model, one uses the Frobenius norm of the cumulants tensors of order $3$ 
and $4$, see Eq.~\eqref{eq::cumnorm}. Results of such analysis are 
presented in Figure~\ref{fig::bio}. A straight forward conclusion is that 
analysed features do not have multivariate Gaussian distribution, what should 
be taken into account in their analysis. Further multivariate higher order 
cumulants indicates that probabilistic models of these features changes 
expressly on the daily basis. As discussed in~\cite{dominodomino2018}, these 
changes appears in concordance with crucial periods of the development and 
maintenance of pregnancy.

\begin{enumerate}
	\item Authors suggest the first pick of $\| \CC_d\|$ that 
begins in the corpus uteri in day $-4$, see 
Figures~\ref{fig::bioc33}~\ref{fig::bioc43} is the signal of the estrus cycle  
synchronization. Next it spreads along the uterine horn (days $-3$ to $0$ in 
Figures~\ref{fig::bioc32}~\ref{fig::bioc42}) into the tip of the uterine horn 
(day $-3$ in Figures~\ref{fig::bioc31}~\ref{fig::bioc41}).
\item The second pick of $\| \CC_d\|$ occurs shortly after the Artificial 
Insemination (AI) in the corpus uteri (day $1$, 
Figures~\ref{fig::bioc33}~\ref{fig::bioc43}), next it moves to the middle 
of the utters 
horn (day $2$, Figures~\ref{fig::bioc32}~\ref{fig::bioc42}) and finally to the 
tip of the uterine horn (day $3,4$, 
Figures~\ref{fig::bioc31}~\ref{fig::bioc41}). Authors suggest that such 
behaviour is related to the sperm cells transport.
\item The third pick occurs only at the middle in the uterine horn (days $6$ to 
$8$, 
Figures~\ref{fig::bioc32}~\ref{fig::bioc42}) and in the tip (days $8$ to $10$, 
Figures~\ref{fig::bioc31}~\ref{fig::bioc41}). It is probably caused by the he 
descent of embryos into the uterine lumen.
\end{enumerate}

To summarize, in~\cite{dominodomino2018} cumulants tensors of order $3$ and $4$ 
were used to anticipate information about the 
probabilistic model of data. It is why authors used the $3$\textsuperscript{rd} 
cumulant tensor based features selection procedure to achieve a rank of 
features importance. For this purpose authors 
applied the features selection 
procedure introduced in~\cite{geng2015joint} for hyper-spectral 
data analysis. This procedure is called the JSBS (Joint Skewness Band 
Selection) and is
discussed in this book in Section~\ref{sec::cumsfsel}. Results of the ranking 
of the importance of features for all days and channels are presented in Table 
$1$ 
of~\cite{dominodomino2018}. Based on these authors speculate for investigated 
animals, features: MinP, 
MaxP and Pause are 
crucial 
elements of the middle of the uterine horn activity. This is an important 
result in evolving models of 
the uterine activity, that partially approves the outcome 
of~\cite{garfield2007biophysical}, hence we have two features in a frequency 
domain.

Concluding, modelling based on higher order cumulants tensors offers a unique 
opportunity to understand the mechanisms underlying uterine contractility and 
development of early pregnancy in animals.
Further such novel approaches may point the direction of research aimed to 
develop uterine models, hence improve management of the early pregnancy both in 
animals and in human, due to close similarities 
\cite{fanchin2009uterine} of models in both cases. The $3$\textsuperscript{rd} 
cumulant's tensor bases JSBS features selection used in~\cite{dominodomino2018} 
is a member of the wider family of features selection methods that are 
introduced and discussed in next section of this book.
\end{example}

Using the analogical approach in \cite{domino2018comp} authors have applied 
higher 
order culuants tensors to analyse features of the EMG signal of the estrus 
cycle in 
pigs. Here, in contrary to the classical statistical approach, the higher order
cumulants tensors model shows evidences of large EMG activity during 
induced estrus. This is in contrary to the spontaneous estrus. Authors of 
\cite{domino2018comp} concludes, that 
higher order cumulants tensors are indicator revealing a probabilistic model of 
dynamic changes in the myometrial electrical activity as a result of positive 
and negative feedback of the reproductive hormones. 

Referring to the last two paragraphs of Example~\ref{eg::pigs}, and the 
application of the cumulants tensors in the features selection for 
hyper-spectral data~\cite{geng2015joint} we can move to the discussion of 
cumulant's based features selection algorithms in next section.

\section{Features selection}\label{sec::cumfdet}

In this section we discuss algorithms aimed to select a subset of marginals 
$\mathbf{r} \in (1:n)$ 
of multivariate data that carries a meaningful information tied to the 
non-Gaussian joint distribution of such subset. Suppose we have 
$n$-variate random vector $\XXX^{(n)}$ 
with marginals $\XXX_1, 
\ldots, \XXX_n$ and data are given in a form of matrix $\mathbf{X} \in \R^{t 
\times n}$, where columns
correspond to marginals. Following~\cite{geng2015joint, domino2018use, 
glombband2018} we can introduce a 
family of features selection algorithms that uses higher order cumulants 
tensors to select a subset of marginals with non-Gaussian joint 
distribution.

\begin{definition}\label{def::fcut}
Let $\mathcal{T} \in \R^{[n, d]}$ be the super-symmetric tensor. Following 
\cite{geng2015joint, glombband2018}, we define its $r$\textsuperscript{th} 
fibres cut, 
as the following tensor
$\mathcal{T}_{(-r)} = \mathcal{T'} \in \R^{[(n-1), d]}$, where
\begin{equation}
t'_{i'_1, \ldots i'_d} = t_{i_1, \ldots, i_d}: i'_k = \begin{cases} i_k 
 \ &\text{if} \ \ \ \ \ \ \  i_k < r \\
i_{k-1}
 \ &\text{if} \ r \leq i_k < n,
\end{cases} 
\end{equation}
and $i_k' \in (1:n-1)$. Referring to \cite{kolda2009tensor} we simply remove 
all $r$\textsuperscript{th} fibres from $\mathcal{T}$, obviously such 
transformation preserves super-symmetry.
\end{definition} 

It is easy to show referring to Definition~\ref{def::cumsdef}, that given 
$\XXX^{(n)}$ and its cumulant's tensor $\CC_d$ if we change $\XXX^{(n)}$ to 	
${\XXX'}^{(n-1)}$ by removing $r$\textsuperscript{th} marginal, the new 
cumulant's tensor would be $\CC'_d = {\CC_{d}}_{(-r)}$. Hence 
following~\cite{geng2015joint, domino2018use, glombband2018} we can use the 
following target function
\begin{equation}
h(\CC_{d_1}, \CC_{d_2}, \ldots)
\end{equation}
and maximise it at each iterative step, while removing a feature that carries 
little information tied to non-Gaussian joint distribution of features. Such 
procedure, see Algorithm~\ref{alg::itterremove}, 
 is discussed in~\cite{domino2018use}, and~\cite{glombband2018} where it was 
 applied to analyse hyper-spectral data in the small target detection problem. 
 For the implementation in the Julia programming language 
see~\cite{cumf}. In following subsections we discuss 
various target functions.

\begin{algorithm}[t]
 	\caption{Cumulants based features selection, see also ~\cite{glombband2018}.
 		\label{alg::itterremove}}
 	\begin{algorithmic}[1]	
 		\State \textbf{Input}: $\CC_{d_1} \in \R^{[n, d_1]}, \CC_{d_2} \in 
 		\R^{[n, d_2]}, \ldots$ - cumulant's tensors, $h$ - target function, $s 
 		< n$ - Int, a stop 
 		condition. 	
 		\State \textbf{Output:} a subset of 
 		marginals that carries important information.
 		\Function{features select}{$\CC_{d_1}, \CC_{d_2}, \ldots, h, s$}
 		\For {$n' \text{ in }  n , n-1, \ldots, s$}
 		\For {$i \gets 1 \textrm{ to } n'$}
 		\State $m_i = h({\CC_{d_1}}_{(-i)}, {\CC_{d_2}}_{(-i)}, \ldots)$
 		\EndFor
 		\State \text{set} $r \text{ such that } m_r = \text{max}(\{m_1, \ldots, 
 		m_{n'}\})$
 		\State $\CC_{d_1}, \CC_{d_2}, \ldots = {\CC_{d_1}}_{(-r)}, 
 		{\CC_{d_2}}_{(-r)}, \ldots$
 		\EndFor
 		\State \Return reminding $s$ marginals.
 		\EndFunction 
 	\end{algorithmic}
 \end{algorithm}

\subsection{Classical method MEV}\label{sec::mev}

Before moving to higher order cumulants, for the sake of comparison, consider 
the 
MEV (Maximum Ellipsoid Volume) 
\cite{sheffield1985selecting} that uses information stored in the covariance 
matrix - hence the MEV is appropriate for multivariate Gaussian 
distributed data. Following \cite{sheffield1985selecting} the MEV measures 
information stored in the covariance matrix by means of the volume of the 
hyper-ellipsoid spanned on its eigenvectors. Its follows from fact that the 
covariance matrix $\SSSS$, being real valued, symmetric, and positively semi 
definite, can be decomposed according to Eq.~\eqref{eq::cvd}:
\begin{equation}\label{eq::SVDagain}
	\SSSS= \mathbf{A} \SSd \mathbf{A}^{\intercal},
\end{equation}
where $\SSSS_{(d)}$ is a diagonal matrix with elements $\lambda_1, \ldots, 
\lambda_n$ on its diagonal. Here each $\lambda_i \geq 0$ is a non-negative 
eigenvector of $\SSSS$. Columns of $\mathbf{A}$, \ie~$A_1, \ldots, A_n$ are 
normalised 
eigenvectors of 
$\SSSS$.
Since $\mathbf{A}$ is an unitary 
matrix, the determinant of the covariance matrix can be written as:
\begin{equation}\label{eq::detmev}
	\det(\SSSS) = \det(\SSd) = \prod_{i=1}^n \lambda_i.
\end{equation}
This determinant is proportional to the hyper-ellipsoid volume span by vectors 
$\lambda_1 A_1, 
\ldots,$ 
$\lambda_n A_n$. Thus, as discussed in~\cite{sheffield1985selecting}, we can 
use such decomposition in data processing, if we 
have multivariate Gaussian distributed data and the meaningful information is 
tied to the variance of features (we are not interested in the mean vector). In 
such approach, the higher the absolute value of the cross-correlation between 
two 
features (marginals) the more information carried by the first one is carried 
by the second 
one as well. Having a pair of highly (positively or negatively) 
cross-correlated marginals, two columns of the correlation matrix $\SSSS$ would 
be similar (or approximately proportional) and one eigenvalue in 
Eq.~\eqref{eq::detmev} would be small (near zero) affecting strongly 
the determinant.
Suppose now, that we remove one marginal, in such a way, that the determinant 
of reminding 
data is maximised. Here we would remove one of two highly cross-correlated 
marginals, but we would preserve most information. 
Based on these, we can have the iterative dimensionality reduction procedure. 
At each step we 
remove one marginal in such a way, that we maximise the 
determinant of the correlation matrix of reminding data. Performing iterations, 
we would obtain a subset 
of marginals that carry significant part of information. Nevertheless 
the discussion on the lose of information, and acceptable information 
error is necessary to determine the stop condition of iterations. For further 
discussion on the MEV see~\cite{sheffield1985selecting}. Importantly, if the 
multivariate distribution of data is non-Gaussian, not all 
information would be tie to the covariance matrix and the MEV procedure may 
not be optimal. Such scenarios will be discussed in the reminding part of 
this section.

At the end let us refer to the 
simple example. Suppose we have two repeating 
marginals. 
The covariance matrix would have two the same columns, one of its 
eigenvalue would be zero, and the determinant in Eq.~\eqref{eq::detmev} 
would be zero. By removing one of the two marginals (assuming all other 
marginals differs from each other) we have a non-zero determinant of reminding 
data set, but the information carried by data is not reduced.

\subsection{Cumulant based features selection}\label{sec::cumsfsel}

Following~\cite{domino2018use, glombband2018} we propose to modify the MEV 
algorithm to handle non-Gaussian  distributed data. Suppose, we have 
data for with information is tied to non-Gaussian joint distribution of 
features. To 
analyse such data, we apply method based on the High Order Singular Value 
Decomposition (HOSVD), see~\cite{tucker1966some, de2000multilinear, 
jondeau2015moment} of the higher order cumulant's tensor.

\begin{definition}\label{def::hosvdsym}
	Let $\CC \in R^{[n, d]}$ be the super-symmetric tensor. The 
	relation
	\begin{equation}\label{eq::hsvd}
		\CC = \1 \times_{1,\ldots, d} \CC = 
		 \mathbf{A} \mathbf{A}^{\intercal} \times_{1,\ldots, d} \CC = 
		 \mathbf{A} \times_{1,\ldots, d} \CC',
	\end{equation}
	is called its HOSVD decomposition, for used there notation 
	see Definition~\ref{def::tmmodemult}. 
	Here $\mathbf{A} \in \R^{n \times n}$ is an unitary factor matrix 
	and 
	$\mathbf{A}
	\mathbf{A}^{\intercal} = \1$. Tensor 
	\begin{equation}\label{eq::coretensor}
	\R^{[n,d]} \ni \CC' = 
	\mathbf{A}^{\intercal}
	\times_{1,\ldots, d} \CC
	\end{equation} is called the 
	core-tensor. Eq.~\eqref{eq::coretensor} can be shown using element 
	wise notation, where
	\begin{equation}
	t'_{j_1, \ldots, j_n} = \sum_{i_1, \ldots i_n} a_{i_1, j_1} \dots a_{i_n, 
	j_n} t_{i_1, \ldots, i_n}.
	\end{equation}
 To introduce a graphic notation observe that by 
	Definition~\ref{def::matmatgraph} for an 
	unitary matrix the $\mathbf{A} 
	\mathbf{A}^{\intercal} = \1$ can be represented graphically by
		\begin{center}
			\begin{tikzpicture}
			\node[shape=rectangle,draw=black] (A) at (4,0) 
			{$\mathbf{A}$};
			\node[shape=rectangle,draw=black] (B) at (5.8,0) {$\mathbf{A}$};
			\node[shape=rectangle,draw=black] (C) at (1,0) {$\1$};
			\node at (3.44,0.2) {\scriptsize{$(1)$}};
			\node at (4.55,0.2) {\scriptsize{$(2)$}};
			\node at (5.3,0.2) {\scriptsize{$(2)$}};
			\node at (6.32,0.2) {\scriptsize{$(1)$}};
			\path [-] (C) edge node[left] {} (0,0);
			\path [-] (C) edge node[right] {} (2,0);
			\node at (2.5,0){$=$};
			\path [-] (A) edge node[left] {} (B);
			\path [-] (A) edge node[left] {} (3.0,0);
			\path [-] (B) edge node[left] {} (6.8,0);
			\end{tikzpicture}
		\end{center}
Hence, Eq.~\eqref{eq::hsvd} can be represented in a graphic notation (see 
Definition~\ref{def::graphicmoms}), from 
which 
clearly comes Eq.~\eqref{eq::coretensor}
	\begin{center}
		\begin{tikzpicture}
		\node[shape=rectangle,draw=black] (C) at (-1,0) {$\CC$};
		\node at (1.2,0){$=$};
		\path [-] (C) edge node[left] {} (-2.1,1);
		\path [-] (C) edge node[left] {} (.1,1);
		\path [-] (C) edge node[left] {} (-2.1,-1);
		\path [-] (C) edge node[left] {} (.1,-1);
		\node at (-1,1){$\cdots$};
		\node at (-1,-1){$\cdots$};
		\node at (0.2,0.1) {$\vdots$};
		\node[shape=rectangle,draw=black] (D) at (4,0) {$\CC$};
		\node at (4,1.5) {$\cdots$};
		\node[shape=rectangle,draw=black] (C) at (2.8,1) {$\mathbf{A}$};
		\node[shape=rectangle,draw=black] (CC) at (1.8,2.) {$\mathbf{A}$};
		\node at (3.4, .76) {\scriptsize{$(1)$}};
		\node at (2.63, 1.45) {\scriptsize{$(2)$}};
		\node at (3.0, -.55) {\scriptsize{$(1)$}};
		\node at (2.23, -1.25) {\scriptsize{$(2)$}};
		\node at (4.6, .78) {\scriptsize{$(1)$}};
		\node at (5.4, 1.5) {\scriptsize{$(2)$}};
		\node at (5.05, -.55) {\scriptsize{$(1)$}};
		\node at (5.8, -1.25) {\scriptsize{$(2)$}};
		\node at (1.6, 2.5) {\scriptsize{$(1)$}};
		\node at (2.38, 1.75) {\scriptsize{$(2)$}};
		\node at (1.25, -2.25) {\scriptsize{$(1)$}};
		\node at (2.0, -1.55) {\scriptsize{$(2)$}};
		\node at (6.4, 2.5) {\scriptsize{$(1)$}};
		\node at (5.7, 1.75) {\scriptsize{$(2)$}};
		\node at (6.75, -2.25) {\scriptsize{$(1)$}};
		\node at (6.05, -1.55) {\scriptsize{$(2)$}};
		\node at (5.6,0.1) {$\vdots$};
		\node[shape=rectangle,draw=black] (E) at (5.2,1) {$\mathbf{A}$};
		\node[shape=rectangle,draw=black] (EC) at (6.2,2.) {$\mathbf{A}$};
		\node at (4,-1.5) {$\cdots$};
		\node[shape=rectangle,draw=black] (G) at (2.8,-1) {$\mathbf{A}$};
		\node[shape=rectangle,draw=black] (GC) at (1.8,-2.) {$\mathbf{A}$};
		\node[shape=rectangle,draw=black] (I) at (5.2,-1) {$\mathbf{A}$};
		\node[shape=rectangle,draw=black] (IC) at (6.2,-2.) {$\mathbf{A}$};
		\path [-] (D) edge node[left] {} (E);
		\path [-] (D) edge node[left] {} (G);
		\path [-] (D) edge node[left] {} (C);
		\path [-] (D) edge node[left] {} (I);
		\path [-] (EC) edge node[left] {} (6.9,2.7);
		\path [-] (CC) edge node[left] {} (1.1,2.7);
		\path [-] (IC) edge node[left] {} (6.9,-2.7);
		\path [-] (GC) edge node[left] {} (1.1,-2.7);
		\path [-] (E) edge node[left] {} (EC);
		\path [-] (G) edge node[left] {} (GC);
		\path [-] (I) edge node[left] {} (IC);
		\path [-] (C) edge node[left] {} (CC);
		\node at (7.5,0){$=$};
		\node[shape=rectangle,draw=black] (DA) at (10,0) {$\CC'$};
		\node at (10,1) {$\cdots$};
		\node[shape=rectangle,draw=black] (CA) at (8.8,1) {$\mathbf{A}$};
		\node at (9.4, .76) {\scriptsize{$(2)$}};
		\node at (8.63, 1.45) {\scriptsize{$(1)$}};
		\node at (9.0, -.55) {\scriptsize{$(2)$}};
		\node at (8.23, -1.25) {\scriptsize{$(1)$}};
		\node at (11.2,0.1) {$\vdots$};
		\node at (10.6, .78) {\scriptsize{$(2)$}};
		\node at (11.4, 1.5) {\scriptsize{$(1)$}};
		\node at (11.05, -.55) {\scriptsize{$(2)$}};
		\node at (11.8, -1.25) {\scriptsize{$(1)$}};
		\node at (11.2,0.1) {$\vdots$};
		\node[shape=rectangle,draw=black] (EA) at (11.2,1) {$\mathbf{A}$};
		\node at (10,-1) {$\cdots$};
		\node[shape=rectangle,draw=black] (GA) at (8.8,-1) {$\mathbf{A}$};
		\node[shape=rectangle,draw=black] (IA) at (11.2,-1) {$\mathbf{A}$};		
		\path [-] (DA) edge node[left] {} (EA);
		\path [-] (DA) edge node[left] {} (GA);
		\path [-] (DA) edge node[left] {} (CA);
		\path [-] (DA) edge node[left] {} (IA);
		\path [-] (EA) edge node[left] {} (11.9,1.7);
		\path [-] (CA) edge node[left] {} (8.1,1.7);
		\path [-] (IA) edge node[left] {} (11.9,-1.7);
		\path [-] (GA) edge node[left] {} (8.1,-1.7);
		\end{tikzpicture}
	\end{center}
\end{definition}
The Eq.~\eqref{eq::coretensor} can be represented in the  
graphical notation as
	\begin{center}
	\begin{tikzpicture}
	\node[shape=rectangle,draw=black] (C) at (4,0) {$\CC'$};
	\node at (6.8,0){$=$};
	\path [-] (C) edge node[left] {} (3.1,1);
	\path [-] (C) edge node[left] {} (5.1,1);
	\path [-] (C) edge node[left] {} (3.1,-1);
	\path [-] (C) edge node[left] {} (5.1,-1);
	\node at (4,1){$\cdots$};
	\node at (4,-1){$\cdots$};
	\node at (5.2,0.1) {$\vdots$};
	\node at (9,1.5) {$\cdots$};
	\node[shape=rectangle,draw=black] (DA) at (10,0) {$\CC$};
	\node at (10,1) {$\cdots$};
	\node[shape=rectangle,draw=black] (CA) at (8.8,1) {$\mathbf{A}$};
	\node at (9.4, .76) {\scriptsize{$(1)$}};
	\node at (8.63, 1.45) {\scriptsize{$(2)$}};
	\node at (9.0, -.55) {\scriptsize{$(1)$}};
	\node at (8.23, -1.25) {\scriptsize{$(2)$}};
	\node at (11.2,0.1) {$\vdots$};
	\node at (10.6, .78) {\scriptsize{$(1)$}};
	\node at (11.4, 1.5) {\scriptsize{$(2)$}};
	\node at (11.05, -.55) {\scriptsize{$(1)$}};
	\node at (11.8, -1.25) {\scriptsize{$(2)$}};
	\node at (11.2,0.1) {$\vdots$};
	\node[shape=rectangle,draw=black] (EA) at (11.2,1) {$\mathbf{A}$};
	\node at (10,-1) {$\cdots$};
	\node[shape=rectangle,draw=black] (GA) at (8.8,-1) {$\mathbf{A}$};
	\node[shape=rectangle,draw=black] (IA) at (11.2,-1) {$\mathbf{A}$};		
	\path [-] (DA) edge node[left] {} (EA);
	\path [-] (DA) edge node[left] {} (GA);
	\path [-] (DA) edge node[left] {} (CA);
	\path [-] (DA) edge node[left] {} (IA);
	\path [-] (EA) edge node[left] {} (11.9,1.7);
	\path [-] (CA) edge node[left] {} (8.1,1.7);
	\path [-] (IA) edge node[left] {} (11.9,-1.7);
	\path [-] (GA) edge node[left] {} (8.1,-1.7);
	\end{tikzpicture}
\end{center}
Observe, that efinition~\ref{def::hosvdsym} concerns only the HOSVD of the 
super-symmetric tensor 
being of our interest since cumulants tensors are super-symmetric. In the 
general tensor 
case~\cite{de2000multilinear}, the HOSVD is more complex, yielding $d$ distinct
factor matrices.

Let us now discuss the standard method of finding the factor 
matrix~\cite{de2000multilinear} of the HOSVD of the super-symmetric tensor 
case.

\begin{definition}
	Let $\CC \in \R^{[n, d]}$ be a super-symmetric tensor, its 
	contraction with itself in $d-1$ 
	modes (all modes but first) produces the following symmetric matrix
	\begin{equation}\label{eq::bdef}
	\R^{[n,2]}	\ni \mathbf{B} = \sum_{\mathbf{i}\setminus(i_1)} \CC 
	\CC,
	\end{equation} 
	with elements
	\begin{equation}\label{eq::bedef}
	b_{j_1, j_2} = \sum_{i_2, \ldots, i_n} 
	c_{j_1, i_2, \ldots, i_n} 
	c_{j_2, i_2, \ldots, i_n}.
	\end{equation}
	The matrix $\mathbf{B}$ is symmetric, since indices $j_1$ and $j_2$ are 
	exchangeable in Eq.~\eqref{eq::bedef}. Further due to the 
	super-symmetry of $\CC$ it does not matter 
	in which $d-1$ modes we contract $\CC$ by itself. We can introduce a 
	graphic notation analogical to 
	Definition~\ref{def::matmatgraph} and Definition~\ref{def::graphicmoms}. 
	On the graphic notation, each connection between tensors accounts for the 
	contraction in one mode. In 
	Eq.~\eqref{eq::bedef} we have contractions over $n-1$ indices 
	(corresponding to 
	$n-1$ modes). These correspond 
	to the following graphical notation (with $n-1$ connections between 
	$\CC$ and~$\CC$),
	\begin{center}
		\begin{tikzpicture}
		\node[shape=rectangle,draw=black] (A) at (-4,0) {$\mathbf{B}$};
		\path [-] (A) edge node[left] {} (-5,0);
		\path [-] (A) edge node[right] {} (-3,0);
		\node at (-2.5,0){$=$};
		\node[shape=rectangle,draw=black] (C) at (-1,0) {$\CC$};
		\path [-] (C) edge node[left] {} (.1,1);
		\path [-] (C) edge node[left] {} (-2.1,0.);
		\path [-] (C) edge node[left] {} (.1,-1);
		\node at (0.1,0.45) {$\vdots$};
		\node at (0.1,-0.45) {$\vdots$};
		\node[shape=rectangle,draw=black] (D) at (1.2,0) {$\CC$};
		\path [-] (C) edge node[right] {} (D);
		\path [-] (D) edge node[left] {} (0.1,1);
		\path [-] (D) edge node[left] {} (2.2,0);
		\path [-] (D) edge node[left] {} (0.1,-1);
		\end{tikzpicture}
	\end{center}
	 We do 
	not number 
	edges due to symmetry / super-symmetry of all objects 
	involved.
\end{definition}
Replacing $\CC$ by its HOSVD, we have
	\begin{center}
		\begin{tikzpicture}
		\node[shape=rectangle,draw=black] (C) at (-3,0) {$\CC$};
		\path [-] (C) edge node[left] {} (-1.9,1);
		\path [-] (C) edge node[left] {} (-4.1,0.);
		\path [-] (C) edge node[left] {} (-1.9,-1);
		\node at (-1.9,0.45) {$\vdots$};
		\node at (-1.9,-0.45) {$\vdots$};
		\node[shape=rectangle,draw=black] (D) at (-0.8,0) {$\CC$};
		\path [-] (C) edge node[right] {} (D);
		\path [-] (D) edge node[left] {} (-1.9,1);
		\path [-] (D) edge node[left] {} (.2,0);
		\path [-] (D) edge node[left] {} (-1.9,-1);
		\node at (0.5,0){$=$};
		\node[shape=rectangle,draw=black] (A) at (3,0) {$\CC'$};
		\node[shape=rectangle,draw=black] (CA) at (4.1,1) {$\mathbf{A}$};
		\node[shape=rectangle,draw=black] (BA) at (4.1,-1) {$\mathbf{A}$};
		\node[shape=rectangle,draw=black] (AA) at (4.3,0) {$\mathbf{A}$};
		\node[shape=rectangle,draw=black] (DA) at (1.6,0) {$\mathbf{A}$};
		\node[shape=rectangle,draw=black] (CB) at (6.1,1) {$\mathbf{A}$};
		\node[shape=rectangle,draw=black] (BB) at (6.1,-1) {$\mathbf{A}$};
		\node[shape=rectangle,draw=black] (AB) at (5.9,0) {$\mathbf{A}$};
		\node[shape=rectangle,draw=black] (DB) at (8.6,0) {$\mathbf{A}$};
		\path [-] (A) edge node[left] {} (CA);
		\path [-] (A) edge node[left] {} (BA);
		\path [-] (A) edge node[left] {} (AA);
		\path [-] (A) edge node[right] {} (DA);
		\path [-] (DA) edge node[left] {} (.75,0.);
		\node at (5.1,0.75) {$\vdots$};
		\node at (5.1,-0.55) {$\vdots$};
		\node[shape=rectangle,draw=black] (B) at (7.2,0) {$\CC'$};
		\path [-] (AB) edge node[right] {} (B);
		\path [-] (AA) edge node[right] {} (AB);
		\path [-] (BA) edge node[right] {} (BB);
		\path [-] (CA) edge node[right] {} (CB);
		\path [-] (B) edge node[left] {} (BB);
		\path [-] (DB) edge node[left] {} (9.7,0);
		\path [-] (B) edge node[left] {} (CB);
		\path [-] (B) edge node[left] {} (DB);
		\node at (3.82,0.2) {\scriptsize{$(2)$}};
		\node at (4.8,0.2) {\scriptsize{$(1)$}};
		\node at (3.6,.83) {\scriptsize{$(2)$}};
		\node at (4.6,1.15) {\scriptsize{$(1)$}};
		\node at (5.6,1.15) {\scriptsize{$(1)$}};
		\node at (6.6,0.83) {\scriptsize{$(2)$}};
		\node at (3.6,-.83) {\scriptsize{$(2)$}};
		\node at (4.6,-0.83) {\scriptsize{$(1)$}};
		\node at (5.6,-0.83) {\scriptsize{$(1)$}};
		\node at (6.6,-0.83) {\scriptsize{$(2)$}};
		\node at (1.1,0.2) {\scriptsize{$(1)$}};
		\node at (2.12,0.2) {\scriptsize{$(2)$}};
		\node at (5.4,0.2) {\scriptsize{$(1)$}};
		\node at (6.38,0.2) {\scriptsize{$(2)$}};
		\node at (8.1,0.2) {\scriptsize{$(2)$}};
		\node at (9.12,0.2) {\scriptsize{$(1)$}};
		\end{tikzpicture}
	\end{center}
Using $\mathbf{A}^{\intercal} \mathbf{A} = \1$ we have
	\begin{center}
		\begin{tikzpicture}
		\node[shape=rectangle,draw=black] (C) at (-3,0) {$\CC$};
		\path [-] (C) edge node[left] {} (-1.9,1);
		\path [-] (C) edge node[left] {} (-4.1,0.);
		\path [-] (C) edge node[left] {} (-1.9,-1);
		\node at (-1.9,0.45) {$\vdots$};
		\node at (-1.9,-0.45) {$\vdots$};
		\node[shape=rectangle,draw=black] (D) at (-0.8,0) {$\CC$};
		\path [-] (C) edge node[right] {} (D);
		\path [-] (D) edge node[left] {} (-1.9,1);
		\path [-] (D) edge node[left] {} (.2,0);
		\path [-] (D) edge node[left] {} (-1.9,-1);
		\node at (0.5,0){$=$};
		\node[shape=rectangle,draw=black] (A) at (3,0) {$\CC'$};
		\node[shape=rectangle,draw=black] (DA) at (1.6,0) {$\mathbf{A}$};
		\node[shape=rectangle,draw=black] (DB) at (6.6,0) {$\mathbf{A}$};
		\path [-] (A) edge node[left] {} (4.1,1);
		\path [-] (A) edge node[left] {} (4.1,-1);
		\path [-] (A) edge node[left] {} (4.1,0.);
		\path [-] (A) edge node[right] {} (DA);
		\path [-] (DA) edge node[left] {} (.75,0.);
		\node at (4.1,0.5) {$\vdots$};
		\node at (4.1,-0.5) {$\vdots$};
		\node[shape=rectangle,draw=black] (B) at (5.2,0) {$\CC'$};
		\path [-] (DB) edge node[left] {} (B);
		\path [-] (B) edge node[left] {} (4.1,-1);
		\path [-] (DB) edge node[left] {} (7.7,0);
		\path [-] (B) edge node[left] {} (4.1,1);
		\path [-] (B) edge node[left] {} (4.1,0);
		\node at (1.1,0.2) {\scriptsize{$(1)$}};
		\node at (2.12,0.2) {\scriptsize{$(2)$}};
		\node at (6.1,0.2) {\scriptsize{$(2)$}};
		\node at (7.12,0.2) {\scriptsize{$(1)$}};
		\end{tikzpicture}
	\end{center}
	Analogically, using tensor notation
	\begin{equation}\label{eq::hsvdt}
	\sum_{\mathbf{i} \setminus i_1} \CC \CC = 
	\mathbf{A} \left( 
	\sum_{\mathbf{i} \setminus i_1} \CC' \CC' \right) 
	\mathbf{A}^{\intercal}.
	\end{equation}
	Recall now that 
	Eq.~\eqref{eq::hsvdt} is similar to Eq.~\eqref{eq::SVDagain} in 
	such a way in which the HOSVD is similar to the Eigenvalue Decomposition.
	Our strategy is to perform the Singular Value Decomposition of the
	$\sum_{\mathbf{i} \setminus i_1} \CC \CC$ matrix, yielding 
	$\sum_{\mathbf{i} \setminus i_1} \CC' \CC'$ that is a 
	diagonal matrix of eigenvalues of $\sum_{\mathbf{i} \setminus i_1} 
	\CC \CC$. 
	Importantly, in Eq.~\eqref{eq::bdef} we can rearrange summation over 
	$i_2, \ldots, 
	i_d$ to the summation over one index 
	\begin{equation}\label{eq::j}
		j = 1+\sum_{k=2}^d\left(i_k-1\right)d^{k-2}.
	\end{equation}
	Hence $j$ is uniquely determined by $i_2, 
	\ldots, 
	i_d$. Now Eq.~\eqref{eq::bdef} is equivalent to
	\begin{equation}\label{eq::hsvdunfold}
		\mathbf{B} = \mathbf{V}^{\intercal} \mathbf{V} 
		= \sum_{\mathbf{i} \setminus i_1} \CC \CC,
	\end{equation}
	where $\mathbf{V} \in \R^{n^{d-1} \times n}$ is a matrix being 
	a transposition of the	unfold in first 
	mode of the tensor $\CC$ according to definition~\cite{kolda2009tensor}. 
	Elements of $\mathbf{V}$ are
	\begin{equation}\label{eq::unf}
		v_{j, i_1} = c_{i_1, i_2, \ldots, i_d},
	\end{equation}
	where index $j$ is given by Eq.~\eqref{eq::j}. Due to 
	Eq.~\eqref{eq::hsvdunfold} 
	we 
	expect $\mathbf{B}$ to be positive definite. To enlighten the meaning of 
	the $\mathbf{V}$ in higher order cross-correlation analysis let us refer 
	to the following remark and example.
\begin{remark}\label{rem::cinf}
	Suppose,
	$\mathbf{V}$ is an unfold of the $d$\textsuperscript{th} order cumulant of 
	original 
	data $\X$, see Eq.~\eqref{eq::unf}. We, can treat 
	$\mathbf{V}$ as a data vector of $n$ marginals and 
	$n^{d-1}$ realisations. Its $i$\textsuperscript{th} column ($V_i$) contains 
	all possible $d$\textsuperscript{th} order cross-correlation between the 
	$i$\textsuperscript{th} marginal of original data $\X$ and all other 
	marginals. See also~\cite{domino2018use} for further 
	discussion on the $4$\textsuperscript{th} cumulant's case. Hence 
	$\mathbf{V}$ carries information about higher order corss-correlation 
	between margins of original data. Each 
	column of $V_i$ contains a lot of combinations of $d$\textsuperscript{th} 
	order correlations with all other marginals, hence there is a hope for some 
	sort of the Central Limit Theorem, yielding that $\mathbf{V}$ would be 
	Gaussian or almost Gaussian distributed. Furthermore there is a suggestion, 
	based on experiments, that $\mathbf{V} \in \R^{n^{d-1}\times n}$ is much 
	nearer to 
	the	multivariate Gaussian distribution than the original data $\mathbf{X}$, 
	see 
	Example~\ref{ex::vgauss}.  
\end{remark}
\begin{example}\label{ex::vgauss}
	Let us take the artificial data $\mathbf{X} \in \R^{t \times n}$ sampled 
	for the $t$-Student copula parametrised by the Topelitz $\mathbf{R}$ matrix 
	with $\rho = 0.95$, see Figure~\ref{fig::top}, and $\nu = 5$. As priorly we 
	assume standard normal marginals. Next we compute $\CC_4(\X)$ and its 
	unfold $\mathbf{V} \in \R^{n^{3}\times n}$ as in Eq.~\eqref{eq::unf}.
	We use $n = 50$ and $t = 50^3$ to compare data of the same size.	
	\begin{figure}
		\subfigure[Comparison of $h_{\text{norm}, 4}(\X)$ and
		$h_{\text{norm}, 
		4}(\mathbf{V})$\label{fig::ncv}]{\includegraphics{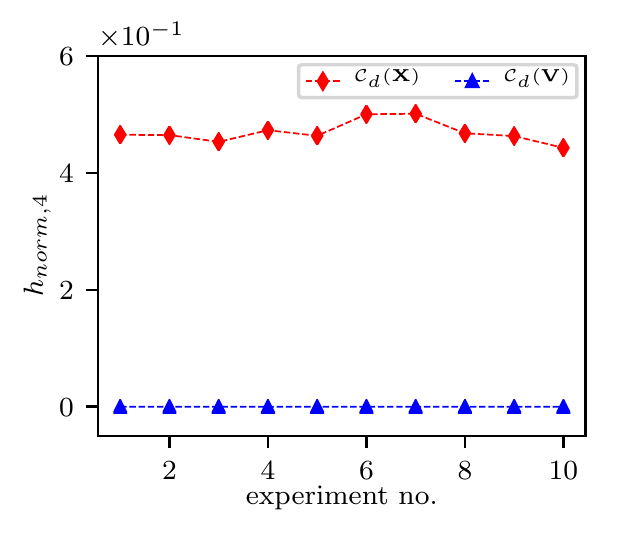}}
		\subfigure[Norms $\| \CC_d( \mathbf{V}) 
		\|$ for experiments 
		from~\ref{fig::ncv}\label{fig::nv}]{\includegraphics{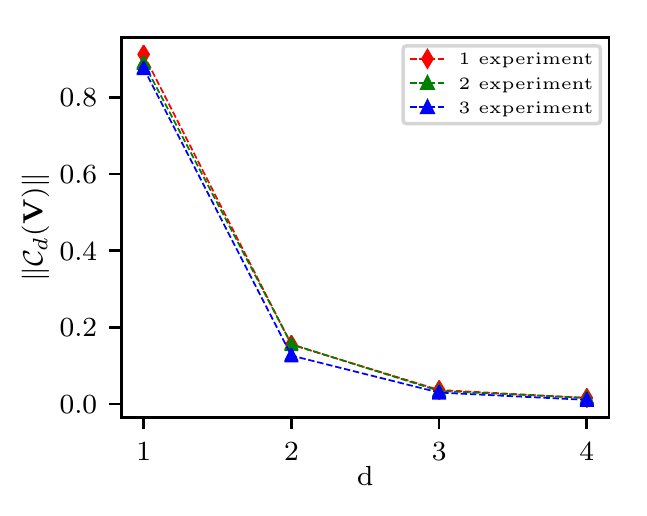}}
		\caption{Cumulants based measures the distance form the multivariate 
		Gaussian distribution for $\X$ and $\mathbf{V}$. Original data $\X$ are 
		generated by means of the $t$-Student copula parametrised by the 
		Topelitz matrix $\mathbf{R}$ and the scalar parameter $\nu = 5$, with 
		standard Gaussian marginals. The $\mathbf{V}$ is an unfold of the 
		$\CC_4(\X)$. Using the $h_{\text{norm}, 4}$ measure, one can conclude 
		that 
		$\mathbf{V}$ is nearer to the multivariate Gaussian distribution than 
		$\X$.}\label{fig::normshosvd}
	\end{figure}
In Figure~\ref{fig::normshosvd} we present measures based on the Frobenius 
norms of cumulants of $\mathbf{V}$ and $\mathbf{X}$. It results, that 
$\mathbf{V}$ is much nearer to the Gaussian 
distribution than $\X$, see $h_{\text{norm}, 4}$ in Figure~\ref{fig::ncv}. 
Further most of information of $\mathbf{V}$ is stored in the mean vector, see 
Figure~\ref{fig::nv}, hence we can not use zero mean $\mathbf{V}$ for analysis.
\end{example}
Based on the discussion above we use $\mathbf{B}$, as in 
Eq.~\eqref{eq::hsvdunfold} to measure higher order cross-correlation in 
$\X$. In 
fact $\mathbf{B}$ is a sum of the covariance matrix of $\mathbf{V}$ and an 
outer produce of the mean vector and mean vector transposition of $\mathbf{V}$. 
Further $\mathbf{V}$ is assumed to be Gaussian or almost Gaussian assumed. 
Henceforth the transformation $\CC \rightarrow 
\CC'$, see Eq.~\eqref{eq::hsvdt} is consistent with the transformation of 
the $\mathbf{B}$ matrix. Although, since 
$\mathbf{A}$ is unitary we have the following equality in the sense of the 
Frobenius norm
 \begin{equation}\label{eq::cnorms}
\| \mathbf{B} \| = \| \mathbf{V}^{\intercal} \mathbf{V} \| = \Big{\|} 
\sum_{\mathbf{i} \setminus i_1} \CC \CC 
\Big{\|} = 
 \Big{\|} 
 \sum_{\mathbf{i} \setminus i_1} \CC' \CC'\Big{\|} \text{ and }
\| \CC \| =  \| \CC' \|\end{equation}
As $\mathbf{V}$ is expected to be Gaussian distributed or at least near the 
Gaussian distribution, we can measure the information while using the 
Frobenius norm. This will be further related to the mean squared 
approximation of Gaussian distributed data.
We can assume that by the transformation $\CC \rightarrow 
\CC'$ we do not lose significantly
information~\cite{de2000multilinear} stored in $\CC$. 

Referring to Eq.~\eqref{eq::cnorms}, in general $\| \CC \| \neq \| 
\sum_{\mathbf{i} \setminus i_1} \CC \CC \|$. Nevertheless we 
expect some monotone relation between $\| \CC \|$ and $\| 
\sum_{\mathbf{i} \setminus i_1} \CC \CC \|$, to tie information carried by 
$\CC$ with information carried by $\sum_{\mathbf{i} \setminus i_1} 
\CC' \CC'$. See \cite{de2000multilinear} and discussion in 
Remark~\ref{rem::cinf}. Further, since 
$\mathbf{A}$ is unitary we have:
 \begin{equation}
\det \left(\sum_{\mathbf{i} \setminus i_1} \CC \CC \right) = 
\det \left(
\sum_{\mathbf{i} \setminus i_1} \CC' \CC' \right) = \prod 
\left( \text{diag} \left(
\sum_{\mathbf{i} \setminus i_1} \CC' \CC' \right) \right).
\end{equation}
The last term is the product of diagonal elements of $\sum_{\mathbf{i} 
\setminus i_1} \CC' \CC'$, that is a diagonal matrix.

Based on proposed relation between $\CC$ and the product of diagonal 
elements of $\sum_{\mathbf{i} \setminus i_1} \CC' \CC'$, in the 
informative sense, we will use $\det\left( 
\sum_{\mathbf{i}\setminus{i_1}} \CC_d \  \CC_d \right)$ to measure information 
tied to $d$\textsuperscript{th} order cross-correlations between marginals of 
data for which $\CC_d$ is computed.
 Hence we 
introduce following general target function for non-Gaussian features selection.
\begin{equation}\label{eq::h}
h_{\det,d}(\CC_2, \CC_d) = \frac{\det\left( \sum_{\mathbf{i}\setminus{i_1}} 
\CC_d \  \CC_d \right)}{\det (\CC_2)^d}
\end{equation}
The normalisation (denominator) assures that in univariate case we 
reproduce the module 
of asymmetry if $d = 3$ and the module of kurtosis if $d = 4$. 
See~\cite{glombband2018} for further discussion on the normalisation and 
the application of Eq.~\eqref{eq::h} to analyse real life hyper-spectral 
data. 

Following the argumentation from subsection~\ref{sec::mev} the determinant in 
the numerator would be near zero in two cases.
The first case concerns two similar columns of the matrix $\mathbf{B} = 
\mathbf{V}^{\intercal} 
\mathbf{V} = \det \left( 
\sum_{\mathbf{i}\setminus{i_1}} 
\CC_d \  \CC_d \right)$. This would correspond to
two features with similar $d$\textsuperscript{th} order cross-correlations with 
all other features. Hence, from the informative point of view, one of these 
features would be 
redundant. The 
second case concerns zero all $d$\textsuperscript{th} order cross-correlations 
of the $i$\textsuperscript{th} feature. Such feature may be suspected to be  
Gaussian distributed. Hence, we expect to remove such Gaussian features by
Algorithm~\ref{alg::itterremove}.

Experiments concerning features selection based on $h_{\text{det}, d}$ are 
widely 
discussed in~\cite{domino2018hiding, domino2018use}. For some experiments on 
artificial data, see also next subsection.

\subsection{Experiments}\label{sec::experiments}

In this subsection we test algorithms introduced in this chapter on artificial 
data generated, see 
Section~\ref{sec::datagen}. For each experiment's 
realisation we start with data $\X \in \R^{t 
	\times n}$ randomly sampled from multivariate Gaussian distribution 
	$\mathcal{N}(0, \mathbf{R})$. The $\mathbf{R}$ is the positive definite 
	symmetric correlation 
matrix with ones on a diagonal, generated at random for each experiment, see
random method in Subsection~\ref{sec::experimentmats} and 
Figure~\ref{fig::rand_cor}. Next we change chosen at random subset of marginals 
$|\mathbf{r}| = k$ by means of $t$-Student, Fr\'echet or Archimedean nested 
copulas. The transformation is performed 
in such a way that the correlation matrix of changed data is similar 
to  the covariance matrix of an original data.  
Next, we perform the features selection procedure. We remove iterative 
low-informative features using one of following target functions: 
$h_{\text{det}, d}$, $h_{\text{norm}, d}$ or the MEV. See 
also~\cite{glombband2018} for the application of the $h_{\text{det}, d}$ on 
real life hyper-spectral data and the iteration algorithm. We perform the 
iterative elimination of 
features as long as there reminds $k$ features left, that are most informative 
according to the target function. Next we check how many of these 
$k$ features are correctly selected as features from the non-Gaussian copula 
subset. We denote number of correctly selected features by
$|\mathbf{r}_{\text{detected}}|$, obviously in our setting the perfect outcome 
would give $|\mathbf{r}_{detected}| = k$. In Figure~\ref{fig::fdetect} we 
present the empirical probability of the number of detections, given 
different target functions in Figure~\ref{fig::detmethods} or different 
non-Gaussian copulas of changed subset of features in 
Figure~\ref{fig::detcops}. For 
further experiments see \cite{domino2018use} in the case ot the $t$-Student 
copula and \cite{domino2018hiding} in the case of Archimedean Copulas.

\begin{figure}
    \begin{center}
	\subfigure[Clayton nested 
	copula\label{fig::detmethods}]{\includegraphics{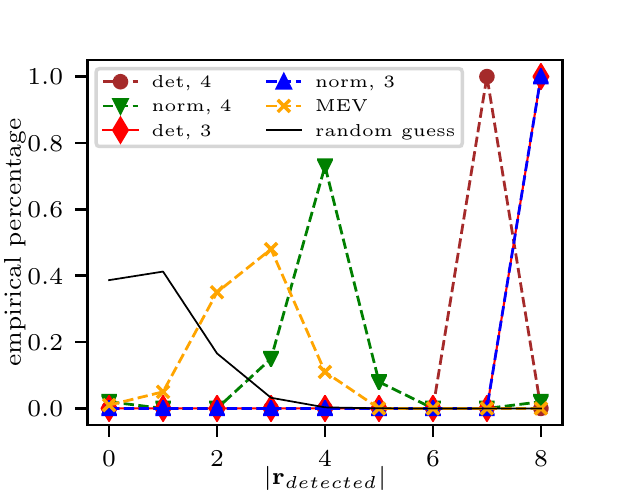}}
	\subfigure[$h_{\text{det}, 4}$ target 
	function\label{fig::detcops}]{\includegraphics{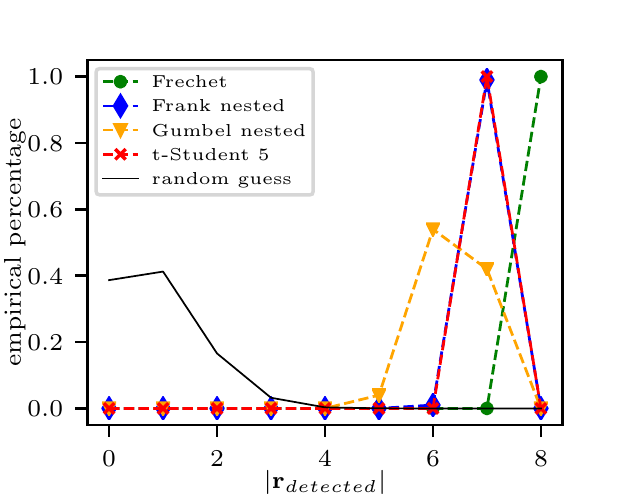}}
	\caption{Features selection experiments $t = 10^5, n = 75, k = 
	8$, and $100$ experiments realisations. Cumulant's based methods are 
	compared with the MEV and the random 
	guess. On the horizontal axis there are numbers of correctly detected 
	features, 
	the more to the right the better the results. Results of the MEV are poor 
	as it 
	was 
	expected. Results of the
	$h_{\text{det}, 4}$ for Fr\'echet, $t$-Student and Frank copula are good. 
	For Clayton copula methods using the 
	$3$\textsuperscript{rd} cumulant seams to work better than these using the 
	$4$\textsuperscript{th} cumulant. }\label{fig::fdetect}
\end{center}
\end{figure}

Referring to Figure~\ref{fig::detmethods} we have a poor performance of the MEV 
similar to the random guess. This is due to the fact that the transformation 
of marginal's subset does not affect significantly the 
covariance matrix, that is used by the MEV method to extract information, see 
Subsection~\ref{sec::experimentmats}.
On the other hand, we have a good performance of the $3$\textsuperscript{rd} 
cumulant 
based method, as absolute values of $3$\textsuperscript{rd} cumulant's tensor's
elements for the Clayton copula (with Gaussian marginals) are relatively large, 
see Figure~\ref{fig::cl_c3} and compare it with Figure~\ref{fig::cl_c4} for 
$4$\textsuperscript{th} cumulant's elements for this copula. In 
these experiments we use cross-correlation between marginals, as in 
Figure~\ref{fig::rand_cor}, due the random $\mathbf{R}$ matrix generation 
algorithm. For the 
$4$\textsuperscript{th} cumulant case the $h_{\text{det}, 4}$ function is much 
better than $h_{\text{norm}, 4}$ and almost as good as the 
$3$\textsuperscript{rd} cumulant's based methods. Hence sophisticated 
HOSVD based method discussed in Subsection~\ref{sec::cumsfsel} have 
an advantage over the simple norm of the cumulant's tensor.

Nevertheless the $4$\textsuperscript{th} cumulant's based methods are more 
general, since there are copulas (such as a $t$-Student one, widely used in the 
financial data analysis) for which elements of the $3$\textsuperscript{rd} 
cumulant's tensor are zero, given symmetric marginal distributions. Hence, 
 we will use the target function $h_{\text{det}, 4}$ to 
analyse more general case for various copulas. Referring to 
Figure~\ref{fig::detcops} we have good detection for all copulas, far from the 
random guess case. Best results gives the Fr\'echet copula, next we have 
$t$-Student and Frank copulas. We observe, the worse 
detectability of features modelled by the Gumbel copula if using the 
$4$\textsuperscript{th} cumulant's based method. This comes form the fact that 
for 
the Gumbel copula with high cross-correlations, $4$\textsuperscript{th} 
cumulant's elements are small, see Figure~\ref{fig::gu_c4}. For the 
corss-correlation level between marginals during experiments see 
Figure~\ref{fig::rand_cor}.

\section{Features extraction}\label{sec::fextract}

Having discussed the cumulant's based features selection, we can move to 
another dimensionality reduction scenario, the features extraction. 
Let us start as in the previous section with $n$-variate random vector 
$\XXX^{(n)}$. Its $t$ realisations are given in a form of 
a matrix $\mathbf{X} \in \R^{t \times n}$, where $n$ marginals are features.
Features extraction is the projection of features onto directions that carry 
most meaningful information~\cite{duda2012pattern}. In other words we want to 
find such linear combinations of marginals that carry most meaningful 
information. Suppose we have the 
following factor matrix $\mathbf{A}' \in \R^{n \times n'}$, where $n' < n$. In 
our case we would use the orthonormal columns corresponding to $n'$ 
'independent' features, hence
$(\mathbf{A'})^{\intercal} \mathbf{A'} = \1$, but $\mathbf{A'} 
(\mathbf{A'})^{\intercal} \neq \1$. We can reduce dimensionality of data 
by performing the transformation
\begin{equation}
\R^{t \times n'} \ni \Y' = \X \mathbf{A'}
\end{equation}
Columns of $\Y'$ represent linear combinations of original features. If 
matrix $\mathbf{A}$ is properly determined most of information from 
$\X$ will remain in $\Y'$. 
 
	Suppose that $n$ is large, it would be computationally complex to 
	use 
	whole $\X \in \R^{t \times n}$ in some sophisticated machine learning 
	algorithm~\eg~classifying realisations $\x_j \in \R^n$ to given classes, 
	detecting outlier realisations or fitting some model to data. Further it 
	may require large computer memory to store or process a 
	whole $\X$ especially if $t$ is large. But, if we chose such 
	$\mathbf{A'}$ that most of meaningful information from $\X$ is still 
	present in $\Y' 
	\in \R^{t \times n'}$ and $n' \ll n$, we would have a significant gain on 
	computational complexity or computational memory, 
	see~\cite{duda2012pattern}.

\begin{example}
In the case of financial data analysis, consider investment portfolio 
management, where we are basically interested in finding such portfolios of 
assets that are characterised by low risk. It can be done, in a statistical 
approach, by the analysis of
statistics of historical assets values. Obviously, we assume that the 
probabilistic model of future data will be similar to the probabilistic model 
of past data. 
Suppose 
$\X \in \R^{t \times n}$ represents data of $t$ past records of $n$ assets 
returns. To minimise risk, we need to find such linear 
combination of assets that has low variability, what is an opposite approach to 
the 
information maximisation. We have
\begin{equation}
	\R^{t \times n''} \ni \textbf{Y}'' = \X \mathbf{A''},
\end{equation}
suppose another factor matrix $\mathbf{A}'$ is determined in such a way, that 
it projects data on directions caring most information about the variability of 
data. Hence we can search for directions, orthogonal to those in $\mathbf{A}'$ 
projecting data onto directions with low variability. Assuming such directions 
are orthonormal to each other, they can be stored in the orthonormal factor 
matrix 
$\mathbf{A}''$. One of the methods of obtaining $\mathbf{A}''$ will be 
searching for columns orthogonal to columns from $\mathbf{A}'$. 
\end{example}

The most popular method of determining the factor matrix $\mathbf{A'}$ (and 
$\mathbf{A''}$) is the 
Singular Value Decompositions (SVD) of a covariance matrix. Let us concentrate 
on $\mathbf{A'}$, following
Eq.~\eqref{eq::SVDagain} we have
\begin{equation}
	\CC_2 = \mathbf{A} \SSd \mathbf{A}^{\intercal}.
\end{equation}
If we shuffle columns of $\mathbf{A}$ in such a way 
that eigenvalues are in decreasing order (these eigenvalues are real 
and non-negative) first columns of matrix $\mathbf{A}$ would correspond to 
largest information or highest variability of linear combination of features in 
reference to the covariance matrix.
Here by taking first $n'$ columns of $\mathbf{A}$ we can determine the factor 
matrix that preserves most of information, see~\cite{krzanowski2000principles}. 
The 
covariance matrix $\SSd' = \mathbf{A'}^{\intercal} 
\CC_2 \mathbf{A'}$ of new data $\mathbf{Y} = \X \mathbf{A'}$ would be a 
diagonal one composed of highest $n'$ eigenvectors of $\CC_2$.

For a graphic representation clarity, we assume centred $\tilde{\X}$, such that
$\CC_1({\tilde{\X}}) = 0$, hence  $\CC_1({\tilde{\Y}}) = 0$. We can discuss 
$\tilde{\Y}$ using a graphic notation. Starting with:
	\begin{center}
		\begin{tikzpicture}
		\node[shape=rectangle,draw=black] (C) at (-3,0) {$\tilde{\Y}$};
		\path [-] (C) edge node[left] {} (-4,0);
		\path [-] (C) edge node[right] {} (-2,0);
		\node[shape=rectangle,draw=black] (AA) at (0,0) 
		{$\tilde{\X}$};
		\node[shape=rectangle,draw=black] (BA) at (1.8,0) 
		{$\mathbf{A}'$};
		\node at (-0.5,0.2) {\scriptsize{$(1)$}};
		\node at (.55,0.2) {\scriptsize{$(2)$}};
		\node at (1.25,0.2) {\scriptsize{$(1)$}};
		\node at (2.38,0.2) {\scriptsize{$(2)$}};
				\node at (-3.52,0.2) {\scriptsize{$(1)$}};
				\node at (-2.45,0.2) {\scriptsize{$(2)$}};
		\node at (-1.5,0){$=$};
		\path [-] (AA) edge node[left] {} (BA);
		\path [-] (AA) edge node[left] {} (-1.0,0);
		\path [-] (BA) edge node[left] {} (2.8,0);
		\end{tikzpicture}
	\end{center}
	the covariance matrix of $\tilde{\Y}$ is
	\begin{center}
		\begin{tikzpicture}
		\node[shape=rectangle,draw=black] (C) at (-2.7,0) {$\CC_2'$};
		\node at (9.25,0.18) {\scriptsize{$(1)$}};
		\node at (10.35,0.18) {\scriptsize{$(2)$}};
		\node at (7.6,0.18) {\scriptsize{$(1)$}};
		\node at (8.62,0.18) {\scriptsize{$(2)$}};
				\node at (6.0,0.18) {\scriptsize{$(2)$}};
				\node at (7.0,0.18) {\scriptsize{$(1)$}};
		\node at (5.15,0.18) {\scriptsize{$(1)$}};
		\node at (4.05,0.18) {\scriptsize{$(2)$}};
		\node[shape=rectangle,draw=black] (A) at (4.6,0) {$\mathbf{A'}$};
		\node[shape=rectangle,draw=black] (B) at (6.5,0) {$\tilde{\X}$};
		\node[shape=rectangle,draw=black] (E) at (8.1,0) {$\tilde{\X}$};
		\node[shape=rectangle,draw=black] (D) at (9.8,0) 
		{$\tilde{\mathbf{A'}}$};
		\path [-] (C) edge node[left] {} (-3.4,0);
		\path [-] (C) edge node[right] {} (-2,0);
		\path [-] (E) edge node[left] {} (D);
		\path [-] (B) edge node[left] {} (E);
		\node at (-1.5,0){$=$};
		\path [-] (A) edge node[left] {} (B);
		\path [-] (A) edge node[left] {} (3.5,0);
		\path [-] (D) edge node[left] {} (10.7,0);
		\node[shape=rectangle,draw=black] (AA) at (0,0) 
		{$\tilde{\mathbf{Y}}$};
		\node[shape=rectangle,draw=black] (BA) at (1.8,0) 
		{$\tilde{\mathbf{Y}}$};
		\node at (-0.5,0.2) {\scriptsize{$(2)$}};
		\node at (.55,0.2) {\scriptsize{$(1)$}};
		\node at (1.3,0.2) {\scriptsize{$(1)$}};
		\node at (2.32,0.2) {\scriptsize{$(2)$}};
		\node at (3.2,0){$=$};
		\path [-] (AA) edge node[left] {} (BA);
		\path [-] (AA) edge node[left] {} (-1.0,0);
		\path [-] (BA) edge node[left] {} (2.8,0);
		\end{tikzpicture}
	\end{center}
	but
	\begin{center}
		\begin{tikzpicture}
		\node[shape=rectangle,draw=black] (C) at (-3,0) {$\CC_2$};
		\path [-] (C) edge node[left] {} (-4,0);
		\path [-] (C) edge node[right] {} (-2,0);
		\node[shape=rectangle,draw=black] (AA) at (0,0) 
		{$\tilde{\X}$};
		\node[shape=rectangle,draw=black] (BA) at (1.8,0) 
		{$\tilde{\X}$};
		\node at (-0.5,0.2) {\scriptsize{$(2)$}};
		\node at (.55,0.2) {\scriptsize{$(1)$}};
		\node at (1.3,0.2) {\scriptsize{$(1)$}};
		\node at (2.32,0.2) {\scriptsize{$(2)$}};
		\node at (-1.5,0){$=$};
		\path [-] (AA) edge node[left] {} (BA);
		\path [-] (AA) edge node[left] {} (-1.0,0);
		\path [-] (BA) edge node[left] {} (2.8,0);
		\end{tikzpicture}
	\end{center}
The covariance matrix of $\tilde{\Y}$ will be:
\begin{center}
	\begin{tikzpicture}
	\node[shape=rectangle,draw=black] (A) at (4,0) 
	{$\mathbf{A'}$};
	\node[shape=rectangle,draw=black] (B) at (5.5,0) {$\CC_2$};
	\node[shape=rectangle,draw=black] (D) at (7.,0) {$\mathbf{A'}$};
	\node[shape=rectangle,draw=black] (C) at (1,0) {$\CC_2'$};
	\node at (6.45,0.18) {\scriptsize{$(1)$}};
	\node at (7.55,0.18) {\scriptsize{$(2)$}};
	\node at (4.55,0.18) {\scriptsize{$(1)$}};
	\node at (3.45,0.18) {\scriptsize{$(2)$}};
	\path [-] (C) edge node[left] {} (0,0);
	\path [-] (C) edge node[right] {} (2,0);
	\path [-] (B) edge node[left] {} (D);
	\node at (2.5,0){$=$};
	\path [-] (A) edge node[left] {} (B);
	\path [-] (A) edge node[left] {} (3.0,0);
	\path [-] (D) edge node[left] {} (8.,0);
	\end{tikzpicture}
\end{center}
In the case of non-centred data $\X$ we need to observe that if $\Y' = \X 
\mathbf{A}'$ than $\tilde{\Y'} = \tilde{\X} \mathbf{A'}$ and use 
$\CC_2(\tilde{\Y}) = \CC_2(\Y)$.

If we use  $\mathbf{A'} \in 
\R^{n \times n'}$ where $n > n'$ and $\mathbf{A'}^{\intercal} \mathbf{A'} = 
\1$, the transformation $\Y' = \X 
\mathbf{A}'$ will loose some information in the sense of the Frobenius norm
	\begin{equation}
	\|\CC_2 \| < \|\CC_2' \| \text{ or } \| \Y'\| < \| \X \|.
	\end{equation}
hence we have to chose a proper value of $n'$ to make an information lose 
acceptable. In 
these approximation the lose of information is minimised by the minimisation 
of 
the mean square error measured by the Frobenius norm. Such approach is 
optimal for Gaussian distributed multivariate data. 
 Unfortunately most of real life data, such as financial one are usually 
non-Gaussian distributed. Hence, these simple approach is not optimal. Higher 
order statistics may play a vital role in more adequate approach for these data.

\subsection{High Order Singular Value Decomposition}

Suppose now we have a following linear data transformation $\Y' = \X 
\mathbf{A'}$. We can take zero mean example $\tilde{\Y'} = \tilde{\X} 
\mathbf{A'}$ for simplicity. 
One can show, that the $d$\textsuperscript{th} order cumulant for $d \geq 2$ 
(for $d=1$ we do not have equity between the cumulant of centred and 
non-centred 
data) of 
$\Y$ or $\tilde{\Y}$ would be
\begin{equation}\label{eq::cumtransf}
\R^{[n',d]} \ni \CC_d' = 
 \left(\mathbf{A'}\right)^{\intercal} \times_{1,\ldots, d} \CC_d,
\end{equation}
where $\CC_d \in \R^{[n, d]}$ is the $d$\textsuperscript{th} order cumulant of 
$\X$ or $\tilde{\X}$, and $\mathbf{A} \in \R^{n \times n'}$ is a factor matrix. 
Let us refer to a
notation for the $d$\textsuperscript{th} order central moment, see 
Definition~\ref{def::graphicmoms}, of $\tilde{\Y'} = \tilde{\X} \mathbf{A'}$. 
For the proof of Eq.~\eqref{eq::cumtransf} refer to graphical notation
\begin{center}
	\begin{tikzpicture}
	\node[shape=rectangle,draw=black] (C) at (-1,0) {$\MM_d'$};
	\node at (1.2,0){$=$};
	\path [-] (C) edge node[left] {} (-2.1,1);
	\path [-] (C) edge node[left] {} (.1,1);
	\path [-] (C) edge node[left] {} (-2.1,-1);
	\path [-] (C) edge node[left] {} (.1,-1);
	\node at (-1,1){$\cdots$};
	\node at (-1,-1){$\cdots$};
	\node at (0.2,0.1) {$\vdots$};
	\node[shape=rectangle,draw=black] (D) at (4,0) {$\1_{(d)}$};
	\node at (4,1.5) {$\cdots$};
	\node[shape=rectangle,draw=black] (C) at (2.8,1) {$\tilde{\X}$};
	\node[shape=rectangle,draw=black] (CC) at (1.8,2.) {$\mathbf{A'}$};
	\node at (3.4, .76) {\scriptsize{$(1)$}};
	\node at (2.63, 1.45) {\scriptsize{$(2)$}};
	\node at (3.0, -.55) {\scriptsize{$(1)$}};
	\node at (2.23, -1.25) {\scriptsize{$(2)$}};
	\node at (4.6, .78) {\scriptsize{$(1)$}};
	\node at (5.4, 1.5) {\scriptsize{$(2)$}};
	\node at (5.05, -.55) {\scriptsize{$(1)$}};
	\node at (5.8, -1.25) {\scriptsize{$(2)$}};
	\node at (1.6, 2.5) {\scriptsize{$(2)$}};
	\node at (2.38, 1.75) {\scriptsize{$(1)$}};
	\node at (1.25, -2.25) {\scriptsize{$(2)$}};
	\node at (2.0, -1.55) {\scriptsize{$(1)$}};
	\node at (6.4, 2.5) {\scriptsize{$(2)$}};
	\node at (5.65, 1.75) {\scriptsize{$(1)$}};
	\node at (6.75, -2.25) {\scriptsize{$(2)$}};
	\node at (6.05, -1.55) {\scriptsize{$(1)$}};
	\node at (5.6,0.1) {$\vdots$};
	\node[shape=rectangle,draw=black] (E) at (5.2,1) {$\tilde{\X}$};
	\node[shape=rectangle,draw=black] (EC) at (6.2,2.) {$\mathbf{A'}$};
	\node at (4,-1.5) {$\cdots$};
	\node[shape=rectangle,draw=black] (G) at (2.8,-1) {$\tilde{\X}$};
	\node[shape=rectangle,draw=black] (GC) at (1.8,-2.) {$\mathbf{A'}$};
	\node[shape=rectangle,draw=black] (I) at (5.2,-1) {$\tilde{\X}$};
	\node[shape=rectangle,draw=black] (IC) at (6.2,-2.) {$\mathbf{A'}$};
	\path [-] (D) edge node[left] {} (E);
	\path [-] (D) edge node[left] {} (G);
	\path [-] (D) edge node[left] {} (C);
	\path [-] (D) edge node[left] {} (I);
	\path [-] (EC) edge node[left] {} (6.9,2.7);
	\path [-] (CC) edge node[left] {} (1.1,2.7);
	\path [-] (IC) edge node[left] {} (6.9,-2.7);
	\path [-] (GC) edge node[left] {} (1.1,-2.7);
	\path [-] (E) edge node[left] {} (EC);
	\path [-] (G) edge node[left] {} (GC);
	\path [-] (I) edge node[left] {} (IC);
	\path [-] (C) edge node[left] {} (CC);
	\node at (7.5,0){$=$};
	\node[shape=rectangle,draw=black] (DA) at (10,0) {$\MM_d$};
	\node at (10,1) {$\cdots$};
	\node[shape=rectangle,draw=black] (CA) at (8.8,1) {$\mathbf{A'}$};
	\node at (9.4, .76) {\scriptsize{$(1)$}};
	\node at (8.63, 1.45) {\scriptsize{$(2)$}};
	\node at (9.0, -.55) {\scriptsize{$(1)$}};
	\node at (8.23, -1.25) {\scriptsize{$(2)$}};
	\node at (11.2,0.1) {$\vdots$};
	\node at (10.6, .78) {\scriptsize{$(1)$}};
	\node at (11.4, 1.5) {\scriptsize{$(2)$}};
	\node at (11.05, -.55) {\scriptsize{$(1)$}};
	\node at (11.8, -1.25) {\scriptsize{$(2)$}};
	\node at (11.2,0.1) {$\vdots$};
	\node[shape=rectangle,draw=black] (EA) at (11.2,1) {$\mathbf{A'}$};
	\node at (10,-1) {$\cdots$};
	\node[shape=rectangle,draw=black] (GA) at (8.8,-1) {$\mathbf{A'}$};
	\node[shape=rectangle,draw=black] (IA) at (11.2,-1) {$\mathbf{A'}$};		
	\path [-] (DA) edge node[left] {} (EA);
	\path [-] (DA) edge node[left] {} (GA);
	\path [-] (DA) edge node[left] {} (CA);
	\path [-] (DA) edge node[left] {} (IA);
	\path [-] (EA) edge node[left] {} (11.9,1.7);
	\path [-] (CA) edge node[left] {} (8.1,1.7);
	\path [-] (IA) edge node[left] {} (11.9,-1.7);
	\path [-] (GA) edge node[left] {} (8.1,-1.7);
	\end{tikzpicture}
\end{center}
shows the relation between central moments of $\Y'$ and $\X$, namely
\begin{equation}\label{eq::momtransf}
\R^{[n',d]} \ni \MM_d' = 
\left(\mathbf{A'}\right)^{\intercal} \times_{1,\ldots, d} \MM_d.
\end{equation}
Hence, Eq.~\eqref{eq::cumtransf} holds for cumulant of order $d=2$ and 
$d=3$ as these are central moments.  
For $d > 3$, one need to refer to the graphic representation of corresponding 
cumulant, see Remarks~\ref{rem::c4graph}~\ref{rem::c5graph}~\ref{rem::c6graph}, 
use Eq.~\eqref{eq::momtransf} and the fact that 
Eq.~\eqref{eq::cumtransf} holds for $2, \ldots, d-2$.

Now our task is to find such factor matrix that gives highly informative 
features with respect to the $d$\textsuperscript{th} order cumulant. This can 
by performed by means of the HOSVD procedure of the corresponding cumulant's 
tensor. Suppose 
$\mathbf{A} \in \R^{n \times n}$ is an unitary column matrix of eigenvectors of 
$\sum_{\mathbf{i}\setminus{i_1}} \CC_d \CC_d$, see 
Eq.~\eqref{eq::hsvdt}. The core-tensor fulfils $\CC_d' = 
\left(\mathbf{A}\right)^{\intercal} \times_{1,\ldots, d} \CC_d$, and referring 
to Eq.~\eqref{eq::hsvdt} we get
\begin{equation}\label{eq::hosvdc}
	\sum_{\mathbf{i}\setminus{i_1}} \CC_d \CC_d =\mathbf{A} \left( 
	\sum_{\mathbf{i}\setminus{i_1}} \CC_d' \CC_d' \right) 
	\mathbf{A}^{\intercal}.
\end{equation}
where $\sum_{\mathbf{i}\setminus{i_1}} \CC_d' \CC_d'$ is diagonal eigenvalue 
matrix. 
Hence, if 
we shuffle columns of 
$\mathbf{A}$ is such a way that its eigenvalues are sorted in 
decreasing order (given they are real and non-negative) 
first columns of $\mathbf{A}$ should give a linear combination of data with 
high information measured by the $d$\textsuperscript{th} cumulant. To discuss 
this in more details consider $\mathbf{V_y} \in \R^{n^{d-1} \times n}$, that is 
the matrix representation of $\CC'$ with elements given by Eq.~\eqref{eq::unf}. 
We 
have
\begin{equation}
\left(\mathbf{V_y}\right)^{\intercal} \mathbf{V_y} = 
\sum_{\mathbf{i}\setminus{i_1}} \CC_d' 
\CC_d'.
\end{equation}
Following Remark~\ref{rem::cinf} the $i$\textsuperscript{th} column of 
$\mathbf{V_y}$ contains $d$\textsuperscript{th} order cross-correlations of the 
$i$\textsuperscript{th} marginal of $\mathbf{Y} = \mathbf{X} \mathbf{A}$ with 
all other marginals of $\mathbf{Y}$.
Further as discussed in Remark~\ref{rem::cinf} data in $\mathbf{V_y}$ should be 
$n$-variate Gaussian distributed or have a probabilistic model near to the 
Gaussian one. Due to the Singular Value 
Decomposition matrix $\left(\mathbf{V_y}\right)^{\intercal} \mathbf{V_y}$ is 
diagonal. Concluding, columns of $\mathbf{V_y}$ can be ordered 
with respect to the information they carry about the $d$\textsuperscript{th} 
order cross-correlation of data. Assuming the Gaussian model of 
$\mathbf{V_y}$ we can use the mean square error approximation in the analysis 
of the decomposition.

Analogically to the Singular Value Decomposition of the covariance matrix case, 
if we use unitary $\mathbf{A} \in \R^{n \times n}$ we have the following 
information conservation in the sense of the Frobenius norm,
\begin{equation}
\| \left(\mathbf{A}\right)^{\intercal} \times_{1,\ldots, d} \CC_d  \| = \|\CC_d 
\|  \text{ and } \| \X 
\mathbf{A} \| = \| \X \|.
\end{equation}	
Given the eigenvalues ordering, one can conclude, that elements of the 
following core tensor $\CC_d' = 
\left(\mathbf{A}\right)^{\intercal} 
\times_{1,\ldots, d} 
\CC_d $ indexed by low multi-index (left right corner in the $d = 3$ case), 
should have 
relatively large
impact to the Forbenious norm of this core tensor $\| 
\CC_d'\|$ \cite{de2000multilinear}. Such ordering leads to $\mathbf{A}' \in 
\R^{n 
\times n'}$ composed of first $n'$ columns of $\mathbf{A}$. Such $\mathbf{A}'$ 
can be used to achieve a linear 
combination of data carrying most of information with respect to 
$d$\textsuperscript{th} order cumulant. We have here
\begin{equation}
\| \left(\mathbf{A'}\right)^{\intercal} \times_{1,\ldots, d} \CC_d  \| < 
\|\CC_d \|.
\end{equation}
Hence, we have to chose such $n'$ to make an information lose 
acceptable.	In the
next subsection we will discuss a real life data application of the procedure.

\subsection{Multi-cumulant higher order singular value 
decomposition}\label{sec::als}

In this subsection, following~\cite{domino2016use}, we discuss the procedure of 
the real life financial data analysis where we want to find such factor matrix 
that gives low risk investment portfolios. Investment portfolio is composed of 
many assets, negative quantity of assets is possible and correspond to the 
`short sale'. The value of the portfolio is the linear combination of values of 
these assets (often being shares prices). To determine the low risk portfolio 
we want to reduce its variability measured by cumulants of order 
$2,3,\ldots, d$.
In~\cite{domino2016use} the author has shown, that such approach, given $d = 
6$, gives relatively save portfolios of shares traded on the Warsaw Stock 
Exchange, during the crisis of year $2015$. On average, the losses on such 
portfolios were smaller than the loss of the benchmark being the stock exchange 
index (WIG20).

Do determine safe portfolios we use the 
Alternating Least Squares (ALS) algorithm, which is some composition of the 
HOSVD of cumulants tensors and the Singular Value Decomposition of the 
covariance matrix
\cite{morton2009algebraic, 
	de2004dimensionality, savas2010quasi}. 
The ALS searches for the local maximum of the function
\begin{equation}\label{eq::als1step}
\xi_{d}(\mathbf{A'}) = \frac{1}{2!} \|\mathbf{A'}^{\intercal} \CC_2 
\mathbf{A'}\| + \sum_{k=3}^{d} \frac{1}{k!} 
\|\left(\mathbf{A'}\right)^{\intercal} \times_{1,\ldots, k} \CC_k\|,
\end{equation}
given $\mathbf{A'}^{\intercal} \mathbf{A'} = 
\1$ where $\mathbf{A'} \in \R^{n \times n'}$ and $n' < n$, where $n'$ is a 
parameter.
Following \cite{domino2016use}, we discuss here the following iterative 
approach.
\begin{enumerate}
	\item In the first step we construct the matrix
	\begin{equation} 
	\R^{[n, 2]} \ni \mathbf{T}_{\xi} = \frac{1}{2!} \CC_2 \CC_2 + \sum_{k=3}^d 
	\frac{1}{k!} 
	\sum_{\mathbf{i}\setminus(i_1)} \CC_k \CC_k.
	\end{equation}
	\item We perform the eigenvalue decomposition of $\textbf{T}_{\xi}$,
	\begin{equation}\label{eq::alsA}
	\textbf{T}_{\xi} = \mathbf{A} \textbf{T}_{(d)} \mathbf{A}^{\intercal}
	\end{equation}
	where $\mathbf{A}$ is unitary matrix which columns are eigenvectors of 
	$\textbf{T}_{\xi}$ and $\textbf{T}_{(d)}$ is diagonal matrix of real 
	non-negative eigenvalues sorted in decreasing order. Here $\mathbf{A}' \in 
	\R^{n \times n'}$ is supposed to carry meaningful information. Such 
	information would be tied both to its standard measure being a covariance 
	of features and not-standard measures being higher order cross-correlations 
	of order 
	($3,4,\ldots, d$).
	\item In the $p^{\mathrm{th}}$ step we take a factor matrix from the 
	$(p-1)$\textsuperscript{th} step, compute
	\begin{equation}\label{eq::alspstep}
		\mathbf{T}_{\xi, (k)} = \sum_{\mathbf{i}\setminus(i_1)} \mathcal{T}_{k} 
	\mathcal{T}_{k}
	\text{ where }  \mathcal{T}_{k} = 
	\mathbf{A'}^{\intercal}_{(p-1)}\times_{2,\ldots, k} \CC_k,
	\end{equation}
	 what can be represented graphically as
		\begin{center}
			\begin{tikzpicture}
			\node[shape=rectangle,draw=black] (C) at (0,0) {$\mathbf{T}_{\xi, 
			(k)}$};
			\path [-] (C) edge node[left] {} (-1.,0.);
			\path [-] (C) edge node[left] {} (1,0);
			\node at (1.4,0){$=$};
			\node[shape=rectangle,draw=black] (A) at (3,0) {$\CC_k$};
			\node[shape=rectangle,draw=black] (CA) at (4.1,1) {$\mathbf{A'}$};
			\node[shape=rectangle,draw=black] (BA) at (4.1,-1) {$\mathbf{A'}$};
			\node[shape=rectangle,draw=black] (AA) at (4.3,0) {$\mathbf{A'}$};
			\node[shape=rectangle,draw=black] (CB) at (6.1,1) {$\mathbf{A'}$};
			\node[shape=rectangle,draw=black] (BB) at (6.1,-1) {$\mathbf{A'}$};
			\node[shape=rectangle,draw=black] (AB) at (5.9,0) {$\mathbf{A'}$};
			\path [-] (A) edge node[left] {} (CA);
			\path [-] (A) edge node[left] {} (BA);
			\path [-] (A) edge node[left] {} (AA);
			\path [-] (A) edge node[left] {} (1.75,0.);
			\node at (5.1,0.75) {$\vdots$};
			\node at (5.1,-0.55) {$\vdots$};
			\node[shape=rectangle,draw=black] (B) at (7.2,0) {$\CC_k$};
			\path [-] (AB) edge node[right] {} (B);
			\path [-] (AA) edge node[right] {} (AB);
			\path [-] (BA) edge node[right] {} (BB);
			\path [-] (CA) edge node[right] {} (CB);
			\path [-] (B) edge node[left] {} (BB);
			\path [-] (B) edge node[left] {} (8.7,0);
			\path [-] (B) edge node[left] {} (CB);
			\node at (3.72,0.2) {\scriptsize{$(1)$}};
			\node at (4.85,0.2) {\scriptsize{$(2)$}};
			\node at (3.55,.83) {\scriptsize{$(1)$}};
			\node at (4.65,1.15) {\scriptsize{$(2)$}};
			\node at (5.55,1.15) {\scriptsize{$(2)$}};
			\node at (6.7,0.83) {\scriptsize{$(1)$}};
			\node at (3.55,-.83) {\scriptsize{$(1)$}};
			\node at (4.65,-0.83) {\scriptsize{$(2)$}};
			\node at (5.55,-0.83) {\scriptsize{$(2)$}};
			\node at (6.65,-0.83) {\scriptsize{$(1)$}};
			\node at (5.35,0.2) {\scriptsize{$(2)$}};
			\node at (6.45,0.2) {\scriptsize{$(1)$}};
			\end{tikzpicture}
		\end{center}
	compute 
	\begin{equation}\mathbf{T}_{\xi} = \sum_{k = 2}^d 
	\frac{1}{k!}\mathbf{T}_{\xi, (k)},
	\end{equation}
	and repeat point $2$.
	\item The procedure is repeated to satisfy the stop condition.
\end{enumerate}

Recall that if we take $n = n'$ the Eq.~\eqref{eq::alspstep} would give the 
Eq.~\eqref{eq::als1step}, see graphic notation for justification and the 
fact that for $n = n'$ we have $\mathbf{A} 
\mathbf{A}^{\intercal} = \1$. Nevertheless since the factor matrix in not 
unequally determined by Eq.~\eqref{eq::alspstep}, some iterations may be
still useful,~\eg~requiring $\det(\mathbf{A}) = +1$ in a stop 
condition to make an outcome deterministic, see~\cite{domino2016use}. If taking 
$n = n'$, last 
columns of $\mathbf{A}$ will give a linear combination of features with low 
variability. The variability measures both by the variance and higher order 
cumulants $3, \ldots, d$ as well. These linear combination of features with low 
variability, correspond to safe portfolios especially for a crisis. During the 
crisis 
we have simultaneous high drops of values of many assets, what is modelled by 
higher order correlations \cite{sornette2009stock, morton2009algebraic}.

Alternatively, to take a full advantage of the iterative procedure, we may use 
$\mathbf{A}' \in \R^{n \times n'}$, where $n' < n$.
The output would be a factor matrix giving portfolios with high variability. 
Next, one have to search 
for the factor matrix 
$\mathbf{A}''$ with columns orthonormal to the output matrix. Such 
$\mathbf{A}''$ is suppose to give a factor matrix reflecting linear 
combinations with variability - safe portfolios.

\subsection{Experiments}

To show that last columns of the factor matrix of the HOSVD of the 
$d$\textsuperscript{th} cumulant's tensors gives mainly a linear combination of 
such marginals that have 
little $d$ order cross-correlation, we propose a following experiment.
In the experimental set we use $t 
= 10^5$ 
and $n = 5$, and for each experimental realisation we generate randomly 
positive definite symmetric correlation matrix $\mathbf{R}$ with ones on the 
diagonal, using random method discussed in Section~\ref{sec::datamalfexp}, see 
Figure~\ref{fig::rand_cor} there. Next, we sample $\X$ from 
$\NN(0, 
\mathbf{R})$, and change first two marginals of $\X$ by means of a non-Gaussian 
copula. Next we perform the HOSVD of the 
$\CC_d$ computed for changed data, see Eq.~\eqref{eq::hosvdc}, to get a 
factor matrix $\mathbf{A} \in 
\R^{n \times n}$. Such factor matrix, is a matrix of columns that are 
eigenvectors 
of 
$\sum_{\mathbf{i}\setminus{i_1}} \CC_d \CC_d$. In our procedure, these 
eigenvectors are ordered by 
the information significance, which is measured by the $d$\textsuperscript{th} 
order cumulant of the corresponding linear combination of features. Hence if we 
take the last column of 
the factor 
matrix - $A_5 \in \R^n$, it corresponds to the lowest eigenvalue of 
$\sum_{\mathbf{i}\setminus{i_1}} \CC_d \CC_d$. Such column should generally 
correspond to the
linear combination of marginals that have small 
$d$\textsuperscript{th} 
order cross-correlation. This linear combination should have small 
$d$\textsuperscript{th} order cumulant, see~\cite{morton2009algebraic}. Given 
the last column $A_5$ and its first two elements $a_{1,5}$ and $a_{2,5}$ that 
correspond to the impact of non-Gaussian features, we use the 
norm
\begin{equation}\label{eq::w}
w = \frac{\sqrt{a_{1,5}^2+a_{2,5}^2}}{\| A_5 \|},
\end{equation}
which should be low if there is a little impact of the non-Gaussian distributed 
subset in the linear combination of features given by $A_5$. 

\begin{figure}
    \begin{center}
	\subfigure[Gumbel 
	copula\label{fig::hosvdarch}]{\includegraphics{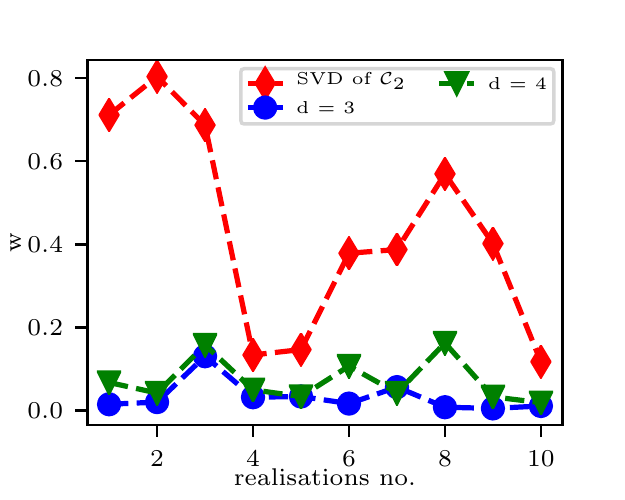}}
	\subfigure[$t$-Student copula $\nu = 5$ 
	\label{fig::hosvdt}]{\includegraphics{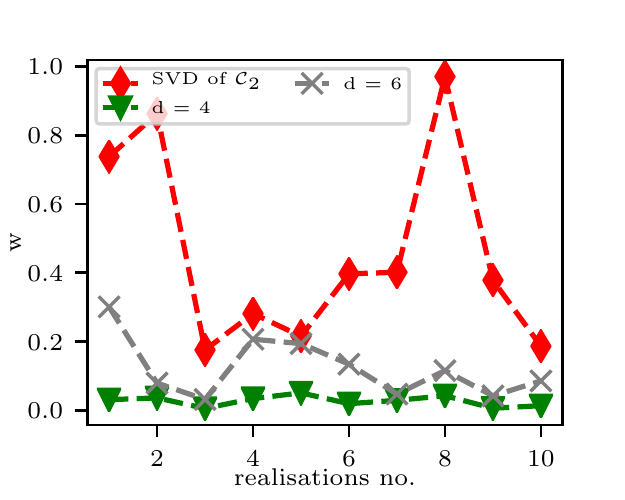}}
	\caption{An impact of non-Gaussian features on linear combination that is 
		supposed to have small $d$\textsuperscript{th} 
		cumulant. Here, following Eq.~\eqref{eq::w}, $w$ is the measure of the 
		participation of non-Gaussian 
		features in the linear combination of features, given by the last 
		column of the factor matrix. The lower $w$ the lower participation of 
		non-Gaussian features. For Gumbel copula best results gives the HOSVD 
		of $3$\textsuperscript{rd} cumulant, while for $t$-Student copula the 
		HOSVD of the $4$\textsuperscript{th} cumulant.
		}\label{fig::hosvd}
    \end{center}
\end{figure}

Results are presented in 
Figure~\ref{fig::hosvd}, ideally we expect $w = 0$. We have $w \approx 0$ for 
many realisations of the experiment given the HOSVD of $\CC_3$ 
for data changed by the Gumbel copula and of the HOSVD of $\CC_4$ for 
data changed by the $t$-Student copula. In the Gumbel copula case, 
observe that absolute values of elements 
$3$\textsuperscript{rd} cumulant's tensor's are basically larger than those  
corresponding to the $4$\textsuperscript{th} cumulant's tensor's, compare 
Figures~\ref{fig::gu_c3} and~\ref{fig::gu_c4}. Nevertheless the performance of 
the HOSVD of $\CC_4$ is almost as good as those of the HOSVD of 
$\CC_3$, especially in comparison with the performance of the SVD of the 
covariance matrix.

In the $t$-Student copula case, see Figure~\ref{fig::hosvdt}, we did not 
present 
outcomes of the HOSVD of odd cumulants, hence these cumulants are zero and the 
HOSVD of these cumulants would produce random results. Importantly 
in the $t$-Student copula case, the HOSVD of $\CC_4$ 
produces good results, better than the HOSVD of $\CC_6$. This is probably due 
to 
higher estimation error of $\CC_6$ in comparison with $\CC_4$, see Appendix A 
in~\cite{domino2018efficient}. Nevertheless, the HOSVD of $\CC_6$ is still 
better than 
the SVD of $\CC_2$.

As a conclusion one can observe, that last columns of the factor matrix of the 
HOSVD of higher order 
cumulants tensors are good candidates for such linear combination of features 
in 
which non-Gaussian features have a little representation. Nevertheless the 
order of the cumulant for which the HOSVD is performed has to be determined on 
the base of the probabilistic mode of data.

\chapter{Discussion}

In this book, we discuss stochastic models of real-life data that break the 
Central Limit Theorem assumptions both in a univariate and a multivariate case. 
Such models give motivation for the developments of algorithms applicable to 
the analysis of non-Gaussian distributed data. To discuss and develop some of 
such algorithms we analyse probabilistic models of data, including copulas and 
univariate marginal distributions. Further, to analyse these probabilistic 
models, we use higher order multivariate statistics and, notably, higher order 
multivariate cumulants. This leads to particular algorithms presented in 
Chapter~\ref{cha::cumml} of the book.  

In Chapter~\ref{cha::cumml} we discuss different methods of dimensionality 
reduction applicable for non-Gaussian distributed data. These methods provide 
the core of various features selection or feature extraction schemes, all based 
on higher order cumulants tensors. One should keep in mind that dimensionality 
reduction is vital in many machine learning scenarios where we either want to 
reduce the dimensionality of data to such that carries only interesting us 
information or to filter out data that are not informative or carry a noise.  
There we can either improve a score of further analysis or reduce the size of 
the data set to be further processed.

Discussed dimensionality reduction methods were tested successfully on 
specifically prepared artificial data with higher order dependencies within a 
subset of marginals. Marginals from this subset where modelled by some 
non-Gaussian copula, while reminding were modelled by the Gaussian copula. The 
discussion on the optimal dimensionality reduction schemes given various 
probabilistic models is provided as well in Chapter~\ref{cha::cumml}. For 
example, if we have higher order dependencies modelled by the $t$-Student 
copula, we need even higher order cumulants to analyse it, since odd ones are 
supposed to be zero. The wide discussion on copulas and copula-based algorithms 
for data preparation is presented in Chapter~\ref{chap::cops}, while a wide 
discussion on cumulants is presented in Chapter~\ref{ch::hostats}.

Apart from this, we present some applications of cumulants based dimensionality 
reduction algorithms in real-life data analysis. In particular we concentrate 
on financial data analysis, hyper-spectral imaging data analysis, or biomedical 
data analysis. We hope for the further applications on real data, given 
the fact cumulants based dimensionality reduction algorithms are tested both on 
artificial and real-life data. Finally, these algorithms are implemented in a 
modern and efficient Julia programming language and available in the 
\texttt{GitHub} repository.

Potential applications involve the further analysis of various non-Gaussian 
distributed data, both from sources mentioned in preceding paragraph, and from 
other sources not mentioned there. The vast source of non-Gaussian data prvided 
from the hyper-spectral imaging. 
Here we can refer to various hyper-spectral analysis scenarios. One of these is 
the search for features (or features combinations) that carry meaningful 
information about small targets in detection scenarios.  Another is the 
introduction of non-Gaussian models of hyper-spectra data.  

Large sources of non-Gaussian data are financial systems as well. Referring to 
financial data one can use stochastic models as discussed in 
Chapter~\ref{ch::univ} to predict crises or other periods of high and unusual 
variability of financial data. In such periods extreme changes in financial 
data values, simultaneous on many marginals, are possible. Such dynamics can be 
analysed by higher order cumulants. In practice, different algebraic 
multi-linear methods may be applied to find such linear combinations of 
assets---the portfolio---that have small variability and are relatively safe. 
For example, one can improve the ALS approach discuss in 
Chapter~\ref{cha::cumml}, to be applicable to the practical problem of finding
portfolios with low variability and positive portions of all assets, and to 
exclude the problematic short sale (negative portions) case. Such algorithm can 
be tested on a vast amount of financial data from around the world.

Analogically one can investigate non-Gaussian distributions of various 
multi-features biomedical data---such as multi-sources EMG or EEG 
data---employing higher order cumulants tensors. The proposed approach of 
cumulant-based multilag autocorrelation analysis of such signals is presented 
at the end of Chapter~\ref{ch::hostats}. By using advanced statistical models 
to analyse existing biomedical data, new knowledge can be retained without 
further data acquisition involving experiments on animals or human based 
research. 

There are much more non-Gaussian real-life data that can be analysed by higher 
order cumulants, such as traffic 
data. Analysis of these data can be beneficial in computer network modelling or 
road, rail or air traffic modelling. Here non-Gaussian patterns can be analysed 
by copulas or higher order cumulants, while data dynamics ba various stochastic 
models. There is also a potential for quantum computing in handling traffic 
data.
Observe finally, that non-Gaussian are also weather data and solar activity 
data. These can be 
analysed in the climate and weather research, especially as extreme events are 
now more probable due to the climate changes. Such analysis can be beneficial 
both in analysing a human impact on the climate and in civil engineering. 

We believe that this book will significantly contribute to the 
development of algorithms applicable for the non-Gaussian data analysis. As 
mentioned before, such analysis is crucial for various scientific disciplines 
and humanity development.

\bibliographystyle{plain}
\bibliography{manuscript}

\end{document}